\documentclass[12pt,english]{article}
\usepackage{lmodern}

\usepackage[T1]{fontenc}
\usepackage{textcomp}
\usepackage[latin9]{inputenc}
\usepackage{geometry}
\usepackage{color}
\usepackage{babel}
\usepackage{array}
\usepackage{float}
\usepackage{booktabs}
\usepackage{multirow}
\usepackage{varwidth}
\usepackage{amsmath}
\usepackage{amsthm}
\usepackage{amssymb}
\usepackage{stmaryrd}
\usepackage{graphicx}
\usepackage{rotfloat}
\usepackage{setspace}
\usepackage[authoryear,round]{natbib}
\PassOptionsToPackage{normalem}{ulem}
\usepackage{ulem}
\usepackage[pdfusetitle,
 bookmarks=true,bookmarksnumbered=false,bookmarksopen=false,
 breaklinks=false,pdfborder={0 0 1},backref=false,colorlinks=true]
 {hyperref}
\hypersetup{
 linkcolor=red, citecolor=blue, unicode, psdextra, pdfencoding=auto}

\makeatletter

\providecommand{\tabularnewline}{\\}
\newenvironment{cellvarwidth}[1][t]
    {\begin{varwidth}[#1]{\linewidth}}
    {\@finalstrut\@arstrutbox\end{varwidth}}

\numberwithin{equation}{section}
\theoremstyle{plain}
\newtheorem{assumption}{\protect\assumptionname}
\theoremstyle{plain}
\newtheorem{prop}{\protect\propositionname}[section]
\theoremstyle{definition}
\newtheorem{example}{\protect\examplename}[section]
\theoremstyle{plain}
\newtheorem{thm}{\protect\theoremname}[section]
\theoremstyle{plain}
\newtheorem{lem}{\protect\lemmaname}[section]

\usepackage{mathtools}
\usepackage{amsmath}
\usepackage{amsthm}
\usepackage{amssymb}
\usepackage{babel}
\usepackage{xcolor,soul}
\usepackage{appendix}
\usepackage{booktabs}
\usepackage{tabularx}
\usepackage{threeparttable}
\usepackage{dcolumn}
\usepackage{fancyhdr}
\usepackage[T1]{fontenc}
\usepackage{color}
\usepackage{titlesec}
\usepackage{lscape}
\usepackage{pdflscape}

\makeatletter

\usepackage{todonotes}
\allowdisplaybreaks

\usepackage{subfigure}
\usepackage{caption}

\makeatletter
\let\oldappendix\appendix
\renewcommand{\appendix}{%
  \oldappendix
  \renewcommand{\thesection}{S.\Alph{section}}
  \renewcommand{\theequation}{\thesection.\arabic{equation}}
  \renewcommand{\thefigure}{\thesection.\arabic{figure}}
  \renewcommand{\thetable}{\thesection.\arabic{table}}
  \setcounter{section}{0}
}
\makeatother

\usepackage{array}
\usepackage{booktabs}

\fancypagestyle{appendix}{%
  \fancyhf{}
  \fancyfoot[C]{Supplemental Appendix \thepage}%
  \renewcommand*{\headrulewidth}{0pt}
}

\makeatother

\providecommand{\assumptionname}{Assumption}
\providecommand{\examplename}{Example}
\providecommand{\lemmaname}{Lemma}
\providecommand{\propositionname}{Proposition}
\providecommand{\theoremname}{Theorem}

\begin{document}
\title{Social Interactions in Endogenous Groups \thanks{\scriptsize We are grateful to Eric Auerbach, Vincent Boucher, Denis Chetverikov, Gorden Dahl, Aureo de Paula, Steven Durlauf, Bryan Graham, Jinyong Hahn, James Heckman, Dongwoo Kim, Ivana Komunjer, Brian Krauth, Michael Leung, Arthur Lewbel, Zhipeng Liao, Lance Lochner, Adriana Lleras-Muney, Rosa Matzkin, Angelo Mele, Konrad Menzel, Krishna  Pendakur, Geert Ridder, Andres Santos, and Xiaoxia Shi for their helpful comments. We are also grateful to seminar and conference participants at Caltech, CUHK, Georgetown, Iowa, Jinan (IESR), MIT/Harvard, Rice, Texas A\&M University, UCL, UCLA, UC Riverside, UPenn, USC, Wisconsin-Madison, Xiamen, Yale, 2020 Econometric Society World Congress, 2021 Annual Conference of the Canadian Economics Association, 2021 Asia Meeting of the Econometric Society, 2021 Australasia Meeting of the Econometric Society, 2021 Canadian Econometric Study Group Meetings, 2021 China Meeting of the Econometric Society, 2021 IAAE Annual Conference, 2022 California Econometrics Conference, 2023 Interactions Workshop (Northwestern University), 2023 North American Summer Meeting, and 2025 ASSA/ES North American Winter Meeting. Sun is grateful for financial support from the Social Sciences and Humanities Research Council of Canada through its Insight Development Grants Program. All errors are our own.}}
\author{Shuyang Sheng\thanks{\scriptsize Shenzhen Finance Institute, School of Management and Economics, The Chinese University of Hong Kong, Shenzhen. Email: shengshuyang@cuhk.edu.cn}\and
Xiaoting Sun\thanks{\scriptsize Department of Economics, Simon Fraser University. Email: xiaoting\_sun@sfu.ca}}
\maketitle
\begin{abstract}
This paper investigates social interactions in endogenous groups.
We specify a two-sided many-to-one matching model, where individuals
select groups based on preferences, while groups admit individuals
based on qualifications until reaching capacities. Endogenous formation
of groups leads to selection bias in peer effect estimation, which
is complicated by equilibrium effects and alternative groups. We propose
novel methods to simplify selection bias and develop a sieve OLS estimator
for peer effects that is $\sqrt{n}$-consistent and asymptotically
normal. Using Chilean data, we find that ignoring selection into high
schools leads to overestimated peer influence and distorts the estimation
of school effectiveness.


\newpage{}
\end{abstract}

\section{Introduction}

Social interactions models are useful tools for exploring the interdependence
of individual outcomes in various contexts, such as education, earnings,
health, and crime. One notable feature of social interactions is their
high tendency to occur among individuals within the same social or
economic group. For instance, students interact with fellow students
within their school, and residents engage with other residents within
their neighborhood.\footnote{There is a massive literature on peer effects in schools, classrooms,
or dorms (e.g., \citealp{Evans1992}; \citealp{sacerdote2001peer};
\citealp{Duflo2011tracking}). See \citet{Epple2011} for a survey.
Examples of neighborhood effects include \citet{katz2001moving} and
\citet{bayer2008place}.} Because individuals select the schools they apply to or the neighborhoods
they reside in, the peers they ultimately interact with are often
determined endogenously. Therefore, when we observe that individuals
with more advantageous peers achieve better outcomes, it is unclear
whether these outcomes are driven by the influence of their peers
or by their choices of peer groups \citep{Epple2011,Sacerdote2011}.
This paper develops new econometric methods for identifying and estimating
causal peer effects in the presence of endogenous groups.

Selection into groups has been a central topic in the literature on
social interactions and group treatment effects (e.g., school value-added).
Early studies on social interactions account for group selection in
the framework of correlated effects (\citealp{Manski1993}; \citealp{moffitt2001policy}).
This approach addresses selection by including group fixed effects,
the validity of which relies on the restrictive assumption that an
individual's group membership becomes exogenous once group-level unobservables
(e.g., school resources, neighborhood amenities) are controlled for
\citep{lee2007identification,bramoulle2009identification,Sacerdote2011}.
In contrast, empirical research has leveraged natural experiments
with random group assignment to circumvent the selection issue.\footnote{Examples include classroom assignment \citep{Duflo2011tracking} and
dormitory assignment \citep{sacerdote2001peer,zimmerman2003peer}.} In school choice, the literature has exploited exogenous variation
under centralized school assignment (e.g., random lotteries or discontinuities
around non-random test score cutoffs) to overcome selection into schools
\citep{kirkeboen2016field,abdulkadirouglu2017research,abdulkadirouglu2022breaking,angrist2024credible}.
These empirical approaches require observing information on individuals'
rank-order lists, priorities, and admission cutoffs. In this paper,
we propose an approach to account for selection into groups that does
not rely on random assignment or detailed information on applications
and admissions. Our approach is applicable to observational data where
only group memberships are observed, making it particularly relevant
for contexts such as decentralized school choice and college admissions.

Specifically, we develop a model of group formation to examine how
individuals sort into groups and how to effectively account for the
impact of this sorting. We characterize group formation using a two-sided
many-to-one matching model with nonparametric unobservables, where
individuals select among groups based on their preferences, while
groups rank and accept individuals according to qualifications until
reaching capacity constraints \citep{azevedo_supply_2016,HSS_twosided}.
This framework closely mirrors real-world admission processes, such
as those for schools \citep{gazmuri2017school,HSS_twosided}, colleges
\citep{dale_estimating_2002}, residency programs \citep{roth1984evolution,agarwal2015empirical},
and nursing homes \citep{gandhi2022picking}. Our framework also covers
one-sided group formation as a special case, in which individuals
unilaterally determine the groups they join, such as in neighborhood
choice \citep{brock2001interactions,brock2002multinomial,brock2005multinomial,ioannides2008interactions}.
In contrast to one-sided group formation, where an individual's choice
set is assumed to include all groups, two-sided group formation allows
the choice set to be latent and determined endogenously \citep{Barseghyan2021,Agarwal2022}.
To the best of our knowledge, we are the first to apply two-sided
matching to analyze selection into groups.

The endogenous formation of groups leads to selection bias in the
estimation of peer effects in linear-in-means social interactions
models \citep{Manski1993}. For example, students with higher capabilities
or more advantageous family backgrounds may be sorted into more selective
colleges, leading to the overestimation of peer effects \citep{dale_estimating_2002}.
The formulation of our selection bias supports the approach proposed
by \citet{dale_estimating_2002}, who overcame selection bias by matching
students who applied to and were accepted by comparable sets of colleges.
In contrast, we correct for selection bias in the absence of such
application and admission information.

The selection bias in our model is further complicated by equilibrium
effects. Because equilibrium groups depend on the (observed and unobserved)
characteristics of all $n$ individuals in a market, selection bias
is a high-dimensional function that involves the observed characteristics
of the $n$ individuals. We overcome this dimensionality problem using
the limiting approximation of a market as $n$ approaches infinity
\citep{azevedo_supply_2016}.

Moreover, we impose an exchangeability assumption on the distribution
of unobservables, which yields two crucial properties that aid in
addressing selection bias. First, it allows selection bias to be represented
as a group-invariant selection function of preference and qualification
indices, enabling the identification of the effects of group-level
variables (e.g., group averages, group dummies). Second, under exchangeability,
selection bias is symmetric in the impacts of alternative groups.
For example, selection bias for attending college A is symmetrically
affected by the preference and qualification indices associated with
colleges B and C. This symmetry enables us to aggregate the indices
of alternative groups using their elementary symmetric functions \citep{altonji_matzkin},
thereby reducing the number of nuisance parameters in the nonparametric
estimation of selection bias.

Our results on selection bias make three important contributions to
the literature. First, our model yields selection bias that depends
on individual characteristics involved in group formation. This finding
highlights that simply including group fixed effects is insufficient
to correct for selection bias. The result also complements the work
of \citet{Altonji2018}, who showed that selection bias can be captured
by group averages of covariates under specific linear restrictions.
In contrast, our framework imposes no such linear restrictions. We
show that effectively correcting for selection bias requires exploiting
individual-specific information in group formation. Second, our results
emphasize the importance of accounting for the impacts of alternative
groups.\footnote{For instance, in college admissions, whether a student attends a particular
college depends not only on the attributes of that college, but also
on the attributes of the colleges they choose to forgo. This feature
suggests that the attributes of alternative groups play a crucial
role in selection correction.} Unlike existing studies that directly restrict the impacts of alternative
groups \citep[e.g., ][]{Dahl2002}, our approach aggregates these
impacts by leveraging the symmetry of selection bias under exchangeability.
This makes the selection correction tractable, even when the number
of groups is moderately large. Related to our work, \citet{abdulkadirouglu2020parents}
proposed a similar selection correction method in centralized school
choice, where selection occurs only through unobserved preferences.
Their selection correction aggregates the impacts of alternative schools
constructed using rank-order lists. We provide a microfoundation for
the aggregation, but in the context of decentralized markets with
selection on both unobserved preferences and qualifications, where
rank-order lists are unavailable. Third, we separately identify group
treatment effects (e.g., school fixed effects that measure school
effectiveness) and self-selection, thereby allowing for the disentanglement
of correlated effects due to group-level unobservables from those
due to self-selection into groups \citep{Manski1993,moffitt2001policy}.\footnote{\citet{Manski1993} was the first to discuss correlated effects, defining
them as situations in which ``individuals in the same group tend
to behave similarly because they have similar individual characteristics
or face similar institutional environments.'' \citet{moffitt2001policy}
further categorized correlated effects into two types: those arising
from shared environmental influences (e.g., unobserved school characteristics),
and those arising from sorting, where individuals with similar traits
tend to group together.}

To identify the social interaction parameters, we partial out selection
bias as in a partially linear model \citep{Robinson1988}. A key challenge
in identifying linear-in-means social interactions is that peer outcomes
and peer characteristics may be linearly dependent, leading to the
well-known reflection problem \citep{Manski1993}. Some studies have
overcome this problem by leveraging group size variation or intransitive
triads (\citealp{lee2007identification}; \citealp{graham2008identifying};
\citealp{Davezies2009}; \citealp{bramoulle2009identification}).
By comparison, \citet{brock2001interactions} demonstrated that self-selection
can aid in the identification of social interactions. The intuition
is straightforward: self-selection enters the outcome equation as
an individual-specific term that is excluded from group-level averages,
thereby providing an exclusion restriction to disentangle endogenous
social effects from exogenous social effects. \citet{brock2001interactions}
considered one-sided group formation with parametric unobservables.
We extend their insight to a more general setting with two-sided group
formation under nonparametric unobservables. Similarly, \citet{gu2024peer}
leveraged self-selection to achieve the identification of peer effects
in a binary selection setting, where individuals decide whether to
participate in a program, under parametric unobservables.

We propose a sieve OLS estimator for the social interaction parameters,
where we control for selection bias using a sieve approximation. While
using OLS instead of IV might appear counter-intuitive---given that
peer outcomes could be endogenous due to simultaneity---we demonstrate
that this endogeneity vanishes asymptotically in large groups where
within-group connections are dense (e.g., group averages). This finding
is consistent with the results in \citet{Lee2002}, who demonstrated
the consistency and efficiency of OLS estimators for peer effects
under a deterministic adjacency matrix. The asymptotic properties
of our estimator cannot be established using existing methods because
our adjacency matrix is stochastic due to randomness in group memberships
and within-group connections. By recognizing that the leading terms
in our estimator take the form of weighted $U$-statistics with random
weights, we develop an innovative approach that generalizes the asymptotic
methods for standard weighted $\mathit{U}$-statistics \citep{Lee1990}
to accommodate randomness in weights. Furthermore, we impose additional
conditions to ensure that the network dependence introduced by random
weights diminishes sufficiently fast. Under these conditions, we prove
that our estimator is $\sqrt{n}$-consistent and asymptotically normal.
Finally, we provide simulation evidence showing that our estimator
performs well.

We apply our approach to analyze social interactions among high school
students in Chile. We use data from the SIMCE along with administrative
records from the Ministry of Education. Our sample consists of 6,872
tenth-grade students enrolled in 53 high schools in the Biobío Region
in 2006. The Chilean education system is highly segregated, with substantial
sorting across high schools. While our dataset provides comprehensive
student information both before high school admissions (eighth grade)
and after (from tenth grade through college graduation), it lacks
detailed data on the high school admission process (e.g., rank-order
lists, priorities, and admission thresholds), which are essential
for conventional identification strategies such as regression discontinuity
or lottery-based designs. Our approach is particularly useful in this
decentralized market, as it accounts for selection bias using only
school enrollment information, without requiring data on the admission
process. We define a student's peers as their classmates and measure
student outcomes using academic performance both in the short term
(tenth-grade math and language scores and high school graduation)
and in the long term (post-secondary and college enrollment and graduation).
Our objective is to estimate the causal effects of peer outcomes and
peer characteristics on these academic outcomes, in the presence of
endogenous selection into high schools.

We find that both peer outcomes and peer characteristics have significant
effects on tenth graders. Including school fixed effects only partially
corrects for selection bias compared to simple OLS. Additionally controlling
for our selection correction further adjusts the estimates of peer
effects. For example, the selection-corrected coefficients of peer
outcomes are 10-45\% lower than those obtained using school fixed
effects alone. A variance decomposition of predicted outcomes into
peer influence, school effectiveness, and self-selection reveals that
self-selection accounts for the largest share of explained variance
in most outcomes, highlighting the importance of properly correcting
for selection bias in school evaluation. Notably, peer influence is
overestimated when selection is not adequately accounted for, particularly
in highly selective schools. Moreover, failing to account for selection
distorts the estimation of school effectiveness, primarily by underestimating
the value-added of less selective schools. The selection-corrected
estimates suggest that more selective schools may actually provide
lower value-added. These findings underscore the critical need for
proper selection correction to evaluate school performance and ensure
equitable resource allocation.

Our work is related to the econometric literature that addresses the
endogeneity of peer relationships in social interactions models (\citealp{goldsmith2013social,qu2015estimating,hsieh2016social,hsieh2018smoking,hsieh2020specification,johnsson2021estimation,Auerbach2022,griffith2024random}).
These studies typically assume the presence of unobserved individual
heterogeneity that affects both the formation of links and individual
outcomes, thereby generating endogeneity in peer relationships. This
type of endogeneity can be addressed by accounting for individual
heterogeneity. In contrast to these studies, which focus on link formation,
we develop a framework for group formation, where the unobserved factors
driving endogenous groups can be pair-specific (e.g., student-school
specific). We complement the methods in this literature by proposing
an approach to correct for selection into groups.

Our work is also related to the literature on sample selection models
\citep[e.g., ][]{heckman1979sample,das2003nonparametric} and extends
this literature to social interactions with endogenous group selection.
Our framework is mostly related to that of \citet{das2003nonparametric},
who examined selection into a single sample governed by multiple rules.
In contrast, we focus on selection into multiple groups. We propose
a novel approach to reduce the dimensionality of selection bias when
the number of groups is moderately large, thereby providing a practical
method for addressing selection bias despite dimensionality challenges.

The remainder of the paper is organized as follows. Section \ref{sec:Model}
introduces the model. Section \ref{sec:Selection_bias} derives selection
bias. Section \ref{sec:Identification} investigates identification.
Section \ref{sec:Estimation} proposes estimation methods and establishes
asymptotic properties. Section \ref{sec:Simulations} evaluates our
approach through simulations. Section \ref{sec:Empirical} studies
social interactions in Chilean high schools. Section \ref{sec:Conclusion}
concludes the paper. All proofs are provided in Supplemental Appendix
\ref{online:proofs}.\footnote{All the numbered items designated with an \textquotedblleft S\textquotedblright{}
are shown in the Supplemental Appendix.}

\section{\protect\label{sec:Model}Model}

\subsection{Social Interactions}

Consider a set of individuals $\mathcal{N}=\{1,2,\dots,n\}$ who can
join a set of groups $\mathcal{G}=\{1,\dots,G\}$. We assume that
the number of groups $G$ is finite and each group $g\in\mathcal{G}$
has a predetermined capacity $n_{g}$ that is proportional to $n$.
The groups are non-overlapping (e.g., colleges, neighborhoods), so
one joins only one group. Let $g_{i}$ denote the group that individual
$i$ joins and $\boldsymbol{g}\equiv(g_{1},\dots,g_{n})'$ the $n\times1$
vector that stacks $g_{i}$.

After the groups are formed, individuals interact with their groupmates
following a linear-in-means social interactions model \citep{Manski1993,brock2001interactions}
\begin{align}
y_{i} & =\sum_{j=1}^{n}w_{ij}y_{j}\gamma_{1}+\sum_{j=1}^{n}w_{ij}x'_{j}\gamma_{2}+x'_{i}\gamma_{3}+\epsilon_{i}.\label{eq:linmean}
\end{align}
In this specification, $y_{i}\in\mathbb{R}$ represents the outcome
of interest (e.g., GPA, earnings), $x_{i}\in\mathbb{R}^{d_{x}}$ is
a vector of observed individual characteristics (e.g., parental education,
family income), $\epsilon_{i}\in\mathbb{R}$ is an unobserved characteristic
(e.g., ability). Let $\boldsymbol{y}$ denote the $n\times1$ vector
that stacks $y_{i}$, $\boldsymbol{x}$ the $n\times d_{x}$ matrix
that stacks $x'_{i}$, and $\boldsymbol{\epsilon}$ the $n\times1$
vector that stacks $\epsilon_{i}$. We assume that $i$'s outcome
$y_{i}$ depends on $\sum_{j=1}^{n}w_{ij}y_{j}$ and $\sum_{j=1}^{n}w_{ij}x_{j}$,
the weighted averages of outcomes and observed characteristics of
$i$'s peers, where $w_{ij}\in\mathbb{R}$ denotes the weight of peer
$j$ on individual $i$. Following the terminology in \citet{Manski1993},
$\gamma_{1}$ represents the endogenous social effect, and $\gamma_{2}$
represents the exogenous/contextual social effect. The parameter of
interest is $\gamma=(\gamma_{1},\gamma'_{2},\gamma'_{3})'\in\mathbb{R}^{2d_{x}+1}$.

Because individuals are influenced solely by their groupmates, the
adjacency matrix $\boldsymbol{w}\equiv(w_{ij})\in\mathbb{R}^{n^{2}}$
exhibits a group structure, that is, $w_{ij}=0$ if $i$ and $j$
belong to different groups ($g_{i}\neq g_{j}$). A typical specification
is given by $w_{ij}=\frac{1}{n_{g_{i}}}$ if $g_{i}=g_{j}$, where
each group member has equal weight \citep{Manski1993}. More generally,
we can specify $\boldsymbol{w}$ to capture additional network structures
in a group. For example, groupmates may form friendship ties, and
only friends may have a nonzero influence.

If an individual $i$'s decision to join a group or form friendships
is influenced by unobserved characteristics that are correlated with
$\epsilon_{i}$, the adjacency matrix $\boldsymbol{w}$ would be correlated
with $\boldsymbol{\epsilon}$ and become endogenous. In this paper,
we treat the sorting of individuals into groups as endogenous but
assume that link formation within groups (if any) is exogenous, as
formulated in Assumption \ref{ass:adj_exog} (in Section \ref{sec:Selection_bias}).
This assumption is not a limitation but a simplification, which allows
us to focus on the endogeneity in $\boldsymbol{g}$ that arises from
selection into groups.

In the literature that relies on observational data---where random
assignment into groups is not available---the issue of sorting into
groups is predominantly addressed in the framework of correlated effects
\citep{Manski1993,moffitt2001policy}. It is often assumed that there
exist unobserved characteristics that have common effects on the outcome
of all individuals within a group (e.g., school resources, neighborhood
amenities), and conditional on these group-level unobservables, an
individual's group membership becomes exogenous \citep{lee2007identification,bramoulle2009identification,Sacerdote2011}.
While using group fixed effects to address sorting into groups may
seem reasonable, it does not allow for unobserved characteristics
that simultaneously affect both group selection and individual outcomes
when these characteristics vary across individuals within the same
group (\citealp{bramoulle2009identification}; \citealp{patacchini2017heterogeneous}).\footnote{Another limitation of using group fixed effects is that it requires
sufficient variation in the regressors $\sum_{j=1}^{n}w_{ij}y_{j}$
and $\sum_{j=1}^{n}w_{ij}x_{j}$ within a group. If these regressors
represent group averages that include oneself, we cannot control for
group fixed effects due to perfect multicollinearity.}

In the following section, we develop a structural model of group formation
to examine how individuals sort into groups and how to effectively
account for endogenous sorting. 

\subsection{\protect\label{sec:Model.group}Group Formation}

\citet{brock2001interactions,brock2002multinomial,brock2005multinomial},
\citet{ioannides2008interactions}, and \citet{gu2024peer} have proposed
various models of group formation to analyze the sorting of individuals
into groups. These studies consider the framework where individuals
can unilaterally decide whether to join a group or which group to
join, resulting in a standard binary or multinomial discrete choice
model. While this approach is well-suited for applications such as
job training participation and neighborhood choice, there are scenarios
where groups also actively decide whether to admit individuals, as
seen in school choice, college admissions, medical residency placements,
and nursing home admissions. In these contexts, the formation of groups
is determined by the bilateral decisions of both individuals and groups. 

We consider a model of group formation that accounts for two-sided
decisions of individuals and groups. On one side, individuals choose
among groups based on their preferences. On the other side, groups
rank and admit individuals based on qualification criteria until reaching
their capacity constraints.\footnote{For example, in college admissions, a student's qualifications reflect
the preferences of colleges. In school choice, schools are typically
assigned based on administrative rules that prioritize factors such
as geographic proximity and sibling enrollment.} If groups have unlimited capacities and capacity constraints become
non-binding, the framework reduces to one-sided decisions made by
individuals, as in \citet{brock2005multinomial} and \citet{ioannides2008interactions}.

Our framework can be equivalently characterized as a two-sided many-to-one
matching model without transfers, where individuals in a group are
considered ``matched with'' the group (\citealp{azevedo_supply_2016};
\citealp{HSS_twosided}). Therefore, we follow this matching literature
to specify the model primitives and analyze equilibrium outcomes.
To the best of our knowledge, no previous studies have applied two-sided
matching to analyze selection into groups. 

\paragraph{Utility}

For individual $i\in\mathcal{N}$ and group $g\in\mathcal{G}$, let
$u_{ig}$ denote $i$'s utility of joining group $g$ and $v_{ig}$
denote $i$'s qualification for group $g$:
\begin{equation}
u_{ig}=z'_{i}\delta_{g}^{u}+\xi_{ig}\text{ and }v_{ig}=z'_{i}\delta_{g}^{v}+\eta_{ig},\label{eq:uig_vgi}
\end{equation}
where $z_{i}\in\mathbb{R}^{d_{z}}$ represents a vector of observed
characteristics that are either individual-specific (e.g., test score,
parental education, family income) or pair-specific (e.g., distance
to school, a student's minority status interacted with the past minority
composition of a school).\footnote{To illustrate that $z_{i}$ can flexibly include both individual-specific
and pair-specific variables, consider a two-group example where $v_{ig}=s_{i}\delta_{g,s}+d_{ig}\delta_{g,d}+\eta_{ig}$,
with $s_{i}$ being an individual-specific variable and $d_{ig}$
a pair-specific variable. This specification maps to equation (\ref{eq:uig_vgi})
by setting $z_{i}=(s_{i},d_{i1},d_{i2})'$, $\delta_{1}^{v}=(\delta_{1,s},\delta_{1,d},0)'$,
and $\delta_{2}^{v}=(\delta_{2,s},0,\delta_{2,d})'$.} The components of $z_{i}$ that appear in $u_{ig}$ and $v_{ig}$
may differ; in such cases, the corresponding coefficients for absent
components are set to zero. The group-specific coefficients $\delta_{g}^{u},\delta_{g}^{v}\in\mathbb{R}^{d_{z}}$
capture heterogeneous effects of $z_{i}$ across groups. The terms
$\xi_{ig}\in\mathbb{R}$ and $\eta_{ig}\in\mathbb{R}$ represent pair-specific
unobserved shocks to utility and qualification, respectively (e.g.,
family tradition, extracurricular activities).\footnote{Individuals may take anticipated outcomes into account when comparing
between groups (e.g., students considering post-graduation job prospects).
However, the utility specification in equation (\ref{eq:uig_vgi})
can be interpreted as a reduced form that implicitly accounts for
these expectations. Developing a model that explicitly incorporates
anticipated outcomes in the utility is beyond the scope of our analysis.} Let $\xi_{i}=(\xi_{i1},\dots,\xi_{iG})'$ and $\eta_{i}=(\eta_{i1},\dots,\eta_{iG})'$.
We assume that the joint distribution of $(\epsilon_{i},\xi_{i},\eta_{i})$
is nonparametric, which has the advantage of allowing $\epsilon_{i}$
to have flexible dependence with $\xi_{i}$ and $\eta_{i}$.\footnote{\label{fn:np_groupFE}The nonparametric specification implies that
$z_{i}$ cannot include a constant or group-specific variables, as
their effects cannot be separated from the nonparametric distribution
of $(\xi_{i},\eta_{i})$. In Section \ref{sec:Selection_bias}, we
impose additional restrictions on the distribution of $(\xi_{i},\eta_{i})$,
which then allows us to include group-level effects.}

To maintain a tractable framework, our model does not explicitly allow
for peer effects in group formation, meaning that the utility and
qualification functions in equation (\ref{eq:uig_vgi}) do not depend
on prospective groupmates. However, this theoretical simplification
does not necessarily preclude the empirical consideration of peer
influences in group formation. In fact, peer effects can be proxied
using the outcomes and characteristics of past group members (such
as average test scores, gender or racial composition, and average
family income), provided that these measures remain stable over time.

\paragraph{Equilibrium}

Following the matching literature \citep{roth1992two}, we assume
that the matching outcome is stable. \citet{azevedo_supply_2016}
showed that a stable matching exists and can be characterized by group
\textit{cutoffs}. For $g\in\mathcal{G}$, define the cutoff $p_{g}$
of group $g$ as the lowest qualification among the group members
if the capacity constraint is binding; otherwise, the cutoff is set
to $-\infty$. Namely, $p_{g}=\inf_{i:g_{i}=g}v_{ig}\text{ if }\sum_{i\in\mathcal{N}}1\{g_{i}=g\}=n_{g}$,
and $p_{g}=-\infty$ otherwise.\footnote{Following \citet{azevedo_supply_2016}, we assume that all individuals
are acceptable to all groups. Therefore, if a group has vacancies
remaining, any individual wishing to join can do so, resulting in
a cutoff of negative infinity. This assumption can be relaxed by introducing
an acceptance threshold whereby groups only admit individuals above
a certain acceptability level \citep{HSS_twosided}. We maintain this
simpler specification as it does not affect the main conceptual framework.}

Given a vector of cutoffs $p=(p_{1},\dots,p_{G})'$, let $\mathcal{C}_{i}(p)=\{g\in\mathcal{G}:\:v_{ig}\geqslant p_{g}\}\subseteq\mathcal{G}$
denote individual $i's$ choice set---the subset of groups for which
$i$ qualifies. Within $\mathcal{C}_{i}(p)$, $i$ chooses the group
that yields the highest utility, $g_{i}=\arg\max_{g\in\mathcal{C}_{i}(p)}u_{ig}$.
This is a multinomial discrete choice problem with the choice set
$\mathcal{C}_{i}(p)$ determined endogenously by the cutoffs $p$.\footnote{For simplicity of exposition, we assume that individuals always prefer
to join a group. This simplification can be relaxed by introducing
a utility $u_{i0}=\xi_{i0}$ for the outside option (i.e., not joining
any group). Such a relaxation requires no substantial modifications
to our theoretical framework.} The structure highlights a key distinction between one-sided and
two-sided group formation: in one-sided group formation, an individual's
choice set is assumed to include all groups, whereas two-sided group
formation allows the choice set to be latent and determined endogenously
\citep{Barseghyan2021,Agarwal2022}.

Individual $i$ joins group $g$ if and only if (i) $i$ qualifies
for group $g$, and (ii) for any other group $h\ne g$, either $i$
does not prefer group $h$, or $i$ does not qualify for group $h$.
This can be expressed as
\begin{eqnarray}
 &  & 1\{g_{i}=g\}\nonumber \\
 & = & 1\{v_{ig}\geq p_{g}\}\cdot\prod_{h\neq g}1\{u_{ih}<u_{ig}\:\text{or}\:v_{ih}<p_{h}\}\nonumber \\
 & = & 1\{\eta_{ig}\geq p_{g}-z'_{i}\delta_{g}^{v}\}\cdot\prod_{h\neq g}1\{\xi_{ih}-\xi_{ig}<z'_{i}(\delta_{g}^{u}-\delta_{h}^{u})\:\text{or}\:\eta_{ih}<p_{h}-z'_{i}\delta_{h}^{v}\}.\label{eq:glink}
\end{eqnarray}
From this expression, the group that $i$ joins depends on $i$'s
observed and unobserved characteristics $(z_{i},\xi_{i},\eta_{i})$
as well as the cutoffs $p$. We write $g_{i}=g(z_{i},\xi_{i},\eta_{i};p)$.\footnote{In the one-sided setting, the capacities are infinite, the cutoffs
$p_{g}$ are set to $-\infty$, and the consideration sets $\mathcal{C}_{i}(p)$
for all $i\in\mathcal{N}$ are equal to $\mathcal{G}$. The optimal
decision in equation (\ref{eq:glink}) reduces to a multinomial discrete
choice problem $1\{g_{i}=g\}=\prod_{k\neq g}1\{u_{ik}<u_{ig}\}=\prod_{k\neq g}1\{\xi_{ik}-\xi_{ig}<z'_{i}(\delta_{g}^{u}-\delta_{k}^{u})\}$.
The group that $i$ joins depends on $(z_{i},\xi_{i})$ only, that
is, $g_{i}=g(z_{i},\xi_{i})$.}

In a stable matching, a vector of cutoffs $p$ clears the supply and
demand for each group.\footnote{An equilibrium cutoff vector $p$ satisfies the market-clearing equations:
$\sum_{i\in\mathcal{N}}1\{g(z_{i},\xi_{i},\eta_{i};p)=g\}\leq n_{g}$
for all $g\in\mathcal{G}$, and $\sum_{i\in\mathcal{N}}1\{g(z_{i},\xi_{i},\eta_{i};p)=g\}=n_{g}$
if $p_{g}>-\infty$ \citep{azevedo_supply_2016}.} Let $\boldsymbol{\boldsymbol{z}}$ denote the $n\times d_{z}$ matrix
that stacks $z'_{i}$, $\boldsymbol{\xi}$ the $n\times G$ matrix
that stacks $\xi_{i}$, and $\boldsymbol{\eta}$ the $n\times G$
matrix that stacks $\eta_{i}$. An equilibrium cutoff vector can be
represented as $p(\boldsymbol{z},\boldsymbol{\xi},\boldsymbol{\eta})$.\footnote{There may exist multiple equilibrium cutoffs in a finite-$n$ market.
We denote by $p(\boldsymbol{z},\boldsymbol{\xi},\boldsymbol{\eta})$
the equilibrium that is selected by nature.} Given $p(\boldsymbol{z},\boldsymbol{\xi},\boldsymbol{\eta})$, the
equilibrium group that $i$ joins is given by equation (\ref{eq:glink})
and can be represented as $g_{i}=g(z_{i},\xi_{i},\eta_{i};p(\boldsymbol{z},\boldsymbol{\xi},\boldsymbol{\eta}))$.
The equilibrium groups of all the individuals can be expressed as
$\boldsymbol{g}(\boldsymbol{z},\boldsymbol{\xi},\boldsymbol{\eta};p(\boldsymbol{z},\boldsymbol{\xi},\boldsymbol{\eta}))$.

\section{\protect\label{sec:Selection_bias}Selection Bias}

In this section, we use the model described in Section \ref{sec:Model.group}
to derive the bias arising from selection into groups. This selection
bias is highly complex due to the presence of equilibrium effects
and the role of alternative groups. To address this, we propose novel
methods to simplify the selection bias and improve tractability.

Throughout the paper, we maintain the following assumptions.
\begin{assumption}
\label{ass:adj_exog}The adjacency matrix $\boldsymbol{w}$ is independent
of $\boldsymbol{\epsilon}$ conditional on \textup{$(\boldsymbol{x},\boldsymbol{z},\boldsymbol{g})$}.
\end{assumption}

\begin{assumption}
\label{ass:regular}(i) $(x_{i},z_{i},\epsilon_{i},\xi_{i},\eta_{i})$
is i.i.d. for $i\in\mathcal{N}$. (ii) The joint cdf of $(\epsilon_{i},\xi_{i},\eta_{i})$
is continuously differentiable. (iii) $(x_{i},z_{i})$ is independent
of $(\epsilon_{i},\xi_{i},\eta_{i})$ for $i\in\mathcal{N}$. 
\end{assumption}
As previously stated, Assumption \ref{ass:adj_exog} implies that
endogenous selection occurs only during group formation. For an adjacency
matrix $\boldsymbol{w}$ that represents group averages with binding
capacities, this assumption is trivially satisfied because $\boldsymbol{w}$
is determined by $\boldsymbol{g}$. For a more general $\boldsymbol{w}$,
the assumption requires that once groups are formed, how group members
form additional connections must be exogenous. This assumption is
satisfied if, for example, students in a school make friends independently
of $\boldsymbol{\epsilon}$ or are randomly assigned to dorms or classes
where they interact.\footnote{We can potentially relax the assumption by incorporating endogenous
friendship formation within a group, following the setup in e.g.,
\citet{johnsson2021estimation}. However, this added complexity would
not yield additional insights, so we do not pursue it here.} Assumption \ref{ass:regular} imposes a set of regularity conditions
that are standard in social interactions. It assumes that the observables
and unobservables are i.i.d. across individuals, the unobservables
have a smooth joint cdf, and the observables are exogenous.

Under Assumption \ref{ass:adj_exog}, equation (\ref{eq:linmean})
exhibits a selection bias if
\begin{align}
\mathbb{E}[\epsilon_{i}|\boldsymbol{x},\boldsymbol{z},\boldsymbol{g},\boldsymbol{w}]=\mathbb{E}[\epsilon_{i}|\boldsymbol{x},\boldsymbol{z},\boldsymbol{g}] & \neq0.\label{eq:selection_bias}
\end{align}
Recall that the equilibrium groups $\boldsymbol{g}$ can be expressed
as $\boldsymbol{g}(\boldsymbol{z},\boldsymbol{\xi},\boldsymbol{\eta};p(\boldsymbol{z},\boldsymbol{\xi},\boldsymbol{\eta}))$,
with each individual component given by $g_{i}=g(z_{i},\xi_{i},\eta_{i};p(\boldsymbol{z},\boldsymbol{\xi},\boldsymbol{\eta}))$,
$i\in\mathcal{N}$. Suppose that the outcome unobservable $\epsilon_{i}$
is correlated with the unobservables in group formation $(\xi_{i},\eta_{i})$
(selection on the unobservables). Under Assumption \ref{ass:regular}(i)(iii),
this correlation may lead to a selection bias through two channels.
First, the equilibrium group that individual $i$ joins $g_{i}$ depends
on $i$'s unobservables in group formation $(\xi_{i},\eta_{i})$,
which are correlated with $\epsilon_{i}$. Therefore, $\epsilon_{i}$
can be correlated with $g_{i}$ directly through $(\xi_{i},\eta_{i})$,
resulting in a selection bias. Second, observe that the equilibrium
cutoffs $p(\boldsymbol{z},\boldsymbol{\xi},\boldsymbol{\eta})$ depend
on the unobservable profile in group formation $(\boldsymbol{\xi},\boldsymbol{\eta})$
which includes $(\xi_{i},\eta_{i})$. Moreover, these equilibrium
cutoffs affect not only the group that $i$ joins $g_{i}$, but also
the groups that others join $g_{j}$, $j\neq i$. Consequently, $\epsilon_{i}$
can be correlated with the entire group structure $\boldsymbol{g}$,
including both $g_{i}$ and $g_{j}$, $j\ne i$, indirectly through
$p(\boldsymbol{z},\boldsymbol{\xi},\boldsymbol{\eta})$, again resulting
in a selection bias.

We illustrate the selection bias arising through the first channel
using college admissions as an example. Suppose that unobserved ability
($\epsilon_{i}$) that affects labor market outcomes ($y_{i}$) is
positively correlated with unobserved ability ($\eta_{ig}$) that
influences college qualifications ($v_{ig}$). Additionally, let $x_{i}$
and $z_{i}$ represent family income, which has a positive impact
on both college qualifications and labor market outcomes. Through
the college admission process, students with higher qualifications
($v_{ig}=z_{i}\delta_{g}^{v}+\eta_{ig}$) are sorted into more selective
colleges. This sorting implies that students in more selective colleges
are either more capable (high $\eta_{ig}$) or come from wealthier
families (high $z_{i}$). As a result, the sorting creates a positive
correlation between $\epsilon_{i}$ and the average family income
of peers in $i$'s college ($\sum_{j}w_{ij}x_{j}$), leading to upward
bias in estimates of peer effects \citep{dale_estimating_2002}.

The general equilibrium effects through the second channel present
a challenge in correcting for selection bias. Because the equilibrium
cutoffs $p(\boldsymbol{z},\boldsymbol{\xi},\boldsymbol{\eta})$ depend
on the characteristics of the $n$ individuals in a market, the selection
bias in equation (\ref{eq:selection_bias}) is a high-dimensional
function that involves the observed characteristics of the $n$ individuals.
Below, we propose an approach to address this dimensionality issue.

\subsection{\protect\label{sec:limiting}Limiting Approximation}

Our approach is to employ the limiting approximation of a market as
its size $n$ approaches infinity. We show that, in a large market,
the correlation between $\epsilon_{i}$ and $\boldsymbol{g}$ through
equilibrium cutoffs vanishes, thereby reducing the dimensionality
of selection bias.

To this end, let $p_{n}=(p_{n,1},\dots,p_{n,G})'$ denote the equilibrium
cutoffs in a market with $n$ individuals. \citet{azevedo_supply_2016}
demonstrated that, as the market size grows large ($n\to\infty$),
the equilibrium cutoffs in a finite-$n$ market converge to a unique
limit, denoted by $p^{*}=(p_{1}^{*},\dots,p_{G}^{*})'$, referred
to as the limiting equilibrium cutoffs. Unlike the finite-$n$ cutoffs
$p_{n}$, the limiting cutoffs $p^{*}$ are deterministic because
they are determined by the distribution of characteristics.\footnote{The convergence result in \citet{azevedo_supply_2016} requires that
the number of groups is finite and that the size of each group grows
proportionally large. To apply their result, we adopt the same assumption.
While extending the analysis to allow for a growing number of groups
is an interesting direction, we leave this for future work.} Building on the cutoff convergence result established by \citet{azevedo_supply_2016},
we show that the selection bias in a finite-$n$ market also converges
to the limiting counterpart with $p^{*}$ in place of $p_{n}$.
\begin{prop}[Limiting approximation]
\label{prop:biascov} Under Assumption \ref{ass:regular}(i)-(ii),
we have
\begin{equation}
\mathbb{E}[\epsilon_{i}|\boldsymbol{x},\boldsymbol{z},\boldsymbol{g}(\boldsymbol{z},\boldsymbol{\xi},\boldsymbol{\eta};p_{n})]\overset{p}{\rightarrow}\mathbb{E}[\epsilon_{i}|x_{i},z_{i},g(z_{i},\xi_{i},\eta_{i};p^{*})],\text{ as }n\rightarrow\infty.\label{eq:limiting_bias}
\end{equation}
\end{prop}
The proposition indicates that we can approximate the finite-market
selection bias using its limiting counterpart. Because the limiting
cutoffs are deterministic, each individual's selection bias depends
only on their own characteristics. This reduces the dimensionality
of the selection bias from $O(n)$ to a finite number.

Large-market approximation has been widely used in the matching literature
\citep{choo_who_2006,menzel_large_2015,azevedo_supply_2016,fack_beyond_2019,galichon2022cupid,HSS_twosided}.
We follow the literature and assume that selection bias takes the
limiting form.\footnote{\citet{azevedo_supply_2016} demonstrated that the finite-$n$ cutoffs
converge to the limiting cutoffs at the rate of $\sqrt{n}$. However,
addressing the approximation error in the subsequent asymptotic analysis
remains an open question for future research.}
\begin{assumption}
\label{ass:limit}The groups in the data are formed based on the limiting
cutoffs $p^{*}$.
\end{assumption}

\subsection{\protect\label{sec:exchange}Group-Invariant and Symmetric Selection
Function}

We now derive the limiting selection bias in (\ref{eq:limiting_bias})
in an explicit form. Let $\tau_{ig}\equiv z'_{i}\delta_{g}^{u}\in\mathbb{R}$
and $\iota_{ig}\equiv z'_{i}\delta_{g}^{v}-p_{g}\in\mathbb{R}$ denote
the preference and qualification (net cutoff) indices of individual
$i$ for group $g\in\mathcal{G}$, respectively. By equation (\ref{eq:glink})
and the exogeneity of $(x_{i},z_{i})$ (Assumption \ref{ass:regular}(iii)),
we can represent the limiting selection bias in (\ref{eq:limiting_bias})
as
\begin{eqnarray}
\mathbb{E}[\epsilon_{i}|x_{i},z_{i},g_{i}=g] & = & \mathbb{E}[\epsilon_{i}|x_{i},z_{i},\eta_{ig}\geq p_{g}-z'_{i}\delta_{g}^{v},\nonumber \\
 &  & \xi_{ih}-\xi_{ig}<z'_{i}\delta_{g}^{u}-z'_{i}\delta_{h}^{u}\:\text{or}\:\eta_{ih}<p_{h}-z'_{i}\delta_{h}^{v},\forall h\neq g]\nonumber \\
 & \equiv & \lambda_{g}(\tau_{i1},\iota_{i1}\dots,\tau_{iG},\iota_{iG}),\label{eq:npbias}
\end{eqnarray}
where $\lambda_{g}(\cdot):\mathbb{R}^{2G}\rightarrow\mathbb{R}$ is
a nonparametric function of $(\tau_{ig},\iota_{ig}){}_{g=1}^{G}$,
referred to as the selection function.\footnote{\label{fn:p_unbinding}A group's qualification index is relevant only
when its capacity constraint binds. For groups with non-binding capacity
constraints, their qualification indices can be omitted. Specifically,
let $\overline{\mathcal{G}}\subseteq\mathcal{G}$ denote the subset
of capacity-binding groups, with cardinality $\overline{G}$. Then
$\lambda_{g}(\cdot)$ is a function of $(\tau_{ig})_{g\in\mathcal{G}}\in\mathbb{R}^{G}$
and $(\iota_{ig})_{g\in\overline{\mathcal{G}}}\in\mathbb{R}^{\overline{G}}$.} The expression of the selection bias in (\ref{eq:npbias}) supports
the approach proposed by \citet{dale_estimating_2002}, who overcame
selection bias by matching students who applied to and were accepted
by comparable sets of colleges. In our context, however, such application
and admission information is not observed. Equation (\ref{eq:npbias})
extends the results in standard sample selection models (\citealp{heckman1979sample};
\citealp{das2003nonparametric}) to social interactions with endogenous
group selection. In \citet{heckman1979sample}'s model, there is a
single ``group'' (the sample) and individuals decide whether to
join that group (be selected into the sample) through a binary choice.
\citet{das2003nonparametric} extended \citet{heckman1979sample}'s
framework by considering a binary choice determined by multiple rules
under nonparametric specifications. Our approach shares a similar
structure, as it can also be viewed as a binary choice (i.e., $1\{g_{i}=g\}$)
determined by multiple rules (e.g., $1\{u_{ig}>u_{ih}\}$ and $1\{v_{ig}>p_{g}\}$),
as shown in equation (\ref{eq:glink}). However, \citet{das2003nonparametric}
focused on selection into a single sample governed by multiple rules,
while our approach addresses selection among multiple groups.\footnote{\citet{das2003nonparametric} represented selection bias as a nonparametric
function of the propensity scores corresponding to each selection
rule. In our setting, however, the propensity scores associated with
the selection rules in equation (\ref{eq:glink}) (e.g., $1\{u_{ig}>u_{ih}\}$
and $1\{v_{ig}>p_{g}\}$) are \textit{not} available, as we do not
observe individuals' rankings of the groups or whether they qualify
for a group. Instead, we express selection bias as a function of indices.} In social interactions, \citet{brock2001interactions,brock2002multinomial,brock2005multinomial},
\citet{ioannides2008interactions}, and \citet{gu2024peer} studied
social interactions under one-sided group formation with parametric
unobservables. We generalize their results to two-sided group formation
with nonparametric unobservables.

The subscript $g$ in $\lambda_{g}(\cdot)$ indicates that its functional
form can differ across groups. Below we illustrate $\lambda_{g}(\cdot)$
in the case of three groups.
\begin{example}
\label{ex:selbias}Consider three groups ($G=3$) with the preference
and qualification indices $\tau_{ig}=z'_{i}\delta_{g}^{u}$ and $\iota_{ig}=z'_{i}\delta_{g}^{v}-p_{g}$,
$g=1,2,3$. Let $\xi_{i}=(\xi_{i1},\xi_{i2},\xi_{i3})'\in\mathbb{R}^{3}$
and $\eta_{i}=(\eta_{i1},\eta_{i2},\eta_{i3})'\in\mathbb{R}^{3}$
represent the vectors of unobservables. Denote by $f(\epsilon_{i},\xi_{i},\eta_{i})$
the joint pdf of $(\epsilon_{i},\xi_{i},\eta_{i})$ and $f(\xi_{i},\eta_{i})$
the joint pdf of $(\xi_{i},\eta_{i})$. For group $g=1,2,3$, define
$R_{g}(\tau_{i1},\iota_{i1},\dots,\tau_{i3},\iota_{i3})$ as the subset
of unobservables $(\xi{}_{i},\eta{}_{i})\in\mathbb{R}^{6}$ for which
individual $i$ joins group $g$. Formally, $R_{g}(\tau_{i1},\iota_{i1},\dots,\tau_{i3},\iota_{i3})=\{(\xi{}_{i},\eta{}_{i})\in\mathbb{R}^{6}:\eta_{ig}\geq-\iota_{ig},\xi_{ih}-\xi_{ig}<\tau_{ig}-\tau_{ih}\text{ or }\eta_{ih}<-\iota_{ih},\forall h\neq g\}$.
The selection bias of individual $i$ when joining group $g$ is
\begin{eqnarray}
\mathbb{E}[\epsilon_{i}|x_{i},z_{i},g_{i}=g] & = & \frac{\int_{R_{g}(\tau_{i1},\iota_{i1},\dots,\tau_{i3},\iota_{i3})}\epsilon_{i}f(\epsilon_{i},\xi_{i},\eta_{i})\textrm{d}\epsilon_{i}\textrm{d}\xi_{i}\textrm{d}\eta_{i}}{\int_{R_{g}(\tau_{i1},\iota_{i1},\dots,\tau_{i3},\iota_{i3})}f(\xi_{i},\eta_{i})\textrm{d}\xi_{i}\textrm{d}\eta_{i}}\nonumber \\
 & \equiv & \lambda_{g}(\tau_{i1},\iota_{i1},\dots,\tau_{i3},\iota_{i3}),\enspace g=1,2,3.\label{eq:eg-selbias}
\end{eqnarray}
The denominator in equation (\ref{eq:eg-selbias}) equals the conditional
probability that $i$ joins group $g$, $\Pr(g_{i}=g|z_{i})$.
\end{example}
The selection bias in equation (\ref{eq:npbias}) present two features
that complicate its correction. First, the selection function $\lambda_{g}(\cdot)$
may be group-specific if the distribution of the unobservables $(\xi_{ig},\eta_{ig})$
differs across groups.\footnote{Equation (\ref{eq:npbias}) also indicates that the indices for group
$g$ play a distinct role in the selection function compared to those
for the other groups. Nevertheless, this group-specific feature can
be addressed by separating the indices for group $g$ from those for
the other groups, as formulated in equation (\ref{eq:npbias-exch}).} A group-specific selection function poses a challenge in identifying
the effects of group-level variables, as these effects cannot be distinguished
from nonparametric selection bias that is specific to each group.\footnote{Existing literature (e.g., \citealp{brock2005multinomial}) does not
encounter this issue because with parametric assumptions, the functional
form of $\lambda_{g}(\cdot)$ is known.} This problem resembles an issue in panel data models, where the effects
of time-invariant variables cannot be distinguished from individual
fixed effects. In our context, there could be variables of interest
that are at group level. For example, let $w_{i}$ denote the $i$th
row of $\boldsymbol{w}$. if $w_{i}\boldsymbol{y}$ and $w_{i}\boldsymbol{x}$
are group averages that include individual $i$ as in \citet{Manski1993},
they are invariant within a group.\footnote{\label{fn:exclude_variation}If $w_{i}\boldsymbol{y}$ and $w_{i}\boldsymbol{x}$
are group averages that exclude $i$, they converge to including-oneself
group averages as group size goes to infinity (see footnote \ref{fn:exclude_limit}).
Hence, the variation of $w_{i}\boldsymbol{y}$ and $w_{i}\boldsymbol{x}$
within a group vanishes to zero as the group size grows.} Moreover, in certain applications group fixed effects may themselves
be the parameters of interest. For instance, school fixed effects
may be used as measures of school effectiveness in school choice.
Second, the selection function depends on the indices of all groups,
which may lead to a dimensionality issue if the number of groups is
moderately large.

We propose a novel method to overcome these problems. Note that by
appropriately arranging the indices, the selection function can be
made group-invariant and symmetric in alternative groups, provided
that the distribution of unobservables exhibits symmetry across groups.
Motivated by this insight, we assume that the joint distribution of
unobservables is exchangeable across groups.
\begin{assumption}[Exchangeability]
\label{ass:exch}The joint pdf of $(\epsilon_{i},\xi_{i},\eta_{i})$
is exchangeable in group identities, that is,
\[
f(\epsilon_{i},\xi_{i1},\dots,\xi_{iG},\eta_{i1},\dots,\eta_{iG})=f(\epsilon_{i},\xi_{ik_{1}},\dots,\xi_{ik_{G}},\eta_{ik_{1}},\dots,\eta_{ik_{G}}),
\]
for any permutation $(k_{1},\dots,k_{G})$ of \textup{$(1,\dots,G)$.}
\end{assumption}
Assumption \ref{ass:exch} requires that the joint distribution of
unobservables remains invariant under the relabeling of groups. The
concept of exchangeability has been widely used in various contexts,
such as differentiated product markets \citep{gandhi2019measuring},
panel data \citep{altonji_matzkin}, matching \citep{fox_unobserved_2018},
and network formation \citep{menzel2021central}, although the specific
methodologies vary. Exchangeability can accommodate a complex dependence
structure among the unobservables of an individual. Specifically,
it allows for (i) the dependence between the unobserved preference
$\xi_{ig}$ and the unobserved qualification $\eta_{ig}$ for each
group $g\in G$, and (ii) the dependence between the unobservables
$(\xi_{ig},\eta_{ig})$ for group $g$ and the unobservables $(\xi_{ih},\eta_{ih})$
for another group $h\neq g$. However, because exchangeability imposes
symmetry across groups in the distribution of unobservables, it rules
out scenarios where the variance of unobservables differs across groups
or where the correlation between the unobservables of two groups differs
across pairs of groups.\footnote{For example, in college admissions, students' unobserved preferences
may be more strongly correlated between elite colleges than between
elite and non-elite colleges. Similarly, elite colleges may share
unobserved evaluation criteria for qualitative factors (e.g., extracurriculars)
that differ from those used by non-elite colleges.} In such scenarios, we can relax Assumption \ref{ass:exch} by categorizing
groups into distinct types (e.g., public and private schools) and
imposing exchangeability only among groups of the same type.\footnote{The subsequent results can be extended to exchangeability conditional
on group type. However, for notational simplicity, we retain the basic
setting of exchangeability across all groups in the main text. In
our empirical application, we impose exchangeability separately for
public and private schools. See Section \ref{sec:Empirical} for more
details.} Furthermore, note that we can include group fixed effects in $u_{ig}$
to account for unobserved heterogeneity across groups, but the individual-varying
unobservables must not exhibit group heterogeneity.\footnote{While we could also include group fixed effects in $v_{ig}$, they
become redundant when groups rank individuals based on $v_{ig}$ and
cannot be distinguished from the cutoffs.} Assumption \ref{ass:exch} is satisfied under the usual logit or
probit specifications.\footnote{For example, $\xi_{ig}=a_{i}+\tilde{\xi}_{ig}$, where $a_{i}$ represents
an individual effect that is i.i.d. across $i$, and $\tilde{\xi}_{ig}$
follows a type I extreme value or Gaussian distribution, and is i.i.d.
across both $i$ and $g$.}

Observe from equation (\ref{eq:npbias}) that only utility differences
matter. For $g\in\mathcal{G}$, we define the utility difference $\Delta_{g}\tau_{ih}\equiv\tau_{ih}-\tau_{ig}$
for $h\neq g$. By separating the index $\iota_{ig}$ for group $g$
from those for the other groups, the selection bias in equation (\ref{eq:npbias})
can be reformulated as
\begin{equation}
\mathbb{E}[\epsilon_{i}|x_{i},z_{i},g_{i}=g]\equiv\lambda_{g}^{e}(\iota_{ig};\Delta_{g}\tau_{ih},\iota_{ih},\forall h\neq g).\label{eq:npbias-exch}
\end{equation}
Under exchangeability, the selection function $\lambda_{g}^{e}(\cdot)$
becomes invariant across groups, that is, there exists $\lambda^{e}(\cdot)$
such that $\lambda_{g}^{e}(\cdot)\equiv\lambda^{e}(\cdot)$ for all
$g\in\mathcal{G}$. Furthermore, exchangeability implies that the
ordering of the index pairs $(\Delta_{g}\tau_{ih},\iota_{ih})$ across
all $h\neq g$ is irrelevant -- the selection function is symmetric
in the indices for groups other than $g$.
\begin{prop}[Group invariance and symmetry]
\label{prop:exch}Under Assumptions \ref{ass:adj_exog}--\ref{ass:exch},
the selection function $\lambda_{g}^{e}(\cdot)$ defined in equation
(\ref{eq:npbias-exch}) satisfies the following properties: (i) it
is invariant across $g$, that is, $\lambda_{g}^{e}(\cdot)\equiv\lambda^{e}(\cdot)$
for all $g\in\mathcal{G}$; (ii) it is symmetric in the index pairs
$(\Delta_{g}\tau_{ih},\iota_{ih})$ across all $h\neq g$.
\end{prop}
\begin{example}[Example \ref{ex:selbias} continued]
\label{ex:selbias-exch}Consider the three-group example discussed
in Example \ref{ex:selbias}. Rewrite the set $R_{g}(\tau_{i1},\iota_{i1},\dots,\tau_{i3},\iota_{i3})$
in Example \ref{ex:selbias} as $R_{g}(\iota_{ig};\Delta_{g}\tau_{ih},\iota_{ih},\forall h\neq g)$,
where we separate out group $g$'s index $\iota_{ig}$ from those
of other groups and take utility difference. We can rewrite the selection
bias in (\ref{eq:eg-selbias}) as
\begin{eqnarray}
\mathbb{E}[\epsilon_{i}|x_{i},z_{i},g_{i}=g] & = & \frac{\int_{R_{g}(\iota_{ig};\Delta_{g}\tau_{ih},\iota_{ih},\forall h\neq g)}\epsilon_{i}f(\epsilon_{i},\xi_{i},\eta_{i})\textrm{d}\epsilon_{i}\textrm{d}\xi_{i}\textrm{d}\eta_{i}}{\int_{R_{g}(\iota_{ig};\Delta_{g}\tau_{ih},\iota_{ih},\forall h\neq g)}f(\xi_{i},\eta_{i})\textrm{d}\xi_{i}\textrm{d}\eta_{i}}\nonumber \\
 & \equiv & \lambda_{g}^{e}(\iota_{ig};\Delta_{g}\tau_{ih},\iota_{ih},\forall h\neq g),\enspace g=1,2,3.\label{eq:eg-selbias-exch}
\end{eqnarray}
Under the exchangeability assumption, the unobservable vectors $(\eta_{i1},\xi_{i2}-\xi_{i1},\eta_{i2},\xi_{i3}-\xi_{i1},\eta_{i3})$,
$(\eta_{i2},\xi_{i1}-\xi_{i2},\eta_{i1},\xi_{i3}-\xi_{i2},\eta_{i3})$
and $(\eta_{i3},\xi_{i1}-\xi_{i3},\eta_{i1},\xi_{i2}-\xi_{i3},\eta_{i2})$
have the same joint distributions. As a result, the three selection
functions are identical, that is, $\lambda_{1}^{e}(\cdot)=\lambda_{2}^{e}(\cdot)=\lambda_{3}^{e}(\cdot)\equiv\lambda^{e}(\cdot)$.
Moreover, note that the joint distribution of $(\eta_{i1},\xi_{i2}-\xi_{i1},\eta_{i2},\xi_{i3}-\xi_{i1},\eta_{i3})$
is symmetric in $(\xi_{i2}-\xi_{i1},\eta_{i2})$ and $(\xi_{i3}-\xi_{i1},\eta_{i3})$.
This implies that $\lambda^{e}(\iota_{i1};\Delta_{1}\tau_{i2},\iota_{i2},\Delta_{1}\tau_{i3},\iota_{i3})$
is symmetric in $(\Delta_{1}\tau_{i2},\iota_{i2})$ and $(\Delta_{1}\tau_{i3},\iota_{i3})$.
Similarly, we can see that $\lambda^{e}(\iota_{i2};\Delta_{2}\tau_{i1},\iota_{i1},\Delta_{2}\tau_{i3},\iota_{i3})$
is symmetric in $(\Delta_{2}\tau_{i1},\iota_{i1})$ and $(\Delta_{2}\tau_{i3},\iota_{i3})$,
and $\lambda^{e}(\iota_{i3};\Delta_{3}\tau_{i1},\iota_{i1},\Delta_{3}\tau_{i2},\iota_{i2})$
is symmetric in $(\Delta_{3}\tau_{i1},\iota_{i1})$ and $(\Delta_{3}\tau_{i2},\iota_{i2})$.
\end{example}
Proposition \ref{prop:exch} has important implications. First, the
group invariance in Proposition \ref{prop:exch}(i) implies that we
can use a single selection function for all the groups. Second, the
symmetry established in Proposition \ref{prop:exch}(ii) can further
reduce the dimensionality of the selection function. Note that the
selection function depends not only on the index of the group an individual
joins $\iota_{ig_{i}}$, but also on the indices of all other groups
$(\Delta_{g_{i}}\tau_{ih},\iota_{ih})$, $h\neq g_{i}$. Because the
selection function is symmetric in the indices $(\Delta_{g_{i}}\tau_{ih},\iota_{ih})$
for $h\neq g_{i}$, we can equivalently express the selection bias
in equation (\ref{eq:npbias-exch}) using elementary symmetric functions
of these indices.

Specifically, we denote the elementary symmetric functions of $(\Delta_{g_{i}}\tau_{ih},\iota_{ih}){}_{\forall h\ne g_{i}}\in\mathbb{R}^{2(G-1)}$
as $\pi_{i,-g_{i}}$, which can be represented by the coefficients
of the polynomial function $\prod_{h\neq g_{i}}(1+(\Delta_{g_{i}}\tau_{ih},\iota_{ih})t)$
in the indeterminates $t=(t_{1},t_{2})'$ \citep{Weyl1946}.\footnote{\label{fn:tau_sym}The first two orders of the elementary symmetric
functions are given by the sums of all individual terms ($\sum_{h\neq g_{i}}\Delta_{g_{i}}\tau_{ih}$
and $\sum_{h\neq g_{i}}\iota_{ih}$) and the sums of all pairwise
products ($\sum_{(h_{1},h_{2})\neq g_{i}}\Delta_{g_{i}}\tau_{ih_{1}}\Delta_{g_{i}}\tau_{ih_{2}}$,
$\sum_{(h_{1},h_{2})\neq g_{i}}\iota_{ih_{1}}\iota_{ih_{2}}$, and
$\sum_{(h_{1},h_{2})\neq g_{i}}\Delta_{g_{i}}\tau_{ih_{1}}\iota_{ih_{2}}$,
where $\sum_{(h_{1},h_{2})\neq g_{i}}$ denotes the sum over all combinations
of distinct $h_{1}$ and $h_{2}$ in $\mathcal{G}\backslash\{g_{i}\}$).
The higher-order functions can be derived similarly.} Define $\pi_{i}\equiv(\iota_{ig_{i}};\pi_{i,-g_{i}})$. By the fundamental
theorem of symmetric functions in conjunction with the Weierstrass
approximation theorem, any symmetric function can be approximated
arbitrarily closely by a polynomial function of the elementary symmetric
functions \citep{altonji_matzkin}. Therefore, there exists a function
$\lambda(\cdot)$ such that
\begin{equation}
\mathbb{E}[\epsilon_{i}|x_{i},z_{i},g_{i}]=\lambda^{e}(\iota_{ig_{i}};\Delta_{g_{i}}\tau_{ih},\iota_{ih},\forall h\neq g_{i})\equiv\lambda(\pi_{i}).\label{eq:npbias-sym}
\end{equation}

Using the symmetric representation reduces the number of nuisance
parameters in a sieve approximation of the selection function. For
instance, if we use linear basis functions, $\lambda^{e}(\cdot)$
has $2G-1$ approximating functions, whereas $\lambda(\cdot)$ has
only $3$, including one for group $g_{i}$ and two for the remaining
$G-1$ groups combined. If we consider basis functions of order two,
$\lambda^{e}(\cdot)$ has $(2G-1)G$ approximating functions, whereas
$\lambda(\cdot)$ has $9$.\footnote{Follow the discussion in footnote \ref{fn:tau_sym}. For $\lambda^{e}(\cdot)$,
we have $(2G-1)G$ functions of order two: $2G-1$ squared indices
and $(2G-1)(G-1)$ pairwise interactions. For $\lambda(\cdot)$, we
have the following $9$ functions of order two: $(\iota_{ig_{i}})^{2}$,
$\iota_{ig_{i}}\sum_{h\neq g_{i}}\Delta_{g_{i}}\tau_{ih}$, $\iota_{ig_{i}}\sum_{h\neq g_{i}}\iota_{ih}$,
$(\sum_{h\neq g_{i}}\Delta_{g_{i}}\tau_{ih})^{2}$, $(\sum_{h\neq g_{i}}\iota_{ih})^{2}$,
$(\sum_{h\neq g_{i}}\Delta_{g_{i}}\tau_{ih})(\sum_{h\neq g_{i}}\iota_{ih})$,
$\sum_{(h_{1},h_{2})\neq g_{i}}\Delta_{g_{i}}\tau_{ih_{1}}\Delta_{g_{i}}\tau_{ih_{2}}$,
$\sum_{(h_{1},h_{2})\neq g_{i}}\iota_{ih_{1}}\iota_{ih_{2}}$, and
$\sum_{(h_{1},h_{2})\neq g_{i}}\Delta_{g_{i}}\tau_{ih_{1}}\iota_{ih_{2}}$.} Aggregating the indices for groups other than $g_{i}$ by elementary
symmetric functions has the advantage that the number of nuisance
parameters in a sieve approximation for a given order does not depend
on the number of groups. This reduces the dimensionality of a sieve
approximation if the number of groups is moderately large. \citet{gandhi2019measuring}
employed a similar strategy to reduce the dimensionality of instruments
in a BLP model with many products.

\paragraph{Discussion.}

Selection bias $\lambda(\pi_{i})$ is individual-specific because
it depends on $\pi_{i}=\pi(z_{i},g_{i})$, which involves individual
characteristics $z_{i}$. While endogenous selection occurs only during
group formation, individuals with different values of $z_{i}$ may
be subject to varying levels of selection bias. This feature mirrors
the results in standard sample selection models. For instance, the
selection bias in \citet{heckman1979sample}---the inverse Mills
ratio---depends on individual characteristics included in the selection
equation. This finding provides empirical guidance that contrasts
with the group-fixed-effect approach, which assumes constant selection
bias within a group. \citet{Altonji2018} demonstrated that selection
bias can be captured by group averages of covariates. However, their
result relies on specific linear restrictions and may not be generalized
to nonparametric settings. In our setting, neither group fixed effects
nor group averages in the outcome equation adequately addresses group
selection. To effectively correct for selection bias, it is essential
to exploit the individual-level information contained in $z_{i}$.

Another key feature of selection bias is that it depends not only
on the indices of the group an individual joins, but also on those
of alternative groups. This finding highlights the importance of accounting
for the impacts of alternative groups, which can be challenging due
to dimensionality issues. \citet{Dahl2002} imposed an assumption
that controlling for the probability of an individual's first-best
choice exhausts the impacts of all groups. In contrast, we aggregate
the impacts of alternative groups by leveraging the symmetry of the
selection function under exchangeability. Related to our work, \citet{abdulkadirouglu2020parents}
proposed a selection correction method in centralized school choice,
where selection occurs only through unobserved preferences. Their
selection correction aggregates the impacts of alternative schools,
constructed using rank-order lists under multinomial logit unobservables,
which resembles the linear components of the elementary symmetric
functions in our selection bias derived under exchangeability. We
provide a microfoundation for the aggregation and derive the exact
form of selection bias under more general nonparametric unobservables.
Unlike their approach, ours is applicable to decentralized markets
with selection on both unobserved preferences and qualifications,
where rank-order lists are unavailable.

In addition, the literature on social interactions has long struggled
with disentangling two distinct sources of correlated effects: those
arising from group-level unobservables and those arising from self-selection
into groups based on individual-level unobservables. Distinguishing
between these sources is crucial, as they have distinct policy implications.
For example, in the context of school choice, should a policymaker
focus on improving school effectiveness or altering self-selection
among students? By imposing the exchangeability assumption, our approach
allows for the inclusion of group fixed effects alongside selection
correction, thereby enabling the separate identification of school
effectiveness and self-selection.

\section{\protect\label{sec:Identification}Identification}

We now turn to the identification of the social effects $\gamma$.
\citet{HSS_twosided} established identification results for the group
formation parameters, including the slope parameters $\delta\equiv(\delta_{g}^{u},\delta_{g}^{v})_{g\in\mathcal{G}}$
and the cutoffs $p$. We thus proceed by treating $\pi_{i}$ as known.

Let $\nu_{i}\equiv\epsilon_{i}-\lambda(\pi_{i})$ represent the residual
of $\epsilon_{i}$ after eliminating selection bias. Write equation
(\ref{eq:linmean}) as
\begin{align}
y_{i} & =w_{i}\boldsymbol{y}\gamma_{1}+w_{i}\boldsymbol{x}\gamma_{2}+x'_{i}\gamma_{3}+\lambda(\pi_{i})+\nu_{i}\nonumber \\
 & \equiv X'_{i}\gamma+\lambda(\pi_{i})+\nu_{i},\label{eq:linmean_bias}
\end{align}
where $X_{i}\equiv(w_{i}\boldsymbol{y},w_{i}\boldsymbol{x},x'_{i})'\in\mathbb{R}^{d_{X}}$
denotes a vector of regressors with dimension $d_{X}\equiv2d_{x}+1$.
This is a partially linear model \citep{Robinson1988}. To eliminate
$\lambda(\pi_{i}),$ we take the expectation of equation (\ref{eq:linmean_bias})
conditional on $\pi_{i}$ and subtract it from equation (\ref{eq:linmean_bias}):
\begin{equation}
y_{i}-\mathbb{E}[y_{i}|\pi_{i}]=(X_{i}-\mathbb{E}[X_{i}|\pi_{i}])'\gamma+\nu_{i}.\label{eq:linmean_partialout}
\end{equation}
The following rank condition guarantees the identification of $\gamma$. 
\begin{assumption}[Rank]
\label{ass:rank}For each $i\in\mathcal{N}$, the matrix $\mathbb{E}[(X_{i}-\mathbb{E}[X_{i}|\pi_{i}])(X_{i}-\mathbb{E}[X_{i}|\pi_{i}])']$
is non-singular, i.e., for any $a\neq0$, $a\in\mathbb{R}^{d_{X}}$,
there is no measurable function $h(\pi_{i})$ such that $X'_{i}a=h(\pi_{i})$.\textcolor{blue}{{} }
\end{assumption}
This assumption was imposed by \citet{cosslett1991semiparametric}
and discussed by \citet{newey2009two}. The rank condition is satisfied
if and only if there is no linear combination of $w_{i}\boldsymbol{y}$,
$w_{i}\boldsymbol{x}$, and $x_{i}$ that is a function of $\pi_{i}$
almost surely.\footnote{To see the equivalence between the two statements in Assumption \ref{ass:rank},
note that $\mathbb{E}[(X_{i}-\mathbb{E}[X_{i}|\pi_{i}])(X_{i}-\mathbb{E}[X_{i}|\pi_{i}])']$
is singular if and only if there exists $a\neq0$, $a\in\mathbb{R}^{d_{X}}$,
such that $a'\mathbb{E}[(X_{i}-\mathbb{E}[X_{i}|\pi_{i}])(X_{i}-\mathbb{E}[X_{i}|\pi_{i}])']a=0.$
This can be rewritten as $\mathbb{E}[((X_{i}-\mathbb{E}[X_{i}|\pi_{i}])'a)^{2}]=0$,
which holds if and only if $(X_{i}-\mathbb{E}[X_{i}|\pi_{i}])'a=0$
with probability one. Since $\mathbb{E}[X_{i}|\pi_{i}])'a$ is a measurable
function of $\pi_{i}$ (call it $h(\pi_{i})$), we have $X_{i}\prime a=h(\pi_{i})~$with probability 1.
Conversely, if there exists some $a$ such that $X'_{i}a=h(\pi_{i})$,
we have $(X_{i}-\mathbb{E}[X_{i}|\pi_{i}])'a=X'_{i}a-\mathbb{E}[X'_{i}a|\pi_{i}]=h(\pi_{i})-h(\pi_{i})=0$
with probability one.} A sufficient condition of it is that (i) $w_{i}\boldsymbol{y}$,
$w_{i}\boldsymbol{x}$, and $x_{i}$ are linearly independent; and
(ii) the conditional distribution of $\pi_{i}$ given $X_{i}$ has
an absolutely continuous component with conditional density that is
positive on the entire real line for almost all $X_{i}$ \citep{newey2009two}.
The violation of condition (i) is referred to as the reflection problem
\citep{Manski1993,blume2011identification}, which we will discuss
in Section \ref{sec:id_reflection}. For condition (ii), recall that
$\pi_{i}=\pi(z_{i},g_{i})$. Suppose $z_{i}$ contains at least one
component that is not present in $x_{i}$ and has an absolutely continuous
density that is positive on its support. This excluded variable would
provide sufficient variation in $\pi_{i}$ (conditional on $X_{i}$)
to satisfy the rank condition.

\subsection{\protect\label{sec:id_reflection}The reflection problem}

To investigate the reflection problem, we consider the social equilibrium
in equation (\ref{eq:linmean_bias}). Let $\boldsymbol{\lambda}=\boldsymbol{\lambda}(\boldsymbol{\pi})$
denote the $n\times1$ vector that stacks $\lambda(\pi_{i})$, where
$\boldsymbol{\pi}\equiv(\pi'_{1},\dots,\pi'_{n})'$, and $\boldsymbol{\nu}$
the $n\times1$ vector that stacks $\nu_{i}$. Write equation (\ref{eq:linmean_bias})
in a matrix form
\begin{equation}
\boldsymbol{y}=\boldsymbol{w}\boldsymbol{y}\gamma_{1}+\boldsymbol{w}\boldsymbol{x}\gamma_{2}+\boldsymbol{x}\gamma_{3}+\boldsymbol{\lambda}+\boldsymbol{\nu}.\label{eq:lin-mean-stack}
\end{equation}
Assume $|\gamma_{1}|<1$ and $\interleave\boldsymbol{w}\interleave_{\infty}=\max_{i\in\mathcal{N}}\sum_{j=1}^{n}|w_{ij}|=1$,
so $\boldsymbol{s}\equiv I_{n}-\gamma_{1}\boldsymbol{w}$ is invertible
and $\boldsymbol{s}^{-1}=(I_{n}-\gamma_{1}\boldsymbol{w})^{-1}=\sum_{k=0}^{\infty}\gamma_{1}^{k}\boldsymbol{w}^{k}$.
The social equilibrium is
\[
\boldsymbol{w}\boldsymbol{y}=\boldsymbol{s}^{-1}(\boldsymbol{w}^{2}\boldsymbol{x}\gamma_{2}+\boldsymbol{w}\boldsymbol{x}\gamma_{3}+\boldsymbol{w}\boldsymbol{\lambda}+\boldsymbol{w}\boldsymbol{\nu}).
\]
Its $i$th equation is given by 
\begin{equation}
w_{i}\boldsymbol{y}=w_{i}\boldsymbol{x}\gamma_{3}+\sum_{k=0}^{\infty}\gamma_{1}^{k}w_{i}^{k+2}\boldsymbol{x}(\gamma_{1}\gamma_{3}+\gamma_{2})+\sum_{k=0}^{\infty}\gamma_{1}^{k}w_{i}^{k+1}\boldsymbol{\lambda}+\sum_{k=0}^{\infty}\gamma_{1}^{k}w_{i}^{k+1}\boldsymbol{\nu},\label{eq:wiy}
\end{equation}
where $w_{i}^{k}$ denotes the $i^{th}$ row of $\boldsymbol{w}^{k}$.
The expression implies that $w_{i}\boldsymbol{y}$, $w_{i}\boldsymbol{x}$,
and $x_{i}$ are linearly independent if (i) the support of $x_{i},w_{i}\boldsymbol{x},w_{i}^{2}\boldsymbol{x},w_{i}^{3}\boldsymbol{x},\dots$
is not contained in a proper linear subspace of $\mathbb{R}^{2d_{x}+1}$
and $\gamma_{1}\gamma_{3}+\gamma_{2}\neq0$, or (ii) the support of
$x_{i},w_{i}\boldsymbol{x},w_{i}\boldsymbol{\lambda},w_{i}^{2}\boldsymbol{\lambda},\dots$
is not contained in a proper linear subspace of $\mathbb{R}^{2d_{x}+1}$.

Sufficient conditions have been established in the existing literature
for case (i). For example, $w_{i}^{2}\boldsymbol{x}$, $w_{i}\boldsymbol{x}$,
and $x_{i}$ are linearly independent if there is an intransitive
triad in each group \citep{bramoulle2009identification}, or if there
is variation in group sizes when we consider group averages that exclude
oneself \citep{lee2007identification,graham2008identifying,Davezies2009}.
However, this identification strategy fails for group averages that
include oneself because $\boldsymbol{w}^{2}=\boldsymbol{w}$ \citep{Manski1993,bramoulle2009identification}.\footnote{In large groups, the difference between group averages that include
or exclude oneself vanishes (see footnote \ref{fn:exclude_variation}).
Consequently, identification through variation in group sizes becomes
less effective as groups grow large.}

The presence of selection offers an alternative method of identification
through case (ii). It is evident from equation (\ref{eq:wiy}) that
$w_{i}\boldsymbol{y}$, $w_{i}\boldsymbol{x}$, and $x_{i}$ are linearly
independent if $w_{i}\boldsymbol{\lambda}$, $w_{i}\boldsymbol{x}$,
and $x_{i}$ are linearly independent. This identification strategy
is applicable regardless of whether group averages include or exclude
oneself, or whether there are networks within groups. The result is
consistent with the insight of \citet[Section 3.6]{brock2001interactions}
that identification can be achieved through self-selection, provided
that there is variation in the selection within a group. Because selection
bias acts as an individual-level variable whose average is not included
in the contextual effect, its presence precludes $w_{i}\boldsymbol{y}$
and $w_{i}\boldsymbol{x}$ from being linearly dependent \citep{Manski1993,brock2001interactions}.

\section{\protect\label{sec:Estimation}Estimation}

We now discuss how to estimate the social effects $\gamma$. Denote
the true value of $\gamma$ by $\gamma_{0}\in\mathbb{R}^{d_{X}}$.
Following equation (\ref{eq:linmean_partialout}), we propose a three-step
sieve OLS estimator for $\gamma_{0}$. Recall that the indices $\pi_{i}=\pi(z_{i},g_{i},\theta_{0})$
in selection bias depend on the parameters in group formation, $\theta_{0}\equiv(\delta'_{0},p^{*\prime})'$.
In the first step, we estimate these indices by the estimator $\hat{\pi}_{i}=\pi(z_{i},g_{i},\hat{\theta})$,
where $\hat{\theta}=(\hat{\delta}',\hat{p}')'$ is an estimator of
$\theta_{0}$. Using the estimated indices, we estimate $\gamma_{0}$
using sieve OLS in the next two steps.

Specifically, let $b^{K}(\pi_{i})=(b_{1K}(\pi_{i}),\dots,b_{KK}(\pi_{i}))'$
be a $K\times1$ vector of approximating functions for individual
$i$, and $B_{K}(\boldsymbol{\pi})=(b^{K}(\pi_{1}),\dots,b^{K}(\pi_{n}))'$
the $n\times K$ matrix of all approximating functions, where $\boldsymbol{\pi}=(\pi'_{1},\dots,\pi'_{n})'$.
By replacing $\boldsymbol{\pi}$ with its estimator $\hat{\boldsymbol{\pi}}=(\hat{\pi}'_{1},\dots,\hat{\pi}'_{n})'$,
we obtain an estimate of the approximating functions $\hat{B}_{K}=B_{K}(\hat{\boldsymbol{\pi}})=(b^{K}(\hat{\pi}_{1}),\dots,b^{K}(\hat{\pi}_{n}))'$.
Next, we estimate the conditional expectation $\mu_{0}^{X}(\pi_{i})\equiv\mathbb{E}[X_{i}|\pi_{i}]$
using the sieve estimator $\hat{\mu}^{X}(\hat{\pi}_{i})\equiv\boldsymbol{X}'\hat{B}_{K}(\hat{B}'_{K}\hat{B}_{K})^{-1}b^{K}(\hat{\pi}_{i})$,
where $\boldsymbol{X}$ is the $n\times d_{X}$ matrix that stacks
$X'_{i}$. Define the matrix $\hat{M}_{K}\equiv I-\hat{B}_{K}(\hat{B}'_{K}\hat{B}_{K})^{-1}\hat{B}'_{K}$.
The three-step sieve OLS estimator of $\gamma_{0}$ is given by
\begin{align}
\hat{\gamma} & \equiv(\boldsymbol{X}'\hat{M}_{K}\boldsymbol{X}){}^{-1}\boldsymbol{X}'\hat{M}_{K}\boldsymbol{y}\nonumber \\
 & =\left(\frac{1}{n}\sum_{i=1}^{n}(X_{i}-\hat{\mu}^{X}(\hat{\pi}_{i}))X'_{i}\right)^{-1}\left(\frac{1}{n}\sum_{i=1}^{n}(X_{i}-\hat{\mu}^{X}(\hat{\pi}_{i}))y_{i}\right).\label{eq:gamma_hat}
\end{align}

With the addition of the following assumptions, we show in Theorems
\ref{thm:gamma_consist} and \ref{thm:gamma_clt} that $\hat{\gamma}$
is $\sqrt{n}$ consistent and asymptotically normal.
\begin{assumption}[Bounded Covariates]
\label{ass:compact}(i) The support of $z_{i}$ is bounded. (ii)
The support of $x_{i}$ is bounded.
\end{assumption}
\begin{assumption}[Group formation parameters]
\label{ass:theta}(i) The true parameter $\theta_{0}$ lies in the
interior of a compact set $\Theta$. (ii) $\hat{\theta}-\theta_{0}=n^{-1}\sum_{i=1}^{n}\phi_{\theta}(z_{i},\theta_{0})+o_{p}(n^{-1/2})$,
where $\mathbb{E}[\phi_{\theta}(z_{i},\theta_{0})]=0$ and $\mathbb{E}[\|\phi_{\theta}(z_{i},\theta)\|^{2}]<\infty$.
\end{assumption}
\begin{assumption}[Sieve]
\label{ass:sieve}Let $K\rightarrow\infty$ and $K/n\rightarrow0$.
The basis functions $b^{K}(\pi)\in\mathbb{R}^{K}$ satisfy the following
conditions. (i) $\mathbb{E}[b^{K}(\pi)b^{K}(\pi)']=I_{K}$. (ii) There
exist $\beta^{X}$ and a constant $a>0$ such that $\sup_{\pi}\|\mu_{0}^{X}(\pi)-b^{K}(\pi)'\beta^{X}\|=O(K^{-a})$.
(iii) $\sup_{\pi}\|b^{K}(\pi)\|\leq\varrho_{0}(K)$ for constants
\textup{$\varrho_{0}(K)$} such that $\varrho_{0}(K)^{2}K/n\rightarrow0$.
(iv) $\sup_{\pi}\|\partial b^{K}(\pi)/\partial\pi'\|\leq\varrho_{1}(K)$
for constants \textup{$\varrho_{1}(K)$} such that $\varrho_{1}(K)/\sqrt{n}\rightarrow0$.
\end{assumption}
\begin{assumption}[Adjacency Matrix]
\label{ass:w}Let $\psi_{i}\equiv(x'_{i},z'_{i},g_{i})'$ and $\boldsymbol{\psi}\equiv(\boldsymbol{x},\boldsymbol{z},\boldsymbol{g})$.
The matrices $\boldsymbol{w}=(w_{ij})\in\mathbb{R}^{n^{2}}$ and $\boldsymbol{s}=(I_{n}-\gamma_{1}\boldsymbol{w})^{-1}$
satisfy the following conditions. (i) $\interleave\boldsymbol{w}\interleave_{\infty}=\max_{i\in\mathcal{N}}\sum_{j=1}^{n}|w_{ij}|=1$.
(ii) \textup{$\mathbb{E}[\|\boldsymbol{w}\|_{\infty}^{8}]=\mathbb{E}[\max_{i,j\in\mathcal{N}}(w_{ij})^{8}]=O(n^{-8})$}.
(iii) For the matrix $\boldsymbol{q}=(q_{ij})$ in the form of $\boldsymbol{w}$
or $\boldsymbol{sw}^{t}$, $t=1,2$, we have $\max_{i,j,k,l\in\mathcal{N}:\{i,j\}\cap\{k,l\}=\emptyset}\mathbb{E}[(\mathbb{E}[q_{ij}q_{kl}|\boldsymbol{\psi}]-\mathbb{E}[q_{ij}|\psi_{i},\psi_{j}]\mathbb{E}[q_{kl}|\psi_{k},\psi_{l}])^{2}]=o(n^{-4}/K)$
and $\max_{i,j\in\mathcal{N}}\mathbb{E}[(\mathbb{E}[q_{ij}|\boldsymbol{\psi}]-\mathbb{E}[q_{ij}|\psi_{i},\psi_{j}])^{4}]=o(n^{-4}/K^{2})$,
where $K$ satisfies Assumption \ref{ass:sieve}. (iv) For the matrix
$\boldsymbol{q}$ in the form of $\boldsymbol{w}$, $\boldsymbol{w}'\boldsymbol{w}$,
$\boldsymbol{s}\boldsymbol{w}^{t}$, $\boldsymbol{w}'\boldsymbol{s}\boldsymbol{w}^{t}$,
or $(\boldsymbol{w}')^{r}\boldsymbol{s}'\boldsymbol{s}\boldsymbol{w}^{t}$,
$r,t=1,2$, $\max_{i,j,k,l\in\mathcal{N}:\{i,j\}\cap\{k,l\}=\emptyset}\mathbb{E}[(\mathbb{E}[q_{ij}q_{kl}|\boldsymbol{\psi}]-\mathbb{E}[q_{ij}|\psi_{i},\psi_{j}]\mathbb{E}[q_{kl}|\psi_{k},\psi_{l}])^{2}]=o(n^{-4})$.\footnote{Note that for the matrix $\boldsymbol{q}$ in the form of $\boldsymbol{w}$
or $\boldsymbol{sw}^{t}$, $t=1,2$, part (iv) is implied by part
(iii) because $K\rightarrow\infty$.} (v) Suppose there exist a vector of i.i.d. variables $\tilde{\psi}_{i}$
and $\tilde{\boldsymbol{\psi}}\equiv(\tilde{\psi}_{1},\dots,\tilde{\psi}_{n})'$
such that (a) $\tilde{\psi}_{i}$ contains $\psi_{i}$, (b) $\tilde{\psi}_{i}$
has finite fourth moment, (c) $\boldsymbol{w}$ and $\boldsymbol{\epsilon}$
are independent conditional on \textup{$\tilde{\boldsymbol{\psi}}$},
and (d) \textup{$\tilde{\boldsymbol{\psi}}$} and $\boldsymbol{\epsilon}$
are independent conditional on $\boldsymbol{\psi}$. For the matrix
$\boldsymbol{q}$ in the form of $\boldsymbol{w}$ or $\boldsymbol{s}\boldsymbol{w}^{t}$,
$t=1,2$, $\max_{i,j,k\in\mathcal{N}:j\neq k}\mathbb{E}[(\mathbb{E}[q_{ij}q_{ik}|\tilde{\boldsymbol{\psi}}]-\mathbb{E}[q_{ij}|\tilde{\psi}_{i},\tilde{\psi}_{j}]\mathbb{E}[q_{ik}|\tilde{\psi}_{i},\tilde{\psi}_{k}])^{2}]=o(n^{-4})$.
\end{assumption}
\begin{assumption}[Smoothness]
\label{ass:smooth}(i) The unobservable $\epsilon_{i}$ satisfies
$\mathbb{E}[\epsilon_{i}^{8}]<\infty$. (ii) For any $\theta\in\Theta$,
$\mathbb{E}[X_{i}|\pi(z_{i},g_{i},\theta)]$ and $\mathbb{E}[\epsilon_{i}|\pi(z_{i},g_{i},\theta)]$
are continuously differentiable in $\pi(z_{i},g_{i},\theta)$.
\end{assumption}
\begin{thm}[Consistency of $\hat{\gamma}$]
\label{thm:gamma_consist}Under Assumptions \ref{ass:adj_exog}--\ref{ass:smooth},
$\hat{\gamma}-\gamma_{0}=o_{p}(1)$.
\end{thm}
\begin{thm}[Asymptotic distribution of $\hat{\gamma}$]
\label{thm:gamma_clt}Under Assumptions \ref{ass:adj_exog}--\ref{ass:smooth},
$\sqrt{n}\Omega_{n}^{-1/2}M_{n}(\hat{\gamma}-\gamma_{0})\stackrel{d}{\rightarrow}N(0,I_{d_{X}})$,
where the matrices $M_{n}$ and $\Omega_{n}$ are defined in the proof.
\end{thm}
Assumption \ref{ass:compact} assumes that the covariates $z_{i}$
and $x_{i}$ are bounded. Assumption \ref{ass:theta}(i) requires
that the group formation parameter $\theta_{0}$ lies in a compact
set.\footnote{Although we set a cutoff to $-\infty$ if the capacity is not binding,
such a cutoff does not affect individuals' choices and is excluded
from $\theta_{0}$ (see footnote \ref{fn:p_unbinding}). Strictly
speaking, we also need to assume that any market-clearing cutoff is
bounded away from $-\infty$. Although the demand and supply for a
group may be equal at a cutoff of $-\infty$, the group's capacity
must take a particular value for that to occur. This is because a
cutoff of $-\infty$ no longer makes the group selective, and the
demand is solely determined by the number of individuals who do not
prefer or qualify for any other group. Such a solution is non-generic
(i.e., it requires a precise alignment of parameters that rarely occurs
in practice), and we assume that this unlikely case is ruled out.} These boundedness assumptions are standard in the literature. Assumption
\ref{ass:theta}(ii) requires that the estimator $\hat{\theta}$ has
an asymptotically linear representation, which can be satisfied by
semiparametric estimators \citep[e.g., ][]{lee1995semiparametric,Sun2019}
or parametric estimators (e.g., the constrained maximum likelihood
estimator proposed in Supplemental Appendix \ref{online:gf_estimate}).
Assumption \ref{ass:sieve} imposes standard regularity conditions
for the sieve estimation. Assumption \ref{ass:sieve}(i) is a normalization.\footnote{Alternatively, we can assume that the smallest eigenvalue of $\mathbb{E}[b^{K}(\pi)b^{K}(\pi)']$
is bounded away from zero uniformly in $K$. Assuming this, let $Q_{0}=\mathbb{E}[b^{K}(\pi)b^{K}(\pi)']$
and $Q_{0}^{-1/2}$ the symmetric square root of $Q_{0}^{-1}$. Then
$\tilde{b}^{K}(\pi)=Q_{0}^{-1/2}b^{K}(\pi)$ is a nonsingular transformation
of $b^{K}(\pi)$ that satisfies $\mathbb{E}[\tilde{b}^{K}(\pi)\tilde{b}^{K}(\pi)']=I_{K}$.
Notably, nonparametric series estimators are invariant under nonsingular
transformations of $b^{K}(\pi)$: let $\tilde{\beta}^{Z}=Q_{0}^{1/2}\beta^{Z}$
then $\tilde{b}^{K}(\pi)'\tilde{\beta}^{Z}=b^{K}(\pi)'\beta^{Z}$.
Furthermore, $\tilde{b}^{K}(\pi)$ satisfies Assumption \ref{ass:sieve}(iii)(iv)
if and only if $b^{K}(\pi)$ does. Therefore, all parts of Assumption
\ref{ass:sieve} are satisfied with $b^{K}(\pi)$ replaced by $\tilde{b}^{K}(\pi)$
\citep[p.480]{liracine2007}.} Assumption \ref{ass:sieve}(ii)-(iv) impose rate conditions on the
basis functions, similar to those used in the literature (\citealp{newey1994asymptotic};
\citealp{liracine2007}). Assumption \ref{ass:smooth}(i) requires
that $\epsilon_{i}$ have a finite eighth moment. Assumption \ref{ass:smooth}(ii)
is used to account for the estimation errors in the first and second
steps, which we establish following the approach of \citet{Hahn2013}.
Assumption \ref{ass:w} imposes crucial restrictions on the adjacency
matrix $\boldsymbol{w}$, which are essential for establishing the
asymptotic theorems. Assumption \ref{ass:w}(i) is the usual row normalization,
and the remaining parts of Assumption \ref{ass:w} will be discussed
below.

We propose an OLS estimator instead of an IV estimator for $\gamma_{0}$,
which might seem counter-intuitive given that the regressor $w_{i}\boldsymbol{y}$
could be endogenous due to simultaneity. Nevertheless, we show that
the endogeneity of $w_{i}\boldsymbol{y}$ vanishes asymptotically
under the assumption that the adjacency matrix $\boldsymbol{w}$ is
dense (Assumption \ref{ass:w}(ii)), which implies that each component
of $\boldsymbol{w}$ is bounded by $O_{p}(n^{-1})$. This assumption
is satisfied in large groups if $\boldsymbol{w}$ is specified by
group averages or dense networks within each group (see Supplemental
Appendix \ref{online:w} for examples of $\boldsymbol{w}$). The intuition
is that if an individual has a growing number of peers, the average
in $w_{i}\boldsymbol{y}$ will converge to a population expectation,
which is no longer endogenous. Our result aligns with the findings
in \citet{Lee2002}, who demonstrated the consistency and efficiency
of OLS estimators for peer effects under a deterministic adjacency
matrix.

The asymptotic analysis of $\hat{\gamma}$ is complicated by two sources
of randomness in the adjacency matrix $\boldsymbol{w}$. First, $\boldsymbol{w}$
depends on the group memberships $\boldsymbol{g}$, which are random
and can be correlated with $\boldsymbol{\epsilon}$. Second, $\boldsymbol{w}$
may include additional randomness arising from the networks within
each group. These features imply that traditional asymptotic methods,
which often assume a deterministic adjacency matrix, are not applicable.
Consequently, we develop new methods to establish the asymptotic properties
of $\hat{\gamma}$.

We begin by observing that the leading terms in $\hat{\gamma}$ take
the form of weighted $U$-statistics of order 2, with weights given
by the $(i,j)$ components of $\boldsymbol{w}$ as well as its polynomials
and series, denoted by $\boldsymbol{q}$. These polynomials and series
arise because $w_{i}\boldsymbol{y}$ is included as a regressor. It
is worth noting that our case is more challenging than the standard
weighted $U$-statistics analyzed in \citet[Section 3.7.5]{Lee1990},
where weights are assumed to be fixed. In our setting, the weights
can be stochastic as they depend on group memberships and within-group
connections. Furthermore, these random weights may introduce network
dependence, which complicates the asymptotic analysis. To address
this problem, we extend the methods in \citet[Section 3.7.5]{Lee1990}
by imposing additional restrictions on the network dependence in $\boldsymbol{q}$,
conditional on $\boldsymbol{\psi}=(\boldsymbol{x},\boldsymbol{z},\boldsymbol{g})$,
as specified in Assumption \ref{ass:w}(iii)-(v).

Specifically, we establish the consistency of the sieve estimation
in the second step using the following facts: (a) $\boldsymbol{q}$
is dense in the same way as $\boldsymbol{w}$, which follows from
row normalization and the density of $\boldsymbol{w}$ (Assumption
\ref{ass:w}(i)-(ii)), (b) the conditional dependence between network
connections $q_{ij}$ and $q_{kl}$ on disjoint nodes $\{i,j\}$ and
$\{k,l\}$ diminishes sufficiently fast, and (c) the difference between
conditioning on global information ($\boldsymbol{\psi}$) and local
information ($\psi_{i}$ and $\psi_{j}$) becomes negligible at an
appropriate rate (Assumption \ref{ass:w}(iii)). Building on these
results, we establish the consistency of $\hat{\gamma}$ under a similar
limited dependence condition on $\boldsymbol{q}$ (Assumption \ref{ass:w}(iv)).

To derive the asymptotic distribution of $\hat{\gamma}$, we extend
the Hoeffding decomposition for standard weighted $U$-statistics
\citep[Section 3.7.5]{Lee1990} by (a) conditioning on all individual-level
variables, including group memberships and, if applicable, individual
fixed effects in network formation within each group (see Example
\ref{ex:w.fe} in Supplemental Appendix \ref{online:w}), and (b)
assuming that, conditional on such individual-level variables, the
dependence between network connections $q_{ij}$ and $q_{ik}$ with
a shared node $i$ diminishes sufficiently fast (Assumption \ref{ass:w}(v)).
Under this assumption, the Hoeffding decomposition remains valid despite
random weights, enabling us to derive the asymptotic distribution
of $\hat{\gamma}$.

We verify in Supplemental Appendix \ref{online:w} that Assumption
\ref{ass:w} is satisfied for group averages (both including and excluding
oneself) and for dyadic networks with fixed effects \citep{Graham2017}.
There might be sufficient conditions on $\boldsymbol{w}$ that achieve
the desired asymptotic results while accommodating sparsity and/or
strategic interactions in network formation within groups (\citealp{leung2015two};
\citealp{ridder2020estimation}; \citealp{menzel2021central}; \citealp{Leung_Moon_2023}),
but we leave these directions for future research.

\section{\protect\label{sec:Simulations}Simulations}

\subsection{Setup}

In this section, we evaluate our approach through a simulation study.
We generate a market of $2,000$ individuals. Each individual is assigned
i.i.d. $x_{i}\sim N(5,25)$ and $\epsilon_{i}\sim N(0,1)$, where
$x_{i}$ is independent of $\epsilon_{i}$. The individuals interact
according to the linear-in-means model in equation (\ref{eq:linmean}).
We consider two scenarios: one without endogenous effects and one
with endogenous effects. In the absence of endogenous effects, we
set the parameter values to $\gamma=(0,1,1)$. In the presence of
endogenous effects, we set $\gamma=(0.5,1,1)$.

The market consists of five groups with capacities of 280, 340, 200,
460, and 400, resulting in a total of 1,680 seats. Individuals choose
which group to join based on the model described in Section \ref{sec:Model.group}.
The utility of individual $i$ when joining group $g$ is specified
as $u_{ig}=\zeta_{g}+\delta_{1}^{u}z_{1,ig}^{u}+\delta_{2}^{u}z_{2,i}+\xi_{ig},$
where $\zeta_{g}$ is a group-specific fixed effect, $z_{1,ig}^{u}$
is a pair-specific characteristic that is i.i.d. across $i$ and $g$,
following $N(0,9)$, and $z_{2,i}$ is an individual-specific characteristic
that is i.i.d. following $N(2,1)$. We allow $z_{2,i}$ to be correlated
with $x_{i}$ such that $\text{Cov}(z_{2,i},x_{i})=2$. An individual
may also choose not to join any group, in which case their utility
is given by $u_{i0}=\xi_{i0}$. The unobserved preference $\xi_{ig}$
is i.i.d. across $i$ and $g=0,1,\dots,5$, following the type I extreme
value distribution. Individual $i$'s qualification for joining group
$g$ is specified as $v_{ig}=\delta_{1}^{v}z_{1,ig}^{v}+\delta_{2}^{v}z_{2,i}+\eta_{ig},$
where the pair-specific characteristic\textbf{ }$z_{1,ig}^{v}$ follows
$N(0,9)$ and the unobservable $\eta_{ig}$ follows $N(\epsilon_{i},1)$,
both i.i.d. across $i$ and $g$. Note that $\eta_{ig}$ is correlated
with $\epsilon_{i}$, thereby leading to endogenous groups. We set
the group fixed effects to\textbf{ $\zeta=(\zeta_{1},\zeta_{2},\zeta_{3},\zeta_{4},\zeta_{5})=(9,6,4,2,0)$
}and the parameter values $(\delta_{1}^{u},\delta_{2}^{u},\delta_{1}^{v},\delta_{2}^{v})=(-1,1,1,1)$.
Based on these model primitives, the stable groups are determined
through the individual-proposing Deferred-Acceptance algorithm \citep{gale_college_1962}.
The capacity constraints in all markets are binding.

Given the groups, we consider two specifications for the adjacency
matrix $\boldsymbol{w}$. In the first specification, we use group
averages that exclude the individual themselves. In the second specification,
we average over friends within a group, where the friendships are
generated independently with a constant probability of $0.5$. For
each specification, we estimate $\gamma$ using data from a single
market. The experiment procedure is repeated independently 200 times
and we report the average bias, standard errors, and root mean squared
errors (RMSE) of the 200 estimates of $\gamma$.\footnote{The group formation parameters are estimated using constrained maximum
simulated likelihood, where the cutoffs are treated as auxiliary parameters
that satisfy market-clearing conditions. See Supplemental Appendix
\ref{online:gf_estimate} for more details on the estimation method
and results.}

\subsection{Estimation Results}

Table \ref{tab:gamma_exog} presents the estimation results in the
absence of endogenous effects ($\gamma_{1}=0$). For the specification
with group averages (Panel A), the OLS estimate of $\gamma_{2}$ is
biased upward, indicating the presence of selection bias (Column 1).
Including group fixed effects (FE) does not mitigate this bias (Column
2). In Column 3, we control for a second-order polynomial series of
the elementary symmetric functions of utility difference and qualification
indices, as detailed in Section \ref{sec:exchange}. The sieve OLS
estimate of $\gamma$ is unbiased, demonstrating the effectiveness
of the selection correction. The specification with networks produces
similar results (Panel B). Both OLS and OLS with group FE yield biased
estimates (Columns 4 and 5), whereas sieve OLS provides unbiased estimates
(Column 6). In addition, sieve OLS exhibits smaller standard errors
and RMSE compared to OLS and OLS with group FE.

The estimation results in the presence of endogenous effects ($\gamma_{1}\neq0$)
are presented in Table \ref{tab:gamma_endo}. For the specification
with group averages (Panel A), OLS yields biased estimates of $\gamma_{1}$
and $\gamma_{2}$ (Column 1). Including group FE significantly exacerbates
the bias, likely due to multicollinearity between the group averages
and group dummies (Column 2). By applying the polynomial selection
correction, sieve OLS yields unbiased estimates (Column 3). The specification
with networks shows a similar pattern (Panel B). The estimates obtained
from OLS and OLS with group FE are biased (Columns 4 and 5), whereas
sieve OLS yields unbiased estimates (Column 6). Consistent with the
results in Table \ref{tab:gamma_exog}, sieve OLS has smaller root
mean squared errors compared to both OLS specifications. 

In sum, the simulation results indicate that including group FE is
insufficient to correct for selection bias, whereas sieve OLS provides
an effective approach for selection correction.

\section{\protect\label{sec:Empirical}Social Interactions in Chilean High
Schools}

\subsection{Data}

In this section, we apply our approach to analyze social interactions
among high school students in Chile. We use data from the SIMCE, provided
by the Agency for the Quality of Education in Chile \citep{simce}.\footnote{SIMCE is an abbreviation for Sistema de Medición de la Calidad de
la Educación (Education Quality Measurement System).} The SIMCE dataset provides information on math and language scores,
as well as information on parental education, parental income, and
other student and family characteristics collected through a parental
questionnaire sent home with students. To track a student's academic
performance in subsequent years, we merge the SIMCE dataset with Ministry
of Education administrative records, which provide detailed information
on educational attainment, ranging from high school completion to
college graduation.

Our sample consists of 6,872 tenth-grade students enrolled in 53 high
schools in the Biobío Region in 2006. Of these schools, 23 are public
and 30 are private. Public schools are required to accept any student
willing to enroll, while private schools can select students based
on their admission criteria.\footnote{Although public schools cannot select students, peer effect estimates
based on these schools may be biased because students' unobserved
preferences for schools ($\xi_{i}$) may be correlated with $\epsilon_{i}$.} Panel A of Table \ref{tab:summary} provides descriptive statistics
for the students in our sample. On average, their mothers have 9.58
years of education, 51\% of the students are female, and 87\% of the
students are enrolled in Fonasa.\footnote{Fonasa (Fondo Nacional de Salud) is Chile's tax-funded public health
insurance system that provides free or subsidized healthcare for those
unable to afford private insurance.} We use academic performance across various educational stages as
our measures of outcomes. Tenth-grade math and language scores come
from the SIMCE standardized test, measured as percentile ranks ranging
from 0 to 1. High school graduation is an indicator that equals 1
if a student completes high school on time at the end of twelfth grade
and 0 otherwise. In our sample, 72\% of the students graduate from
high school on time. Additionally, we consider four long-run outcomes:
post-secondary enrollment, college enrollment, post-secondary graduation,
and college graduation. Among the students in our sample, 63\% enroll
in post-secondary education, with 30\% attending college. Furthermore,
50\% complete post-secondary education, and 24\% graduate from college.

We construct the peer measures in equation (\ref{eq:linmean}) using
class averages that exclude the student themselves. It is documented
that a large fraction of peer effects in education arises at the classroom
level \citep[e.g.,][]{ammermueller2009peer}. In our sample, the average
class size (37) is approximately one quarter of the average school
size (130). Therefore, classroom averages provide a dense adjacency
matrix that satisfies Assumption \ref{ass:w}. Furthermore, this specification
allows us to distinguish peer influence from school effectiveness,
the latter of which is captured by school FE. 

Panels B and C of Table \ref{tab:summary} present summary statistics
for peer and school averages, respectively. If students were randomly
assigned to schools, school averages would be relatively homogeneous,
with standard deviations significantly smaller than those of individual
characteristics. However, the standard deviations of school averages
reported in Panel C remain comparable to those of individual characteristics
in Panel A, indicating sorting across schools.\footnote{If students are randomly assigned to schools, the variance of the
average of a characteristic $x_{i}$ in school $s$ is given by $\text{Var}(\bar{x}_{s})=\text{Var}(\frac{1}{n_{s}}\sum_{i=1}^{n_{s}}x_{i})=\frac{1}{n_{s}}\text{Var}(x_{i})$,
where $n_{s}$ represents the school size. In our sample, the average
school size is 130, so the ratio between the standard deviation of
an individual characteristic and the standard deviation of its school
average under random assignment should be approximately $\sqrt{130}\approx11.4$,
which is much larger than what we observe in the sample.} This finding aligns with the evidence on high levels of socioeconomic
segregation between schools \citep{valenzuela2019acrecentando}.

Table \ref{tab:btw_vs_within_sch} decomposes the total variances
of peer characteristics and outcomes into the variances within and
between schools. We find that 88\% and 91\% of the variation in peer
mother\textquoteright s education and peer Fonasa enrollment, respectively,
occur between schools, while the fraction is a bit lower for peer
fraction female (73\%). Peer outcomes have similar fractions of variation
occurring between schools: 73-74\% for test scores, 66-76\% for college-related
outcomes, and the lowest fraction of 48\% for high school graduation.
These results suggest substantial sorting across schools, as the majority
of the variation in peer characteristics and outcomes takes place
between schools rather than within them. In Section \ref{sec:test_random},
we provide further evidence that classroom assignment within schools
does not follow systematic patterns, implying that Assumption \ref{ass:adj_exog}
is satisfied in our context.

\subsection{Estimation Results}

We estimate the parameters in high school admissions using the MPEC
(Mathematical Programming with Equilibrium Constraints) algorithm
\citep{su2012constrained}, as described in Supplemental Appendix
\ref{online:gf_estimate}. We specify preferences and qualifications
as linear functions of individual characteristics (logarithm of family
income, eighth-grade composite score calculated as the sum of math
and language scores, mother's education, and distance to school) and
their interactions with school characteristics (logarithm of tuition,
average composite score, and average mother's education).\footnote{Our empirical model of high school admissions follows closely that
in \citet{HSS_twosided}. However, instead of the Bayesian approach
used in \citet{HSS_twosided}, we adopt a frequentist approach to
estimate the parameters. Our approach has two main advantages: it
reduces the computational burden and, more importantly, it facilitates
the adjustment of standard errors in the second stage, where we estimate
social interactions by frequentist methods.} Based on these estimates (provided in Supplemental Appendix \ref{online:gf_empirical}),
we estimate social interactions among tenth graders and compute standard
errors using the methodology developed in Section \ref{sec:Estimation}. 

Our basic specification of equation (\ref{eq:linmean}) includes individual
characteristics (female, mother's education, and Fonasa enrollment),
the peer averages of these characteristics, and the peer outcome.
We consider both short-run outcomes (tenth-grade math and language
scores and high school graduation) and long-run outcomes (post-secondary
and college enrollment and graduation). For each outcome, we estimate
equation (\ref{eq:linmean}) by (a) OLS, and then progressively control
for (b) school FE, and (c) selection correction. The selection correction
is constructed using a second-order polynomial series of the elementary
symmetric functions of utility difference and qualification indices
in high school admissions.

To account for institutional differences between public and private
schools, we relax the framework in Section \ref{sec:exchange} by
imposing exchangeability conditional on school type (public or private)
rather than across all schools. Conditional exchangeability by school
type requires only that schools of the same type have exchangeable
unobservables, while allowing for systematic differences in unobservables
between public and private schools.\footnote{Conditional exchangeability by school type implies that public and
private schools have two distinct selection functions, with each selection
function depending on two sets of elementary symmetric functions:
one for the indices of public schools and the other for the indices
of private schools. See Supplemental Appendix \ref{online:exch_by_type}
for details on sieve estimation in this case.} By leveraging symmetric functions, our approach significantly reduces
the dimensionality of sieve estimation, requiring only 35 basis functions
compared to 3,485 basis functions if symmetric functions were not
used (see Supplemental Appendix \ref{online:exch_by_type}). Additionally,
by exchangeability our approach can separate two sources of correlated
effects within a school: (i) school effectiveness, captured by the
school FE, and (ii) self-selection into schools, captured by the selection
correction.

Tables \ref{tab:est_short} and \ref{tab:est_long} report the estimates
for short-run and long-run outcomes, respectively. Across all outcomes,
simple OLS yields the largest estimates of endogenous peer effects
(coefficients of peer outcomes). Including school FE reduces these
OLS estimates by 12-46\%. Additionally controlling for our selection
correction further lowers the estimates; the selection-corrected estimates
of endogenous peer effects are 10-45\% lower than those obtained using
school FE alone. These findings demonstrate that the OLS estimates
of endogenous peer effects are biased upward due to sorting into schools.
While the inclusion of school FE mitigates selection bias to some
extent, it does not completely eliminate it. In contrast, our selection
correction effectively addresses the selection issue and provides
unbiased estimates of endogenous peer effects.

The selection-corrected estimates in Tables \ref{tab:est_short} and
\ref{tab:est_long} show that peer outcomes have positive and statistically
significant effects across all outcomes. In high school, one standard
deviation increases in peer math and language scores (0.204 and 0.185)
raise a student's math and language scores by 11.1 and 7.5 percentile
ranks, respectively. This peer effect persists in later outcomes with
a similar magnitude. For example, one standard deviation increases
in peer college enrollment and graduation (0.302 and 0.260) raise
a student's college enrollment and graduation rates by 14.2 and 10.6
percentage points, respectively. These are substantial effects compared
to the mean rates (30.3\% and 23.8\%).

In addition, Tables \ref{tab:est_short} and \ref{tab:est_long} provide
similar results for exogenous peer effects (coefficients of peer characteristics).
Across all outcomes, peer mother's education has positive and significant
effects in OLS regressions with school FE. The estimates decrease
when we correct for selection, suggesting upward selection bias. The
selection-corrected estimates remain positive and significant for
all outcomes (except for tenth-grade math and language scores). Peer
fraction female shows a similar pattern across short-run outcomes
and college enrollment and graduation: the selection correction further
reduces the estimates compared to OLS with school FE. The selection-corrected
estimates are positive and significant for tenth-grade math score,
high school graduation, and post-secondary enrollment and graduation,
albeit with a negative effect for college graduation. This finding
is consistent with the documented evidence that the fraction of female
peers is positively correlated with academic achievement \citep{Sacerdote2011}.
For peer Fonasa enrollment, the selection correction also lowers the
OLS estimates with school FE, leading to larger negative effects.\footnote{This suggests that peer Fonasa enrollment may be positively correlated
with the selection correction, conditional on other controls. } As expected, the selection-corrected estimates are negative and significant
for all outcomes (except for tenth-grade math and language scores).

In a nutshell, OLS regressions with school FE only partially account
for selection and tend to overestimate peer effects. The F-statistics
for the selection correction are statistically significant across
all outcomes (except for college enrollment), underscoring the importance
of our selection correction approach. The selection-corrected estimates
provide evidence that both peer outcomes and peer characteristics
have significant effects on tenth graders.

\subsection{Peer Influence, School Effectiveness, and Self-Selection}

In our framework, a student's outcome are shaped by three factors
related to the school they attend: (i) the influence of their peers
within the school, (ii) the academic effectiveness of the school,
including aspects such as teacher quality and school resources and
investments; and (iii) unobserved factors that affect self-selection
into the school. In this section, we aim to answer the following questions:
To what extent can a student's academic performance be attributed
to peer influence, school effectiveness, and self-selection? How would
the evaluation of peer influence and school effectiveness be impacted
if self-selection were ignored?

Using the selection-corrected estimates, we calculate the portion
of a student's predicted outcome that depends on the school they attend.
This portion is defined as the sum of three components: (i) peer influence,
calculated as $\sum_{j=1}^{n}w_{ij}y_{j}\hat{\gamma}_{1}+\sum_{j=1}^{n}w_{ij}x_{j}\hat{\gamma}_{2}$,
where $\hat{\gamma}_{1}$ and $\hat{\gamma}_{2}$ are the estimates
of $\gamma_{1}$ and $\gamma_{2}$, (ii) school effectiveness, represented
by the estimated school FE for the school the student attends, and
(iii) self-selection, calculated as $\hat{\lambda}(\hat{\pi}_{i})$,
the estimated polynomial series used for the selection correction.\footnote{Our measure of peer influence captures both peer quality ($\sum_{j=1}^{n}w_{ij}y_{j}$
and $\sum_{j=1}^{n}w_{ij}x_{j}$) and peer effects ($\gamma_{1}$
and $\gamma_{2}$).} We then decompose the total variance of the school-dependent predicted
outcome into the variances of peer influence, school effectiveness,
and self-selection, as well as the covariances between any two of
these components.

The decomposition results are presented in Table \ref{tab:yvar_decompose}.
We find that self-selection accounts for the largest fraction of the
total variation in school-dependent predicted outcomes, except for
post-graduate enrollment and college graduation, where it accounts
for the second largest fraction. Peer influence and school effectiveness
contribute comparably to short-run outcomes; however, peer influence
becomes more important for long-run outcomes.\footnote{There is evidence that peer interactions during adolescence have a
lasting impact in later years. For example, \citet{LlerasMuney2024}
found that friendships formed during adolescence have significant
influence on labor market outcomes.} In contrast, school effectiveness accounts for the smallest fraction
of the explained variation in long-run outcomes. These findings underscore
the critical role of self-selection in explaining a student's predicted
outcomes. Failing to account for self-selection can lead to significant
bias when measuring the contributions of peer influence and school
effectiveness.

The last three rows of Table \ref{tab:yvar_decompose} present covariances
among peer influence, school effectiveness, and self-selection. The
correlation between self-selection and peer influence is generally
positive, suggesting that students and their peers are sorted into
schools in a positively assortative manner. In contrast, the correlation
between self-selection and school effectiveness is negative, indicating
that more selective schools may provide lower value-added. This result
aligns with the existing literature that documents limited school
effectiveness for highly selective schools \citep{abdulkadirouglu2014elite,dobbie2014impact}.

To understand the impact of ignoring selection on the evaluation of
peer influence and school effectiveness, we further examine these
correlations at the school level. Figure \ref{fig:peer_selection}
plots the average selection in each school against the average peer
influence in the school, with peer influence calculated using both
biased estimates (OLS with school FE) and unbiased estimates (sieve
OLS). Both biased and unbiased estimates demonstrate that peer influence
is positively correlated with selection at the school level, consistent
with the result in Table \ref{tab:yvar_decompose}. Figure \ref{fig:school_selection}
plots the average selection in each school against school effectiveness,
with school effectiveness similarly calculated using both biased estimates
(OLS with school FE) and unbiased estimates (sieve OLS).\footnote{Since we measure school effectiveness using school fixed effects,
we include the dummies for all schools in our specifications of school
FE and sieve OLS. To achieve identification, we drop the constant
term and normalize $\mathbb{E}[\epsilon_{i}]=0$. By construction,
$\mathbb{E}[\nu_{i}|\pi_{i}]=\mathbb{E}[\epsilon_{i}-\lambda(\pi_{i})|\pi_{i}]=\lambda(\pi_{i})-\lambda(\pi_{i})=0$
and therefore $\mathbb{E}[\lambda(\pi_{i})]=\mathbb{E}[\epsilon_{i}]-\mathbb{E}[\nu_{i}]=0$
by the law of iterated expectations. In the sieve OLS specification,
we implement the normalization by demeaning all the basis functions
in a sieve approximation.} The biased estimates suggest a misleading positive correlation between
average selection and school effectiveness, potentially leading to
the incorrect conclusion that more selective schools are more effective.
In contrast, the unbiased estimates reveal a negative correlation
between average selection and school effectiveness, reaffirming the
earlier finding that more selective schools tend to provide lower
value-added. 

Because peer influence and school effectiveness are correlated with
selection in opposite directions, and they are also correlated with
each other in varying directions (Table \ref{tab:yvar_decompose}),
it is challenging to predict the directions of selection bias in these
variables. Instead, we directly compare their biased and unbiased
estimates in Figures \ref{fig:peer_dist} and \ref{fig:school_dist}.
Figure \ref{fig:peer_dist} plots the distributions of school-average
peer influence, calculated using both biased and unbiased estimates
across all outcomes. Compared to the unbiased distributions, the biased
ones are shifted to the right, indicating an overestimation of peer
influence, and exhibit increased dispersion across schools. Peer influence
is overestimated in all schools, with larger upward bias for more
selective schools (Figure \ref{fig:peer_selection}). Figure \ref{fig:school_dist}
plots the distributions of school effectiveness, calculated using
both biased and unbiased estimates. The biased estimates of school
effectiveness tend to be smaller and less dispersed across schools.
School effectiveness is primarily underestimated for less selective
schools (Figure \ref{fig:school_selection}).

These findings highlight the significant impact of selection bias
on the evaluation of peer influence and school effectiveness. Addressing
the selection problem is critical for deriving accurate implications
for policymakers. Specifically, school performance should not be assessed
solely based on students' outcomes, as these outcomes are heavily
influenced by self-selection. While peer effects are substantial,
they are often overestimated when selection is not adequately addressed,
and the overestimation is more pronounced in highly selective schools.
Moreover, selection tends to distort the evaluation of school effectiveness,
primarily by underestimating the value-added of less selective schools,
which often serve less advantaged student populations. Correcting
for selection is essential to accurately recognize school performance
and ensure equitable resource allocation.

\subsection{\protect\label{sec:test_random}Testing for Random Assignment into
Classrooms}

There may be potential concerns about additional sorting at the classroom
level. To check whether classroom assignment is random, we apply a
regression-based test originally proposed by \citet{sacerdote2001peer}
and later refined with size correction by \citet{jochmans2023testing}.
The underlying intuition of the test is straightforward: under random
assignment, a student's characteristics should not be correlated with
the average characteristics of their classroom peers, after controlling
for sorting across schools.

We apply \citet{jochmans2023testing}'s test to a variety of student
characteristics that schools may consider when assigning students
into classrooms: female, mother's education, Fonasa enrollment, family
income, and test scores. For each characteristic, we perform (i) a
baseline test that controls for school FE, and (ii) a selection-corrected
test that controls for both school FE and selection correction. Table
\ref{tab:test} reports the p-values of these tests. In the baseline
tests, the p-values for mother's education and eighth-grade composite
score are significant, suggesting potential sorting at the classroom
level based on these variables. However, after controlling for our
selection correction, none of the characteristics yield significant
p-values. These results provide evidence that classroom assignment
is random, as long as that selection into schools is properly controlled
for. 

\section{\protect\label{sec:Conclusion}Conclusion}

In this paper, we study social interactions in endogenous groups.
We develop a model of group formation to analyze how individuals select
into groups and how to account for the impact of group selection.
Our model accommodates two-sided decision-making, where individuals
choose groups based on their preferences, while groups admit individuals
based on qualifications until reaching their capacities. This framework
mirrors many admission processes in the real world.

We make significant contributions to the literature in several aspects.
First, we characterize group formation using a two-sided many-to-one
matching model with nonparametric unobservables. Building on this
framework, we demonstrate that endogenous selection into groups leads
to selection bias in the estimation of peer effects. Second, we propose
innovative methods to address the dimensionality challenges in correcting
for this selection bias. Specifically, we employ the limiting approximation
of a market as it grows large to reduce the high dimensionality due
to equilibrium effects. Additionally, we impose an exchangeability
assumption, under which the selection bias can be expressed through
a group-invariant selection function that remains tractable even with
a moderately large number of groups. Third, we propose a sieve OLS
estimator for the social interaction parameters, which achieves $\sqrt{n}$-consistency
and asymptotic normality. These asymptotic properties are established
using novel asymptotic methods under additional conditions (e.g.,
a dense adjacency matrix, limited network dependency). We verify that
these conditions hold for commonly used specifications of adjacency
matrices, including group averages and dyadic networks with fixed
effects.

We apply our approach to investigate social interactions among tenth
graders in Chile. We find that including school fixed effects is insufficient
to account for selection into high schools, whereas our selection
correction method yields unbiased estimates of peer effects. A variance
decomposition shows that self-selection accounts for the largest share
of explained variance in most outcomes, highlighting the importance
of properly accounting for endogenous selection in school evaluation.
Ignoring selection into high schools tends to overestimate peer influence,
particularly in highly selective schools. Moreover, the selection-corrected
estimates suggest that more selective schools may actually provide
lower value-added. These results underscore the significant impact
of selection bias on the evaluation of peer influence and school effectiveness.
Properly correcting for selection is crucial for policymakers to achieve
equitable resource allocation.

\begin{spacing}{1.2} \bibliographystyle{ecta}
\bibliography{endo_group}
\end{spacing} 

\newpage\vspace*{0pt}
\begin{figure}[H]
\begin{centering}
\caption{\protect\label{fig:peer_selection}Relationships Between Selection
and Peer Influence}
\includegraphics[totalheight=0.8\textheight]{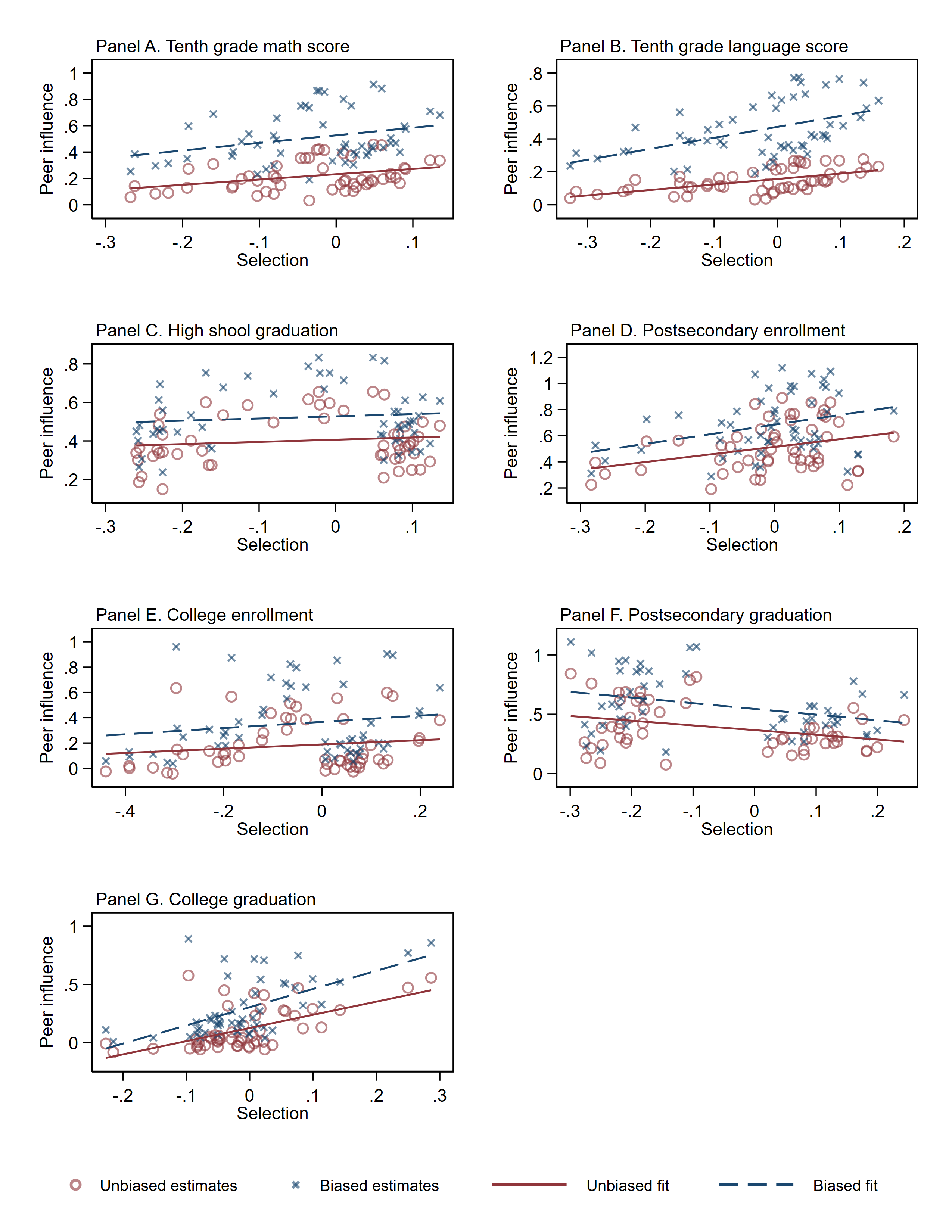}
\par\end{centering}
\begin{tablenotes}\item \scriptsize
\textit{Notes}: School average peer influence against selection across different educational outcomes. For each school, peer influence is estimated using both OLS models with school FE (shown in blue crosses and fitted dashed line) and sieve OLS models with school FE (shown in red circles and fitted solid lines). The biased estimates (blue) come from regressions of outcomes on school dummies and a set of individual student characteristics (female, mother's education, and Fonasa enrollment) along with peers' average outcomes and attributes. The unbiased estimates (red) come from sieve OLS models that extend OLS models by adding polynomial basis functions of elementary symmetric group formation indices up to order 2. Mean selection is calculated as the school-level average of individual selection $\hat{\lambda}(\hat{\pi}_i)$. Peer influence at the school level is calculated as the school-level average of individual peer influence $\sum_{j=1}^{n}w_{ij}y_{j}\hat{\gamma}_{1}+\sum_{j=1}^{n}w_{ij}x_{j}\hat{\gamma}_{2}$. Each panel presents scatter plots and fitted linear relationships between peer influence and mean selection. 
\end{tablenotes}
\end{figure}
\begin{figure}[H]
\begin{centering}
\caption{\protect\label{fig:school_selection}Relationships Between Selection
and School Effectiveness}
\includegraphics[totalheight=0.8\textheight]{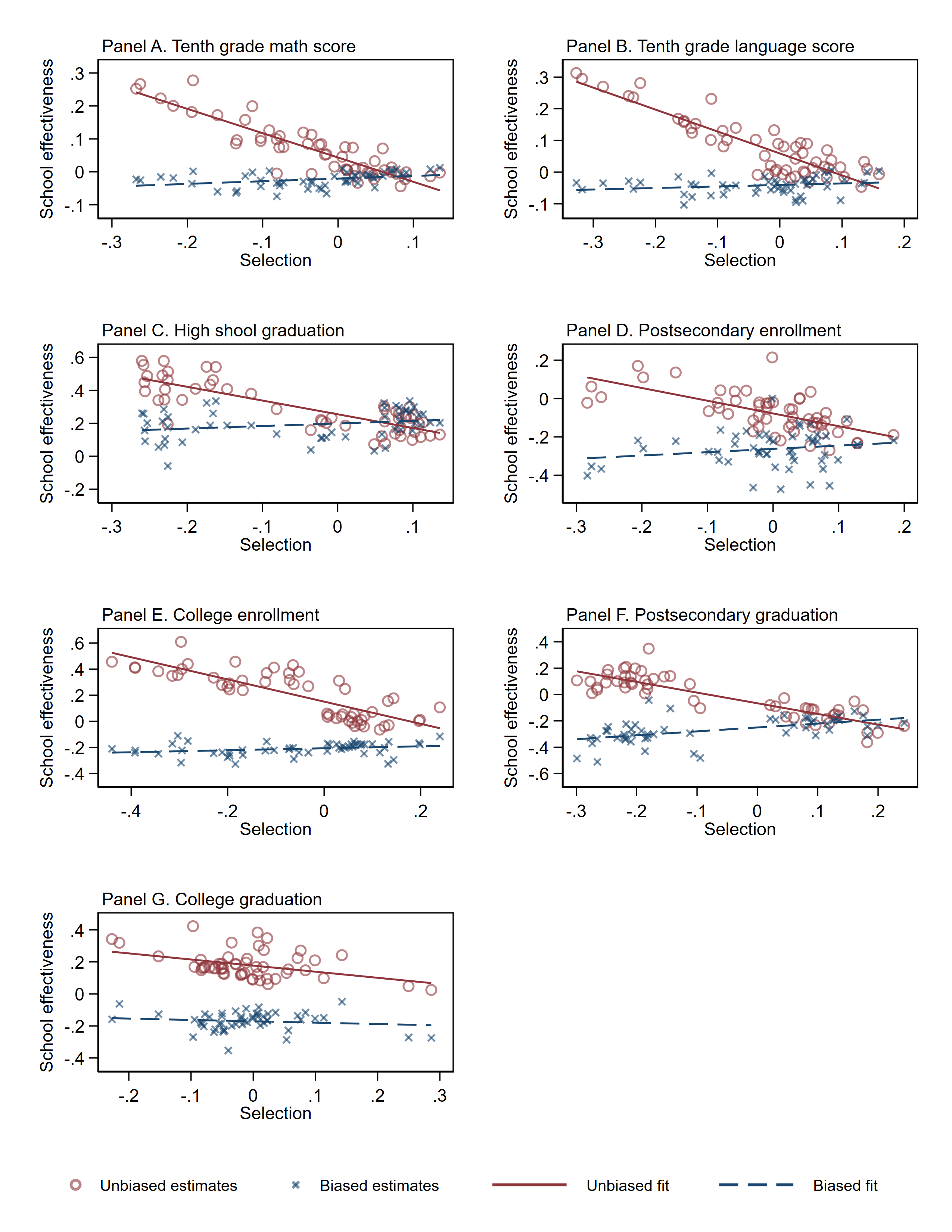}
\par\end{centering}
\begin{tablenotes}\item \footnotesize
\textit{Notes}: School effectiveness against school mean selection across different educational outcomes. For each school, school effectiveness is estimated using both OLS models with school FE (shown in blue crosses and fitted dashed lines) and sieve OLS models with school FE (shown in red circles and fitted solid lines), with model specifications detailed in Figure \ref{fig:peer_selection}. Each panel presents scatter plots and fitted linear relationships between school effectiveness and mean selection. 
\end{tablenotes}
\end{figure}
\begin{figure}[H]
\begin{centering}
\caption{\protect\label{fig:peer_dist}Distributions of Peer Influence: Biased
versus Unbiased Estimates}
\includegraphics[totalheight=0.8\textheight]{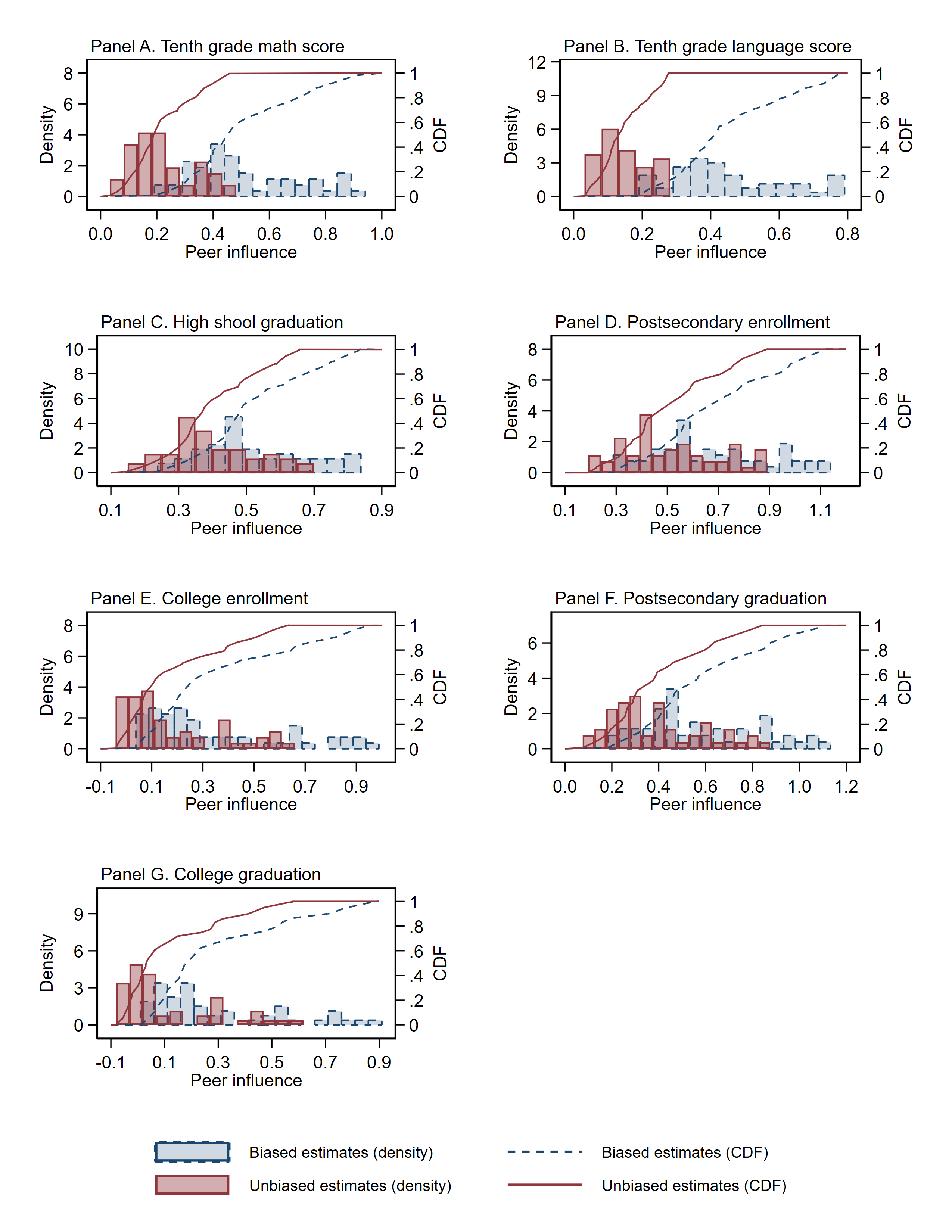}
\par\end{centering}
\begin{tablenotes}\item \footnotesize
\textit{Notes}: Histograms (bars, left y-axis) and cumulative distribution functions (lines, right y-axis) of school-average peer influence across outcomes, calculated using both biased estimates (OLS with school FE, shown in blue dashed bars/lines) and unbiased estimates (sieve OLS with school FE, shown in red solid bars/lines).
\end{tablenotes}
\end{figure}
\begin{figure}[H]
\begin{centering}
\caption{\protect\label{fig:school_dist}Distributions of School Effectiveness:
Biased versus Unbiased Estimates}
\includegraphics[totalheight=0.8\textheight]{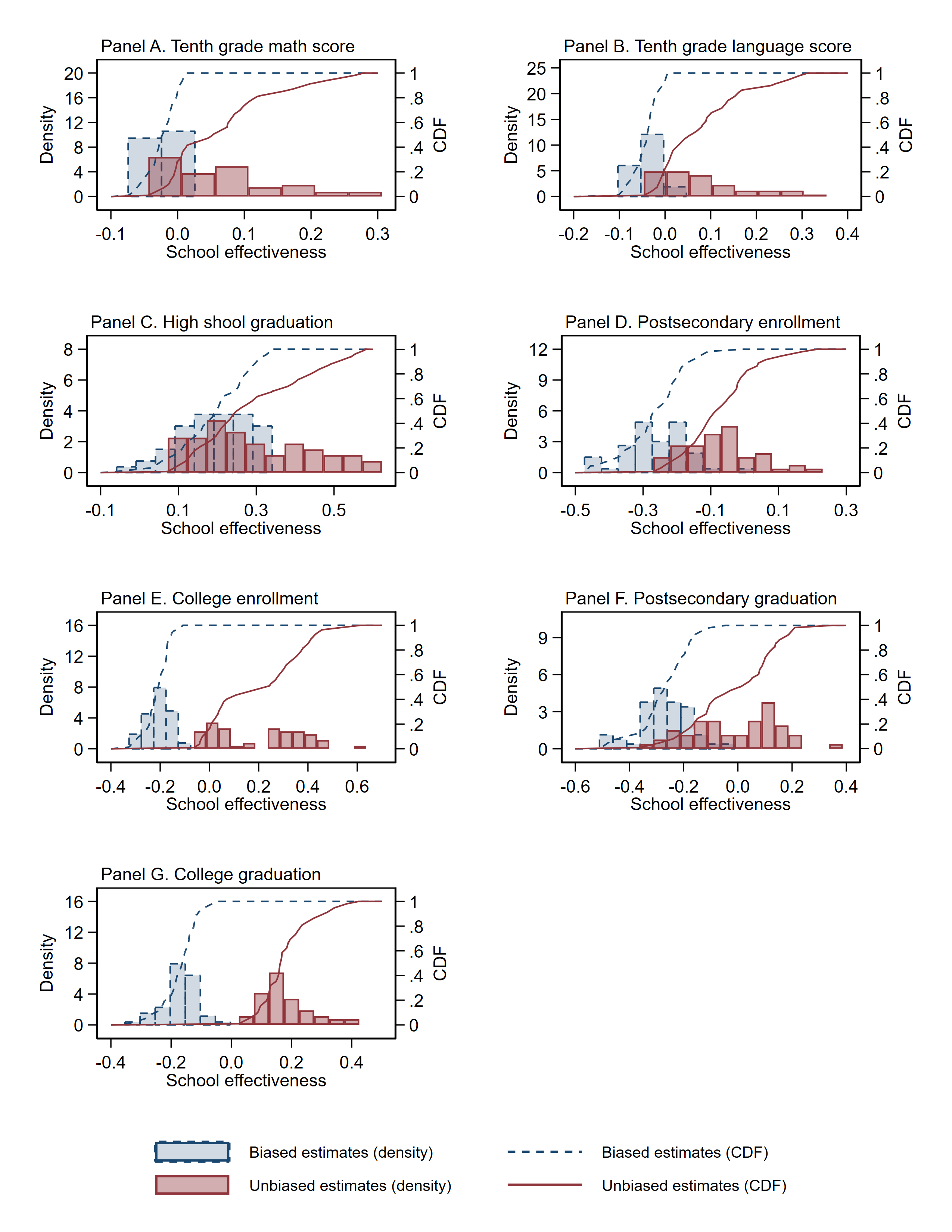}
\par\end{centering}
\begin{tablenotes}\item \footnotesize
\textit{Notes}: Histograms (bars, left y-axis) and cumulative distribution functions (lines, right y-axis) of school effectiveness across outcomes, calculated using both biased estimates (OLS with school FE, shown in blue dashed bars/lines) and unbiased estimates (sieve OLS with school FE, shown in red solid bars/lines).
\end{tablenotes}
\end{figure}
\begin{table}[H]
\centering{}\caption{\protect\label{tab:gamma_exog}Simulation Estimates: Without Endogenous
Effects ($\gamma_{1}=0$)}
{\footnotesize{}%
\begin{tabular}{ll>{\centering}p{1.45cm}>{\centering}p{1.6cm}>{\centering}p{1.65cm}>{\centering}p{1.65cm}>{\centering}p{1.65cm}>{\centering}p{1.65cm}}
\hline 
\hline &  & \multicolumn{3}{c}{{\footnotesize\textit{Panel A: Pure Groups}}} & \multicolumn{3}{c}{{\footnotesize\textit{Panel B: Networks}}}\tabularnewline
\cline{3-8}
 &  & {\footnotesize (1)}{\footnotesize\par}

{\footnotesize OLS} & {\footnotesize (2)}{\footnotesize\par}

{\footnotesize OLS} & {\footnotesize (3)}{\footnotesize\par}

{\footnotesize Sieve OLS} & {\footnotesize (4)}{\footnotesize\par}

{\footnotesize OLS} & {\footnotesize (5)}{\footnotesize\par}

{\footnotesize OLS} & {\footnotesize (6)}{\footnotesize\par}

{\footnotesize Sieve OLS}\tabularnewline
\hline 
{\footnotesize$\gamma_{2}$} & {\footnotesize Bias} & {\footnotesize 0.444} & {\footnotesize 0.555} & {\footnotesize 0.002} & {\footnotesize 0.295} & {\footnotesize 0.140} & {\footnotesize 0.001}\tabularnewline
 & {\footnotesize Std. Dev.} & {\footnotesize 0.139} & {\footnotesize 1.205} & {\footnotesize 0.080} & {\footnotesize 0.071} & {\footnotesize 0.110} & {\footnotesize 0.061}\tabularnewline
 & {\footnotesize RMSE} & {\footnotesize 0.465} & {\footnotesize 1.324} & {\footnotesize 0.080} & {\footnotesize 0.304} & {\footnotesize 0.178} & {\footnotesize 0.061}\tabularnewline
{\footnotesize$\gamma_{3}$} & {\footnotesize Bias} & {\footnotesize -0.006} & {\footnotesize -0.006} & {\footnotesize 0.000} & {\footnotesize -0.005} & {\footnotesize -0.006} & {\footnotesize 0.000}\tabularnewline
 & {\footnotesize Std. Dev.} & {\footnotesize 0.005} & {\footnotesize 0.006} & {\footnotesize 0.004} & {\footnotesize 0.005} & {\footnotesize 0.005} & {\footnotesize 0.004}\tabularnewline
 & {\footnotesize RMSE} & {\footnotesize 0.007} & {\footnotesize 0.008} & {\footnotesize 0.004} & {\footnotesize 0.007} & {\footnotesize 0.008} & {\footnotesize 0.004}\tabularnewline
\hline 
\multicolumn{2}{l}{{\footnotesize Selection Correction}} & {\footnotesize No} & {\footnotesize Group FE} & {\footnotesize Sieve} & {\footnotesize No} & {\footnotesize Group FE} & {\footnotesize Sieve}\tabularnewline
\hline 
\end{tabular}}\begin{tablenotes}[flushleft]\footnotesize \item 
\textit{Notes}: Average biases, standard deviations, and root mean squared errors of the social interaction parameter estimates obtained from 200 Monte Carlo samples in the absence of endogenous interactions. Each Monte Carlo sample consists of 5 groups and 1,680 individuals. Panel A presents the results for group averages that exclude oneself, and Panel B presents the results for within-group networks. Columns 1 and 4 estimate the parameters using OLS. Columns 2 and 5 estimate the parameters using OLS with group fixed effects. Columns 3 and 6 estimate the parameters using sieve OLS, where the basis functions are specified by a polynomial series of symmetric indices of utility differences and qualifications up to the second order.
\end{tablenotes}
\end{table}
\begin{table}[H]
\begin{centering}
\caption{\protect\label{tab:gamma_endo}Simulation Estimates: With Endogenous
Effects ($\gamma_{1}\protect\neq0$)}
{\footnotesize{}%
\begin{tabular}{ll>{\centering}p{1.45cm}>{\centering}p{1.6cm}>{\centering}p{1.65cm}>{\centering}p{1.65cm}>{\centering}p{1.65cm}>{\centering}p{1.65cm}}
\hline 
\hline &  & \multicolumn{3}{c}{{\footnotesize\textit{Panel A: Pure Groups}}} & \multicolumn{3}{c}{{\footnotesize\textit{Panel B: Networks}}}\tabularnewline
\cline{3-8}
 &  & {\footnotesize (1)}{\footnotesize\par}

{\footnotesize OLS} & {\footnotesize (2)}{\footnotesize\par}

{\footnotesize OLS} & {\footnotesize (3)}{\footnotesize\par}

{\footnotesize Sieve OLS} & {\footnotesize (4)}{\footnotesize\par}

{\footnotesize OLS} & {\footnotesize (5)}{\footnotesize\par}

{\footnotesize OLS} & {\footnotesize (6)}{\footnotesize\par}

{\footnotesize Sieve OLS}\tabularnewline
\hline 
$\gamma_{1}$ & {\footnotesize Bias} & {\footnotesize 0.377} & {\footnotesize -308.214} & {\footnotesize -0.046} & {\footnotesize 0.136} & {\footnotesize -0.845} & {\footnotesize 0.003}\tabularnewline
 & {\footnotesize Std. Dev.} & {\footnotesize 0.220} & {\footnotesize 3.572} & {\footnotesize 0.256} & {\footnotesize 0.035} & {\footnotesize 0.398} & {\footnotesize 0.029}\tabularnewline
 & {\footnotesize RMSE} & {\footnotesize 0.436} & {\footnotesize 308.235} & {\footnotesize 0.259} & {\footnotesize 0.140} & {\footnotesize 0.933} & {\footnotesize 0.029}\tabularnewline
{\footnotesize$\gamma_{2}$} & {\footnotesize Bias} & {\footnotesize -1.387} & {\footnotesize 302.827} & {\footnotesize 0.222} & {\footnotesize -0.192} & {\footnotesize 0.827} & {\footnotesize -0.01}\tabularnewline
 & {\footnotesize Std. Dev.} & {\footnotesize 1.073} & {\footnotesize 4.096} & {\footnotesize 1.229} & {\footnotesize 0.119} & {\footnotesize 0.389} & {\footnotesize 0.099}\tabularnewline
 & {\footnotesize RMSE} & {\footnotesize 1.752} & {\footnotesize 302.855} & {\footnotesize 1.246} & {\footnotesize 0.226} & {\footnotesize 0.913} & {\footnotesize 0.099}\tabularnewline
{\footnotesize$\gamma_{3}$} & {\footnotesize Bias} & {\footnotesize -0.010} & {\footnotesize -0.012} & {\footnotesize 0.001} & {\footnotesize -0.008} & {\footnotesize -0.007} & {\footnotesize 0.000}\tabularnewline
 & {\footnotesize Std. Dev.} & {\footnotesize 0.005} & {\footnotesize 0.007} & {\footnotesize 0.005} & {\footnotesize 0.005} & {\footnotesize 0.005} & {\footnotesize 0.004}\tabularnewline
 & {\footnotesize RMSE} & {\footnotesize 0.011} & {\footnotesize 0.014} & {\footnotesize 0.005} & {\footnotesize 0.009} & {\footnotesize 0.009} & {\footnotesize 0.004}\tabularnewline
\hline 
\multicolumn{2}{l}{{\footnotesize Selection Correction}} & {\footnotesize No} & {\footnotesize Group FE} & {\footnotesize Sieve} & {\footnotesize No} & {\footnotesize Group FE} & {\footnotesize Sieve}\tabularnewline
\hline 
\end{tabular}}{\footnotesize\par}
\par\end{centering}
\begin{tablenotes}[flushleft]\footnotesize \item 
\textit{Notes}: Average biases, standard deviations, and root mean squared errors of the social interaction parameter estimates obtained from 200 Monte Carlo samples in the presence of endogenous interactions. Each Monte Carlo sample consists of 5 groups and 1,680 individuals. Panel A presents the results for group averages that exclude oneself, and Panel B presents the results for within-group networks. Columns 1 and 4 estimate the parameters using OLS. Columns 2 and 5 estimate the parameters using OLS with group fixed effects. Columns 3 and 6 estimate the parameters using sieve OLS, where the basis functions are specified by a polynomial series of symmetric indices of utility differences and qualifications up to the second order.
\end{tablenotes}
\end{table}
\begin{table}[H]
\begin{centering}
{\footnotesize\caption{\protect\label{tab:summary}Descriptive Statistics for Tenth Graders
in Biobío Region in 2006}
}{\footnotesize{}%
\begin{tabular}{lcccccccc}
\hline 
\hline & \multicolumn{2}{c}{\begin{cellvarwidth}[t]
\centering
{\footnotesize Panel A. Individual}{\footnotesize\par}

{\footnotesize ($N=6,872$)}
\end{cellvarwidth}} &  & \multicolumn{2}{c}{\begin{cellvarwidth}[t]
\centering
{\footnotesize Panel B. Peer}{\footnotesize\par}

{\footnotesize ($N=6,872$)}
\end{cellvarwidth}} &  & \multicolumn{2}{c}{\begin{cellvarwidth}[t]
\centering
{\footnotesize Panel C. School}{\footnotesize\par}

{\footnotesize ($N=53$)}
\end{cellvarwidth}}\tabularnewline
\cline{2-9}
 & {\footnotesize Mean} & {\footnotesize Std. Dev.} &  & {\footnotesize Mean} & {\footnotesize Std. Dev.} &  & {\footnotesize Mean} & {\footnotesize Std. Dev.}\tabularnewline
\hline 
{\footnotesize Female} & {\footnotesize 0.508} & {\footnotesize 0.500} &  & {\footnotesize 0.508} & {\footnotesize 0.232} &  & {\footnotesize 0.502} & {\footnotesize 0.168}\tabularnewline
{\footnotesize Mother's education} & {\footnotesize 9.580} & {\footnotesize 3.396} &  & {\footnotesize 9.580} & {\footnotesize 1.978} &  & {\footnotesize 9.924} & {\footnotesize 2.377}\tabularnewline
{\footnotesize Fonasa} & {\footnotesize 0.869} & {\footnotesize 0.338} &  & {\footnotesize 0.869} & {\footnotesize 0.162} &  & {\footnotesize 0.822} & {\footnotesize 0.218}\tabularnewline
{\footnotesize Tenth-grade math score} & {\footnotesize 0.499} & {\footnotesize 0.288} &  & {\footnotesize 0.499} & {\footnotesize 0.204} &  & {\footnotesize 0.501} & {\footnotesize 0.199}\tabularnewline
{\footnotesize Tenth-grade language score} & {\footnotesize 0.500} & {\footnotesize 0.289} &  & {\footnotesize 0.500} & {\footnotesize 0.185} &  & {\footnotesize 0.500} & {\footnotesize 0.175}\tabularnewline
{\footnotesize High school graduation} & {\footnotesize 0.718} & {\footnotesize 0.450} &  & {\footnotesize 0.718} & {\footnotesize 0.188} &  & {\footnotesize 0.698} & {\footnotesize 0.170}\tabularnewline
{\footnotesize Postsecondary enrollment} & {\footnotesize 0.631} & {\footnotesize 0.483} &  & {\footnotesize 0.631} & {\footnotesize 0.257} &  & {\footnotesize 0.636} & {\footnotesize 0.254}\tabularnewline
{\footnotesize College enrollment} & {\footnotesize 0.303} & {\footnotesize 0.459} &  & {\footnotesize 0.303} & {\footnotesize 0.302} &  & {\footnotesize 0.337} & {\footnotesize 0.301}\tabularnewline
{\footnotesize Postsecondary graduation} & {\footnotesize 0.502} & {\footnotesize 0.500} &  & {\footnotesize 0.502} & {\footnotesize 0.260} &  & {\footnotesize 0.499} & {\footnotesize 0.258}\tabularnewline
{\footnotesize College graduation} & {\footnotesize 0.238} & {\footnotesize 0.426} &  & {\footnotesize 0.238} & {\footnotesize 0.260} &  & {\footnotesize 0.265} & {\footnotesize 0.260}\tabularnewline
{\footnotesize Class size/School size} & \multicolumn{1}{c}{-} & - &  & {\footnotesize 37} & {\footnotesize 6.455} &  & {\footnotesize 130} & {\footnotesize 107.189}\tabularnewline
\hline 
\end{tabular}}{\footnotesize\par}
\par\end{centering}
\begin{tablenotes}[flushleft] \footnotesize \item
\textit{Notes}:Descriptive statistics for tenth graders in the Biobío Region in 2006. Both tenth grade math and language scores come from SIMCE tests and are expressed as percentile ranks (0-1). Mother's education represents years of schooling. Fonasa indicates enrollment in the Fondo Nacional de Salud (FONASA), Chile's public health insurance system funded by taxes, which provides free or subsidized healthcare for those unable to afford private insurance. High school graduation equals one if a student completes high school at the end of 12th grade. Postsecondary/College enrollment and graduation are binary indicators measuring enrollment in and completion of higher education. The sample size for tenth grade scores is 6,073 students. The mean and standard deviation of class size are calculated using 184 classes. 
\end{tablenotes}
\end{table}
\begin{table}[H]
\centering
{\small\caption{\protect\label{tab:btw_vs_within_sch}Variance Decomposition of Peer
Characteristics and Outcomes}
}{\small{}%
\begin{tabular}{l>{\centering}p{0.1\textwidth}>{\centering}p{0.1\textwidth}>{\centering}p{0.1\textwidth}>{\centering}p{0.12\textwidth}}
\hline 
\hline & {\footnotesize Total variance} & {\footnotesize Within schools} & {\footnotesize Between schools} & {\footnotesize\% Between schools}\tabularnewline
\hline 
{\footnotesize Peer female} & {\footnotesize 0.054} & {\footnotesize 0.015} & {\footnotesize 0.039} & {\footnotesize 72.7}\tabularnewline
{\footnotesize Peer mother's education} & {\footnotesize 3.912} & {\footnotesize 0.457} & {\footnotesize 3.455} & {\footnotesize 88.3}\tabularnewline
{\footnotesize Peer Fonasa} & {\footnotesize 0.026} & {\footnotesize 0.003} & {\footnotesize 0.024} & {\footnotesize 90.5}\tabularnewline
{\footnotesize Peer tenth-grade math score} & {\footnotesize 0.037} & {\footnotesize 0.013} & {\footnotesize 0.027} & {\footnotesize 73.3}\tabularnewline
{\footnotesize Peer tenth-grade language score} & {\footnotesize 0.030} & {\footnotesize 0.010} & {\footnotesize 0.022} & {\footnotesize 74.0}\tabularnewline
{\footnotesize Peer high school graduation} & {\footnotesize 0.035} & {\footnotesize 0.018} & {\footnotesize 0.017} & {\footnotesize 48.4}\tabularnewline
{\footnotesize Peer postsecondary enrollment} & {\footnotesize 0.066} & {\footnotesize 0.022} & {\footnotesize 0.044} & {\footnotesize 66.0}\tabularnewline
{\footnotesize Peer college enrollment} & {\footnotesize 0.091} & {\footnotesize 0.023} & {\footnotesize 0.068} & {\footnotesize 74.7}\tabularnewline
{\footnotesize Peer postsecondary graduation} & {\footnotesize 0.067} & {\footnotesize 0.021} & {\footnotesize 0.046} & {\footnotesize 68.1}\tabularnewline
{\footnotesize Peer college graduation} & {\footnotesize 0.068} & {\footnotesize 0.016} & {\footnotesize 0.052} & {\footnotesize 76.0}\tabularnewline
\hline 
\end{tabular}}{\small\par}

\begin{tablenotes}[flushleft] \footnotesize \item
\textit{Notes}: Variance decomposition of peer characteristics and outcomes into within-school and between-school components. Let $x_{i}$ represent a peer characteristic or outcome of individual $i$, measured by the classroom average excluding $i$. Define $\bar{x}_{s}$ as the average of $x_{i}$ in school $s$, and $\bar{x}$ as the average of $x_{i}$ in the entire sample. The total variance of $x_{i}$ in the sample can be decomposed as $\frac{1}{n}\sum_{s=1}^{S}\sum_{i=1}^{n_{s}}(x_{i}-\bar{x})^{2}=\frac{1}{n}\sum_{s=1}^{S}\sum_{i=1}^{n_{s}}(x_{i}-\bar{x}_{s})^{2}+\frac{1}{n}\sum_{s=1}^{S}n_{s}(\bar{x}_{s}-\bar{x})^{2}$,
where the first term represents the within-school component and the second term captures the between-school component.
\end{tablenotes}
\end{table}
\begin{sidewaystable}[H]
\centering{}{\small\caption{\protect\label{tab:est_short}Peer Effect Estimates: Short-Run Outcomes}
}{\footnotesize{}%
\begin{tabular*}{1\textwidth}{@{\extracolsep{\fill}}lccccccccc}
\hline 
\hline & \multicolumn{3}{c}{{\footnotesize Tenth-grade math score}} & \multicolumn{3}{c}{{\footnotesize Tenth-grade language score}} & \multicolumn{3}{c}{{\footnotesize High school graduation}}\tabularnewline
\cline{2-10}
 & {\footnotesize (1)} & {\footnotesize (2)} & {\footnotesize (3)} & {\footnotesize (4)} & {\footnotesize (5)} & {\footnotesize (6)} & {\footnotesize (7)} & {\footnotesize (8)} & {\footnotesize (9)}\tabularnewline
 & {\footnotesize OLS} & {\footnotesize OLS} & {\footnotesize Sieve OLS} & {\footnotesize OLS} & {\footnotesize OLS} & {\footnotesize Sieve OLS} & {\footnotesize OLS} & {\footnotesize OLS} & {\footnotesize Sieve OLS}\tabularnewline
\hline 
{\footnotesize\textit{Peer avg./frac.}} &  &  &  &  &  &  &  &  & \tabularnewline
{\footnotesize$\:$Outcome} & {\footnotesize 0.929{*}{*}{*}} & {\footnotesize 0.818{*}{*}{*}} & {\footnotesize 0.543{*}{*}{*}} & {\footnotesize 0.886{*}{*}{*}} & {\footnotesize 0.733{*}{*}{*}} & {\footnotesize 0.403{*}{*}{*}} & {\footnotesize 0.799{*}{*}{*}} & {\footnotesize 0.459{*}{*}{*}} & {\footnotesize 0.413{*}{*}{*}}\tabularnewline
 & {\footnotesize (0.015)} & {\footnotesize (0.018)} & {\footnotesize (0.022)} & {\footnotesize (0.018)} & {\footnotesize (0.021)} & {\footnotesize (0.027)} & {\footnotesize (0.027)} & {\footnotesize (0.026)} & {\footnotesize (0.023)}\tabularnewline
{\footnotesize$\:$Female} & {\footnotesize 0.068{*}{*}{*}} & {\footnotesize 0.069{*}{*}{*}} & {\footnotesize 0.027{*}{*}} & {\footnotesize -0.012} & {\footnotesize 0.009} & {\footnotesize -0.009} & {\footnotesize -0.024} & {\footnotesize 0.073{*}{*}{*}} & {\footnotesize 0.068{*}{*}{*}}\tabularnewline
 & {\footnotesize (0.010)} & {\footnotesize (0.011)} & {\footnotesize (0.011)} & {\footnotesize (0.010)} & {\footnotesize (0.012)} & {\footnotesize (0.012)} & {\footnotesize (0.021)} & {\footnotesize (0.023)} & {\footnotesize (0.024)}\tabularnewline
{\footnotesize$\:$Mother's edu.} & {\footnotesize -0.002} & {\footnotesize 0.009{*}{*}{*}} & {\footnotesize -0.003} & {\footnotesize -0.002} & {\footnotesize 0.012{*}{*}{*}} & {\footnotesize -0.003} & {\footnotesize 0.003} & {\footnotesize 0.027{*}{*}{*}} & {\footnotesize 0.018{*}{*}{*}}\tabularnewline
 & {\footnotesize (0.002)} & {\footnotesize (0.002)} & {\footnotesize (0.002)} & {\footnotesize (0.002)} & {\footnotesize (0.002)} & {\footnotesize (0.003)} & {\footnotesize (0.004)} & {\footnotesize (0.003)} & {\footnotesize (0.003)}\tabularnewline
{\footnotesize$\:$Fonasa} & {\footnotesize -0.006} & {\footnotesize -0.010} & {\footnotesize -0.023} & {\footnotesize -0.001} & {\footnotesize -0.010} & {\footnotesize 0.003} & {\footnotesize -0.030} & {\footnotesize -0.070{*}{*}} & {\footnotesize -0.073{*}{*}}\tabularnewline
 & {\footnotesize (0.017)} & {\footnotesize (0.015)} & {\footnotesize (0.015)} & {\footnotesize (0.018)} & {\footnotesize (0.017)} & {\footnotesize (0.017)} & {\footnotesize (0.039)} & {\footnotesize (0.033)} & {\footnotesize (0.034)}\tabularnewline
{\footnotesize\textit{Individual attr.}} &  &  &  &  &  &  &  &  & \tabularnewline
{\footnotesize$\:$Female} & {\footnotesize -0.070{*}{*}{*}} & {\footnotesize -0.070{*}{*}{*}} & {\footnotesize -0.032{*}{*}{*}} & {\footnotesize 0.016{*}{*}{*}} & {\footnotesize 0.016{*}{*}{*}} & {\footnotesize 0.061{*}{*}{*}} & {\footnotesize 0.031{*}{*}{*}} & {\footnotesize 0.034{*}{*}{*}} & {\footnotesize 0.049{*}{*}{*}}\tabularnewline
 & {\footnotesize (0.005)} & {\footnotesize (0.005)} & {\footnotesize (0.005)} & {\footnotesize (0.005)} & {\footnotesize (0.005)} & {\footnotesize (0.005)} & {\footnotesize (0.010)} & {\footnotesize (0.011)} & {\footnotesize (0.011)}\tabularnewline
{\footnotesize$\:$Mother's edu.} & {\footnotesize 0.006{*}{*}{*}} & {\footnotesize 0.006{*}{*}{*}} & {\footnotesize 0.027{*}{*}} & {\footnotesize 0.008{*}{*}{*}} & {\footnotesize 0.008{*}{*}{*}} & {\footnotesize 0.027{*}{*}} & {\footnotesize 0.000} & {\footnotesize 0.001} & {\footnotesize 0.006}\tabularnewline
 & {\footnotesize (0.001)} & {\footnotesize (0.001)} & {\footnotesize (0.011)} & {\footnotesize (0.001)} & {\footnotesize (0.001)} & {\footnotesize (0.012)} & {\footnotesize (0.002)} & {\footnotesize (0.002)} & {\footnotesize (0.009)}\tabularnewline
{\footnotesize$\:$Fonasa} & {\footnotesize -0.002} & {\footnotesize -0.002} & {\footnotesize 0.001} & {\footnotesize 0.000} & {\footnotesize -0.000} & {\footnotesize 0.008} & {\footnotesize 0.021} & {\footnotesize 0.020} & {\footnotesize 0.025}\tabularnewline
 & {\footnotesize (0.007)} & {\footnotesize (0.007)} & {\footnotesize (0.007)} & {\footnotesize (0.007)} & {\footnotesize (0.007)} & {\footnotesize (0.008)} & {\footnotesize (0.015)} & {\footnotesize (0.015)} & {\footnotesize (0.015)}\tabularnewline
\hline 
{\footnotesize School FE} & {\footnotesize No} & {\footnotesize Yes} & {\footnotesize Yes} & {\footnotesize No} & {\footnotesize Yes} & {\footnotesize Yes} & {\footnotesize No} & {\footnotesize Yes} & {\footnotesize Yes}\tabularnewline
{\footnotesize Selection correction} & {\footnotesize No} & {\footnotesize No} & {\footnotesize Yes} & {\footnotesize No} & {\footnotesize No} & {\footnotesize Yes} & {\footnotesize No} & {\footnotesize No} & {\footnotesize Yes}\tabularnewline
{\footnotesize F-stat of School FE} &  & {\footnotesize 0.821} & {\footnotesize 2.879{*}{*}{*}} &  & {\footnotesize 1.144} & {\footnotesize 2.257{*}{*}{*}} &  & {\footnotesize 3.034{*}{*}{*}} & {\footnotesize 2.959{*}{*}{*}}\tabularnewline
{\footnotesize F-stat of Selection} &  &  & {\footnotesize 3.175{*}{*}{*}} &  &  & {\footnotesize 2.969{*}{*}{*}} &  &  & {\footnotesize 2.880{*}{*}{*}}\tabularnewline
{\footnotesize$R^{2}$} & {\footnotesize 0.493} & {\footnotesize 0.873} & {\footnotesize 0.910} & {\footnotesize 0.384} & {\footnotesize 0.846} & {\footnotesize 0.899} & {\footnotesize 0.184} & {\footnotesize 0.772} & {\footnotesize 0.776}\tabularnewline
{\footnotesize Observations} & {\footnotesize 6,073} & {\footnotesize 6,073} & {\footnotesize 6,073} & {\footnotesize 6,073} & {\footnotesize 6,073} & {\footnotesize 6,073} & {\footnotesize 6,872} & {\footnotesize 6,872} & {\footnotesize 6,872}\tabularnewline
\hline 
\end{tabular*}}\begin{tablenotes}[flushleft] \scriptsize \item
\textit{Notes}: Estimates of peer effects on short-term academic outcomes. For each outcome, we employ three progressive estimation strategies: (a) OLS without controls for school or selection (columns 1, 4, 7), (b) OLS with school fixed effects (columns 2, 5, 8), and (c) sieve OLS with both school fixed effects and selection correction (columns 3, 6, 9). The selection correction is constructed using a second-order polynomial series of the elementary symmetric functions of utility difference and qualification indices in high school admissions, as detailed in Supplemental Appendix \ref{online:exch_by_type}. Standard errors constructed based on \ref{thm:gamma_clt} are in parentheses. ***, **, * indicate statistical significance at 1\%, 5\%, and 10\% levels, respectively. F-statistics test the joint significance of school fixed effects and polynomial basis functions, respectively. 
\end{tablenotes}
\end{sidewaystable}
\begingroup
\setlength{\tabcolsep}{1pt}
\begin{sidewaystable}[H]
\centering{}{\scriptsize\caption{\protect\label{tab:est_long}Peer Effect Estimates: Long-Run Outcomes}
}{\scriptsize{}%
\begin{tabular*}{1\textwidth}{@{\extracolsep{\fill}}lcccccccccccc}
\hline 
\hline & \multicolumn{3}{c}{{\scriptsize Post-secondary enrollment}} & \multicolumn{3}{c}{{\scriptsize College enrollment}} & \multicolumn{3}{c}{{\scriptsize Post-secondary graduation}} & \multicolumn{3}{c}{{\scriptsize College graduation}}\tabularnewline
\cline{2-13}
 & {\scriptsize (1)} & {\scriptsize (2)} & {\scriptsize (3)} & {\scriptsize (4)} & {\scriptsize (5)} & {\scriptsize (6)} & {\scriptsize (7)} & {\scriptsize (8)} & {\scriptsize (9)} & {\scriptsize (10)} & {\scriptsize (11)} & {\scriptsize (12)}\tabularnewline
 & {\scriptsize OLS} & {\scriptsize OLS} & {\scriptsize Sieve OLS} & {\scriptsize OLS} & {\scriptsize OLS} & {\scriptsize Sieve OLS} & {\scriptsize OLS} & {\scriptsize OLS} & {\scriptsize Sieve OLS} & {\scriptsize OLS} & {\scriptsize OLS} & {\scriptsize Sieve OLS}\tabularnewline
\hline 
{\scriptsize\textit{Peer avg./frac.}} &  &  &  &  &  &  &  &  &  &  &  & \tabularnewline
{\scriptsize$\:$Outcome} & {\scriptsize 0.766{*}{*}{*}} & {\scriptsize 0.414{*}{*}{*}} & {\scriptsize 0.305{*}{*}{*}} & {\scriptsize 0.864{*}{*}{*}} & {\scriptsize 0.670{*}{*}{*}} & {\scriptsize 0.470{*}{*}{*}} & {\scriptsize 0.759{*}{*}{*}} & {\scriptsize 0.420{*}{*}{*}} & {\scriptsize 0.320{*}{*}{*}} & {\scriptsize 0.834{*}{*}{*}} & {\scriptsize 0.575{*}{*}{*}} & {\scriptsize 0.407{*}{*}{*}}\tabularnewline
 & {\scriptsize (0.026)} & {\scriptsize (0.033)} & {\scriptsize (0.028)} & {\scriptsize (0.024)} & {\scriptsize (0.028)} & {\scriptsize (0.025)} & {\scriptsize (0.029)} & {\scriptsize (0.034)} & {\scriptsize (0.029)} & {\scriptsize (0.030)} & {\scriptsize (0.036)} & {\scriptsize (0.032)}\tabularnewline
{\scriptsize$\:$Female} & {\scriptsize -0.037} & {\scriptsize 0.136{*}{*}{*}} & {\scriptsize 0.153{*}{*}{*}} & {\scriptsize -0.004} & {\scriptsize 0.023} & {\scriptsize 0.011} & {\scriptsize -0.072{*}{*}{*}} & {\scriptsize 0.059{*}{*}} & {\scriptsize 0.066{*}{*}} & {\scriptsize -0.041{*}{*}{*}} & {\scriptsize -0.019} & {\scriptsize -0.037{*}{*}}\tabularnewline
 & {\scriptsize (0.024)} & {\scriptsize (0.028)} & {\scriptsize (0.028)} & {\scriptsize (0.016)} & {\scriptsize (0.017)} & {\scriptsize (0.017)} & {\scriptsize (0.024)} & {\scriptsize (0.028)} & {\scriptsize (0.027)} & {\scriptsize (0.015)} & {\scriptsize (0.016)} & {\scriptsize (0.016)}\tabularnewline
{\scriptsize$\:$Mother's edu.} & {\scriptsize 0.004} & {\scriptsize 0.043{*}{*}{*}} & {\scriptsize 0.034{*}{*}{*}} & {\scriptsize -0.005} & {\scriptsize 0.021{*}{*}{*}} & {\scriptsize 0.013{*}{*}{*}} & {\scriptsize 0.008{*}} & {\scriptsize 0.047{*}{*}{*}} & {\scriptsize 0.035{*}{*}{*}} & {\scriptsize 0.001} & {\scriptsize 0.028{*}{*}{*}} & {\scriptsize 0.018{*}{*}{*}}\tabularnewline
 & {\scriptsize (0.004)} & {\scriptsize (0.004)} & {\scriptsize (0.005)} & {\scriptsize (0.004)} & {\scriptsize (0.004)} & {\scriptsize (0.004)} & {\scriptsize (0.005)} & {\scriptsize (0.004)} & {\scriptsize (0.005)} & {\scriptsize (0.004)} & {\scriptsize (0.004)} & {\scriptsize (0.004)}\tabularnewline
{\scriptsize$\:$Fonasa} & {\scriptsize 0.067{*}} & {\scriptsize -0.017} & {\scriptsize -0.067{*}{*}} & {\scriptsize -0.023} & {\scriptsize -0.086{*}{*}{*}} & {\scriptsize -0.127{*}{*}{*}} & {\scriptsize 0.046} & {\scriptsize -0.092{*}{*}{*}} & {\scriptsize -0.144{*}{*}{*}} & {\scriptsize -0.030} & {\scriptsize -0.132{*}{*}{*}} & {\scriptsize -0.172{*}{*}{*}}\tabularnewline
 & {\scriptsize (0.035)} & {\scriptsize (0.032)} & {\scriptsize (0.034)} & {\scriptsize (0.034)} & {\scriptsize (0.032)} & {\scriptsize (0.032)} & {\scriptsize (0.040)} & {\scriptsize (0.035)} & {\scriptsize (0.036)} & {\scriptsize (0.040)} & {\scriptsize (0.033)} & {\scriptsize (0.033)}\tabularnewline
{\scriptsize\textit{Individual attr.}} &  &  &  &  &  &  &  &  &  &  &  & \tabularnewline
{\scriptsize$\:$Female} & {\scriptsize 0.086{*}{*}{*}} & {\scriptsize 0.090{*}{*}{*}} & {\scriptsize 0.120{*}{*}{*}} & {\scriptsize 0.020{*}{*}} & {\scriptsize 0.021{*}{*}} & {\scriptsize 0.058{*}{*}{*}} & {\scriptsize 0.123{*}{*}{*}} & {\scriptsize 0.126{*}{*}{*}} & {\scriptsize 0.155{*}{*}{*}} & {\scriptsize 0.057{*}{*}{*}} & {\scriptsize 0.058{*}{*}{*}} & {\scriptsize 0.088{*}{*}{*}}\tabularnewline
 & {\scriptsize (0.010)} & {\scriptsize (0.010)} & {\scriptsize (0.010)} & {\scriptsize (0.009)} & {\scriptsize (0.009)} & {\scriptsize (0.009)} & {\scriptsize (0.011)} & {\scriptsize (0.011)} & {\scriptsize (0.011)} & {\scriptsize (0.009)} & {\scriptsize (0.009)} & {\scriptsize (0.009)}\tabularnewline
{\scriptsize$\:$Mother's edu.} & {\scriptsize 0.020{*}{*}{*}} & {\scriptsize 0.020{*}{*}{*}} & {\scriptsize 0.017{*}} & {\scriptsize 0.019{*}{*}{*}} & {\scriptsize 0.020{*}{*}{*}} & {\scriptsize -0.001} & {\scriptsize 0.015{*}{*}{*}} & {\scriptsize 0.016{*}{*}{*}} & {\scriptsize 0.013} & {\scriptsize 0.013{*}{*}{*}} & {\scriptsize 0.014{*}{*}{*}} & {\scriptsize -0.005}\tabularnewline
 & {\scriptsize (0.002)} & {\scriptsize (0.002)} & {\scriptsize (0.010)} & {\scriptsize (0.001)} & {\scriptsize (0.001)} & {\scriptsize (0.007)} & {\scriptsize (0.002)} & {\scriptsize (0.002)} & {\scriptsize (0.010)} & {\scriptsize (0.001)} & {\scriptsize (0.001)} & {\scriptsize (0.006)}\tabularnewline
{\scriptsize$\:$Fonasa} & {\scriptsize -0.001} & {\scriptsize -0.003} & {\scriptsize -0.001} & {\scriptsize -0.003} & {\scriptsize -0.005} & {\scriptsize 0.004} & {\scriptsize -0.002} & {\scriptsize -0.006} & {\scriptsize -0.003} & {\scriptsize -0.011} & {\scriptsize -0.014} & {\scriptsize -0.004}\tabularnewline
 & {\scriptsize (0.014)} & {\scriptsize (0.014)} & {\scriptsize (0.015)} & {\scriptsize (0.014)} & {\scriptsize (0.015)} & {\scriptsize (0.015)} & {\scriptsize (0.016)} & {\scriptsize (0.016)} & {\scriptsize (0.017)} & {\scriptsize (0.015)} & {\scriptsize (0.015)} & {\scriptsize (0.016)}\tabularnewline
\hline 
{\scriptsize School FE} & {\scriptsize No} & {\scriptsize Yes} & {\scriptsize Yes} & {\scriptsize No} & {\scriptsize Yes} & {\scriptsize Yes} & {\scriptsize No} & {\scriptsize Yes} & {\scriptsize Yes} & {\scriptsize No} & {\scriptsize Yes} & {\scriptsize Yes}\tabularnewline
{\scriptsize Selection correction} & {\scriptsize No} & {\scriptsize No} & {\scriptsize Yes} & {\scriptsize No} & {\scriptsize No} & {\scriptsize Yes} & {\scriptsize No} & {\scriptsize No} & {\scriptsize Yes} & {\scriptsize No} & {\scriptsize No} & {\scriptsize Yes}\tabularnewline
{\scriptsize F-stat of School FE} &  & {\scriptsize 3.629{*}{*}{*}} & {\scriptsize 3.502{*}{*}{*}} &  & {\scriptsize 1.406{*}{*}} & {\scriptsize 1.494{*}{*}} &  & {\scriptsize 2.798{*}{*}{*}} & {\scriptsize 2.284{*}{*}{*}} &  & {\scriptsize 1.567{*}{*}{*}} & {\scriptsize 1.773{*}{*}{*}}\tabularnewline
{\scriptsize F-stat of Selection} &  &  & {\scriptsize 3.750{*}{*}{*}} &  &  & {\scriptsize 1.331{*}} &  &  & {\scriptsize 1.891{*}{*}{*}} &  &  & {\scriptsize 1.518{*}{*}}\tabularnewline
{\scriptsize$R^{2}$} & {\scriptsize 0.271} & {\scriptsize 0.734} & {\scriptsize 0.746} & {\scriptsize 0.418} & {\scriptsize 0.596} & {\scriptsize 0.632} & {\scriptsize 0.257} & {\scriptsize 0.633} & {\scriptsize 0.646} & {\scriptsize 0.354} & {\scriptsize 0.511} & {\scriptsize 0.544}\tabularnewline
{\scriptsize Observations} & {\scriptsize 6,872} & {\scriptsize 6,872} & {\scriptsize 6,872} & {\scriptsize 6,872} & {\scriptsize 6,872} & {\scriptsize 6,872} & {\scriptsize 6,872} & {\scriptsize 6,872} & {\scriptsize 6,872} & {\scriptsize 6,872} & {\scriptsize 6,872} & {\scriptsize 6,872}\tabularnewline
\hline 
\end{tabular*}}\begin{tablenotes}[flushleft] \scriptsize \item
\textit{Notes}: Estimates of peer effects on long-term academic outcomes. For each outcome, we employ three progressive estimation strategies: (a) OLS without controls for school or selection (columns 1, 4, 7, 10), (b) OLS with school fixed effects (columns 2, 5, 8, 11), and (c) sieve OLS with both school fixed effects and selection correction (columns 3, 6, 9, 12). The selection correction is constructed using a second-order polynomial series of the elementary symmetric functions of utility difference and qualification indices in high school admissions, as detailed in Supplemental Appendix \ref{online:exch_by_type}. Standard errors constructed based on \ref{thm:gamma_clt} are in parentheses. ***, **, * indicate statistical significance at 1\%, 5\%, and 10\% levels, respectively. F-statistics test the joint significance of school fixed effects and polynomial basis functions. 
\end{tablenotes}
\end{sidewaystable}
\endgroup
\begin{sidewaystable}[H]
\centering{}\caption{\protect\label{tab:yvar_decompose}Variance Decomposition of Predicted
Outcomes}
\begin{tabular}{l>{\centering}p{0.08\textwidth}>{\centering}p{0.08\textwidth}>{\centering}p{0.08\textwidth}>{\centering}p{0.08\textwidth}>{\centering}p{0.08\textwidth}>{\centering}p{0.08\textwidth}>{\centering}p{0.08\textwidth}}
\hline 
\hline & {\footnotesize Math score} & {\footnotesize Language score} & {\footnotesize High school graduation} & {\footnotesize Post-sec. enrollment} & {\footnotesize College enrollment} & {\footnotesize Post-sec. graduation} & {\footnotesize College graduation}\tabularnewline
 & {\footnotesize (1)} & {\footnotesize (2)} & {\footnotesize (3)} & {\footnotesize (5)} & {\footnotesize (6)} & {\footnotesize (7)} & {\footnotesize (8)}\tabularnewline
\hline 
{\footnotesize Total variance of predicted outcome} & {\footnotesize 0.042} & {\footnotesize 0.040} & {\footnotesize 0.024} & {\footnotesize 0.049} & {\footnotesize 0.100} & {\footnotesize 0.051} & {\footnotesize 0.076}\tabularnewline
{\footnotesize Variance of peer influence} & {\footnotesize 0.013} & {\footnotesize 0.006} & {\footnotesize 0.015} & {\footnotesize 0.029} & {\footnotesize 0.033} & {\footnotesize 0.030} & {\footnotesize 0.025}\tabularnewline
{\footnotesize Variance of school effectiveness} & {\footnotesize 0.006} & {\footnotesize 0.007} & {\footnotesize 0.016} & {\footnotesize 0.009} & {\footnotesize 0.030} & {\footnotesize 0.022} & {\footnotesize 0.005}\tabularnewline
{\footnotesize Variance of self-selection} & {\footnotesize 0.026} & {\footnotesize 0.034} & {\footnotesize 0.022} & {\footnotesize 0.018} & {\footnotesize 0.046} & {\footnotesize 0.036} & {\footnotesize 0.018}\tabularnewline
{\footnotesize 2{*}Cov(peer infl., school effe.)} & {\footnotesize -0.001} & {\footnotesize -0.002} & {\footnotesize -0.004} & {\footnotesize -0.008} & {\footnotesize 0.025} & {\footnotesize 0.018} & {\footnotesize 0.014}\tabularnewline
{\footnotesize 2{*}Cov(peer infl., selection)} & {\footnotesize 0.010} & {\footnotesize 0.011} & {\footnotesize 0.005} & {\footnotesize 0.013} & {\footnotesize 0.014} & {\footnotesize -0.011} & {\footnotesize 0.016}\tabularnewline
{\footnotesize 2{*}Cov(school effe., selection)} & {\footnotesize -0.014} & {\footnotesize -0.018} & {\footnotesize -0.028} & {\footnotesize -0.012} & {\footnotesize -0.048} & {\footnotesize -0.044} & {\footnotesize -0.002}\tabularnewline
\hline 
\end{tabular}\begin{tablenotes}[flushleft] \scriptsize \item
\textit{Notes}: Variance decomposition of predicted outcome into three major components: peer influence, school effectiveness, and self-selection. The total variance of predicted outcomes represents the combined variation arising from these three components and their covariances, which can be decomposed into the sum of the individual component variances plus their respective covariance terms. These estimates are derived from sieve OLS models that regress each outcome on peer averages (both outcomes and attributes), individual attributes, school fixed effects, and selection controls, where the selection controls incorporate second-order polynomial basis functions of elementary indices of group formation. The three components are constructed as follows: (i) peer influence,
calculated as $\sum_{j=1}^{n}w_{ij}y_{j}\hat{\gamma}_{1}+\sum_{j=1}^{n}w_{ij}x_{j}\hat{\gamma}_{2}$,
where $\hat{\gamma}_{1}$ and $\hat{\gamma}_{2}$ are the estimates
of $\gamma_{1}$ and $\gamma_{2}$, (ii) school effectiveness, represented
by the estimated school FE for the school attending, and (iii) self
selection, calculated as $\hat{\lambda}(\hat{\pi}_{i})$, the estimated
polynomial series used for the selection correction.
\end{tablenotes}
\end{sidewaystable}
\begin{table}[H]
\centering
\caption{\protect\label{tab:test}Testing for Random Assignment into Classrooms}
\begin{tabular}{lcc}
\hline 
\hline & {\footnotesize Baseline} & {\footnotesize Selection corrected}\tabularnewline
\hline 
{\footnotesize Female} & {\footnotesize 0.109} & {\footnotesize 0.100}\tabularnewline
{\footnotesize Mother's education} & {\footnotesize 0.008} & {\footnotesize 0.243}\tabularnewline
{\footnotesize Fonasa} & {\footnotesize 0.063} & {\footnotesize 0.163}\tabularnewline
{\footnotesize Eighth-grade composite score} & {\footnotesize 0.002} & {\footnotesize 0.401}\tabularnewline
{\footnotesize Eighth-grade family income} & {\footnotesize 0.113} & {\footnotesize 0.334}\tabularnewline
{\footnotesize Tenth-grade family income} & {\footnotesize 0.073} & {\footnotesize 0.509}\tabularnewline
\hline 
\end{tabular}

\begin{tablenotes}[flushleft] \footnotesize \item
\textit{Notes}: P-values from tests of random assignment to classrooms for various student characteristics, following \cite{jochmans2023testing}. The table displays (i) a baseline test that controls for school FE, and (ii) a selection-corrected test that controls for both school FE and selection correction. The null hypothesis is that assignment to classrooms within schools is random.
\end{tablenotes}
\end{table}
\newpage{}

\appendix
\setcounter{page}{1}
\fancypagestyle{appendix}{%
    \fancyhf{} 
    \fancyfoot[C]{Supplemental Appendix \thepage} 
    \renewcommand{\headrulewidth}{0pt} 
}

\pagestyle{appendix}

\numberwithin{equation}{section}
\numberwithin{figure}{section}
\numberwithin{table}{section}
\begin{center}
{\LARGE Supplemental Appendix to}\\[1cm]{\LARGE Social Interactions
in Endogenous Groups}\\[1cm] {\large Shuyang Sheng \hspace{1cm} Xiaoting
Sun}\\[1cm]
\par\end{center}

\section{\protect\label{online:gf_estimate}Estimating Group Formation Parameters}

To establish the empirical relevance of our estimation approach, we
study a parametric simulation setting of group formation process.
We assume known distributions for the unobserved heterogeneity terms
$\xi_{ig}$ and $\eta_{ig}$. Let $z_{i}$ be the vector of observable
characteristics, which contains the individual-group specific variables$\{z_{1,ig}^{u},z_{1,ig}^{v}\}_{g=1}^{G}$
and individual specific variable $z_{2i}$. For each group $g$, the
conditional probability of an individual joining $g$ can be written
as
\begin{eqnarray*}
\sigma_{g}(\delta,\zeta,p;z_{i}) & \equiv & \mathbb{P}(g_{i}=g|z_{i})\\
 & = & \int\!\frac{\exp\!\left(\zeta_{g}+U_{ig}\right)\cdot\mathbf{1}\!\left(V_{ig}+\eta_{ig}>p_{g}\right)}{1+\sum_{h=1}^{G}\exp\!\left(\zeta_{h}+U_{ih}\right)\cdot\mathbf{1}\!\left(V_{ih}+\eta_{ih}>p_{h}\right)}\,dF(\eta_{i}),
\end{eqnarray*}
where $U_{ig}=\delta_{1}^{u}z_{1,ig}^{u}+\delta_{2}^{u}z_{2,i}$ and
$V_{ig}=\delta_{1}^{v}z_{1,ig}^{v}+\delta_{2}^{v}z_{2,i}$ represent
the deterministic part of the utility and qualification respectively,
$\eta_{i}=(\eta_{i1},\dots,\eta_{iG})'$, and $\mathbf{1}\left(\cdot\right)$
is an indicator function. Let $\boldsymbol{\sigma}(\delta,\alpha,p;z_{i})=(\sigma_{1}(\delta,\alpha,p;z_{i}),\dots,\sigma_{G}(\delta,\alpha,p;z_{i}))'$
denote the vector of conditional probabilities across all groups.

We estimate the group formation parameters $\delta$ and $\zeta$
by maximizing log-likelihood, where $p$ is treated as auxiliary parameters
that satisfy the market clearing condition. The estimator solves the
constrained optimization problem:
\[
\max_{\delta,\zeta,p}\;\mathcal{L}_{n}(\delta,\zeta,p;z_{i})\qquad\text{s.t.}\quad D_{n}(\delta,\zeta,p;z_{i})=S_{n},
\]
where the log likelihood function is given by:
\begin{eqnarray*}
\mathcal{L}_{n}(\delta,\zeta,p;z_{i}) & = & \sum_{i=1}^{n}\log l(\delta,\zeta,p;z_{i})\\
 & = & \sum_{i=1}^{n}\sum_{g=0}^{G}\mathbf{1}(g_{i}=g)\log\sigma_{g}(\delta,\zeta,p;z_{i}),
\end{eqnarray*}
the demand is given by
\[
D_{n}(\delta,\zeta,p)=\frac{1}{n}\sum_{i=1}^{n}\boldsymbol{\sigma}(\delta,\zeta,p;z_{i}),
\]
and $S_{n}$ is the vector of group capacities.

Since the match probabilities involve an integral without a closed
form, we approximate it through simulation. For each individual $i$,
we generate 300 independent random draws of the vector $\eta_{i}$
from its distribution to approximate the probabilities through Monte
Carlo integration. To address the numerical challenges posed by non-smooth
indicator functions, we adopt the smoothed A-R simulator \citep{train2009discrete}.
Specifically, we replace the indicator function $\mathbf{1}(p_{g}-(V_{ig}+\eta_{ig})<0)$
with a smoothed logistic function:$\left(1+\exp\left(\frac{p_{g}-(V_{ig}+\eta_{ig})}{\kappa}\right)\right)^{-1}$,
where the scale parameter $\kappa$ is set to 0.05. This smooth approximation
improves numerical properties while maintaining estimation accuracy
and converges to the original indicator function as $\kappa\rightarrow0$.

\citet{aitchison1958maximum} has established the $\sqrt{n}$ consistency
and asymptotic distribution of a constrained maximum likelihood estimator
like ours. We calculate our estimator using the MPEC (Mathematical
Programming with Equilibrium Constraints) algorithm \citep{su2012constrained}.
Table \ref{tab:groupformation} shows the estimation results. We find
that the estimator works well, with minimal bias across parameters
and reasonable standard errors.
\begin{table}
\centering
\caption{{\small\protect\label{tab:groupformation}Simulations: Estimates of
Group Formation Parameters}}
{\small\setlength{\tabcolsep}{4pt}}%
\begin{tabular}{ccccccccccc}
\hline 
\hline & \multicolumn{2}{c}{{\footnotesize\textit{Panel A. Slope Parameters}}} &  &  & \multicolumn{2}{c}{{\footnotesize\textit{Panel B. Group Fixed Effects}}} &  &  & \multicolumn{2}{c}{{\footnotesize\textit{Panel C. Cutoffs}}}\tabularnewline
\cline{1-3}\cline{5-7}\cline{9-11}
 & {\footnotesize Bias} & {\footnotesize Std. Dev.} &  &  & {\footnotesize Bias} & {\footnotesize Std. Dev.} &  &  & {\footnotesize Bias} & {\footnotesize Std. Dev.}\tabularnewline
{\footnotesize$\delta_{1}^{u}$} & {\footnotesize -0.032} & {\footnotesize 0.065} &  & {\footnotesize$\alpha_{1}$} & {\footnotesize 0.641} & {\footnotesize 0.757} &  & {\footnotesize$p_{1}$} & {\footnotesize 0.174} & {\footnotesize 0.254}\tabularnewline
{\footnotesize$\delta_{1}^{v}$} & {\footnotesize -0.051} & {\footnotesize 0.051} &  & {\footnotesize$\alpha_{2}$} & {\footnotesize 0.424} & {\footnotesize 0.547} &  & {\footnotesize$p_{2}$} & {\footnotesize 0.166} & {\footnotesize 0.244}\tabularnewline
{\footnotesize$\delta_{2}^{u}$} & {\footnotesize 0.021} & {\footnotesize 0.174} &  & {\footnotesize$\alpha_{3}$} & {\footnotesize 0.308} & {\footnotesize 0.473} &  & {\footnotesize$p_{3}$} & {\footnotesize 0.095} & {\footnotesize 0.258}\tabularnewline
{\footnotesize$\delta_{2}^{v}$} & {\footnotesize -0.031} & {\footnotesize 0.067} &  & {\footnotesize$\alpha_{4}$} & {\footnotesize 0.140} & {\footnotesize 0.354} &  & {\footnotesize$p_{4}$} & {\footnotesize 0.192} & {\footnotesize 0.180}\tabularnewline
 &  &  &  & {\footnotesize$\alpha_{5}$} & {\footnotesize 0.022} & {\footnotesize 0.308} &  & {\footnotesize$p_{5}$} & {\footnotesize 0.235} & {\footnotesize 0.255}\tabularnewline
\hline 
\end{tabular}\begin{tablenotes}[flushleft]\footnotesize \item 
\textit{Notes}: Average biases, and standard deviations of the group formation parameters obtained from 200 Monte Carlo samples. Each Monte Carlo sample consists of 2,000 individuals and 5 groups. 
\end{tablenotes}
\end{table}

\section{\protect\label{online:gf_empirical}Estimating Parameters in High
School Admissions}

We model the group formation process empirically by analyzing the
matching between students and high schools (grades 9-12) in Chile's
Biobío Region during 2005.\footnote{The Biobío Region includes twelve municipalities: Los Ángeles, Nacimiento,
Yumbel, Laja, Cabrero, Yungay, Mulchén, Negrete, Tucapel, Pemuco,
Quilleco, and Santa Bárbara. The high school admission market in Biobío
Region is relatively independent, with only $2.19\%$ of tenth graders
in Biobío Region schools residing outside the market area and $3.96\%$
of tenth graders living in the Biobío Region attending schools elsewhere
in 2006.} Table \ref{tab:sum_gf} summarizes the student and school characteristics.
Note that school average composite score and average mother's education
was calculated from student characteristics of the tenth graders in
a school in 2003, and thus are pre-determined in the 2005 admissions
that we study.
\begin{table}[t]
\centering
{\footnotesize\caption{\protect\label{tab:sum_gf}Summary Statistics of Student and School
Characteristics}
}{\footnotesize{}%
\begin{tabular}{lccccc}
\hline 
\hline & \multicolumn{2}{c}{{\footnotesize Public schools}} &  & \multicolumn{2}{c}{{\footnotesize Private schools}}\tabularnewline
\hline 
 & {\footnotesize Mean} & {\footnotesize Std. Dev.} &  & {\footnotesize Mean} & {\footnotesize Std. Dev.}\tabularnewline
\hline 
{\footnotesize\textit{Panel A. Student Characteristics}} &  &  &  &  & \tabularnewline
{\footnotesize$\quad$Eighth grade composite score} & {\footnotesize 0.48} & {\footnotesize 0.28} &  & {\footnotesize 0.54} & {\footnotesize 0.30}\tabularnewline
{\footnotesize$\quad$Mother's education (years)} & {\footnotesize 8.75} & {\footnotesize 2.98} &  & {\footnotesize 10.85} & {\footnotesize 3.59}\tabularnewline
{\footnotesize$\quad$Parental income (CLP)} & {\footnotesize 147,388} & {\footnotesize 142,698} &  & {\footnotesize 369,984} & {\footnotesize 479,055}\tabularnewline
{\footnotesize$\quad$Distance to enrolled school (km)} & {\footnotesize 25.65} & {\footnotesize 115.03} &  & {\footnotesize 25.06} & {\footnotesize 105.63}\tabularnewline
{\footnotesize$\quad$Observations} & {\footnotesize 4,700} &  &  & {\footnotesize 2,458} & \tabularnewline
{\footnotesize\textit{Panel B. School Characteristics}} &  &  &  &  & \tabularnewline
{\footnotesize$\quad$Average composite score} & {\footnotesize 0.45} & {\footnotesize 0.11} &  & {\footnotesize 0.55} & {\footnotesize 0.20}\tabularnewline
{\footnotesize$\quad$Average mother's edu. (years)} & {\footnotesize 7.87} & {\footnotesize 1.11} &  & {\footnotesize 10.41} & {\footnotesize 2.59}\tabularnewline
{\footnotesize$\quad$Tuition (CLP)} & {\footnotesize 1,392.04} & {\footnotesize 731.97} &  & {\footnotesize 13,367.66} & {\footnotesize 22,963.47}\tabularnewline
{\footnotesize$\quad$Capacity} & {\footnotesize -} & {\footnotesize -} &  & {\footnotesize 97.93} & {\footnotesize 64.38}\tabularnewline
{\footnotesize$\quad$Biobio student enrollment} & {\footnotesize 191.96} & {\footnotesize 125.32} &  & {\footnotesize 81.90} & {\footnotesize 57.01}\tabularnewline
{\footnotesize$\quad$Observations} & \multicolumn{2}{c}{{\footnotesize 23}} &  & \multicolumn{2}{c}{{\footnotesize 30}}\tabularnewline
\hline 
\end{tabular}}\begin{tablenotes}\footnotesize \item 
\textit{Notes}: Summary statistics of student and school characteristics in Biobío Region. Composite score is measured in percentile rank (from
0 to 1). CLP stands for Chilean peso. Parental income and tuition are measured in 2006 when 1 USD was about 530 CLP.
\end{tablenotes}
\end{table}

Our empirical framework for high school admissions follows closely
\citet{HSS_twosided}, while our estimation approach deviates by employing
frequentist methods rather than the Bayesian approach in \citet{HSS_twosided}.
Specifically, we allow student preferences to be school-type-specific.
For student~$i$, the utility of attending school~$g$ of type~$t$
$\in$ \{public, private\} is 
\begin{equation}
u_{ig}=Z_{ig}^{u\prime}\beta_{t}+\xi_{ig},\label{eq:stu_u}
\end{equation}
where $\xi_{ig}$ is i.i.d.\ extreme value type I; and $Z_{ig}^{u}$
are student-school-specific variables, including a constant term,
and
\begin{itemize}
\item Distance between $i$'s residence and school~$g$;
\item School attributes: tuition (in logarithm), average composite score,
average mother's education;
\item Interactions between school attributes and student characteristics:
tuition interacted with student's parental income, school average
mother's education interacted with student mother's education.
\end{itemize}
Each student has an outside option, $u_{i0}=\xi_{i0}$, with $\xi_{i0}$
being extreme value type I.

As for school preferences, public schools do not have a utility function
because they cannot select students. For private school~$g$, its
qualification function is 
\begin{equation}
v_{ig}=Z_{ig}^{v\prime}\theta+\eta_{ig},\label{eq:sch_u}
\end{equation}
where $\eta_{ig}$ is i.i.d. standard normal; and the vector $Z_{ig}^{v}$
includes a constant term, and
\begin{itemize}
\item Student characteristics: composite score, mother's education;
\item Interactions between student characteristics and school attributes:
student composite score interacted with school average composite score,
and student mother's education interacted with school average mother's
education .
\begin{table}[t]
\centering
\caption{\protect\label{tab:matching_est}Estimation Results: Student and School
Preferences}
\begin{tabular}{lccccc}
\hline 
\hline & \multicolumn{2}{c}{{\footnotesize Public schools}} &  & \multicolumn{2}{c}{{\footnotesize Private Schools}}\tabularnewline
\cline{2-6}
 & {\footnotesize{coef.}} & {\footnotesize{s.e.}} &  & {\footnotesize{coef.}} & {\footnotesize{s.e.}}\tabularnewline
\hline 
{\footnotesize\textit{Panel A. Student Preferences}} &  &  &  &  & \tabularnewline
{\footnotesize$\quad$Constant} & {\footnotesize 3.894{*}{*}{*}} & {\footnotesize 0.319} &  & {\footnotesize 9.521{*}{*}{*}} & {\footnotesize 0.383}\tabularnewline
{\footnotesize$\quad$Distance} & {\footnotesize -0.026{*}{*}{*}} & {\footnotesize 0.000} &  & {\footnotesize -0.019{*}{*}{*}} & {\footnotesize 0.000}\tabularnewline
{\footnotesize$\quad$log(tuition)} & {\footnotesize -2.909{*}{*}{*}} & {\footnotesize 0.128} &  & {\footnotesize -3.521{*}{*}{*}} & {\footnotesize 0.123}\tabularnewline
{\footnotesize$\quad$log(tuition) \texttimes{} log(income)} & {\footnotesize 0.243{*}{*}{*}} & {\footnotesize 0.010} &  & {\footnotesize 0.301{*}{*}{*}} & {\footnotesize 0.010}\tabularnewline
{\footnotesize$\quad$Average composite score} & {\footnotesize 0.796{*}{*}{*}} & {\footnotesize 0.275} &  & {\footnotesize -0.168} & {\footnotesize 0.303}\tabularnewline
{\footnotesize$\quad$Average mother's education} & {\footnotesize -0.032} & {\footnotesize 0.032} &  & {\footnotesize -0.811{*}{*}{*}} & {\footnotesize 0.039}\tabularnewline
{\footnotesize$\quad$Average mother's education \texttimes{} mother's
education} & {\footnotesize 0.010{*}{*}{*}} & {\footnotesize 0.002} &  & {\footnotesize 0.029{*}{*}{*}} & {\footnotesize 0.002}\tabularnewline
{\footnotesize\textit{Panel B. Student Qualification}} &  &  &  &  & \tabularnewline
{\footnotesize$\quad$Constant} &  &  &  & {\footnotesize 0.904{*}{*}{*}} & {\footnotesize 0.256}\tabularnewline
{\footnotesize$\quad$Composite score} &  &  &  & {\footnotesize -2.852{*}{*}{*}} & {\footnotesize 0.299}\tabularnewline
{\footnotesize$\quad$Composite score x average composite score} &  &  &  & {\footnotesize 3.491{*}{*}{*}} & {\footnotesize 0.683}\tabularnewline
{\footnotesize$\quad$Mother's education} &  &  &  & {\footnotesize -0.427{*}{*}{*}} & {\footnotesize 0.038}\tabularnewline
{\footnotesize$\quad$Mother's education \texttimes{} average mother's
education} &  &  &  & {\footnotesize 0.046{*}{*}{*}} & {\footnotesize 0.005}\tabularnewline
\hline 
\end{tabular}\begin{tablenotes}\footnotesize \item 
\textit{Notes}: Estimation results from a two-sided matching model of high school admissions in Chile's Biobío Region during 2005, estimated using MPEC algorithm detailed in Appendix \ref{online:gf_estimate}. Panel A reports student preference parameters for both public and private schools. Panel B reports qualification parameters for private schools only. ***, **, * indicate statistical significance at 1\%, 5\%, and 10\% levels, respectively.
\end{tablenotes}
\end{table}
\end{itemize}
We estimate the model by maximizing a simulated log-likelihood, subject
to equilibrium constraints, as illustrated through Monte Carlo simulations
(Appendix \ref{online:gf_estimate}). The estimation results are summarized
in Table \ref{tab:matching_est}. Panel A shows student preference
estimates, with expected coefficient signs. Both distance and tuition
negatively impact student preferences, though the latter effect diminishes
with parental income through interaction term. Students of more educated
mothers like schools with higher average mother\textquoteright s education,
regardless of school type. Panel B shows admission preferences among
private schools. Schools with higher average mother's education increasingly
prefer students with higher educated mother. The impact of student
score on school preferences depends critically on school quality -
positive preferences for high-achieving students emerge only in schools
whose existing student body maintains an average composite score above
0.8.

\section{\textmd{\protect\label{online:exch_by_type}}Conditional Exchangeability
by School Type}

In our empirical application, we impose exchangeability conditional
on school type (public or private). In this case, while the selection
function is invariant across schools of the same type, it has distinct
functional forms for public and private schools. Moreover, the selection
function is symmetric in the indices of schools within the same type,
rather than across all schools.

Given these properties, in our empirical setting, when constructing
the elementary symmetric functions, instead of aggregating over the
indices of all schools, we aggregate over the indices of public and
private schools separately. Moreover, in our sieve estimation, the
type-specific selection function translates into estimating separate
coefficients for the polynomial basis functions for each school type.
That is, we interact the basis functions with indicator variables
for each school type.

Specifically, let $\mathcal{G}_{1}$ and $\mathcal{G}_{2}$ be a partition
of $\mathcal{G}$, where $\mathcal{G}_{1}$ refers to the set of private
schools and $\mathcal{G}_{2}$ refers to the set of public schools.
Below we list the basis functions up to order 2 used in our empirical
application:
\begin{itemize}
\item $\iota_{ig_{i}}$, $\sum_{h\neq g_{i},h\in\mathcal{G}_{1}}\Delta_{g_{i}}\tau_{ih}$,
$\sum_{h\neq g_{i},h\in\mathcal{G}_{2}}\Delta_{g_{i}}\tau_{ih}$,
$\sum_{h\neq g_{i},h\in\mathcal{G}_{1}}\iota_{ih}$,
\item $\iota_{ig_{i}}^{2}$, $\left(\sum_{h\neq g_{i},h\in\mathcal{G}_{1}}\Delta_{g_{i}}\tau_{ih}\right)^{2}$,
$\left(\sum_{h\neq g_{i},h\in\mathcal{G}_{2}}\Delta_{g_{i}}\tau_{ih}\right)^{2}$,
$\left(\sum_{h\neq g_{i},h\in\mathcal{G}_{1}}\iota_{ih}\right)^{2}$,
\item $\iota_{ig_{i}}\left(\sum_{h\neq g_{i},h\in\mathcal{G}_{1}}\Delta_{g_{i}}\tau_{ih}\right)$,
$\iota_{ig_{i}}\left(\sum_{h\neq g_{i},h\in\mathcal{G}_{2}}\Delta_{g_{i}}\tau_{ih}\right)$,
$\iota_{ig_{i}}\left(\sum_{h\neq g_{i},h\in\mathcal{G}_{1}}\iota_{ih}\right)$,
\\
$\left(\sum_{h\neq g_{i},h\in\mathcal{G}_{1}}\Delta_{g_{i}}\tau_{ih}\right)\left(\sum_{h\neq g_{i},h\in\mathcal{G}_{2}}\Delta_{g_{i}}\tau_{ih}\right)$,
$\left(\sum_{h\neq g_{i},h\in\mathcal{G}_{1}}\Delta_{g_{i}}\tau_{ih}\right)\left(\sum_{h\neq g_{i},h\in\mathcal{G}_{1}}\iota_{ih}\right)$,\\
 $\left(\sum_{h\neq g_{i},h\in\mathcal{G}_{2}}\Delta_{g_{i}}\tau_{ih}\right)\left(\sum_{h\neq g_{i},h\in\mathcal{G}_{1}}\iota_{ih}\right)$,
\item $\sum_{h_{1},h_{2}\neq g_{i},h_{1},h_{2}\in\mathcal{G}_{1}}\Delta_{g_{i}}\tau_{ih_{1}}\Delta_{g_{i}}\tau_{ih_{2}}$,
$\sum_{h_{1},h_{2}\neq g_{i},h_{1},h_{2}\in\mathcal{G}_{2}}\Delta_{g_{i}}\tau_{ih_{1}}\Delta_{g_{i}}\tau_{ih_{2}}$,
$\sum_{h_{1},h_{2}\neq g_{i},h_{1},h_{2}\in\mathcal{G}_{1}}\iota_{ih_{1}}\iota_{ih_{2}}$,\\
 $\sum_{h_{1},h_{2}\neq g_{i},h_{1}\in\mathcal{G}_{1},h_{2}\in\mathcal{G}_{2}}\Delta_{g_{i}}\tau_{ih_{1}}\Delta_{g_{i}}\tau_{ih_{2}}$,
$\sum_{h_{1},h_{2}\neq g_{i},h_{1},h_{2}\in\mathcal{G}_{1}}\Delta_{g_{i}}\tau_{ih_{1}}\iota_{ih_{2}}$,
$\sum_{h_{1},h_{2}\neq g_{i},h_{1}\in\mathcal{G}_{2},h_{2}\in\mathcal{G}_{1}}\Delta_{g_{i}}\tau_{ih_{1}}\iota_{ih_{2}}$.
\end{itemize}
Note that the five terms containing $\iota_{ig_{i}}$ are included
only for students who attend a private school ($g_{i}\in\mathcal{G}_{1}$),
as only private schools have a qualification index, while public schools
cannot select students. For the same reason, the sum of $\iota_{ih}$
is calculated over private schools only. Overall, this implementation
results in 20 basis functions for students attending private schools
and 15 basis functions for students attending public schools, requiring
a total of 35 sieve coefficients to be estimated. The reduction in
nuisance parameters is substantial compared to a specification that
does not use symmetric functions, which would require 82 basis functions
of order one ($G+G_{1}-1=53+30-1$) and 3,403 basis functions of order
two $((G+G_{1}-1)+(G+G_{1}-1)(G+G_{1}-2)/2)$ given our empirical
setting with 23 public schools and 30 private schools.

\section{\protect\label{online:proofs}Proofs}

\paragraph{Notation}

Let $x=(x_{1},\dots,x_{n})'\in\mathbb{R}^{n}$ denote an $n\times1$
vector, and let $A=(a_{ij})\in\mathbb{R}^{n^{2}}$ denote an $n\times n$
matrix. We use $\|\cdot\|$ to denote the Frobenius norm, that is,
$\|x\|\equiv(\sum_{i=1}^{n}x_{i}^{2})^{1/2}$ and $\|A\|\equiv(\text{tr}(AA'))^{1/2}=(\text{\ensuremath{\sum}}_{i=1}^{n}\sum_{j=1}^{n}a_{ij}^{2})^{1/2}$.
For a matrix $A$, we denote the maximum row sum norm as $\interleave A\interleave_{\infty}\equiv\max_{1\leq i\leq n}\sum_{j=1}^{n}|a_{ij}|$
and the maximum column sum norm as $\interleave A\interleave_{1}\equiv\max_{1\leq j\leq n}\sum_{i=1}^{n}|a_{ij}|$.
Note that $\interleave A'\interleave_{\infty}=\interleave A\interleave_{1}$.
Furthermore, for both vectors and matrices, we denote the $l_{\infty}$
and $l_{1}$ norms as $\|\cdot\|_{\infty}$ and $\|\cdot\|_{1}$ respectively.
Specifically, we have $\|x\|_{\infty}\equiv\max_{1\leq i\leq n}|x_{i}|$,
$\|x\|_{1}\equiv\sum_{i=1}^{n}|x_{i}|$, $\|A\|_{\infty}\equiv\max_{1\leq i,j\leq n}|a_{ij}|$,
and $\|A\|_{1}\equiv\sum_{i,j=1}^{n}|a_{ij}|$. It follows that $\|A\|_{\infty}\leq\min\{\interleave A\interleave_{\infty},\interleave A\interleave_{1}\}$,
and $\max\{\interleave A\interleave_{\infty},\interleave A\interleave_{1}\}\leq n\max_{1\leq i,j\leq n}|a_{ij}|=n\|A\|_{\infty}$.
For matrices $A$ and $B$ and vectors $x$ and $y$, we can derive
$\|AB\|_{\infty}\leq\interleave A\interleave_{\infty}\|B\|_{\infty}$,
$\|AB\|_{1}\leq\interleave A\interleave_{1}\|B\|_{1}$, $\|Ax\|_{\infty}\leq\interleave A\interleave_{\infty}\|x\|_{\infty}$,
and $\|Ax\|_{1}\leq\interleave A\interleave_{1}\|x\|_{1}$. Moreover,
we have $|x'Ay|\leq\interleave A\interleave_{\infty}\|x\|_{\infty}\|y\|_{1}\leq n\interleave A\interleave_{\infty}\|x\|_{\infty}\|y\|_{\infty}$.\footnote{These results can be found in Horn and Johnson (1985, Section 5.6)
or proved similarly.} Finally, let $0<C<\infty$ denote a universal constant.

\subsection{Proofs in Sections \ref{sec:Selection_bias} and \ref{sec:Identification}}
\begin{proof}[Proof of Proposition \ref{prop:biascov}]
Following \citet{azevedo_supply_2016}, we show that $p_{n}\overset{p}{\rightarrow}p^{*}$
as $n\rightarrow\infty$, and the limiting cutoffs $p^{*}$ are deterministic.
The selection bias $\mathbb{E}[\epsilon_{i}|\boldsymbol{x},\boldsymbol{z},\boldsymbol{g}(\boldsymbol{z},\boldsymbol{\xi},\boldsymbol{\eta};p)]$
is continuous in $p$ because the cdf of the unobservables is continuous
under Assumption \ref{ass:regular}(ii). Therefore, by the continuous
mapping theorem, we have $\mathbb{E}[\epsilon_{i}|\boldsymbol{x},\boldsymbol{z},\boldsymbol{g}(\boldsymbol{z},\boldsymbol{\xi},\boldsymbol{\eta};p_{n})]\overset{p}{\rightarrow}\mathbb{E}[\epsilon_{i}|\boldsymbol{x},\boldsymbol{z},\boldsymbol{g}(\boldsymbol{z},\boldsymbol{\xi},\boldsymbol{\eta};p^{*})]$.

The selection bias evaluated at $p^{*}$ satisfies
\begin{eqnarray*}
\mathbb{E}[\epsilon_{i}|\boldsymbol{x},\boldsymbol{z},\boldsymbol{g}(\boldsymbol{z},\boldsymbol{\xi},\boldsymbol{\eta};p^{*})] & = & \mathbb{E}[\epsilon_{i}|x_{i},\boldsymbol{x}_{-i},z_{i},\boldsymbol{z}_{-i},g(z_{i},\xi_{i},\eta_{i};p^{*}),\boldsymbol{g}_{-i}(\boldsymbol{z}_{-i},\boldsymbol{\xi}_{-i},\boldsymbol{\eta}_{-i};p^{*})]\\
 & = & \mathbb{E}[\epsilon_{i}|x_{i},z_{i},g(z_{i},\xi_{i},\eta_{i};p^{*})],
\end{eqnarray*}
where $\boldsymbol{x}_{-i}=(x_{j},j\ne i)$ and $\boldsymbol{z}_{-i}$,
$\boldsymbol{g}_{-i}$, $\boldsymbol{\xi}_{-i}$, $\boldsymbol{\eta}_{-i}$
are defined analogously. The last equality follows because given deterministic
cutoffs $p^{*}$, $g_{j}$ only depends on $z_{j}$, $\xi_{j}$, and
$\eta_{j}$ for all $j\neq i$, which are independent of $\epsilon_{i}$
under Assumption \ref{ass:regular}(i). Combining the results proves
the proposition.
\end{proof}

\begin{proof}[Proof of Proposition \ref{prop:exch}]
 Let $f(\epsilon_{i},\xi_{i},\eta_{i})$ denote the joint pdf of
$(\epsilon_{i},\xi_{i},\eta_{i})$, and $f(\xi_{i},\eta_{i})$ the
joint pdf of $(\xi_{i},\eta_{i})$. By equation (\ref{eq:glink})
and the exogeneity of $(x_{i},z_{i})$ (Assumption \ref{ass:regular}(iii)),
individual $i$'s selection bias from joining group $g$ is 
\begin{align}
E[\epsilon_{i}|x_{i},z_{i},g_{i}=g] & =\mathbb{E}[\epsilon_{i}|\eta_{ig}\geq-\iota_{ig},\xi_{ih}-\xi_{ig}<-\Delta_{g}\tau_{ih}\text{ or }\eta_{ih}<-\iota_{ih},\forall h\neq g]\nonumber \\
 & =\frac{\int_{R_{g}(\iota_{ig};\Delta_{g}\tau_{ih},\iota_{ih},\forall h\neq g)}\epsilon_{i}f(\epsilon_{i},\xi_{i},\eta_{i})\textrm{d}\epsilon_{i}\textrm{d}\xi_{i}\textrm{d}\eta_{i}}{\int_{R_{g}(\iota_{ig};\Delta_{g}\tau_{ih},\iota_{ih},\forall h\neq g)}f(\xi_{i},\eta_{i})\textrm{d}\xi_{i}\textrm{d}\eta_{i}}\nonumber \\
 & \equiv\lambda_{g}^{e}(\iota_{ig};\Delta_{g}\tau_{ih},\iota_{ih},\forall h\neq g),\label{eq:lambda_integral}
\end{align}
where the integration region is defined by
\begin{eqnarray*}
 &  & R_{g}(\iota_{ig};\Delta_{g}\tau_{ih},\iota_{ih},\forall h\neq g)\\
 & \equiv & \{(\xi_{i},\eta_{i})\in\mathbb{R}^{2G}:\eta_{ig}\geq-\iota_{ig};\xi_{ih}-\xi_{ig}<-\Delta_{g}\tau_{ih}\text{ or }\eta_{ih}<-\iota_{ih},\forall h\neq g\}.
\end{eqnarray*}

(i) First, we show that $\lambda_{g}^{e}(\cdot)$ is invariant across
$g$. For any two groups $g_{1}$ and $g_{2}$, there exists a permutation
$(k_{1},\dots,k_{G})$ over $(1,\dots,G)$ such that the joint distribution
of $(\epsilon_{i},\eta_{ig_{1}},(\xi_{ih}-\xi_{ig_{1}},\eta_{ih})_{h\neq g_{1}})$
can be derived from the distribution $f(\epsilon_{i},\xi_{i1},...,\xi_{iG},\eta_{i1},...,\eta_{iG})$
in the same way that the joint distribution $(\epsilon_{i},\eta_{g_{2}i},(\xi_{ih}-\xi_{ig_{2}},\eta_{ih})_{h\neq g_{2}})$
is derived from the distribution $f(\epsilon_{i},\xi_{ik_{1}},...,\xi_{ik_{G}},\eta_{ik_{1}},...,\eta_{ik_{G}})$.
By exchangeability (Assmption \ref{ass:exch}), we have $f(\epsilon_{i},\xi_{i1},...,\xi_{iG},\eta_{i1},...,\eta_{iG})=f(\epsilon_{i},\xi_{ik_{1}},...,\xi_{ik_{G}},\eta_{ik_{1}},...,\eta_{ik_{G}})$.
Therefore, the joint distributions of $(\epsilon_{i},\eta_{ig_{1}},(\xi_{ih}-\xi_{ig_{1}},\eta_{ih})_{h\neq g_{1}})\text{ and }(\epsilon_{i},\eta_{ig_{2}},(\xi_{ih}-\xi_{ig_{2}},\eta_{ih})_{h\neq g_{2}})$
are identical. It follows from the fact that $R_{g_{1}}$and $R_{g_{2}}$
are structurally identical that $\lambda_{g_{1}}^{e}(\cdot)=\lambda_{g_{2}}^{e}(\cdot)\equiv\lambda^{e}(\cdot)$.

(ii) Next, we prove that $\lambda^{e}(\cdot)$ is symmetric in the
index pairs $(\Delta_{g}\tau_{ih},\iota_{ih})$ across $h\neq g$.
For any $h_{1},h_{2}\neq g$, exchangeability ensures that the joint
distribution of
\[
(\epsilon_{i},\eta_{ig},\xi_{ih_{1}}-\xi_{ig},\eta_{ih_{1}},\xi_{ih_{2}}-\xi_{ig},\eta_{ih_{2}},(\xi_{ik}-\xi_{ig},\eta_{ik})_{k\neq g,h_{1},h_{2}})
\]
remains unchanged when $h_{1}$ and $h_{2}$ are swapped. This implies
that swapping the index pairs $(\Delta_{g}\tau_{ih_{1}},\iota_{ih_{1}})$
and $(\Delta_{g}\tau_{ih_{2}},\iota_{ih_{2}})$ does not affect the
integrals in (\ref{eq:lambda_integral}). Therefore, $\lambda^{e}(\iota_{ig};\Delta_{g}\tau_{ih},\iota_{ih},\forall h\neq g)$
is symmetric in $(\Delta_{g}\tau_{ih_{1}},\iota_{ih_{1}})$ and $(\Delta_{g}\tau_{ih_{2}},\iota_{ih_{2}})$.
\end{proof}

\subsection{Proofs in Section \ref{sec:Estimation}}

\subsubsection{\protect\label{sec:consist_pf}Consistency of $\hat{\gamma}$}

Theorem \ref{thm:gamma_consist} is proved based on several lemmas
to be presented later in this section. Table \ref{tab:gamma_consist_pfstruct}
states the relationships between Theorem \ref{thm:gamma_consist}
and these lemmas. 

\begin{table}[t]
\centering
\caption{Relationships Between Theorem \ref{thm:gamma_consist} and Its Supporting
Lemmas}
\label{tab:gamma_consist_pfstruct}

\begin{tabular}{llll}
\toprule 
 & Referring to & Referring to & Referring to\tabularnewline
\midrule
Theorem \ref{thm:gamma_consist}  & Lemma \ref{lem:muhat_t_consist} & Lemma \ref{lem:wy_bd} & Lemma \ref{lem:network_bd}\tabularnewline
 &  & Lemma \ref{lem:Q_K} & \tabularnewline
 &  & Lemma \ref{lem:Q_Khat} & Lemma \ref{lem:Q_K}\tabularnewline
 &  & Lemma \ref{lem:betahat_dif} & Lemmas \ref{lem:Q_K}, \ref{lem:Q_Khat}\tabularnewline
 &  & Lemma \ref{lem:wt_betahat} & Lemmas \ref{lem:Q_K}, \ref{lem:network_bd} \tabularnewline
 & Lemma \ref{lem:xt_consist} & Lemma \ref{lem:aqb_consist} & Lemma \ref{lem:network_bd}\tabularnewline
\bottomrule
\end{tabular}
\end{table}

\begin{proof}[Proof of Theorem \ref{thm:gamma_consist}]
Recall that $y_{i}=X'_{i}\gamma_{0}+\epsilon_{i}$. From equation
(\ref{eq:gamma_hat}) we obtain
\begin{align}
\hat{\gamma}-\gamma_{0} & =\left(\frac{1}{n}\sum_{i=1}^{n}(X_{i}-\hat{\mu}^{X}(\hat{\pi}_{i}))X'_{i}\right)^{-1}\frac{1}{n}\sum_{i=1}^{n}(X_{i}-\hat{\mu}^{X}(\hat{\pi}_{i}))\epsilon_{i},\label{eq:gamma_ols}
\end{align}
For $t_{i}=X_{i}$ or $\epsilon_{i}$, by Lemmas \ref{lem:muhat_t_consist}
and \ref{lem:xt_consist}, we can derive
\begin{eqnarray}
\frac{1}{n}\sum_{i=1}^{n}(X_{i}-\hat{\mu}^{X}(\hat{\pi}_{i}))t'_{i} & = & \frac{1}{n}\sum_{i=1}^{n}((X_{i}-\mu_{0}^{X}(\pi_{i}))t'_{i}-\mathbb{E}[(X_{i}-\mu_{0}^{X}(\pi_{i}))t'_{i}])\nonumber \\
 &  & -\frac{1}{n}\sum_{i=1}^{n}(\hat{\mu}^{X}(\hat{\pi}_{i})-\mu_{0}^{X}(\pi_{i}))t'_{i}+\frac{1}{n}\sum_{i=1}^{n}\mathbb{E}[(X_{i}-\mu_{0}^{X}(\pi_{i}))t'_{i}]\nonumber \\
 & = & \frac{1}{n}\sum_{i=1}^{n}\mathbb{E}[(X_{i}-\mu_{0}^{X}(\pi_{i}))t'_{i}]+o_{p}(1).\label{eq:M_lln}
\end{eqnarray}
Note that $n^{-1}\sum_{i=1}^{n}\mathbb{E}[(X_{i}-\mu_{0}^{X}(\pi_{i}))X'_{i}]=n^{-1}\sum_{i=1}^{n}\mathbb{E}[(X_{i}-\mu_{0}^{X}(\pi_{i}))(X_{i}-\mu_{0}^{X}(\pi_{i}))']$
by iterated expectations. By the rank condition in Assumption \ref{ass:rank},
the resulting matrix is positive definite and therefore nonsingular.
By \citet[Corollary 3.1]{wooldridge2010econometric}, it remains to
demonstrate that $n^{-1}\sum_{i=1}^{n}\mathbb{E}[(X_{i}-\mu_{0}^{X}(\pi_{i}))\epsilon_{i}]=o(1)$,
which implies that the regressor $X_{i}$ is asymptotically exogenous.

Note that $\epsilon_{i}=\lambda(\pi_{i})+\nu_{i}$ and $\mathbb{E}[(X_{i}-\mu_{0}^{X}(\pi_{i}))\lambda(\pi_{i})]=0$
by iterated expectations. Hence,
\begin{eqnarray*}
\frac{1}{n}\sum_{i=1}^{n}\mathbb{E}[(X_{i}-\mu_{0}^{X}(\pi_{i}))\epsilon_{i}] & = & \frac{1}{n}\sum_{i=1}^{n}\mathbb{E}[(X_{i}-\mu_{0}^{X}(\pi_{i}))\nu_{i}].
\end{eqnarray*}
Recall that $X_{i}=(w_{i}\boldsymbol{y},w_{i}\boldsymbol{x},x'_{i})'$.
For $t_{i}=x_{i}$ or $w_{i}\boldsymbol{x}$, because $\boldsymbol{w}$
and $\boldsymbol{\nu}=\boldsymbol{\epsilon}-\boldsymbol{\lambda}(\boldsymbol{\pi})$
are independent conditional on $\text{\ensuremath{\boldsymbol{\psi}=(\boldsymbol{x},\boldsymbol{z},\boldsymbol{g})}}$
(Assumption \ref{ass:adj_exog}) and $\mathbb{E}[\nu_{i}|\boldsymbol{\psi}]=\mathbb{E}[\epsilon_{i}|\boldsymbol{\psi}]-\lambda(\pi_{i})=0$,
we have
\[
\mathbb{E}[(t_{i}-\mu_{0}^{t_{i}}(\pi_{i}))\nu_{i}]=\mathbb{E}[(t_{i}-\mu_{0}^{t_{i}}(\pi_{i}))\mathbb{E}[\nu_{i}|\boldsymbol{\psi}]]=0.
\]
For $t_{i}=w_{i}\boldsymbol{y}$, denote $\boldsymbol{s}=(I_{n}-\gamma_{1}\boldsymbol{w})^{-1}$
and $\mu_{0}^{\boldsymbol{w}\boldsymbol{y}}(\boldsymbol{\pi})=(\mu_{0}^{w_{1}\boldsymbol{y}}(\pi_{1}),\dots,\mu_{0}^{w_{n}\boldsymbol{y}}(\pi_{n}))'$,
where $\mu_{0}^{w_{i}\boldsymbol{y}}(\pi_{i})=\mathbb{E}[w_{i}\boldsymbol{y}|\pi_{i}]$.
Recall that $\boldsymbol{y}=\boldsymbol{s}(\boldsymbol{w}\boldsymbol{x}\gamma_{2}+\boldsymbol{x}\gamma_{3}+\boldsymbol{\lambda}+\boldsymbol{\nu})$.
We can derive
\begin{eqnarray*}
n^{-1}\sum_{i=1}^{n}\mathbb{E}[(w_{i}\boldsymbol{y}-\mu_{0}^{w_{i}\boldsymbol{y}}(\pi_{i}))\nu_{i}] & = & n^{-1}\mathbb{E}[(\boldsymbol{w}\boldsymbol{y}-\mu_{0}^{\boldsymbol{w}\boldsymbol{y}}(\boldsymbol{\pi}))'\boldsymbol{\nu}]\\
 & = & n^{-1}\mathbb{E}[(\boldsymbol{w}\boldsymbol{s}(\boldsymbol{w}\boldsymbol{x}\gamma_{2}+\boldsymbol{x}\gamma_{3}+\boldsymbol{\lambda}+\boldsymbol{\nu})-\mu_{0}^{\boldsymbol{w}\boldsymbol{y}}(\boldsymbol{\pi}))'\boldsymbol{\nu}]\\
 & = & n^{-1}\mathbb{E}[\boldsymbol{\nu}'\boldsymbol{s}'\boldsymbol{w}'\boldsymbol{\nu}]\\
 & = & n^{-1}\mathbb{E}[\text{tr}(\boldsymbol{w}\boldsymbol{s})\mathbb{E}[\nu_{i}^{2}|\boldsymbol{\psi}]].
\end{eqnarray*}
The third equality holds because all terms except $\boldsymbol{\nu}'\boldsymbol{s}'\boldsymbol{w}'\boldsymbol{\nu}$
have zero mean due to Assumption \ref{ass:adj_exog} and $\mathbb{E}[\nu_{i}|\boldsymbol{\psi}]=0$,
as previously discussed, and the fourth equality follows from i.i.d.
$\nu_{i}$.\footnote{Since $\nu_{i}$ and $\psi_{i}=(x'_{i},z'_{i},g_{i})$ are i.i.d.
across $i$, we can derive that $\mathbb{E}[\nu_{i}\nu_{j}|\boldsymbol{\psi}]=\mathbb{E}[\nu_{i}|\boldsymbol{\psi}]\mathbb{E}[\nu_{j}|\boldsymbol{\psi}]=0$
for all $i\neq j$.} Note that
\[
|\text{tr}(\boldsymbol{w}\boldsymbol{s})|=|\text{tr}(\boldsymbol{s}\boldsymbol{w})|\leq n\|\boldsymbol{s}\boldsymbol{w}\|_{\infty}\leq n\interleave\boldsymbol{s}\interleave_{\infty}\|\boldsymbol{w}\|_{\infty}\leq\frac{n}{1-|\gamma_{1}|}\|\boldsymbol{\boldsymbol{w}}\|_{\infty},
\]
where the last inequality follows from $\interleave\boldsymbol{s}\interleave_{\infty}\leq\sum_{r=0}^{\infty}|\gamma_{1}|^{r}\interleave\boldsymbol{w}\interleave_{\infty}^{r}=\sum_{r=0}^{\infty}|\gamma_{1}|^{r}=1/(1-|\gamma_{1}|)$
as $\interleave\boldsymbol{w}\interleave_{\infty}=1$. Hence we can
bound
\begin{align*}
n^{-1}\mathbb{E}[\text{tr}(\boldsymbol{s}\boldsymbol{w})\mathbb{E}[\nu_{i}^{2}|\boldsymbol{\psi}]] & \leq C\mathbb{E}[\|\boldsymbol{w}\|_{\infty}\mathbb{E}[\nu_{i}^{2}|\boldsymbol{\psi}]]\\
 & \leq C\mathbb{E}[\|\boldsymbol{w}\|_{\infty}^{2}]^{1/2}\mathbb{E}[\mathbb{E}[\nu_{i}^{2}|\boldsymbol{\psi}]^{2}]^{1/2}\\
 & =O(n^{-1}),
\end{align*}
where we used $\mathbb{E}[\|\boldsymbol{w}\|_{\infty}^{2}]\leq\mathbb{E}[\|\boldsymbol{w}\|_{\infty}^{8}]^{1/4}=O(n^{-2})$
by Assumption \ref{ass:w}(ii) and $\mathbb{E}[\mathbb{E}[\nu_{i}^{2}|\boldsymbol{\psi}]^{2}]\leq\mathbb{E}[\mathbb{E}[\nu_{i}^{4}|\boldsymbol{\psi}]]=\mathbb{E}[\nu_{i}^{4}]<\infty$
by Assumptions \ref{ass:regular}(ii), \ref{ass:compact}(i), \ref{ass:theta}(i),
and \ref{ass:smooth}(i).\footnote{\label{fn:v_bd}Because $\nu_{i}=\epsilon_{i}-\lambda(\pi_{i})$,
by $(a+b)^{4}\leq8(a^{4}+b^{4})$ we can bound $\mathbb{E}[\nu_{i}^{4}]\leq8(\mathbb{E}[\epsilon_{i}^{4}]+\mathbb{E}[\lambda(\pi_{i})^{4}])\leq8\mathbb{E}[\epsilon_{i}^{8}]^{1/2}+C<\infty$,
where $\lambda(\pi_{i})$ is bounded due to the boundedness of $\pi_{i}$
(Assumptions \ref{ass:compact}(i) and \ref{ass:theta}(i)) and the
continuity of $\lambda$ (Assumption \ref{ass:regular}(ii)).}
\end{proof}
\begin{lem}
\label{lem:muhat_t_consist}For $t_{i}=X_{i}$ or $\epsilon_{i}$,
we have\textup{
\begin{equation}
\frac{1}{n}\sum_{i=1}^{n}(\hat{\mu}^{X}(\hat{\pi}_{i})-\mu_{0}^{X}(\pi_{i}))t'_{i}=o_{p}(1).\label{eq:muhat_t}
\end{equation}
}
\end{lem}
\begin{proof}
Recall that $X_{i}=(w_{i}\boldsymbol{y},w_{i}\boldsymbol{x},x'_{i})'$.
By construction, $\hat{\mu}^{X}(\hat{\pi}_{i})=\hat{\beta}^{X}(\hat{\boldsymbol{\pi}})'b^{K}(\hat{\pi}_{i})\in\mathbb{R}^{d_{X}}$
and $\hat{\mu}^{X}(\pi_{i})=\hat{\beta}^{X}(\boldsymbol{\pi})'b^{K}(\pi_{i})\in\mathbb{R}^{d_{X}}$,
where $\hat{\beta}^{X}(\hat{\boldsymbol{\pi}})=(\hat{B}_{K}'\hat{B}_{K})^{-1}\hat{B}_{K}'\boldsymbol{X}$
and $\hat{\beta}^{X}(\boldsymbol{\pi})=(B{}_{K}'B_{K})^{-1}B{}_{K}'\boldsymbol{X}$,
with $\hat{B}_{K}=B_{K}(\hat{\boldsymbol{\pi}})$ and $B_{K}=B_{K}(\boldsymbol{\pi})$.
Denote $\boldsymbol{\mu}_{0}^{X}=(\mu_{0}^{X}(\pi_{1}),\dots,\mu_{0}^{X}(\pi_{n}))'$
and $\boldsymbol{t}=(t_{1},\dots,t_{n})'$. The left-hand side of
equation (\ref{eq:muhat_t}) satisfies
\begin{eqnarray*}
\|n^{-1}\sum_{i=1}^{n}(\hat{\mu}^{X}(\hat{\pi}_{i})-\mu_{0}^{X}(\pi_{i}))t'_{i}\|^{2} & = & n^{-2}\|(\hat{B}{}_{K}\hat{\beta}^{X}(\hat{\boldsymbol{\pi}})-\boldsymbol{\mu}_{0}^{X})'\boldsymbol{t}\|^{2}\\
 & \leq & n^{-2}\|\hat{B}{}_{K}\hat{\beta}^{X}(\hat{\boldsymbol{\pi}})-\boldsymbol{\mu}_{0}^{X}\|^{2}\|\boldsymbol{t}\|^{2},
\end{eqnarray*}
where the inequality follows from the submultiplicativity of the Frobenius
norm.

For $t_{i}=\epsilon_{i}$, because $\epsilon_{i}$ is i.i.d., by the
law of large numbers and Assumption \ref{ass:smooth}(i) $n^{-1}\|\boldsymbol{t}\|^{2}=n^{-1}\sum_{i}\epsilon_{i}^{2}=\mathbb{E}[\epsilon_{i}^{2}]+o_{p}(1)=O_{p}(1)$.
For $t_{i}=X_{i}$, $n^{-1}\|\boldsymbol{t}\|^{2}=n^{-1}\sum_{i}\|X_{i}\|^{2}=n^{-1}\sum_{i}(w_{i}\boldsymbol{y})^{2}+n^{-1}\sum_{i}\|w_{i}\boldsymbol{x}\|^{2}+n^{-1}\sum_{i}\|x_{i}\|^{2}$.
The last two terms are bounded because $\max_{i}\|w_{i}\boldsymbol{x}\|<\infty$
and $\max_{i}\|x_{i}\|<\infty$ (Assumptions \ref{ass:compact}(ii)
and \ref{ass:w}(i)). Moreover, $n^{-1}\sum_{i}(w_{i}\boldsymbol{y})^{2}=n^{-1}(\boldsymbol{w}\boldsymbol{y})'\boldsymbol{w}\boldsymbol{y}=O_{p}(1)$
by Lemma \ref{lem:wy_bd}. We conclude that $n^{-1}\|\boldsymbol{t}\|^{2}=O_{p}(1)$.

By the triangle inequality and $(a+b+c)^{2}\leq3(a^{2}+b^{2}+c^{2})$,
\begin{eqnarray*}
 &  & n^{-1}\|\hat{B}{}_{K}\hat{\beta}^{X}(\hat{\boldsymbol{\pi}})-\boldsymbol{\mu}_{0}^{X}\|^{2}\\
 & \leq & n^{-1}(\|(\hat{B}{}_{K}-B{}_{K})\hat{\beta}^{X}(\hat{\boldsymbol{\pi}})\|+\|B{}_{K}(\hat{\beta}^{X}(\hat{\boldsymbol{\pi}})-\beta^{X})\|+\|B{}_{K}\beta^{X}-\boldsymbol{\mu}_{0}^{X}\|)^{2}\\
 & \leq & 3n^{-1}(\|\hat{B}{}_{K}-B{}_{K}\|^{2}\|\hat{\beta}^{X}(\hat{\boldsymbol{\pi}})\|^{2}+\|B{}_{K}(\hat{\beta}^{X}(\hat{\boldsymbol{\pi}})-\beta^{X})\|^{2}+\|B{}_{K}\beta^{X}-\boldsymbol{\mu}_{0}^{X}\|^{2}).
\end{eqnarray*}
It suffices to show that the last three terms are $o_{p}(1)$.

By equation (\ref{eq:Bhat_rate}), $n^{-1}\|\hat{B}{}_{K}-B{}_{K}\|^{2}=O_{p}(\varrho_{1}(K)^{2}/n)$.
Moreover,
\begin{align}
\|\hat{\beta}^{X}(\hat{\boldsymbol{\pi}})\|^{2} & =\text{tr}(\boldsymbol{X}'\hat{B}{}_{K}(\hat{B}'_{K}\hat{B}_{K})^{-2}\hat{B}'{}_{K}\boldsymbol{X})\nonumber \\
 & \leq O_{p}(n^{-1})\text{tr}(\boldsymbol{X}'\hat{B}{}_{K}(\hat{B}'_{K}\hat{B}_{K})^{-1}\hat{B}'{}_{K}\boldsymbol{X})\nonumber \\
 & \leq O_{p}(n^{-1})\text{tr}(\boldsymbol{X}'\boldsymbol{X})=O_{p}(1).\label{eq:betahat_tauhat}
\end{align}
The first inequality follows from Lemmas \ref{lem:Q_K} and \ref{lem:Q_Khat}.\footnote{By Lemmas \ref{lem:Q_K} and \ref{lem:Q_Khat}, the smallest eigenvalue
of $\hat{Q}_{K}=\hat{B}_{K}'\hat{B}_{K}/n$ converges to one in probability
and hence $(\hat{B}_{K}'\hat{B}_{K}/n)^{-1}\leq CI_{K}$ with probability
approaching one.} The second inequality follows because $\hat{B}{}_{K}(\hat{B}'_{K}\hat{B}_{K})^{-1}\hat{B}'{}_{K}$
is idempotent and thus $\hat{B}{}_{K}(\hat{B}'_{K}\hat{B}_{K})^{-1}\hat{B}'{}_{K}\leq I_{K}$.
The last equality holds because $n^{-1}\text{tr}(\boldsymbol{X}'\boldsymbol{X})=n^{-1}\sum_{i}\|X_{i}\|^{2}=O_{p}(1)$
as previously shown. We conclude that $n^{-1}\|\hat{B}{}_{K}-B{}_{K}\|^{2}\|\hat{\beta}^{X}(\hat{\boldsymbol{\pi}})\|^{2}=o_{p}(1)$.

Observe that $n^{-1}\|B{}_{K}(\hat{\beta}^{X}(\hat{\boldsymbol{\pi}})-\beta^{X})\|^{2}=n^{-1}\|B_{K}(\hat{\beta}^{X}(\hat{\boldsymbol{\pi}})-\beta^{X})\|^{2}=n^{-1}\text{tr}((\hat{\beta}^{X}(\hat{\boldsymbol{\pi}})-\beta^{X})'B'_{K}B{}_{K}(\hat{\beta}^{X}(\hat{\boldsymbol{\pi}})-\beta^{Z}))\leq O_{p}(1)\|\hat{\beta}^{X}(\hat{\boldsymbol{\pi}})-\beta^{X}\|^{2}$,
where the inequality holds because by Lemma \ref{lem:Q_K} $B'_{K}B{}_{K}/n\leq CI_{K}$
with probability approaching one. By the triangle inequality, $\|\hat{\beta}^{X}(\hat{\boldsymbol{\pi}})-\beta^{X}\|\leq\|\hat{\beta}^{X}(\hat{\boldsymbol{\pi}})-\hat{\beta}^{X}(\boldsymbol{\pi})\|+\|\hat{\beta}^{X}(\boldsymbol{\pi})-\beta^{X}\|$.
Lemma \ref{lem:betahat_dif} shows that $\|\hat{\beta}^{X}(\hat{\boldsymbol{\pi}})-\hat{\beta}^{X}(\boldsymbol{\pi})\|=O_{p}(\varrho_{1}(K)/\sqrt{n})$.
Moreover, by Lemma 15.3 in \citet{liracine2007} for $x_{i}$, Lemma
\ref{lem:wt_betahat} for $w_{i}\boldsymbol{x}$ and $w_{i}\boldsymbol{y}$,
we have $\|\hat{\beta}^{X}(\boldsymbol{\pi})-\beta^{X}\|=o_{p}(1)$.
Combining these results yields $n^{-1}\|B{}_{K}(\hat{\beta}^{X}(\hat{\boldsymbol{\pi}})-\beta^{X})\|^{2}=o_{p}(1)$.

Finally, $n^{-1}\|B_{K}\beta^{X}-\boldsymbol{\mu}_{0}^{X}\|^{2}=n^{-1}\sum_{i=1}^{n}\|\beta^{X\prime}b^{K}(\pi_{i})-\mu_{0}^{X}(\pi_{i})\|^{2}\leq\sup_{\pi}\|\beta^{X\prime}b^{K}(\pi)-\mu_{0}^{X}(\pi)\|^{2}=O(K^{-2a})$
by Assumption \ref{ass:sieve}(ii).
\end{proof}
\begin{lem}
\label{lem:xt_consist}For $t_{i}=X_{i}$ or $\epsilon_{i}$, we have
\begin{equation}
\frac{1}{n}\sum_{i=1}^{n}((X_{i}-\mu_{0}^{X}(\pi_{i}))t'_{i}-\mathbb{E}[(X_{i}-\mu_{0}^{X}(\pi_{i}))t'_{i}])=o_{p}(1).\label{eq:ulln}
\end{equation}
\end{lem}
\begin{proof}
Recall that $X_{i}=(w_{i}\boldsymbol{y},w_{i}\boldsymbol{x},x'_{i})'$.
Because both $X_{i}$ and $t_{i}$ are finite-dimensional, we can
prove equation (\ref{eq:ulln}) component-wise. Depending on the components
of $X_{i}$ and $t_{i}$ under consideration, we will divide the proof
into nine cases, as listed in Table \ref{tab:xt_consist_cases}.
\begin{table}[t]
\centering
\caption{\protect\label{tab:xt_consist_cases}The Cases in Lemma \ref{lem:xt_consist}}
\smallskip{}
\begin{tabular}{lcccc}
\toprule 
 &  & \multicolumn{3}{c}{Component of $t_{i}$}\tabularnewline
\cmidrule{3-5}
 &  & $x_{i},\epsilon_{i}$ & $w_{i}\boldsymbol{x}$ & $w_{i}\boldsymbol{y}$\tabularnewline
\midrule
\multirow{3}{*}{Component of $X_{i}$} & $x_{i}$ & Case (a) & Case (b) & Case (c)\tabularnewline
 & $w_{i}\boldsymbol{x}$ & Case (d) & Case (e) & Case (f)\tabularnewline
 & $w_{i}\boldsymbol{y}$ & Case (g) & Case (h) & Case (i)\tabularnewline
\bottomrule
\end{tabular}
\end{table}

Case (a): Since $x_{i}$ and $\epsilon_{i}$ are i.i.d., equation
(\ref{eq:ulln}) follows by the law of large numbers.

Case (b): We can write
\begin{eqnarray}
 &  & n^{-1}\sum_{i=1}^{n}((x_{i}-\mu_{0}^{x}(\pi_{i}))w_{i}\boldsymbol{x}-\mathbb{E}[(x_{i}-\mu_{0}^{x}(\pi_{i}))w_{i}\boldsymbol{x}])\nonumber \\
 & = & n^{-1}\sum_{i=1}^{n}(x_{i}w_{i}\boldsymbol{x}-\mathbb{E}[x_{i}w_{i}\boldsymbol{x}])-n^{-1}\sum_{i=1}^{n}(\mu_{0}^{x}(\pi_{i})w_{i}\boldsymbol{x}-\mathbb{E}[\mu_{0}^{x}(\pi_{i})w_{i}\boldsymbol{x}])\nonumber \\
 & = & n^{-1}(\boldsymbol{x}'\boldsymbol{w}\boldsymbol{x}-\mathbb{E}[\boldsymbol{x}'\boldsymbol{w}\boldsymbol{x}])-n^{-1}(\boldsymbol{\mu}_{0}^{\boldsymbol{x}}(\boldsymbol{\pi})'\boldsymbol{w}\boldsymbol{x}-\mathbb{E}[\boldsymbol{\mu}_{0}^{\boldsymbol{x}}(\boldsymbol{\pi})'\boldsymbol{w}\boldsymbol{x}]),\label{eq:case_x_wx}
\end{eqnarray}
where $\boldsymbol{\mu}_{0}^{\boldsymbol{x}}(\boldsymbol{\pi})=(\mu_{0}^{x}(\pi_{1})',\dots,\mu_{0}^{x}(\pi_{n})')'$.
Applying Lemma \ref{lem:aqb_consist} to the last two terms with $\boldsymbol{a}=\boldsymbol{x}$
or $\boldsymbol{\mu}_{0}^{\boldsymbol{x}}(\boldsymbol{\pi})$, $\boldsymbol{b}=\boldsymbol{x}$,
and $\boldsymbol{q}=\boldsymbol{w}$, we can show that both terms
are $o_{p}(1)$.

Case (c): 
\begin{eqnarray}
 &  & n^{-1}\sum_{i=1}^{n}((x_{i}-\mu_{0}^{x}(\pi_{i}))w_{i}\boldsymbol{y}-\mathbb{E}[(x_{i}-\mu_{0}^{x}(\pi_{i}))w_{i}\boldsymbol{y}])\nonumber \\
 & = & n^{-1}\sum_{i=1}^{n}(x_{i}w_{i}\boldsymbol{y}-\mathbb{E}[x_{i}w_{i}\boldsymbol{y}])-n^{-1}\sum_{i=1}^{n}(\mu_{0}^{x}(\pi_{i})w_{i}\boldsymbol{y}-\mathbb{E}[\mu_{0}^{x}(\pi_{i})w_{i}\boldsymbol{y}])\nonumber \\
 & = & n^{-1}(\boldsymbol{x}'\boldsymbol{w}\boldsymbol{y}-\mathbb{E}[\boldsymbol{x}'\boldsymbol{w}\boldsymbol{y}])-n^{-1}(\boldsymbol{\mu}_{0}^{\boldsymbol{x}}(\boldsymbol{\pi})'\boldsymbol{w}\boldsymbol{y}-\mathbb{E}[\boldsymbol{\mu}_{0}^{\boldsymbol{x}}(\boldsymbol{\pi})'\boldsymbol{w}\boldsymbol{y}]).\label{eq:case_x_wy}
\end{eqnarray}
Recall that $\boldsymbol{y}=\boldsymbol{s}(\boldsymbol{w}\boldsymbol{x}\gamma_{2}+\boldsymbol{x}\gamma_{3}+\boldsymbol{\epsilon})$,
where $\boldsymbol{s}=(I_{n}-\gamma_{1}\boldsymbol{w})^{-1}$, and
thus $\boldsymbol{w}\boldsymbol{y}=\boldsymbol{s}(\boldsymbol{w}^{2}\boldsymbol{x}\gamma_{2}+\boldsymbol{w}\boldsymbol{x}\gamma_{3}+\boldsymbol{w}\boldsymbol{\epsilon})$.
Therefore, we can express the last two terms as
\begin{eqnarray}
 &  & n^{-1}(\boldsymbol{x}'\boldsymbol{w}\boldsymbol{y}-\mathbb{E}[\boldsymbol{x}'\boldsymbol{w}\boldsymbol{y}])\nonumber \\
 & = & n^{-1}(\boldsymbol{x}'\boldsymbol{s}\boldsymbol{w}^{2}\boldsymbol{x}-\mathbb{E}[\boldsymbol{x}'\boldsymbol{s}\boldsymbol{w}^{2}\boldsymbol{x}])\gamma_{2}+n^{-1}(\boldsymbol{x}'\boldsymbol{s}\boldsymbol{w}\boldsymbol{x}-\mathbb{E}[\boldsymbol{x}'\boldsymbol{s}\boldsymbol{w}\boldsymbol{x}])\gamma_{3}\nonumber \\
 &  & +n^{-1}(\boldsymbol{x}'\boldsymbol{s}\boldsymbol{w}\boldsymbol{\epsilon}-\mathbb{E}[\boldsymbol{x}'\boldsymbol{s}\boldsymbol{w}\boldsymbol{\epsilon}])\label{eq:case_x_wy_1}
\end{eqnarray}
and 
\begin{eqnarray}
 &  & n^{-1}(\boldsymbol{\mu}_{0}^{\boldsymbol{x}}(\boldsymbol{\pi})'\boldsymbol{w}\boldsymbol{y}-\mathbb{E}[\boldsymbol{\mu}_{0}^{\boldsymbol{x}}(\boldsymbol{\pi})'\boldsymbol{w}\boldsymbol{y}])\nonumber \\
 & = & n^{-1}(\boldsymbol{\mu}_{0}^{\boldsymbol{x}}(\boldsymbol{\pi})'\boldsymbol{s}\boldsymbol{w}^{2}\boldsymbol{x}-\mathbb{E}[\boldsymbol{\mu}_{0}^{\boldsymbol{x}}(\boldsymbol{\pi})'\boldsymbol{s}\boldsymbol{w}^{2}\boldsymbol{x}])\gamma_{2}+n^{-1}(\boldsymbol{\mu}_{0}^{\boldsymbol{x}}(\boldsymbol{\pi})'\boldsymbol{s}\boldsymbol{w}\boldsymbol{x}-\mathbb{E}[\boldsymbol{\mu}_{0}^{\boldsymbol{x}}(\boldsymbol{\pi})'\boldsymbol{s}\boldsymbol{w}\boldsymbol{x}])\gamma_{3}\nonumber \\
 &  & +n^{-1}(\boldsymbol{\mu}_{0}^{\boldsymbol{x}}(\boldsymbol{\pi})'\boldsymbol{s}\boldsymbol{w}\boldsymbol{\epsilon}-\mathbb{E}[\boldsymbol{\mu}_{0}^{\boldsymbol{x}}(\boldsymbol{\pi})'\boldsymbol{s}\boldsymbol{w}\boldsymbol{\epsilon}]).\label{eq:case_x_wy_2}
\end{eqnarray}
Applying Lemma \ref{lem:aqb_consist} to each term on the right-hand
sides of equations (\ref{eq:case_x_wy_1}) and (\ref{eq:case_x_wy_2})
with $\boldsymbol{a}=\boldsymbol{x}$ or $\boldsymbol{\mu}_{0}^{\boldsymbol{x}}(\boldsymbol{\pi})$,
$\boldsymbol{b}=\boldsymbol{x}$ or $\boldsymbol{\epsilon}$, and
$\boldsymbol{q}=\boldsymbol{s}\boldsymbol{w}^{2}$ or $\boldsymbol{s}\boldsymbol{w}$,
we can show that both (\ref{eq:case_x_wy_1}) and (\ref{eq:case_x_wy_2})
are equal to $o_{p}(1)$.

Case (d): We take $t_{i}=x_{i}$ as an example; the case for $t_{i}=\epsilon_{i}$
can be proved similarly. Taking transpose yields
\begin{eqnarray*}
 &  & n^{-1}\sum_{i=1}^{n}(x{}_{i}(w_{i}\boldsymbol{x}-\mu_{0}^{w_{i}\boldsymbol{x}}(\pi_{i}))-\mathbb{E}[x_{i}(w_{i}\boldsymbol{x}-\mu_{0}^{w_{i}\boldsymbol{x}}(\pi_{i}))])\\
 & = & n^{-1}\sum_{i=1}^{n}(x_{i}w_{i}\boldsymbol{x}-\mathbb{E}[x_{i}w_{i}\boldsymbol{x}])-n^{-1}\sum_{i=1}^{n}(x{}_{i}\mu_{0}^{w_{i}\boldsymbol{x}}(\pi_{i})-\mathbb{E}[x_{i}\mu_{0}^{w_{i}\boldsymbol{x}}(\pi_{i})]).
\end{eqnarray*}
Since $x_{i}\mu_{0}^{w_{i}\boldsymbol{x}}(\pi_{i})$ is independent
across $i,$ the second term on the right-hand side is $o_{p}(1)$
by the law of large numbers. The first term on the right-hand side
coincides with the first term in the second line of equation (\ref{eq:case_x_wx})
and is thus $o_{p}(1)$.

Case (e): We can write
\begin{eqnarray*}
 &  & n^{-1}\sum_{i=1}^{n}((w_{i}\boldsymbol{x}-\mu_{0}^{w_{i}\boldsymbol{x}}(\pi_{i}))'w_{i}\boldsymbol{x}-\mathbb{E}[(w_{i}\boldsymbol{x}-\mu_{0}^{w_{i}\boldsymbol{x}}(\pi_{i}))'w_{i}\boldsymbol{x}])\\
 & = & n^{-1}\sum_{i=1}^{n}((w_{i}\boldsymbol{x})'w_{i}\boldsymbol{x}-\mathbb{E}[(w_{i}\boldsymbol{x})'w_{i}\boldsymbol{x}])-n^{-1}\sum_{i=1}^{n}(\mu_{0}^{w_{i}\boldsymbol{x}}(\pi_{i})'w_{i}\boldsymbol{x}-\mathbb{E}[\mu_{0}^{w_{i}\boldsymbol{x}}(\pi_{i})'w_{i}\boldsymbol{x}])\\
 & = & n^{-1}(\boldsymbol{x}'\boldsymbol{w}'\boldsymbol{w}\boldsymbol{x}-\mathbb{E}[\boldsymbol{x}'\boldsymbol{w}'\boldsymbol{w}\boldsymbol{x}])-n^{-1}(\boldsymbol{\mu}_{0}^{\boldsymbol{w}\boldsymbol{x}}(\boldsymbol{\pi})'\boldsymbol{w}\boldsymbol{x}-\mathbb{E}[\boldsymbol{\mu}_{0}^{\boldsymbol{w}\boldsymbol{x}}(\boldsymbol{\pi})'\boldsymbol{w}\boldsymbol{x}]),
\end{eqnarray*}
where $\boldsymbol{\mu}_{0}^{\boldsymbol{w}\boldsymbol{x}}(\boldsymbol{\pi})=(\mu_{0}^{w_{1}\boldsymbol{x}}(\pi_{1})',\dots,\mu_{0}^{w_{n}\boldsymbol{x}}(\pi_{n})')'$.
Applying Lemma \ref{lem:aqb_consist} to the last two terms with $\boldsymbol{a}=\boldsymbol{x}$
or $\boldsymbol{\mu}_{0}^{\boldsymbol{w}\boldsymbol{x}}(\boldsymbol{\pi})$,
$\boldsymbol{b}=\boldsymbol{x}$, and $\boldsymbol{q}=\boldsymbol{w}$
or $\boldsymbol{w}'\boldsymbol{w}$, we can show that both terms are
$o_{p}(1)$.

Case (f): 
\begin{eqnarray*}
 &  & n^{-1}\sum_{i=1}^{n}((w_{i}\boldsymbol{x}-\mu_{0}^{w_{i}\boldsymbol{x}}(\pi_{i}))'w_{i}\boldsymbol{y}-\mathbb{E}[(w_{i}\boldsymbol{x}-\mu_{0}^{w_{i}\boldsymbol{x}}(\pi_{i}))'w_{i}\boldsymbol{y}])\\
 & = & n^{-1}\sum_{i=1}^{n}((w_{i}\boldsymbol{x})'w_{i}\boldsymbol{y}-\mathbb{E}[(w_{i}\boldsymbol{x})'w_{i}\boldsymbol{y}])-n^{-1}\sum_{i=1}^{n}(\mu_{0}^{w_{i}\boldsymbol{x}}(\pi_{i})'w_{i}\boldsymbol{y}-\mathbb{E}[\mu_{0}^{w_{i}\boldsymbol{x}}(\pi_{i})'w_{i}\boldsymbol{y}])\\
 & = & n^{-1}(\boldsymbol{x}'\boldsymbol{w}'\boldsymbol{w}\boldsymbol{y}-\mathbb{E}[\boldsymbol{x}'\boldsymbol{w}'\boldsymbol{w}\boldsymbol{y}])-n^{-1}(\boldsymbol{\mu}_{0}^{\boldsymbol{w}\boldsymbol{x}}(\boldsymbol{\pi})'\boldsymbol{w}\boldsymbol{y}-\mathbb{E}[\boldsymbol{\mu}_{0}^{\boldsymbol{w}\boldsymbol{x}}(\boldsymbol{\pi})'\boldsymbol{w}\boldsymbol{y}]).
\end{eqnarray*}
Similarly as in case (c), we can express
\begin{eqnarray}
 &  & n^{-1}(\boldsymbol{x}'\boldsymbol{w}'\boldsymbol{w}\boldsymbol{y}-\mathbb{E}[\boldsymbol{x}'\boldsymbol{w}'\boldsymbol{w}\boldsymbol{y}])\nonumber \\
 & = & n^{-1}(\boldsymbol{x}'\boldsymbol{w}'\boldsymbol{s}\boldsymbol{w}^{2}\boldsymbol{x}-\mathbb{E}[\boldsymbol{x}'\boldsymbol{w}'\boldsymbol{s}\boldsymbol{w}^{2}\boldsymbol{x}])\gamma_{2}+n^{-1}(\boldsymbol{x}'\boldsymbol{w}'\boldsymbol{s}\boldsymbol{w}\boldsymbol{x}-\mathbb{E}[\boldsymbol{x}'\boldsymbol{w}'\boldsymbol{s}\boldsymbol{w}\boldsymbol{x}])\gamma_{3}\nonumber \\
 &  & +n^{-1}(\boldsymbol{x}'\boldsymbol{w}'\boldsymbol{s}\boldsymbol{w}\boldsymbol{\epsilon}-\mathbb{E}[\boldsymbol{x}'\boldsymbol{w}'\boldsymbol{s}\boldsymbol{w}\boldsymbol{\epsilon}])\label{eq:case_wx_wy_1}
\end{eqnarray}
and 
\begin{eqnarray}
 &  & n^{-1}(\boldsymbol{\mu}_{0}^{\boldsymbol{w}\boldsymbol{x}}(\boldsymbol{\pi})'\boldsymbol{w}\boldsymbol{y}-\mathbb{E}[\boldsymbol{\mu}_{0}^{\boldsymbol{w}\boldsymbol{x}}(\boldsymbol{\pi})'\boldsymbol{w}\boldsymbol{y}])\nonumber \\
 & = & n^{-1}(\boldsymbol{\mu}_{0}^{\boldsymbol{w}\boldsymbol{x}}(\boldsymbol{\pi})'\boldsymbol{s}\boldsymbol{w}^{2}\boldsymbol{x}-\mathbb{E}[\boldsymbol{\mu}_{0}^{\boldsymbol{w}\boldsymbol{x}}(\boldsymbol{\pi})'\boldsymbol{s}\boldsymbol{w}^{2}\boldsymbol{x}])\gamma_{2}+n^{-1}(\boldsymbol{\mu}_{0}^{\boldsymbol{w}\boldsymbol{x}}(\boldsymbol{\pi})'\boldsymbol{s}\boldsymbol{w}\boldsymbol{x}-\mathbb{E}[\boldsymbol{\mu}_{0}^{\boldsymbol{w}\boldsymbol{x}}(\boldsymbol{\pi})'\boldsymbol{s}\boldsymbol{w}\boldsymbol{x}])\gamma_{3}\nonumber \\
 &  & +n^{-1}(\boldsymbol{\mu}_{0}^{\boldsymbol{w}\boldsymbol{x}}(\boldsymbol{\pi})'\boldsymbol{s}\boldsymbol{w}\boldsymbol{\epsilon}-\mathbb{E}[\boldsymbol{\mu}_{0}^{\boldsymbol{w}\boldsymbol{x}}(\boldsymbol{\pi})'\boldsymbol{s}\boldsymbol{w}\boldsymbol{\epsilon}]).\label{eq:case_wx_wy_2}
\end{eqnarray}
Applying Lemma \ref{lem:aqb_consist} to each term on the right-hand
sides of equations (\ref{eq:case_wx_wy_1}) and (\ref{eq:case_wx_wy_2})
with $\boldsymbol{a}=\boldsymbol{x}$ or $\boldsymbol{\mu}_{0}^{\boldsymbol{x}}(\boldsymbol{\pi})$,
$\boldsymbol{b}=\boldsymbol{x}$ or $\boldsymbol{\epsilon}$, and
$\boldsymbol{q}=\boldsymbol{w}'\boldsymbol{s}\boldsymbol{w}^{2}$,
$\boldsymbol{w}'\boldsymbol{s}\boldsymbol{w}$, $\boldsymbol{s}\boldsymbol{w}^{2}$,
or $\boldsymbol{s}\boldsymbol{w}$, we can show that both (\ref{eq:case_wx_wy_1})
and (\ref{eq:case_wx_wy_2}) are equal to $o_{p}(1)$.

Case (g): We consider $t_{i}=x_{i}$. Taking transpose yields
\begin{eqnarray*}
 &  & n^{-1}\sum_{i=1}^{n}(x{}_{i}(w_{i}\boldsymbol{y}-\mu_{0}^{w_{i}\boldsymbol{y}}(\pi_{i}))-\mathbb{E}[x_{i}(w_{i}\boldsymbol{y}-\mu_{0}^{w_{i}\boldsymbol{y}}(\pi_{i}))])\\
 & = & n^{-1}\sum_{i=1}^{n}(x_{i}w_{i}\boldsymbol{y}-\mathbb{E}[x_{i}w_{i}\boldsymbol{y}])-n^{-1}\sum_{i=1}^{n}(x{}_{i}\mu_{0}^{w_{i}\boldsymbol{\boldsymbol{y}}}(\pi_{i})-\mathbb{E}[x_{i}\mu_{0}^{w_{i}\boldsymbol{\boldsymbol{y}}}(\pi_{i})])
\end{eqnarray*}
The first term on the right-hand side is $o_{p}(1)$ following the
argument for equation (\ref{eq:case_x_wy_1}). Moreover, because $x_{i}\mu_{0}^{w_{i}\boldsymbol{\boldsymbol{y}}}(\pi_{i})$
is independent across $i$, the second term on the right-hand side
is $o_{p}(1)$ by the law of large numbers.

Case (h): 
\begin{eqnarray}
 &  & n^{-1}\sum_{i=1}^{n}((w_{i}\boldsymbol{y}-\mu_{0}^{w_{i}\boldsymbol{y}}(\pi_{i}))w_{i}\boldsymbol{x}-\mathbb{E}[(w_{i}\boldsymbol{y}-\mu_{0}^{w_{i}\boldsymbol{y}}(\pi_{i}))w_{i}\boldsymbol{x}])\nonumber \\
 & = & n^{-1}\sum_{i=1}^{n}(w_{i}\boldsymbol{y}w_{i}\boldsymbol{x}-\mathbb{E}[w_{i}\boldsymbol{y}w_{i}\boldsymbol{x}])-n^{-1}\sum_{i=1}^{n}(\mu_{0}^{w_{i}\boldsymbol{y}}(\pi_{i})w_{i}\boldsymbol{x}-\mathbb{E}[\mu_{0}^{w_{i}\boldsymbol{y}}(\pi_{i})w_{i}\boldsymbol{x}])\nonumber \\
 & = & n^{-1}(\boldsymbol{y}'\boldsymbol{w}'\boldsymbol{w}\boldsymbol{x}-\mathbb{E}[\boldsymbol{y}'\boldsymbol{w}'\boldsymbol{w}\boldsymbol{x}])-n^{-1}(\boldsymbol{\mu}_{0}^{\boldsymbol{w}\boldsymbol{y}}(\boldsymbol{\pi})'\boldsymbol{w}\boldsymbol{x}-\mathbb{E}[\boldsymbol{\mu}_{0}^{\boldsymbol{w}\boldsymbol{y}}(\boldsymbol{\pi})'\boldsymbol{w}\boldsymbol{x}]).\label{eq:case_wy_wx}
\end{eqnarray}
where $\boldsymbol{\mu}_{0}^{\boldsymbol{w}\boldsymbol{y}}(\boldsymbol{\pi})=(\mu_{0}^{w_{1}\boldsymbol{y}}(\pi_{1}),\dots,\mu_{0}^{w_{n}\boldsymbol{y}}(\pi_{n}))'$.
Following the argument for equation (\ref{eq:case_wx_wy_1}), the
first term in the last line is $o_{p}(1)$. In addition, applying
Lemma \ref{lem:aqb_consist} with $\boldsymbol{a}=\boldsymbol{\mu}_{0}^{\boldsymbol{w}\boldsymbol{y}}(\boldsymbol{\pi})$,
$\boldsymbol{b}=\boldsymbol{x}$, and $\boldsymbol{q}=\boldsymbol{w}$,
the second term in the last line is also $o_{p}(1)$.

Case (i): 
\begin{eqnarray*}
 &  & n^{-1}\sum_{i=1}^{n}((w_{i}\boldsymbol{y}-\mu_{0}^{w_{i}\boldsymbol{y}}(\pi_{i}))w_{i}\boldsymbol{y}-\mathbb{E}[(w_{i}\boldsymbol{y}-\mu_{0}^{w_{i}\boldsymbol{y}}(\pi_{i}))w_{i}\boldsymbol{y}])\\
 & = & n^{-1}\sum_{i=1}^{n}(w_{i}\boldsymbol{y}w_{i}\boldsymbol{y}-\mathbb{E}[w_{i}\boldsymbol{y}w_{i}\boldsymbol{y}])-n^{-1}(\boldsymbol{\mu}_{0}^{\boldsymbol{w}\boldsymbol{y}}(\boldsymbol{\pi})'\boldsymbol{w}\boldsymbol{y}-\mathbb{E}[\boldsymbol{\mu}_{0}^{\boldsymbol{w}\boldsymbol{y}}(\boldsymbol{\pi})'\boldsymbol{w}\boldsymbol{y}])\\
 & = & n^{-1}(\boldsymbol{y}'\boldsymbol{w}'\boldsymbol{w}\boldsymbol{y}-\mathbb{E}[\boldsymbol{y}'\boldsymbol{w}'\boldsymbol{w}\boldsymbol{y}])-n^{-1}(\mu_{0}^{w_{i}\boldsymbol{y}}(\pi_{i})w_{i}\boldsymbol{y}-\mathbb{E}[\mu_{0}^{w_{i}\boldsymbol{y}}(\pi_{i})w_{i}\boldsymbol{y}])
\end{eqnarray*}
The second term on the right-hand side can be analyzed similarly to
equation (\ref{eq:case_wx_wy_2}) and is $o_{p}(1)$. To show the
first term on the right-hand side is $o_{p}(1)$, note that
\begin{eqnarray}
 &  & n^{-1}(\boldsymbol{y}'\boldsymbol{w}'\boldsymbol{w}\boldsymbol{y}-\mathbb{E}[\boldsymbol{y}'\boldsymbol{w}'\boldsymbol{w}\boldsymbol{y}])\nonumber \\
 & = & n^{-1}\gamma'_{2}(\boldsymbol{x}'(\boldsymbol{w}')^{2}\boldsymbol{s}'\boldsymbol{s}\boldsymbol{w}^{2}\boldsymbol{x}-\mathbb{E}[\boldsymbol{x}'(\boldsymbol{w}')^{2}\boldsymbol{s}'\boldsymbol{s}\boldsymbol{w}^{2}\boldsymbol{x}])\gamma_{2}\nonumber \\
 &  & +n^{-1}\gamma'_{2}(\boldsymbol{x}'(\boldsymbol{w}')^{2}\boldsymbol{s}'\boldsymbol{s}\boldsymbol{w}\boldsymbol{x}-\mathbb{E}[\boldsymbol{x}'(\boldsymbol{w}')^{2}\boldsymbol{s}'\boldsymbol{s}\boldsymbol{w}\boldsymbol{x}])\gamma_{3}\nonumber \\
 &  & +n^{-1}\gamma'_{2}(\boldsymbol{x}'(\boldsymbol{w}')^{2}\boldsymbol{s}'\boldsymbol{s}\boldsymbol{w}\boldsymbol{\epsilon}-\mathbb{E}[\boldsymbol{x}'(\boldsymbol{w}')^{2}\boldsymbol{s}'\boldsymbol{s}\boldsymbol{w}\boldsymbol{\epsilon}])\nonumber \\
 &  & +n^{-1}\gamma'_{3}(\boldsymbol{x}'\boldsymbol{w}'\boldsymbol{s}'\boldsymbol{s}\boldsymbol{w}^{2}\boldsymbol{x}-\mathbb{E}[\boldsymbol{x}'\boldsymbol{w}'\boldsymbol{s}'\boldsymbol{s}\boldsymbol{w}^{2}\boldsymbol{x}])\gamma_{2}+n^{-1}\gamma'_{3}(\boldsymbol{x}'\boldsymbol{w}'\boldsymbol{s}'\boldsymbol{s}\boldsymbol{w}\boldsymbol{x}-\mathbb{E}[\boldsymbol{x}'\boldsymbol{w}'\boldsymbol{s}'\boldsymbol{s}\boldsymbol{w}\boldsymbol{x}])\gamma_{3}\nonumber \\
 &  & +n^{-1}\gamma'_{3}(\boldsymbol{x}'\boldsymbol{w}'\boldsymbol{s}'\boldsymbol{s}\boldsymbol{w}\boldsymbol{\epsilon}-\mathbb{E}[\boldsymbol{x}'\boldsymbol{w}'\boldsymbol{s}'\boldsymbol{s}\boldsymbol{w}\boldsymbol{\epsilon}])+n^{-1}(\boldsymbol{\epsilon}'\boldsymbol{w}'\boldsymbol{s}'\boldsymbol{s}\boldsymbol{w}^{2}\boldsymbol{x}-\mathbb{E}[\boldsymbol{\epsilon}'\boldsymbol{w}'\boldsymbol{s}'\boldsymbol{s}\boldsymbol{w}^{2}\boldsymbol{x}])\gamma_{2}\nonumber \\
 &  & +n^{-1}(\boldsymbol{\epsilon}'\boldsymbol{w}'\boldsymbol{s}'\boldsymbol{s}\boldsymbol{w}\boldsymbol{x}-\mathbb{E}[\boldsymbol{\epsilon}'\boldsymbol{w}'\boldsymbol{s}'\boldsymbol{s}\boldsymbol{w}\boldsymbol{x}])\gamma_{2}+n^{-1}(\boldsymbol{\epsilon}'\boldsymbol{w}'\boldsymbol{s}'\boldsymbol{s}\boldsymbol{w}\boldsymbol{\epsilon}-\mathbb{E}[\boldsymbol{\epsilon}'\boldsymbol{w}'\boldsymbol{s}'\boldsymbol{s}\boldsymbol{w}\boldsymbol{\epsilon}]).\label{eq:case_wy_wy}
\end{eqnarray}
The result follows by applying Lemma \ref{lem:aqb_consist} to each
term on the right-hand side with $\boldsymbol{a}=\boldsymbol{x}$
or $\boldsymbol{\epsilon}$, $\boldsymbol{b}=\boldsymbol{x}$ or $\boldsymbol{\epsilon}$,
and $\boldsymbol{q}=(\boldsymbol{w}')^{2}\boldsymbol{s}'\boldsymbol{s}\boldsymbol{w}^{2}$,
$(\boldsymbol{w}')^{2}\boldsymbol{s}'\boldsymbol{s}\boldsymbol{w}$,
$\boldsymbol{w}'\boldsymbol{s}'\boldsymbol{s}\boldsymbol{w}^{2}$,
or $\boldsymbol{w}'\boldsymbol{s}'\boldsymbol{s}\boldsymbol{w}$.
\end{proof}

\begin{lem}[Boundness of $\boldsymbol{w}\boldsymbol{y}$]
\label{lem:wy_bd}$n^{-1}(\boldsymbol{w}\boldsymbol{y})'\boldsymbol{w}\boldsymbol{y}=O_{p}(1).$
\end{lem}
\begin{proof}
Let $T=\boldsymbol{w}^{2}\boldsymbol{x}\gamma_{2}+\boldsymbol{w}\boldsymbol{x}\gamma_{3}$
be an $n\times1$ vector, and recall that $\boldsymbol{w}\boldsymbol{y}=\boldsymbol{s}(\boldsymbol{w}^{2}\boldsymbol{x}\gamma_{2}+\boldsymbol{w}\boldsymbol{x}\gamma_{3}+\boldsymbol{w}\boldsymbol{\epsilon})=\boldsymbol{s}(T+\boldsymbol{w}\boldsymbol{\epsilon})$.
We can write
\begin{eqnarray}
n^{-1}(\boldsymbol{w}\boldsymbol{y})'\boldsymbol{w}\boldsymbol{y} & = & n^{-1}(T+\boldsymbol{w}\boldsymbol{\epsilon})'\boldsymbol{s}'\boldsymbol{s}(T+\boldsymbol{w}\boldsymbol{\epsilon})\nonumber \\
 & = & n^{-1}T^{\prime}\boldsymbol{s}'\boldsymbol{s}T+2n^{-1}T^{\prime}\boldsymbol{s}'\boldsymbol{s}\boldsymbol{w}\boldsymbol{\epsilon}+n^{-1}\boldsymbol{\epsilon}'\boldsymbol{w}'\boldsymbol{s}'\boldsymbol{s}\boldsymbol{w}\boldsymbol{\epsilon}.\label{eq:wy_bd}
\end{eqnarray}

By Lemma \ref{lem:network_bd}, both $\boldsymbol{w}$ and $\boldsymbol{s}$
are uniformly bounded in both row and column sums. By the boundedness
of $x_{i}$ and $\gamma$, we can bound $\|T\|_{\infty}\leq\interleave\boldsymbol{w}\interleave_{\infty}^{2}\|\boldsymbol{x}\gamma_{2}\|_{\infty}+\interleave\boldsymbol{w}\interleave_{\infty}\|\boldsymbol{x}\gamma_{3}\|_{\infty}<\infty$.
Therefore, the first term in the last line of (\ref{eq:wy_bd}) is
$n^{-1}T^{\prime}\boldsymbol{s}'\boldsymbol{s}T\leq\interleave\boldsymbol{s}'\boldsymbol{s}\interleave_{\infty}\|T\|_{\infty}^{2}\leq\interleave\boldsymbol{s}\interleave_{1}\interleave\boldsymbol{s}\interleave_{\infty}\|T\|_{\infty}^{2}=O_{p}(1)$.
The second to last term in equation (\ref{eq:wy_bd}) satisfies $n^{-1}|T^{\prime}\boldsymbol{s}'\boldsymbol{s}\boldsymbol{w}\boldsymbol{\epsilon}|\leq\interleave\boldsymbol{s}'\boldsymbol{s}\boldsymbol{w}\interleave_{\infty}\|T\|_{\infty}\|\boldsymbol{\epsilon}/n\|_{1}=O_{p}(1)$,
because $\|\boldsymbol{\epsilon}/n\|_{1}=n^{-1}\sum_{i}|\epsilon_{i}|=\mathbb{E}[|\epsilon_{i}|]+o_{p}(1)=O_{p}(1)$
by the law of large numbers and Assumption \ref{ass:smooth}(i) and
$\interleave\boldsymbol{s}'\boldsymbol{s}\boldsymbol{w}\interleave_{\infty}\leq\interleave\boldsymbol{s}\interleave_{1}\interleave\boldsymbol{s}\interleave_{\infty}\interleave\boldsymbol{w}\interleave_{\infty}=O_{p}(1)$.
Finally, the last term in (\ref{eq:wy_bd}) satisfies $n^{-1}\boldsymbol{\epsilon}'\boldsymbol{w}'\boldsymbol{s}'\boldsymbol{s}\boldsymbol{w}\boldsymbol{\epsilon}\leq n^{-1}\lambda_{\max}(\boldsymbol{w}'\boldsymbol{s}'\boldsymbol{s}\boldsymbol{w})\boldsymbol{\epsilon}'\boldsymbol{\epsilon}=O_{p}(1)$,
because $n^{-1}\boldsymbol{\epsilon}'\boldsymbol{\epsilon}=n^{-1}\sum_{i=1}^{n}\epsilon_{i}^{2}=\mathbb{E}[\epsilon_{i}^{2}]+o_{p}(1)=O_{p}(1)$
and $\lambda_{\max}(\boldsymbol{w}'\boldsymbol{s}'\boldsymbol{s}\boldsymbol{w})\leq\interleave\boldsymbol{w}'\boldsymbol{s}'\boldsymbol{s}\boldsymbol{w}\interleave_{\infty}\leq\interleave\boldsymbol{w}\interleave_{1}\interleave\boldsymbol{s}\interleave_{1}\interleave\boldsymbol{s}\interleave_{\infty}\interleave\boldsymbol{w}\interleave_{\infty}=O_{p}(1)$.
Combining the three terms, we complete the proof.
\end{proof}

\begin{lem}
\label{lem:Q_K}\textup{Let $Q_{K}=B_{K}'B_{K}/n$. Then $\|Q_{K}-I_{K}\|=O_{p}(\varrho_{0}(K)\sqrt{K/n})$.}
\end{lem}
\begin{proof}
The result follows from Lemma 15.2 in \citet[p.481]{liracine2007}.
\end{proof}

\begin{lem}
\label{lem:Q_Khat} Let $\hat{Q}_{K}=\hat{B}_{K}'\hat{B}_{K}/n$.
Then $\|\hat{Q}_{K}-Q_{K}\|=O_{p}(\varrho_{1}(K)/\sqrt{n})$.
\end{lem}
\begin{proof}
Because $\hat{B}{}_{K}'\hat{B}{}_{K}-B_{K}'B{}_{K}=(\hat{B}_{K}-B_{K})^{2}+B_{K}'(\hat{B}_{K}-B_{K})+(\hat{B}_{K}-B_{K})'B_{K}$,
we have $\|\hat{Q}_{K}-Q_{K}\|=\|\hat{B}{}_{K}'\hat{B}{}_{K}-B_{K}'B{}_{K}\|/n\leq\|\hat{B}{}_{K}-B_{K}\|^{2}/n+2\|(\hat{B}_{K}-B_{K})'B_{K}\|/n$.
The $\sqrt{n}$-consistency of $\hat{\theta}$ and boundedness of
$z$ (Assumptions \ref{ass:compact}(i) and \ref{ass:theta}(ii))
imply that $\max_{i}\|\hat{\pi}_{i}-\pi_{i}\|=O_{p}(n^{-1/2})$. Therefore,
\begin{equation}
\|\hat{B}_{K}-B_{K}\|=(\sum_{i=1}^{n}\|b^{K}(\hat{\pi}_{i})-b^{K}(\pi_{i})\|^{2})^{1/2}\leq n^{1/2}\varrho_{1}(K)\max_{i}\text{\ensuremath{\|}}\hat{\pi}_{i}-\pi_{i}\|=O_{p}(\varrho_{1}(K)),\label{eq:Bhat_rate}
\end{equation}
by the mean-value theorem and Assumption \ref{ass:sieve}(iv). Moreover,
\begin{eqnarray*}
\|(\hat{B}_{K}-B_{K})'B_{K}\|/n & = & \text{tr}((\hat{B}_{K}-B_{K})'B_{K}B_{K}'(\hat{B}_{K}-B_{K}))^{1/2}/n\\
 & \leq & O_{p}(1)\text{tr}((\hat{B}_{K}-B_{K})'B_{K}(B_{K}'B_{K})^{-1}B_{K}'(\hat{B}_{K}-B_{K}))^{1/2}/\sqrt{n}\\
 & \leq & O_{p}(1)\|\hat{B}_{K}-B_{K}\|/\sqrt{n}=O_{p}(\varrho_{1}(K)/\sqrt{n}).
\end{eqnarray*}
The first inequality above holds because by Lemma \ref{lem:Q_K} $I_{K}\leq C(B'_{K}B{}_{K}/n)^{-1}$
with probability approaching one. The second inequality follows by
$B_{K}(B_{K}'B_{K})^{-1}B_{K}'$ idempotent. The last equality follows
from equation (\ref{eq:Bhat_rate}). We conclude that $\|\hat{Q}_{K}-Q_{K}\|\leq O_{p}(\varrho_{1}(K)^{2}/n)+O_{p}(\varrho_{1}(K)/\sqrt{n})=O_{p}(\varrho_{1}(K)/\sqrt{n})$.
\end{proof}

\begin{lem}
\label{lem:betahat_dif}$\|\hat{\beta}^{X}(\hat{\boldsymbol{\pi}})-\hat{\beta}^{X}(\boldsymbol{\pi})\|=O_{p}(\varrho_{1}(K)/\sqrt{n})$.
\end{lem}
\begin{proof}
Recall that $\hat{\beta}^{X}(\hat{\boldsymbol{\pi}})=\hat{Q}_{K}^{-1}\hat{B}'_{K}\boldsymbol{X}/n$
and $\hat{\beta}^{X}(\boldsymbol{\pi})=Q_{K}^{-1}B'{}_{K}\boldsymbol{X}/n$.
We have
\begin{eqnarray*}
\|\hat{\beta}^{X}(\hat{\boldsymbol{\pi}})-\hat{\beta}^{X}(\boldsymbol{\pi})\| & = & \textrm{tr}(\boldsymbol{X}'(\hat{Q}_{K}^{-1}\hat{B}'_{K}-Q_{K}^{-1}B'{}_{K})'(\hat{Q}_{K}^{-1}\hat{B}'_{K}-Q_{K}^{-1}B'{}_{K})\boldsymbol{X}/n^{2})^{1/2}\\
 & \leq & \|(\hat{Q}_{K}^{-1}\hat{B}'_{K}-Q_{K}^{-1}B'{}_{K})/\sqrt{n}\|\textrm{tr}(\boldsymbol{X}'\boldsymbol{X}/n)^{1/2}.
\end{eqnarray*}
As shown in Lemma \ref{lem:muhat_t_consist}, $\text{tr}(\boldsymbol{X}'\boldsymbol{X}/n)=n^{-1}\sum_{i=1}^{n}\|X_{i}\|^{2}=O_{p}(1)$.
Moreover, $\|(\hat{Q}_{K}^{-1}\hat{B}'_{K}-Q_{K}^{-1}B'{}_{K})/\sqrt{n}\|\leq\|(\hat{Q}_{K}^{-1}-Q_{K}^{-1})\hat{B}'_{K}/\sqrt{n}\|+\|Q_{K}^{-1}(\hat{B}{}_{K}-B_{K})'/\sqrt{n}\|$.
Observe
\begin{eqnarray*}
\|(\hat{Q}_{K}^{-1}-Q_{K}^{-1})\hat{B}'_{K}/\sqrt{n}\| & = & \textrm{tr}((\hat{Q}_{K}^{-1}-Q_{K}^{-1})\hat{B}'_{K}\hat{B}{}_{K}(\hat{Q}_{K}^{-1}-Q_{K}^{-1})/n)^{1/2}\\
 & = & \textrm{tr}(Q_{K}^{-1}(Q_{K}-\hat{Q}_{K})\hat{Q}_{K}^{-1}(Q_{K}-\hat{Q}_{K})Q_{K}^{-1})^{1/2}\\
 & \leq & O_{p}(1)\textrm{tr}((Q_{K}-\hat{Q}_{K})Q_{K}^{-2}(Q_{K}-\hat{Q}_{K}))^{1/2}\\
 & \leq & O_{p}(1)\|Q_{K}-\hat{Q}_{K}\|=O_{p}(\varrho_{1}(K)/\sqrt{n}),
\end{eqnarray*}
where the inequalities follow from Lemmas \ref{lem:Q_K} and \ref{lem:Q_Khat}.\footnote{By Lemmas \ref{lem:Q_K} and \ref{lem:Q_Khat}, the smallest eigenvalue
of $\hat{Q}_{K}$ converges to one in probability and hence the largest
eigenvalue of $\hat{Q}_{K}^{-1}$ is bounded with probability approaching
one. Similarly, by Lemma \ref{lem:Q_K}, the largest eigenvalue of
$Q_{K}^{-2}$ is bounded with probability approaching one.} The last equality follows from Lemma \ref{lem:Q_Khat}. As for the
second term, by equation (\ref{eq:Bhat_rate}), we have $\|Q_{K}^{-1}(\hat{B}{}_{K}-B{}_{K})'/\sqrt{n}\|=\textrm{tr}((\hat{B}{}_{K}-B{}_{K})Q_{K}^{-2}(\hat{B}{}_{K}-B{}_{K})'/n)^{1/2}\leq O_{p}(1)\|(\hat{B}{}_{K}-B{}_{K})/\sqrt{n}\|=O_{p}(\varrho_{1}(K)/\sqrt{n}).$
\end{proof}
\begin{lem}[Consistency of sieve]
\label{lem:wt_betahat}For $\boldsymbol{t}=\boldsymbol{x}$ or $\boldsymbol{y}$,
$\|\hat{\beta}^{\boldsymbol{w}\boldsymbol{t}}(\boldsymbol{\pi})-\beta^{\boldsymbol{w}\boldsymbol{t}}\|=o_{p}(1)$.
\end{lem}
\begin{proof}
Recall that $\hat{\beta}^{\boldsymbol{w}\boldsymbol{x}}(\boldsymbol{\pi})=Q_{K}^{-1}B'_{K}\boldsymbol{w}\boldsymbol{x}/n$.
We can write $\hat{\beta}^{\boldsymbol{w}\boldsymbol{t}}(\boldsymbol{\pi})-\beta^{\boldsymbol{w}\boldsymbol{t}}=Q_{K}^{-1}B'_{K}(\boldsymbol{w}\boldsymbol{t}-B_{K}\beta^{\boldsymbol{w}\boldsymbol{t}})/n$.
Observe that 
\begin{eqnarray*}
\|\hat{\beta}^{\boldsymbol{w}\boldsymbol{t}}(\boldsymbol{\pi})-\beta^{\boldsymbol{w}\boldsymbol{t}}\| & = & \text{tr}((\boldsymbol{w}\boldsymbol{t}-B_{K}\beta^{\boldsymbol{w}\boldsymbol{t}})'B_{K}Q_{K}^{-2}B'_{K}(\boldsymbol{w}\boldsymbol{t}-B_{K}\beta^{\boldsymbol{w}\boldsymbol{t}})/n^{2})^{1/2}\\
 & \leq & O_{p}(1)\|B'_{K}(\boldsymbol{w}\boldsymbol{t}-B_{K}\beta^{\boldsymbol{w}\boldsymbol{t}})/n\|,
\end{eqnarray*}
where we used that the largest eigenvalue of $Q_{K}^{-2}$ is bounded
with probability approaching one. Write $\boldsymbol{w}\boldsymbol{t}-B_{K}\beta^{\boldsymbol{w}\boldsymbol{t}}=(\boldsymbol{w}\boldsymbol{t}-\boldsymbol{\mu}_{0}^{\boldsymbol{w}\boldsymbol{t}})+(\boldsymbol{\mu}_{0}^{\boldsymbol{w}\boldsymbol{t}}-B_{K}\beta^{\boldsymbol{w}\boldsymbol{t}})$.
We derive
\begin{eqnarray*}
 &  & \|B'_{K}(\boldsymbol{\mu}_{0}^{\boldsymbol{w}\boldsymbol{t}}-B_{K}\beta^{\boldsymbol{w}\boldsymbol{t}})/n\|\\
 & = & \text{tr}((\boldsymbol{\mu}_{0}^{\boldsymbol{w}\boldsymbol{t}}-B_{K}\beta^{\boldsymbol{w}\boldsymbol{t}})'B_{K}B'_{K}(\boldsymbol{\mu}_{0}^{\boldsymbol{w}\boldsymbol{t}}-B_{K}\beta^{\boldsymbol{w}\boldsymbol{t}})/n^{2})^{1/2}\\
 & \leq & O_{p}(1)\text{tr}((\boldsymbol{\mu}_{0}^{\boldsymbol{w}\boldsymbol{t}}-B_{K}\beta^{\boldsymbol{w}\boldsymbol{t}})'B_{K}(B'_{K}B_{K})^{-1}B'_{K}(\boldsymbol{\mu}_{0}^{\boldsymbol{w}\boldsymbol{t}}-B_{K}\beta^{\boldsymbol{w}\boldsymbol{t}})/n)^{1/2}\\
 & \leq & O_{p}(1)\|(\boldsymbol{\mu}_{0}^{\boldsymbol{w}\boldsymbol{t}}-B_{K}\beta^{\boldsymbol{w}\boldsymbol{t}})/\sqrt{n}\|=O_{p}(K^{-a}).
\end{eqnarray*}
The first inequality follows from Lemma \ref{lem:Q_K}.\footnote{By Lemma \ref{lem:Q_K}, the largest eigenvalue of $Q_{K}=B'_{K}B_{K}/n$
converges to one in probability and hence $CI_{K}\leq(B'_{K}B_{K}/n)^{-1}$
with probability approaching one.} The second inequality holds because $B_{K}(B'_{K}B_{K})^{-1}B'_{K}$
is idempotent, and the last equality follows by Assumption \ref{ass:sieve}(ii).
If we can show $B'_{K}(\boldsymbol{w}\boldsymbol{t}-\boldsymbol{\mu}_{0}^{\boldsymbol{w}\boldsymbol{t}})/n=o_{p}(1)$,
then combining the results completes the proof. Because $x_{i}$ is
finite dimensional, we can prove the equation for each element of
$x_{i}$ separately. Without loss of generality, we assume that $x_{i}$
is a scalar for notation simplicity.

We start with $\boldsymbol{t}=\boldsymbol{x}$. Write $B'_{K}(\boldsymbol{w}\boldsymbol{x}-\boldsymbol{\mu}_{0}^{\boldsymbol{w}\boldsymbol{x}})/n=n^{-1}\sum_{i}\sum_{j}b^{K}(\pi_{i})(w_{ij}x_{j}-\mathbb{E}[w_{ij}x_{j}|\pi_{i}])=n^{-1}\sum_{i}\sum_{j}r_{ij}^{wx}$,
where $r_{ij}^{wx}\equiv b^{K}(\pi_{i})(w_{ij}x_{j}-\mathbb{E}[w_{ij}x_{j}|\pi_{i}])$.
Since $\mathbb{E}[r_{ij}^{wx}|\pi_{i}]=0$, we have $\mathbb{E}[r_{ij}^{wx}]=0$.
Then
\begin{equation}
\mathbb{E}\|B'_{K}(\boldsymbol{wx}-\boldsymbol{\mu}_{0}^{\boldsymbol{wx}})/n\|^{2}=n^{-2}\sum_{(i,j)}\sum_{(k,l):\{i,j\}\cap\{k,l\}\neq\emptyset}\mathbb{E}[r_{ij}^{wx\prime}r_{kl}^{wx}]+n^{-2}\sum_{(i,j)}\sum_{(k,l):\{i,j\}\cap\{k,l\}=\emptyset}\mathbb{E}[r_{ij}^{wx\prime}r_{kl}^{wx}].\label{eq:Bwx_sq}
\end{equation}
For any $i,j,k,l\in\mathcal{N}$, $|\mathbb{E}[r_{ij}^{wx\prime}r_{kl}^{wx}]|\leq\mathbb{E}|b^{K}(\pi_{i})'b^{K}(\pi_{k})(w_{ij}x_{j}-\mathbb{E}[w_{ij}x_{j}|\pi_{i}])(w_{kl}x_{l}-\mathbb{E}[w_{kl}x_{l}|\pi_{k}])|\leq O(n^{-2})(\mathbb{E}[(b^{K}(\pi_{i})'b^{K}(\pi_{k}))^{2}])^{1/2}=O(n^{-2}\sqrt{K})$.
The second inequality follows from Assumptions \ref{ass:compact}(ii)
and \ref{ass:w}(ii).\footnote{\label{fn:cov(wx)}By Cauchy-Schwarz inequality, $(a+b)^{4}\leq8(a^{4}+b^{4})$,
Jensen's inequality, and iterated expectations, we have $(\mathbb{E}[(w_{ij}x_{j}-\mathbb{E}[w_{ij}x_{j}|\pi_{i}])^{2}(w_{kl}x_{l}-\mathbb{E}[w_{kl}x_{l}|\pi_{k}])^{2}])^{1/2}\leq(\mathbb{E}[(w_{ij}x_{j}-\mathbb{E}[w_{ij}x_{j}|\pi_{i}])^{4}])^{1/4}(\mathbb{E}[(w_{kl}x_{l}-\mathbb{E}[w_{kl}x_{l}|\pi_{k}])^{4}])^{1/4}\leq4(\mathbb{E}[(w_{ij}x_{j})^{4}])^{1/4}(\mathbb{E}[(w_{kl}x_{l})^{4}])^{1/4}\leq C\mathbb{E}[\|\boldsymbol{w}\|_{\infty}^{4}]^{1/2}=O(n^{-2}).$} The last equality holds because by Assumptions \ref{ass:regular}(i)
and \ref{ass:sieve}(i), $\mathbb{E}[(b^{K}(\pi_{i})'b^{K}(\pi_{k}))^{2}]=\mathbb{E}[b^{K}(\pi_{i})'b^{K}(\pi_{k})b^{K}(\pi_{k})'b^{K}(\pi_{i})]=\mathbb{E}[\text{tr}(b^{K}(\pi_{i})b^{K}(\pi_{i})'b^{K}(\pi_{k})b^{K}(\pi_{k})')]=\text{tr}(\mathbb{E}[b^{K}(\pi_{i})b^{K}(\pi_{i})']\mathbb{E}[b^{K}(\pi_{k})b^{K}(\pi_{k})'])=\text{tr}(I_{K})=K$.
The sum over overlapping $\{i,j\}$ and $\{k,l\}$ contains $O(n^{3})$
terms. Therefore, the first term in equation (\ref{eq:Bwx_sq}) is
$n^{-2}\cdot O(n^{3})\cdot O(n^{-2}\sqrt{K})=O(\sqrt{K}/n)$.

Moreover, for disjoint $\{i,j\}$ and $\{k,l\}$, we have
\begin{eqnarray}
\mathbb{E}[r_{ij}^{wx\prime}r_{kl}^{wx}|\boldsymbol{\psi}] & = & b^{K}(\pi_{i})'b^{K}(\pi_{k})\mathbb{E}[(w_{ij}x_{j}-\mathbb{E}[w_{ij}x_{j}|\pi_{i}])(w_{kl}x_{l}-\mathbb{E}[w_{kl}x_{l}|\pi_{k}])|\boldsymbol{\psi}]\nonumber \\
 & = & b^{K}(\pi_{i})'b^{K}(\pi_{k})(\mathbb{E}[w_{ij}w_{kl}|\boldsymbol{\psi}]x_{j}x_{l}-\mathbb{E}[w_{ij}|\boldsymbol{\psi}]x_{j}\mathbb{E}[w_{kl}x_{l}|\pi_{k}]\nonumber \\
 &  & -\mathbb{E}[w_{ij}x_{j}|\pi_{i}]\mathbb{E}[w_{kl}|\boldsymbol{\psi}]x_{l}+\mathbb{E}[w_{ij}x_{j}|\pi_{i}])\mathbb{E}[w_{kl}x_{l}|\pi_{k}])\nonumber \\
 & = & b^{K}(\pi_{i})'b^{K}(\pi_{k})(\mathbb{E}[w_{ij}|\psi_{i},\psi_{j}]x_{j}-\mathbb{E}[w_{ij}x_{j}|\pi_{i}])(\mathbb{E}[w_{kl}|\psi_{k},\psi_{l}]x_{l}-\mathbb{E}[w_{kl}x_{l}|\pi_{k}])\nonumber \\
 &  & +b^{K}(\pi_{i})'b^{K}(\pi_{k})e_{ij,kl}^{wx},\label{eq:cov(rij)_wx}
\end{eqnarray}
where
\begin{eqnarray*}
e_{ij,kl}^{wx} & \equiv & (\mathbb{E}[w_{ij}w_{kl}|\boldsymbol{\psi}]-\mathbb{E}[w_{ij}|\psi_{i},\psi_{j}]\mathbb{E}[w_{kl}|\psi_{k},\psi_{l}])x_{j}x_{l}\\
 &  & -(\mathbb{E}[w_{ij}|\boldsymbol{\psi}]-\mathbb{E}[w_{ij}|\psi_{i},\psi_{j}])x_{j}\mathbb{E}[w_{kl}x_{l}|\pi_{k}]\\
 &  & -(\mathbb{E}[w_{kl}|\boldsymbol{\psi}]-\mathbb{E}[w_{kl}|\psi_{k},\psi_{l}])x_{l}\mathbb{E}[w_{ij}x_{j}|\pi_{i}].
\end{eqnarray*}
Observe that for disjoint $\{i,j\}$ and $\{k,l\}$ the terms $b^{K}(\pi_{i})(\mathbb{E}[w_{ij}|\psi_{i},\psi_{j}]x_{j}-\mathbb{E}[w_{ij}x_{j}|\pi_{i}])$
and $b^{K}(\pi_{k})(\mathbb{E}[w_{kl}|\psi_{k},\psi_{l}]x_{l}-\mathbb{E}[w_{kl}x_{l}|\pi_{k}])$
are independent, both with mean zero. Hence, the first term in the
last line of equation (\ref{eq:cov(rij)_wx}) has mean zero. Moreover,
by Assumptions \ref{ass:compact}(ii) and \ref{ass:w}(ii)(iii) and
Cauchy-Schwarz inequality, we obtain $\max_{i,j,k,l\in\mathcal{N}:\{i,j\}\cap\{k,l\}=\emptyset}\mathbb{E}[(e_{ij,kl}^{wx})^{2}]\leq o(n^{-4}/K)$.
Therefore, we can derive the bound $|\mathbb{E}[r_{ij}^{wx\prime}r_{kl}^{wx}]|=|\mathbb{E}[b^{K}(\pi_{i})'b^{K}(\pi_{k})e_{ij,kl}^{wx}]|\leq(\mathbb{E}[(b^{K}(\pi_{i})'b^{K}(\pi_{k}))^{2}])^{1/2}(\mathbb{E}[(e_{ij,kl}^{wx})^{2}])^{1/2}\leq\sqrt{K}\cdot o(n^{-2}/\sqrt{K})=o(n^{-2})$
uniformly in disjoint $\{i,j\}$ and $\{k,l\}$. The sum over disjoint
$\{i,j\}$ and $\{k,l\}$ contains $O(n^{4})$ terms. Hence, the second
term in (\ref{eq:Bwx_sq}) can be bounded by $n^{-2}\cdot O(n^{4})\cdot o(n^{-2})=o(1)$.
Combining the results we prove $\mathbb{E}\|B'_{K}(\boldsymbol{wx}-\boldsymbol{\mu}_{0}^{\boldsymbol{wx}})/n\|^{2}=o(1)$
and thus $B'_{K}(\boldsymbol{wx}-\boldsymbol{\mu}_{0}^{\boldsymbol{wx}})/n=o_{p}(1)$.

Next we consider $\boldsymbol{t}=\boldsymbol{y}$. Recall that $\boldsymbol{w}\boldsymbol{y}=\boldsymbol{s}(\boldsymbol{w}^{2}\boldsymbol{x}\gamma_{2}+\boldsymbol{w}\boldsymbol{x}\gamma_{3}+\boldsymbol{w}\boldsymbol{\epsilon})$
and $\boldsymbol{\mu}_{0}^{\boldsymbol{w}\boldsymbol{y}}=\boldsymbol{\mu}_{0}^{\boldsymbol{s}\boldsymbol{w}^{2}\boldsymbol{x}}\gamma_{2}+\boldsymbol{\mu}_{0}^{\boldsymbol{s}\boldsymbol{w}\boldsymbol{x}}\gamma_{3}+\boldsymbol{\mu}_{0}^{\boldsymbol{\boldsymbol{s}\boldsymbol{w}}\boldsymbol{\epsilon}}$.
We can write
\begin{eqnarray*}
 &  & B'_{K}(\boldsymbol{wy}-\boldsymbol{\mu}_{0}^{\boldsymbol{wy}})/n\\
 & = & B'_{K}(\boldsymbol{s}\boldsymbol{w}^{2}\boldsymbol{x}-\boldsymbol{\mu}_{0}^{\boldsymbol{s}\boldsymbol{w}^{2}\boldsymbol{x}})\gamma_{2}/n+B'_{K}(\boldsymbol{s}\boldsymbol{w}\boldsymbol{x}-\boldsymbol{\mu}_{0}^{\boldsymbol{s}\boldsymbol{w}\boldsymbol{x}})\gamma_{3}/n+B'_{K}(\boldsymbol{s}\boldsymbol{w}\boldsymbol{\epsilon}-\boldsymbol{\mu}_{0}^{\boldsymbol{s}\boldsymbol{w}\boldsymbol{\epsilon}})/n.
\end{eqnarray*}
We will demonstrate that all three terms on the right-hand side are
$o_{p}(1)$, following the proof for $B'_{K}(\boldsymbol{wx}-\boldsymbol{\mu}_{0}^{\boldsymbol{wx}})/n$
with some modifications. First, we modify the proof for $B'_{K}(\boldsymbol{wx}-\boldsymbol{\mu}_{0}^{\boldsymbol{wx}})/n$
by replacing $\boldsymbol{w}$ with $\boldsymbol{q}=\boldsymbol{s}\boldsymbol{w}^{2}$
or $\boldsymbol{s}\boldsymbol{w}$, replacing $r_{ij}^{wx}$ with
$r_{ij}^{qx}\equiv b^{K}(\pi_{i})(q_{ij}x_{j}-\mathbb{E}[q_{ij}x_{j}|\pi_{i}])$,
and replacing $e_{ij,kl}^{wx}$ with
\begin{eqnarray*}
e_{ij,kl}^{qx} & \equiv & (\mathbb{E}[q_{ij}q_{kl}|\boldsymbol{\psi}]-\mathbb{E}[q_{ij}|\psi_{i},\psi_{j}]\mathbb{E}[q_{kl}|\psi_{k},\psi_{l}])x_{j}x_{l}\\
 &  & -(\mathbb{E}[q_{ij}|\boldsymbol{\psi}]-\mathbb{E}[q_{ij}|\psi_{i},\psi_{j}])x_{j}\mathbb{E}[q_{kl}x_{l}|\pi_{k}]\\
 &  & -(\mathbb{E}[q_{kl}|\boldsymbol{\psi}]-\mathbb{E}[q_{kl}|\psi_{k},\psi_{l}])x_{l}\mathbb{E}[q_{ij}x_{j}|\pi_{i}],
\end{eqnarray*}
which also satisfies $\max_{i,j,k,l\in\mathcal{N}:\{i,j\}\cap\{k,l\}=\emptyset}\mathbb{E}[(e_{ij,kl}^{qx})^{2}]\leq o(n^{-4}/K)$
under Assumptions \ref{ass:compact}(ii) and \ref{ass:w}(i)-(iii).
Following a similar argument, we can show that $B'_{K}(\boldsymbol{q}\boldsymbol{x}-\boldsymbol{\mu}_{0}^{\boldsymbol{q}\boldsymbol{x}})/n=o_{p}(1)$
for both $\boldsymbol{q}=\boldsymbol{s}\boldsymbol{w}^{2}$ and $\boldsymbol{q}=\boldsymbol{s}\boldsymbol{w}$.\footnote{The argument in Footnote \ref{fn:cov(wx)} still holds when we replace
$\boldsymbol{w}$ by $\boldsymbol{q}=\boldsymbol{s}\boldsymbol{w}^{2}$
or $\boldsymbol{s}\boldsymbol{w}$ because by Lemma \ref{lem:network_bd}
and Assumption \ref{ass:w}(i)(ii) we have $\mathbb{E}[\|\boldsymbol{q}\|_{\infty}^{4}]\leq n^{4}\mathbb{E}[\|\boldsymbol{w}\|_{\infty}^{8}]=O(n^{-4})$
for both forms of $\boldsymbol{q}$.} Second, we modify the proof for $B'_{K}(\boldsymbol{wx}-\boldsymbol{\mu}_{0}^{\boldsymbol{wx}})/n$
by replacing $\boldsymbol{w}$ with $\boldsymbol{q}=\boldsymbol{s}\boldsymbol{w}$,
replacing $\boldsymbol{x}$ with $\boldsymbol{\epsilon}$, and replacing
$r_{ij}^{wx}$ with $r_{ij}^{q\epsilon}\equiv b^{K}(\pi_{i})(q_{ij}\epsilon_{j}-\mathbb{E}[q_{ij}\epsilon_{j}|\pi_{i}])$.
Note that by Assumption \ref{ass:adj_exog} and iterated expectations,
we have $\mathbb{E}[q_{ij}\epsilon_{j}|\pi_{i}]=\mathbb{E}[\mathbb{E}[q_{ij}\epsilon_{j}|\boldsymbol{\psi}]|\pi_{i}]=\mathbb{E}[\mathbb{E}[q_{ij}|\boldsymbol{\psi}]\mathbb{E}[\epsilon_{j}|\boldsymbol{\psi}]|\pi_{i}]=\mathbb{E}[\mathbb{E}[q_{ij}|\boldsymbol{\psi}]\lambda(\pi_{j})|\pi_{i}]=\mathbb{E}[q_{ij}\lambda(\pi_{j})|\pi_{i}]$.
Moreover, we can derive
\begin{eqnarray*}
\mathbb{E}[r_{ij}^{q\epsilon\prime}r_{kl}^{q\epsilon}] & = & \mathbb{E}[b^{K}(\pi_{i})'b^{K}(\pi_{k})(q_{ij}\epsilon_{j}-\mathbb{E}[q_{ij}\epsilon_{j}|\pi_{i}])(q_{kl}\epsilon_{l}-\mathbb{E}[q_{kl}\epsilon_{l}|\pi_{k}])]\\
 & = & \mathbb{E}[b^{K}(\pi_{i})'b^{K}(\pi_{k})(\mathbb{E}[q_{ij}q_{kl}|\boldsymbol{\psi}]\lambda(\pi_{j})\lambda(\pi_{l})-\mathbb{E}[q_{ij}|\boldsymbol{\psi}]\lambda(\pi_{j})\mathbb{E}[q_{kl}\lambda(\pi_{l})|\pi_{k}]\\
 &  & -\mathbb{E}[q_{ij}\lambda(\pi_{j})|\pi_{i}]\mathbb{E}[q_{kl}|\boldsymbol{\psi}]\lambda(\pi_{l})+\mathbb{E}[q_{ij}\lambda(\pi_{j})|\pi_{i}]\mathbb{E}[q_{kl}\lambda(\pi_{l})|\pi_{k}])]\\
 & = & \mathbb{E}[b^{K}(\pi_{i})'b^{K}(\pi_{k})(q_{ij}\lambda(\pi_{j})-\mathbb{E}[q_{ij}\lambda(\pi_{j})|\pi_{i}])(q_{kl}\lambda(\pi_{l})-\mathbb{E}[q_{kl}\lambda(\pi_{l})|\pi_{k}])].
\end{eqnarray*}
Furthermore, for disjoint $\{i,j\}$ and $\{k,l\}$, we can show
\begin{eqnarray*}
\mathbb{E}[r_{ij}^{q\epsilon\prime}r_{kl}^{q\epsilon}|\boldsymbol{\psi}] & = & b^{K}(\pi_{i})'b^{K}(\pi_{k})\mathbb{E}[(q_{ij}\epsilon_{j}-\mathbb{E}[q_{ij}\epsilon_{j}|\pi_{i}])(q_{kl}\epsilon_{l}-\mathbb{E}[q_{kl}\epsilon_{l}|\pi_{k}])|\boldsymbol{\psi}]\\
 & = & b^{K}(\pi_{i})'b^{K}(\pi_{k})(\mathbb{E}[q_{ij}q_{kl}|\boldsymbol{\psi}]\lambda(\pi_{j})\lambda(\pi_{l})-\mathbb{E}[q_{ij}|\boldsymbol{\psi}]\lambda(\pi_{j})\mathbb{E}[q_{kl}\lambda(\pi_{l})|\pi_{k}]\\
 &  & -\mathbb{E}[q_{ij}\lambda(\pi_{j})|\pi_{i}]\mathbb{E}[q_{kl}|\boldsymbol{\psi}]\lambda(\pi_{l})+\mathbb{E}[q_{ij}\lambda(\pi_{j})|\pi_{i}])\mathbb{E}[q_{kl}\lambda(\pi_{l})|\pi_{k}])\\
 & = & b^{K}(\pi_{i})'b^{K}(\pi_{k})(\mathbb{E}[q_{ij}q_{kl}|\boldsymbol{\psi}]\lambda(\pi_{j})\lambda(\pi_{l})-\mathbb{E}[q_{ij}|\boldsymbol{\psi}]\lambda(\pi_{j})\mathbb{E}[q_{kl}\lambda(\pi_{l})|\pi_{k}]\\
 &  & -\mathbb{E}[q_{ij}\lambda(\pi_{j})|\pi_{i}]\mathbb{E}[q_{kl}|\boldsymbol{\psi}]\lambda(\pi_{l})+\mathbb{E}[q_{ij}\lambda(\pi_{j})|\pi_{i}])\mathbb{E}[q_{kl}\lambda(\pi_{l})|\pi_{k}])\\
 & = & b^{K}(\pi_{i})'b^{K}(\pi_{k})(\mathbb{E}[q_{ij}|\psi_{i},\psi_{j}]\lambda(\pi_{j})-\mathbb{E}[q_{ij}\lambda(\pi_{j})|\pi_{i}])(\mathbb{E}[q_{kl}|\psi_{k},\psi_{l}]\lambda(\pi_{l})-\mathbb{E}[q_{kl}\lambda(\pi_{l})|\pi_{k}])\\
 &  & +b^{K}(\pi_{i})'b^{K}(\pi_{k})e_{ij,kl}^{q\epsilon},
\end{eqnarray*}
where
\begin{eqnarray*}
e_{ij,kl}^{q\epsilon} & \equiv & (\mathbb{E}[q_{ij}q_{kl}|\boldsymbol{\psi}]-\mathbb{E}[q_{ij}|\psi_{i},\psi_{j}]\mathbb{E}[q_{kl}|\psi_{k},\psi_{l}])\lambda(\pi_{j})\lambda(\pi_{l})\\
 &  & -(\mathbb{E}[q_{ij}|\boldsymbol{\psi}]-\mathbb{E}[q_{ij}|\psi_{i},\psi_{j}])\lambda(\pi_{j})\mathbb{E}[q_{kl}\lambda(\pi_{l})|\pi_{k}]\\
 &  & -(\mathbb{E}[q_{kl}|\boldsymbol{\psi}]-\mathbb{E}[q_{kl}|\psi_{k},\psi_{l}])\lambda(\pi_{l})\mathbb{E}[q_{ij}\lambda(\pi_{j})|\pi_{i}],
\end{eqnarray*}
which again satisfies $\max_{i,j,k,l\in\mathcal{N}:\{i,j\}\cap\{k,l\}=\emptyset}\mathbb{E}[(e_{ij,kl}^{q\epsilon})^{2}]\leq o(n^{-4}/K)$
under Assumptions \ref{ass:regular}(ii), \ref{ass:compact}(i), \ref{ass:theta}(i),
and \ref{ass:w}(i)-(iii). Therefore, following a similar argument
with $\lambda_{i}$ playing the role of $x_{i}$, we can show that
$B'_{K}(\boldsymbol{q}\boldsymbol{\epsilon}-\boldsymbol{\mu}_{0}^{\boldsymbol{q}\boldsymbol{\epsilon}})/n=o_{p}(1)$
for $\boldsymbol{q}=\boldsymbol{s}\boldsymbol{w}$. The proof is complete.
\end{proof}
\begin{lem}
\label{lem:aqb_consist}Suppose that $a_{i}$ and $b_{i}$ are independent
across $i$. Each $a_{i}$ ($b_{i}$) is either equal to $\epsilon_{i}$,
or a vector function of $\psi_{i}$ with $\text{\ensuremath{\max_{i\in\mathcal{N}}\|a_{i}\|}}<\infty$
(\textup{$\text{\ensuremath{\max_{i\in\mathcal{N}}\|b_{i}\|}}<\infty$}).
Define $\boldsymbol{a}=(a_{1},\dots,a_{n})'$ and $\boldsymbol{b}=(b_{1},\dots,b_{n})'$.
Let $\boldsymbol{q}=(q_{ij})$ denote a matrix that takes one of the
following forms: (a) $\boldsymbol{w}$, (b) $\boldsymbol{w}'\boldsymbol{w}$,
(c) $\boldsymbol{s}\boldsymbol{w}^{t}$, (d) $\boldsymbol{w}'\boldsymbol{s}\boldsymbol{w}^{t}$,
or (e) $(\boldsymbol{w}')^{r}\boldsymbol{s}'\boldsymbol{s}\boldsymbol{w}^{t}$,
$r,t=1,2$, where $\boldsymbol{s}=(I_{n}-\gamma_{1}\boldsymbol{w})^{-1}$.
Then we have \textup{$n^{-1}(\boldsymbol{a}'\boldsymbol{q}\boldsymbol{b}-\mathbb{E}[\boldsymbol{a}'\boldsymbol{q}\boldsymbol{b}])=o_{p}(1)$.}
\end{lem}
\begin{proof}
The result can be proved element-wise. Without loss of generality,
we assume that $a_{i}$ and $b_{i}$ are scalars for notation simplicity.
By Markov's inequality, it suffices if the second moment of $n^{-1}(\boldsymbol{a}'\boldsymbol{q}\boldsymbol{b}-\mathbb{E}[\boldsymbol{a}'\boldsymbol{q}\boldsymbol{b}])$
is $o(1)$. The second moment is given by
\begin{eqnarray}
 &  & n^{-2}\mathbb{E}[(\boldsymbol{a}'\boldsymbol{q}\boldsymbol{b}-\mathbb{E}[\boldsymbol{a}'\boldsymbol{q}\boldsymbol{b}])^{2}]\nonumber \\
 & = & n^{-2}\mathbb{E}\left[\left(\sum_{i}\sum_{j}(q_{ij}a_{i}b_{j}-\mathbb{E}[q_{ij}a_{i}b_{j}]\right)^{2}\right]\\
 & = & n^{-2}\sum_{i}\sum_{j}\sum_{k}\sum_{l}(\mathbb{E}[q_{ij}q_{kl}a_{i}b_{j}a_{k}b_{l}]-\mathbb{E}[q_{ij}a_{i}b_{j}]\mathbb{E}[q_{kl}a_{k}b_{l}])\nonumber \\
 & = & n^{-2}\sum_{(i,j,k,l):\{i,j\}\cap\{k,l\}\neq\emptyset}(\mathbb{E}[q_{ij}q_{kl}a_{i}b_{j}a_{k}b_{l}]-\mathbb{E}[q_{ij}a_{i}b_{j}]\mathbb{E}[q_{kl}a_{k}b_{l}])\nonumber \\
 &  & +n^{-2}\sum_{(i,j,k,l):\{i,j\}\cap\{k,l\}=\emptyset}(\mathbb{E}[q_{ij}q_{kl}a_{i}b_{j}a_{k}b_{l}]-\mathbb{E}[q_{ij}a_{i}b_{j}]\mathbb{E}[q_{kl}a_{k}b_{l}]).\label{eq:aqb_var}
\end{eqnarray}
In the last expression, the first term sums over all indices $i$,
$j$, $k$, and $l$ such that $\{i,j\}$ and $\{k,l\}$ share at
least one common element, and the second term sums over all indices
$i$, $j$, $k$, and $l$ such that $\{i,j\}$ and $\{k,l\}$ do
not overlap.

The first sum in (\ref{eq:aqb_var}) consists of $O(n^{3})$ terms.
By Lemma \ref{lem:network_bd} and Assumption \ref{ass:w}(ii), $\mathbb{E}[\|\boldsymbol{q}\|_{\infty}^{4}]\leq Cn^{4}\mathbb{E}[\|\boldsymbol{w}\|_{\infty}^{8}]=O(n^{-4})$.
Therefore, we can bound
\begin{eqnarray*}
|\mathbb{E}[q_{ij}q_{kl}a_{i}b_{j}a_{k}b_{l}]| & \leq & \mathbb{E}[\|\boldsymbol{q}\|_{\infty}^{4}]^{1/2}\mathbb{E}[(a_{i}b_{j}a_{k}b_{l})^{2}]^{1/2}\\
 & \leq & C\mathbb{E}[\|\boldsymbol{q}\|_{\infty}^{4}]^{1/2}\mathbb{E}[\epsilon_{i}^{8}]^{1/2}\\
 & = & O(n^{-2})
\end{eqnarray*}
uniformly in $i$, $j$, $k$, and $l$, where the second inequality
holds by the definition of $a_{i}$ and $b_{i}$, and the last equality
follows from $\mathbb{E}[\epsilon_{i}^{8}]<\infty$ (Assumption \ref{ass:smooth}(i)).
Hence, the first sum in (\ref{eq:aqb_var}) is $n^{-2}\cdot O(n^{3})\cdot O(n^{-2})=o(1)$.

The second sum in (\ref{eq:aqb_var}) consists of $O(n^{4})$ terms.
For any disjoint $\{i,j\}$ and $\{k,l\}$, we can derive
\begin{eqnarray}
\mathbb{E}[q_{ij}q_{kl}a_{i}b_{j}a_{k}b_{l}] & = & \mathbb{E}[\mathbb{E}[q_{ij}q_{kl}|\boldsymbol{\psi},\boldsymbol{\epsilon}]a_{i}b_{j}a_{k}b_{l}]\nonumber \\
 & = & \mathbb{E}[\mathbb{E}[q_{ij}q_{kl}|\boldsymbol{\psi}]a_{i}b_{j}a_{k}b_{l}]\nonumber \\
 & = & \mathbb{E}[\mathbb{E}[q_{ij}|\psi_{i},\psi_{j}]\mathbb{E}[q_{kl}|\psi_{k},\psi_{l}]a_{i}b_{j}a_{k}b_{l}]+o(n^{-2})\nonumber \\
 & = & \mathbb{E}[\mathbb{E}[q_{ij}|\psi_{i},\psi_{j}]a_{i}b_{j}]\mathbb{E}[\mathbb{E}[q_{kl}|\psi_{k},\psi_{l}]a_{k}b_{l}]+o(n^{-2})\nonumber \\
 & = & \mathbb{E}[\mathbb{E}[q_{ij}|\psi_{i},\psi_{j},\epsilon_{i},\epsilon_{j}]a_{i}b_{j}]\mathbb{E}[\mathbb{E}[q_{kl}|\psi_{k},\psi_{l},\epsilon_{k},\epsilon_{l}]a_{k}b_{l}]+o(n^{-2})\nonumber \\
 & = & \mathbb{E}[q_{ij}a_{i}b_{j}]\mathbb{E}[q_{kl}a_{k}b_{l}]+o(n^{-2}).\label{eq:aqb_ijkl}
\end{eqnarray}
The second equality holds because $\boldsymbol{q}$ and $\boldsymbol{\epsilon}$
are independent conditional on $\boldsymbol{\psi}$. The third equality
follows because 
\begin{eqnarray*}
 &  & |\mathbb{E}[(\mathbb{E}[q_{ij}q_{kl}|\boldsymbol{\psi}]-\mathbb{E}[q_{ij}|\psi_{i},\psi_{j}]\mathbb{E}[q_{kl}|\psi_{k},\psi_{l}])a_{i}b_{j}a_{k}b_{l}]|\\
 & \leq & \mathbb{E}[|\mathbb{E}[q_{ij}q_{kl}|\boldsymbol{\psi}]-\mathbb{E}[q_{ij}|\psi_{i},\psi_{j}]\mathbb{E}[q_{kl}|\psi_{k},\psi_{l}]|^{2}]^{1/2}\mathbb{E}[(a_{i}b_{j}a_{k}b_{l})^{2}]^{1/2}\\
 & \leq & o(n^{-2})\mathbb{E}[\epsilon_{i}^{4}]=o(n^{-2})
\end{eqnarray*}
by the definition of $a_{i}$ and $b_{i}$ and Assumptions \ref{ass:w}(iv)
and \ref{ass:smooth}(i). The derivation also indicates that the $o(n^{-2})$
term is uniformly in $i$, $j$, $k$, and $l$ such that $\{i,j\}\cap\{k,l\}=\emptyset$.
The forth equality in (\ref{eq:aqb_ijkl}) holds because $\mathbb{E}[q_{ij}|\psi_{i},\psi_{j}]a_{i}b_{j}$
and $\mathbb{E}[q_{kl}|\psi_{k},\psi_{l}]a_{k}b_{l}$ are independent
for disjoint $\{i,j\}$ and $\{k,l\}$. The fifth equality in (\ref{eq:aqb_ijkl})
holds because from the conditional independence of $\boldsymbol{q}$
and $\boldsymbol{\epsilon}$ given $\boldsymbol{\psi}$ and i.i.d.
$\psi_{i}$ and $\epsilon_{i}$, we can show that $q_{ij}$ is independent
of $\epsilon_{i}$ and $\epsilon_{j}$ conditional on $\psi_{i}$
and $\psi_{j}$. Hence, the second sum in equation (\ref{eq:aqb_var})
is $n^{-2}\cdot O(n^{4})\cdot o(n^{-2})=o(1)$.
\end{proof}
\begin{lem}[Boundness of network]
\label{lem:network_bd}(i) For the matrix $\boldsymbol{s}=(I_{n}-\gamma_{1}\boldsymbol{w})^{-1}$,
we have $\interleave\boldsymbol{s}\interleave_{\infty}\leq C$ and
$\interleave\boldsymbol{s}\interleave_{1}\leq Cn\|\boldsymbol{w}\|_{\infty}$;
(ii) For the matrix $\boldsymbol{q}$ that takes the form of (a) $\boldsymbol{w}$,
(b) $\boldsymbol{w}'\boldsymbol{w}$, (c) $\boldsymbol{s}\boldsymbol{w}^{t}$,
(d) $\boldsymbol{w}'\boldsymbol{s}\boldsymbol{w}^{t}$, and (e) $(\boldsymbol{w}')^{r}\boldsymbol{s}'\boldsymbol{s}\boldsymbol{w}^{t}$,
$r,t=1,2$, where $\boldsymbol{s}$ is given in part (i), we have
$\|\boldsymbol{q}\|_{\infty}\leq Cn\|\boldsymbol{w}\|_{\infty}^{2}$.
\end{lem}
\begin{proof}
(i) Note that $\interleave I_{n}\interleave_{\infty}=1$ and $\interleave I_{n}\interleave_{1}=1$.
By $\interleave\boldsymbol{w}\interleave_{\infty}=1$ (Assumption
\ref{ass:w}(i)), for any $k\geq1$, we can bound $\interleave\boldsymbol{w}^{k}\interleave_{\infty}\leq\interleave\boldsymbol{w}\interleave_{\infty}^{k}=1$.
Hence, $\interleave\boldsymbol{s}\interleave_{\infty}\leq\sum_{k=0}^{\infty}|\gamma_{1}|^{k}\interleave\boldsymbol{w}^{k}\interleave_{\infty}\leq\sum_{k=0}^{\infty}|\gamma_{1}|^{k}=1/(1-\gamma_{1})<\infty$.
Moreover, for any $k\geq1$, we can bound $\|\boldsymbol{w}^{k}\|_{\infty}\leq\interleave\boldsymbol{w}\interleave_{\infty}^{k-1}\|\boldsymbol{w}\|_{\infty}=\|\boldsymbol{w}\|_{\infty}$
and hence $\interleave\boldsymbol{w}^{k}\interleave_{1}\leq n\|\boldsymbol{w}^{k}\|_{\infty}\leq n\|\boldsymbol{w}\|_{\infty}$.
Therefore, $\interleave\boldsymbol{s}\interleave_{1}\leq\sum_{k=0}^{\infty}|\gamma_{1}|^{k}\interleave\boldsymbol{w}^{k}\interleave_{1}\leq n\|\boldsymbol{w}\|_{\infty}/(1-\gamma_{1})$.

(ii) For case (a) with $\boldsymbol{q}=\boldsymbol{w}$, the result
follows immediately as $n\|\boldsymbol{w}\|_{\infty}\geq\interleave\boldsymbol{w}\interleave_{\infty}=1$.
For case (b) with $\boldsymbol{q}=\boldsymbol{w}'\boldsymbol{w}$,
we can bound $\|\boldsymbol{q}\|_{\infty}=\|\boldsymbol{w}'\boldsymbol{w}\|_{\infty}\leq\interleave\boldsymbol{w}'\interleave_{\infty}\|\boldsymbol{w}\|_{\infty}=\interleave\boldsymbol{w}\interleave_{1}\|\boldsymbol{w}\|_{\infty}\leq n\|\boldsymbol{w}\|_{\infty}^{2}$.
For case (c) with $\boldsymbol{q}=\boldsymbol{s}\boldsymbol{w}^{t}$,
$t=1,2$, we have $\|\boldsymbol{q}\|_{\infty}=\|\boldsymbol{s}\boldsymbol{w}^{t}\|_{\infty}\leq\sum_{k=0}^{\infty}|\gamma_{1}|^{k}\|\boldsymbol{w}^{t+k}\|_{\infty}\leq\|\boldsymbol{w}\|_{\infty}/(1-\gamma_{1})$.
For case (d) with $\boldsymbol{q}=\boldsymbol{w}'\boldsymbol{s}\boldsymbol{w}^{t}$,
$t=1,2$, we can derive $\|\boldsymbol{q}\|_{\infty}=\|\boldsymbol{w}'\boldsymbol{s}\boldsymbol{w}^{t}\|_{\infty}\leq\interleave\boldsymbol{w}\interleave_{1}\|\boldsymbol{s}\boldsymbol{w}^{t}\|_{\infty}\leq n\|\boldsymbol{w}\|_{\infty}^{2}/(1-\gamma_{1})$.
For case (e) with $\boldsymbol{q}=(\boldsymbol{w}')^{r}\boldsymbol{s}'\boldsymbol{s}\boldsymbol{w}^{t}$,
note that we can bound $\interleave\boldsymbol{sw}^{r}\interleave_{1}\leq\sum_{k=0}^{\infty}|\gamma_{1}|^{k}\interleave\boldsymbol{w}^{r+k}\interleave_{1}\leq n\|\boldsymbol{w}\|_{\infty}/(1-\gamma_{1})$
and thus $\|\boldsymbol{q}\|_{\infty}=\|(\boldsymbol{w}')^{r}\boldsymbol{s}'\boldsymbol{s}\boldsymbol{w}^{t}\|_{\infty}\leq\interleave(\boldsymbol{w}')^{r}\boldsymbol{s}'\interleave_{\infty}\|\boldsymbol{s}\boldsymbol{w}^{t}\|_{\infty}=\interleave\boldsymbol{sw}^{r}\interleave_{1}\|\boldsymbol{s}\boldsymbol{w}^{t}\|_{\infty}\leq n\|\boldsymbol{w}\|_{\infty}^{2}/(1-\gamma_{1})^{2}$.
\end{proof}

\subsubsection{Asymptotic Distribution of $\hat{\gamma}$}

Theorem \ref{thm:gamma_clt} is proved based on several lemmas presented
in Section \ref{sec:consist_pf} and later in this section. Table
\ref{tab:gamma_clt_pfstruct} states the relationships between Theorem
\ref{thm:gamma_clt} and these lemmas. 

\begin{table}
\centering

\caption{Relationships Between Theorem \ref{thm:gamma_clt} and Its Supporting
Lemmas}
\label{tab:gamma_clt_pfstruct}

\begin{tabular}{llll}
\toprule 
 & Referring to & Referring to & Referring to\tabularnewline
\midrule
Theorem \ref{thm:gamma_clt} & Lemma \ref{lem:m_asym_linear} & Lemma \ref{lem:muhat_emp} & Lemmas \ref{lem:wy_bd}, \ref{lem:betahat_dif}, \ref{lem:wt_betahat}\tabularnewline
 &  & Lemma \ref{lem:alpha} & Lemma \ref{lem:xt_consist}\tabularnewline
 & Lemma \ref{lem:m_clt} & Lemma \ref{lem:weightedU} & Lemma \ref{lem:network_bd}\tabularnewline
 &  & Lemma \ref{lem:qvv}  & Lemma \ref{lem:network_bd}\tabularnewline
\bottomrule
\end{tabular}
\end{table}

\begin{proof}[Proof of Theorem \ref{thm:gamma_clt}]
Recall from equation (\ref{eq:gamma_ols}) that
\[
\hat{\gamma}-\gamma_{0}=\left(\frac{1}{n}\sum_{i=1}^{n}(X_{i}-\hat{\mu}^{X}(\hat{\pi}_{i}))X'_{i}\right)^{-1}\frac{1}{n}\sum_{i=1}^{n}(X_{i}-\hat{\mu}^{X}(\hat{\pi}_{i}))\epsilon_{i}.
\]
By equation (\ref{eq:M_lln}),
\[
\frac{1}{n}\sum_{i=1}^{n}(X_{i}-\hat{\mu}^{X}(\hat{\pi}_{i}))X'_{i}=\frac{1}{n}\sum_{i=1}^{n}\mathbb{E}[(X_{i}-\mu_{0}^{X}(\pi_{i}))X'_{i}]+o_{p}(1).
\]
Denote $M_{n}\equiv\frac{1}{n}\sum_{i=1}^{n}\mathbb{E}[(X_{i}-\mu_{0}^{X}(\pi_{i}))X'_{i}]$.
Moreover, by Lemmas \ref{lem:m_asym_linear} and \ref{lem:m_clt},
\[
\Omega_{n}^{-1/2}\frac{1}{\sqrt{n}}\sum_{i=1}^{n}(X_{i}-\hat{\mu}^{X}(\hat{\pi}_{i}))\epsilon_{i}\overset{d}{\rightarrow}N(0,I_{d_{X}}),
\]
where $\Omega_{n}$ is defined in Lemma \ref{lem:m_clt}. From these
results and Slutsky's theorem, we obtain
\[
\sqrt{n}\Omega_{n}^{-1/2}M_{n}(\hat{\gamma}-\gamma_{0})\overset{d}{\rightarrow}N(0,I_{d_{X}}).
\]
\end{proof}
\begin{lem}[Asymptotically linear representation of the moment]
\label{lem:m_asym_linear}We have
\begin{equation}
\frac{1}{\sqrt{n}}\sum_{i=1}^{n}(X_{i}-\hat{\mu}^{X}(\hat{\pi}_{i}))\epsilon_{i}=\frac{1}{\sqrt{n}}\sum_{i=1}^{n}((X_{i}-\mu_{0}^{X}(\pi_{i}))\nu_{i}+M_{\theta}\phi_{\theta}(z_{i},\theta_{0}))+o_{p}(1),\label{eq:m.asym.linear}
\end{equation}
where $M_{\theta}=-\mathbb{E}[(\mathbb{E}[X_{i}|z_{i}]-\mu_{0}^{X}(\pi_{i}))\frac{\partial\lambda_{0}(\pi_{i})}{\partial\pi}\frac{\partial\pi(z_{i},g_{i},\theta_{0})}{\partial\theta}]$.
\end{lem}
\begin{proof}
Consider the decomposition
\begin{align}
 & \frac{1}{\sqrt{n}}\sum_{i=1}^{n}(X_{i}-\hat{\mu}^{X}(\hat{\pi}_{i}))\epsilon_{i}\nonumber \\
= & \frac{1}{\sqrt{n}}\sum_{i=1}^{n}(X_{i}-\mu_{0}^{X}(\hat{\pi}_{i}))\epsilon_{i}+\sqrt{n}\int D(\epsilon_{i},\hat{\mu}^{X}(\hat{\pi}_{i})-\mu_{0}^{X}(\hat{\pi}_{i}))dF(z_{i},g_{i},\epsilon_{i})\nonumber \\
 & +\sqrt{n}\int D(\epsilon_{i},\mu^{X}(\hat{\pi}_{i})-\mu_{0}^{X}(\pi_{i}))dF(z_{i},g_{i},\epsilon_{i})\nonumber \\
 & +\frac{1}{\sqrt{n}}\sum_{i=1}^{n}(D(\epsilon_{i},\hat{\mu}^{X}(\hat{\pi}_{i})-\mu_{0}^{X}(\pi_{i}))-\int D(\epsilon_{i},\hat{\mu}^{X}(\hat{\pi}_{i})-\mu_{0}^{X}(\pi_{i}))dF(z_{i},g_{i},\epsilon_{i})),\label{eq:m.dcmp}
\end{align}
where $D(\epsilon_{i},\mu)=-\mu\epsilon_{i}$ for any $\mu\in\mathbb{R}^{d_{X}}$,
$\mu^{X}(\hat{\pi}_{i})=\mathbb{E}[X_{i}|\pi(z_{i},g_{i},\hat{\theta})]$,
and $F(z_{i},g_{i},\epsilon_{i})$ denotes the cdf of $(z_{i},g_{i},\epsilon_{i})$.
The first term is a leading term. The second term is to adjust for
the estimation of $\mu_{0}^{X}$, and the third term is to adjust
for the estimation of $\theta_{0}$ \citep{Hahn2013}, Both terms
contribute to the asymptotic distribution of $\hat{\gamma}$. The
last term is $o_{p}(1)$ by Lemma \ref{lem:muhat_emp}.

The second term in equation (\ref{eq:m.dcmp}) can be analyzed following
\citet{newey1994asymptotic}. For an arbitrary mean square integrable
function $\mu(\pi(z_{i},g_{i},\theta))\in\mathbb{R}^{d_{X}}$ that
is continuously differentiable in $s$, by iterated expectations $\mathbb{E}[D(\epsilon_{i},\mu(\pi(z_{i},g_{i},\theta))]=-\mathbb{E}[\mu(\pi(z_{i},g_{i},\theta))\mu^{\epsilon}(\pi(z_{i},g_{i},\theta))]$,
where $\mu^{\epsilon}(\pi(z_{i},g_{i},\theta))=\mathbb{E}[\epsilon_{i}|\pi(z_{i},g_{i},\theta)]$.
Hence, the correction term in \citet[Proposition 4]{newey1994asymptotic}
takes the form $\alpha^{X}(\omega_{i},\pi(z_{i},g_{i},\theta))=-(X_{i}-\mu^{X}(\pi(z_{i},g_{i},\theta)))\mu^{\epsilon}(\pi(z_{i},g_{i},\theta))$,
where $\omega_{i}=(X_{i},z_{i},g_{i}),$ and thus $\sqrt{n}\int D(\epsilon_{i},\hat{\mu}^{X}(\hat{\pi}_{i})-\mu^{X}(\hat{\pi}_{i}))dF(z_{i},g_{i},\epsilon_{i})=n^{-1/2}\sum_{i=1}^{n}\alpha_{0}^{X}(\omega_{i},\hat{\pi}_{i})$.
Also recall that $\hat{\pi}_{i}=\pi(z_{i},g_{i},\hat{\theta})$ and
$\pi_{i}=\pi(z_{i},g_{i},\theta_{0})$. Define $\alpha_{0}^{X}(\omega_{i},\pi_{i})=-(X_{i}-\mu_{0}^{X}(\pi_{i}))\lambda_{0}(\pi_{i})$.
Under Assumption \ref{ass:smooth}(ii), expanding $\alpha^{X}(\omega_{i},\hat{\pi}_{i})$
around $\theta_{0}$ yields $\alpha^{X}(\omega_{i},\hat{\pi}_{i})=\alpha_{0}^{X}(\omega_{i},\pi_{i})+\frac{\partial\alpha^{X}(\omega_{i},x_{i})}{\partial\theta'}(\hat{\theta}-\theta_{0})+o_{p}(\|\hat{\theta}-\theta_{0}\|)$.
By Lemma \ref{lem:alpha} and Assumption \ref{ass:theta}(ii), $n^{-1}\sum_{i=1}^{n}\frac{\partial\alpha^{X}(\omega_{i},\pi_{i})}{\partial\theta'}=o_{p}(1)$
and $\sqrt{n}(\hat{\theta}-\theta_{0})=O_{p}(1)$. We thus have $\sqrt{n}\int D(\epsilon_{i},\hat{\mu}^{X}(\hat{\pi}_{i})-\mu^{X}(\hat{\pi}_{i}))dF(z_{i},g_{i},\epsilon_{i})=n^{-1/2}\sum_{i=1}^{n}\alpha_{0}^{X}(\omega_{i},\pi_{i})+o_{p}(1).$

The third term in equation (\ref{eq:m.dcmp}) can be analyzed following
\citet{Hahn2013}. Observe that $\frac{\partial D(\epsilon_{i},\mu_{0}^{X}(\pi_{i}))}{\partial\mu^{X}}=-\epsilon_{i}$
and $\mathbb{E}[\frac{\partial D(\epsilon_{i},\mu_{0}^{X}(\pi_{i}))}{\partial\mu^{X}}|\pi_{i}=\pi]=-\lambda_{0}(s)$.
The first term in \citet[Theorem 4]{Hahn2013} takes the form $-\mathbb{E}[(\epsilon_{i}-\lambda_{0}(\pi_{i}))\frac{\partial\mu_{0}^{X}(\pi_{i})}{\partial\pi}\frac{\partial\pi(z_{i},g_{i},\theta_{0})}{\partial\theta}]=0,$
where we used $\epsilon_{i}-\lambda_{0}(\pi_{i})=\nu_{i}$ and $\mathbb{E}[\nu_{i}|z_{i},g_{i}]=0$.
Hence, by \citet[Theorem 4]{Hahn2013},
\begin{eqnarray*}
 &  & \sqrt{n}\int D(\epsilon_{i},\mu^{X}(\hat{\pi}_{i})-\mu_{0}^{X}(\pi_{i}))dF(z_{i},g_{i},\epsilon_{i})\\
 & = & -\mathbb{E}\left[(\mathbb{E}[X_{i}|z_{i}]-\mu_{0}^{X}(\pi_{i}))\frac{\partial\lambda_{0}(\pi_{i})}{\partial s}\frac{\partial\pi(z_{i},g_{i},\theta_{0})}{\partial\theta}\right]\sqrt{n}(\hat{\theta}-\theta_{0})=M_{\theta}\sqrt{n}(\hat{\theta}-\theta_{0}).
\end{eqnarray*}
Because $\sqrt{n}(\hat{\theta}-\theta_{0})=\frac{1}{\sqrt{n}}\sum_{i=1}^{n}\phi_{\theta}(z_{i},\theta_{0})+o_{p}(1)$,
we can represent $\sqrt{n}\int D(\epsilon_{i},\mu^{X}(\hat{\pi}_{i})-\mu_{0}^{X}(\pi_{i}))dF(z_{i},g_{i},\epsilon_{i})=n^{-1/2}\sum_{i=1}^{n}M_{\theta}\phi_{\theta}(z_{i},\theta_{0})+o_{p}(1)$.
Note that $(X_{i}-\mu_{0}^{X}(\pi_{i}))\epsilon_{i}+\alpha^{X}(\omega_{i},\pi_{i})=(X_{i}-\mu_{0}^{X}(\pi_{i}))\nu_{i}$.
Combining the results we obtain equation (\ref{eq:m.asym.linear}).
\end{proof}

\begin{lem}[CLT of the moment]
\label{lem:m_clt}Let $\Phi_{n}=n^{-1/2}\sum_{i=1}^{n}((X_{i}-\mu_{0}^{X}(\pi_{i}))\nu_{i}+M_{\theta}\phi_{\theta}(z_{i},\theta_{0}))$.
Then $\Omega_{n}^{-1/2}\Phi_{n}\overset{d}{\rightarrow}N(0,I_{d_{X}})$,
where $\Omega_{n}=n^{-1}\sum_{i=1}^{n}\mathbb{E}[\varphi_{n}(\tilde{\psi}_{i},\nu_{i})\varphi_{n}(\tilde{\psi}_{i},\nu_{i})']$,
$\varphi_{n}(\tilde{\psi}_{i},\nu_{i})\in\mathbb{R}^{d_{X}}$ is defined
in equation (\ref{eq:gamma_infl}), and $I_{d_{X}}$ is the $d_{X}\times d_{X}$
identity matrix.
\end{lem}
\begin{proof}
Recall that $X_{i}=(w_{i}\boldsymbol{y},w_{i}\boldsymbol{x},x'_{i})'$.
While $x_{i}$ is i.i.d., both $w_{i}\boldsymbol{x}$ and $w_{i}\boldsymbol{y}$
are dependent across $i$. Lemma \ref{lem:weightedU} establishes
the Hoeffding projection 
\begin{align}
n^{-1/2}\sum_{i=1}^{n}(w_{i}\boldsymbol{x})'\nu_{i} & =n^{-1/2}\sum_{i=1}^{n}\sum_{j=1}^{n}\mathbb{E}[w_{ij}x_{j}|\tilde{\psi}_{i},\nu_{i}]\nu_{i}+o_{p}(1)\nonumber \\
 & =n^{-1/2}\sum_{i=1}^{n}\mathbb{E}[w_{i}\boldsymbol{x}|\tilde{\psi}_{i},\nu_{i}]\nu_{i}+o_{p}(1).\label{eq:wx*v}
\end{align}
As for $w_{i}\boldsymbol{y}$, note that $\sum_{i=1}^{n}(w_{i}\boldsymbol{y})'\nu_{i}=\boldsymbol{\nu}'\boldsymbol{w}\boldsymbol{y}=\boldsymbol{\nu}'\boldsymbol{s}(\boldsymbol{w}^{2}\boldsymbol{x}\gamma_{2}+\boldsymbol{w}\boldsymbol{x}\gamma_{3}+\boldsymbol{w}\boldsymbol{\lambda}+\boldsymbol{w}\boldsymbol{\nu})$.
Applying Lemma \ref{lem:weightedU} to each of $n^{-1/2}\boldsymbol{\nu}'\boldsymbol{s}\boldsymbol{w}^{2}\boldsymbol{x}$,
$n^{-1/2}\boldsymbol{\nu}'\boldsymbol{s}\boldsymbol{w}\boldsymbol{x}$,
and $n^{-1/2}\boldsymbol{\nu}'\boldsymbol{s}\boldsymbol{w}\boldsymbol{\lambda}$
with $q=\boldsymbol{s}\boldsymbol{w}^{2}$ or $\boldsymbol{s}\boldsymbol{w}$
and $t_{i}=x_{i}$ or $\lambda(\pi_{i})$, we derive
\begin{align*}
n^{-1/2}\boldsymbol{\nu}'\boldsymbol{s}\boldsymbol{w}^{2}\boldsymbol{x} & =n^{-1/2}\sum_{i=1}^{n}\sum_{j=1}^{n}\mathbb{E}[(\boldsymbol{s}\boldsymbol{w}^{2})_{ij}x_{j}|\tilde{\psi}_{i},\nu_{i}]\nu_{i}+o_{p}(1)\\
n^{-1/2}\boldsymbol{\nu}'\boldsymbol{s}\boldsymbol{w}\boldsymbol{x} & =n^{-1/2}\sum_{i=1}^{n}\sum_{j=1}^{n}\mathbb{E}[(\boldsymbol{s}\boldsymbol{w})_{ij}x_{j}|\tilde{\psi}_{i},\nu_{i}]\nu_{i}+o_{p}(1)\\
n^{-1/2}\boldsymbol{\nu}'\boldsymbol{s}\boldsymbol{w}\boldsymbol{\lambda} & =n^{-1/2}\sum_{i=1}^{n}\sum_{j=1}^{n}\mathbb{E}[(\boldsymbol{s}\boldsymbol{w})_{ij}\lambda(\pi_{j})|\tilde{\psi}_{i},\nu_{i}]\nu_{i}+o_{p}(1),
\end{align*}
where $(\boldsymbol{s}\boldsymbol{w}^{2})_{ij}$ denotes the $(i,j)$
element of $\boldsymbol{s}\boldsymbol{w}^{2}$ and similarly for $(\boldsymbol{s}\boldsymbol{w})_{ij}$.
Moreover, Lemma \ref{lem:qvv} establishes that $n^{-1/2}\boldsymbol{\nu}'\boldsymbol{s}\boldsymbol{w}\boldsymbol{\nu}=o_{p}(1)$.
Combining these results yields the Hoeffding decomposition of $n^{-1/2}\sum_{i=1}^{n}(w_{i}\boldsymbol{y})'\nu_{i}$
as follows:
\begin{equation}
n^{-1/2}\sum_{i=1}^{n}(w_{i}\boldsymbol{y})'\nu_{i}=n^{-1/2}\sum_{i=1}^{n}\mathbb{E}[(w_{i}\boldsymbol{y})^{*}|\tilde{\psi}_{i},\nu_{i}]\nu_{i}+o_{p}(1)\label{eq:wy*v}
\end{equation}
where $(w_{i}\boldsymbol{y})^{*}$ denotes the deterministic part
of $w_{i}\boldsymbol{y},$
\[
(w_{i}\boldsymbol{y})^{*}\equiv\sum_{j=1}^{n}(\boldsymbol{s}\boldsymbol{w}^{2})_{ij}x_{j}\gamma_{2}+\sum_{j=1}^{n}(\boldsymbol{s}\boldsymbol{w})_{ij}x_{j}\gamma_{3}+\sum_{j=1}^{n}(\boldsymbol{s}\boldsymbol{w})_{ij}\lambda(\pi_{j}).
\]
Furthermore, by the definition of $\mu_{0}^{w_{i}\boldsymbol{y}}(\pi_{i})$
and $\mathbb{E}[(\boldsymbol{s}\boldsymbol{w})_{ij}\nu_{j}|\pi_{i}]=\mathbb{E}[\mathbb{E}[(\boldsymbol{s}\boldsymbol{w})_{ij}\nu_{j}|\boldsymbol{\psi}]|\pi_{i}]=\mathbb{E}[\mathbb{E}[(\boldsymbol{s}\boldsymbol{w})_{ij}|\boldsymbol{\psi}]\mathbb{E}[\nu_{j}|\boldsymbol{\psi}]|\pi_{i}]=0$,
we can write
\begin{equation}
n^{-1/2}\sum_{i=1}^{n}\mu_{0}^{w_{i}\boldsymbol{y}}(\pi_{i})\nu_{i}=n^{-1/2}\sum_{i=1}^{n}\mathbb{E}[(w_{i}\boldsymbol{y})^{*}|\pi_{i}]\nu_{i}.\label{eq:Ewy*v}
\end{equation}

Define the function $\varphi_{n}(\tilde{\psi}_{i},\nu_{i})\in\mathbb{R}^{d_{X}}$
by
\begin{equation}
\varphi_{n}(\tilde{\psi}_{i},\nu_{i})=n^{-1/2}\left(\begin{pmatrix}\mathbb{E}[(w_{i}\boldsymbol{y})^{*}|\tilde{\psi}_{i},\nu_{i}]-\mathbb{E}[(w_{i}\boldsymbol{y})^{*}|\pi_{i}]\\
\mathbb{E}[(w_{i}\boldsymbol{x})'|\tilde{\psi}_{i},\nu_{i}]-\mathbb{E}[(w_{i}\boldsymbol{x})'|\pi_{i}]\\
x_{i}-\mathbb{E}[x{}_{i}|\pi_{i}]
\end{pmatrix}\nu_{i}+M_{\theta}\phi_{\theta}(z_{i},\theta_{0})\right)\label{eq:gamma_infl}
\end{equation}
It follows from (\ref{eq:wx*v})-(\ref{eq:Ewy*v}) that $\Phi_{n}=\sum_{i=1}^{n}\varphi_{n}(\tilde{\psi}_{i},\nu_{i})+o_{p}(1)$.
Because $\mathbb{E}[\nu_{i}|\boldsymbol{\psi}]=0$, and $\mathbb{E}[\phi_{\theta}(z_{i},\theta_{0})]=0$,
we can derive $\mathbb{E}[\varphi_{n}(\tilde{\psi}_{i},\nu_{i})]=0$.

Write $\varphi_{ni}=\varphi_{n}(\tilde{\psi}_{i},\nu_{i})$. Observe
that $\{\varphi_{ni},i=1,\dots,n\}$ forms a triangular array. We
apply the Lindeberg-Feller CLT to derive the asymptotic distribution
of $\sum_{i=1}^{n}\varphi_{ni}$. By the Cramer-Wold device it suffices
to show that $a^{\prime}\sum_{i=1}^{n}\varphi_{ni}$ satisfies the
Lindeberg condition for any $d_{X}\times1$ vector of constants $a\in\mathbb{R}^{d_{X}}$.
The Lindeberg condition is that for any $\kappa>0$, $\lim_{n\rightarrow\infty}\sum_{i=1}^{n}\mathbb{E}[\frac{(a^{\prime}\varphi_{ni})^{2}}{a^{\prime}\Omega_{n}a}1\{|a^{\prime}\varphi_{ni}|\geq\kappa\sqrt{a^{\prime}\Omega_{n}a}\}]=0$.
The sum is bounded by $\mathbb{E}[\sum_{i}\frac{(a^{\prime}\varphi_{ni})^{2}}{a^{\prime}\Omega_{n}a}1\{\max_{i}|a^{\prime}\varphi_{ni}|\geq\kappa\sqrt{a^{\prime}\Omega_{n}a}\}]$,
where the random variable $\sum_{i}\frac{(a^{\prime}\varphi_{ni})^{2}}{a^{\prime}\Omega_{n}a}$
has a finite expectation and is therefore $O_{p}(1)$. Moreover, we
can derive $\max_{i}|a^{\prime}\varphi_{ni}|=o_{p}(1)$.\footnote{By Assumptions \ref{ass:compact}, \ref{ass:w}(v), and \ref{ass:smooth}(i),
the components of $(\tilde{\psi}_{i},\nu_{i})$ are either bounded
or have finite fourth moment. By Assumption \ref{ass:w}(ii), we can
bound $\mathbb{E}[\max_{i}(a^{\prime}\varphi_{ni})^{2}]\leq\|a\|^{2}\mathbb{E}[\max_{i}\|\varphi_{ni}\|^{2}]\leq O(n^{-1})=o(1)$.} Therefore $\sum_{i}\frac{(a^{\prime}\varphi_{ni})^{2}}{a^{\prime}\Omega_{n}a}1\{\max_{i}|a^{\prime}\varphi_{ni}|\geq\kappa\sqrt{a^{\prime}\Omega_{n}a}\}=O_{p}(1)o_{p}(1)=o_{p}(1)$.
This random variable is bounded by $\sum_{i}\frac{(a^{\prime}\varphi_{ni})^{2}}{a^{\prime}\Omega_{n}a}$
which has a finite expectation. By dominated convergence, the Lindeberg
condition is satisfied. By Lindeberg-Feller CLT, $\Omega_{n}^{-1/2}\Phi_{n}=\Omega_{n}^{-1/2}\sum_{i=1}^{n}\varphi_{n}(\tilde{\psi}_{i},\nu_{i})+o_{p}(1)\overset{d}{\rightarrow}N(0,I_{d_{X}})$.
\end{proof}

\begin{lem}
\label{lem:muhat_emp}
\[
\frac{1}{\sqrt{n}}\sum_{i=1}^{n}(D(\epsilon_{i},\hat{\mu}^{X}(\hat{\pi}_{i})-\mu_{0}^{X}(\pi_{i}))-\int D(\epsilon_{i},\hat{\mu}^{X}(\hat{\pi}_{i})-\mu_{0}^{X}(\pi_{i}))dF(z_{i},g_{i},\epsilon_{i}))=o_{p}(1).
\]
\end{lem}
\begin{proof}
Let $\mu=\mu(\pi(z_{i},g_{i},\theta))\in\mathbb{R}^{d_{X}}$ be a
function of $\pi(z_{i},g_{i},\theta)$. Define the empirical process
$\mathbb{G}_{n}(\mu)=\frac{1}{\sqrt{n}}\sum_{i}(D(\epsilon_{i},\mu)-\mathbb{E}[D(\epsilon_{i},\mu)])$
indexed by $\mu$. We can represent the left-hand side of the above
equation as $\mathbb{G}_{n}(\hat{\mu}^{X}(\hat{\boldsymbol{\pi}}))-\mathbb{G}_{n}(\mu_{0}^{X}(\boldsymbol{\pi}))$.

Observe that $D(\epsilon_{i},\mu)=-\mu\epsilon_{i}$ is linear in
$\mu$. This together with the stochastic boundedness of $X_{i}$
and $\mathbb{E}[\epsilon_{i}^{2}]<\infty$ (Lemma \ref{lem:wy_bd}
and Assumptions \ref{ass:compact}(ii), \ref{ass:w}(i), and \ref{ass:smooth}(i))
implies that the empirical process $\mathbb{G}_{n}(\mu)$ is stochastically
equicontinuous under $L_{2}$ norm \citep[Theorems 1-2]{Andrews1994}.
It remains to show that $\int\|\hat{\mu}^{X}(\hat{\pi}_{i})-\mu_{0}^{X}(\pi_{i}))\|^{2}dF(z_{i},g_{i})=o_{p}(1)$,
where $F(z_{i},g_{i})$ denotes the cdf of $(z_{i},g_{i})$. We prove
it following \citet[Theorem 1]{newey_convergence_1997}.

By the triangle inequality and $(a+b+c)^{2}\leq3(a^{2}+b^{2}+c^{2})$,
we derive
\begin{eqnarray}
 &  & \int\|\hat{\mu}^{X}(\hat{\pi}_{i})-\mu_{0}^{X}(\pi_{i}))\|^{2}dF(z_{i},g_{i})\nonumber \\
 & \leq & 3\int(\|\hat{\beta}^{X}(\hat{\boldsymbol{\pi}})'(b^{K}(\hat{\pi}_{i})-b^{K}(\pi_{i}))\|^{2}+\|(\hat{\beta}^{X}(\hat{\boldsymbol{\pi}})-\beta^{X})'b^{K}(\pi_{i})\|^{2}\nonumber \\
 &  & \|\beta^{X\prime}b^{K}(\pi_{i})-\mu_{0}^{X}(\pi_{i})\|^{2})dF(z_{i},g_{i}).\label{eq:muhat_L2}
\end{eqnarray}
Consider the three terms in the last equation. The first term satisfies{\small
\[
\int\|\hat{\beta}^{X}(\hat{\boldsymbol{\pi}})'(b^{K}(\hat{\pi}_{i})-b^{K}(\pi_{i}))\|^{2}dF(z_{i},g_{i})\leq O_{p}(\varrho_{1}(K)^{2})\int\max_{1\leq i\leq n}\|\hat{\pi}_{i}-\pi_{i}\|^{2}dF(z_{i},g_{i})=O_{p}(\varrho_{1}(K)^{2}/n),
\]
}where the inequality holds by equation (\ref{eq:betahat_tauhat}),
the mean-value theorem and Assumption \ref{ass:sieve}(iv), and the
equality holds because the $\sqrt{n}$-consistency of $\hat{\theta}$
and boundedness of $z$ imply that $\max_{1\leq i\leq n}\|\hat{\pi}_{i}-\pi_{i}\|=O_{p}(n^{-1/2})$.
As for the second term in (\ref{eq:muhat_L2}), by $\mathbb{E}[b^{K}(\pi_{i})b^{K\prime}(\pi_{i})]=I_{K}$
we obtain
\begin{eqnarray*}
 &  & \int\|(\hat{\beta}^{X}(\hat{\boldsymbol{\pi}})-\beta^{X})'b^{K}(\pi_{i})\|^{2}dF(z_{i},g_{i})\\
 & = & \text{tr}((\hat{\beta}^{X}(\hat{\boldsymbol{\pi}})-\beta^{X})'\int b^{K}(\pi_{i})b^{K\prime}(\pi_{i})dF(z_{i},g_{i})(\hat{\beta}^{X}(\hat{\boldsymbol{\pi}})-\beta^{X}))\\
 & = & \|\hat{\beta}^{X}(\hat{\boldsymbol{\pi}})-\beta^{X}\|^{2}=O_{p}(\varrho_{1}(K)^{2}/n)+o_{p}(1),
\end{eqnarray*}
where the last equality follows from $\|\hat{\beta}^{X}(\hat{\boldsymbol{\pi}})-\beta^{X}\|^{2}\leq2(\|\hat{\beta}^{X}(\hat{\boldsymbol{\pi}})-\hat{\beta}^{X}(\boldsymbol{\pi})\|^{2}+\|\hat{\beta}^{X}(\boldsymbol{\pi})-\beta^{X}\|^{2})$,
Lemmas \ref{lem:betahat_dif} and \ref{lem:wt_betahat}, and \citet[Lemma 15.3]{liracine2007}.
The third term in (\ref{eq:muhat_L2}) has the bound $\int\|\beta^{X\prime}b^{K}(\pi_{i})-\mu_{0}^{X}(\pi_{i})\|^{2}dF(z_{i},g_{i})\leq\sup_{\pi}\|\beta^{X\prime}b^{K}(\pi)-\mu_{0}^{X}(\pi)\|=O(K^{-2a})$
by Assumption \ref{ass:sieve}(ii). Combining the results yields $\int\|\hat{\mu}^{X}(\hat{\pi}_{i})-\mu_{0}^{X}(\pi_{i}))\|^{2}dF(z_{i},g_{i})=o_{p}(1)$
and $\mathbb{G}_{n}(\hat{\mu}^{X}(\hat{\boldsymbol{\pi}}))-\mathbb{G}_{n}(\mu_{0}^{X}(\boldsymbol{\pi}))=o_{p}(1)$.
\end{proof}
\begin{lem}
\label{lem:alpha}$\frac{1}{n}\sum_{i=1}^{n}\frac{\partial\alpha^{X}(\omega_{i},\pi_{i})}{\partial\theta'}=o_{p}(1)$.
\end{lem}
\begin{proof}
Recall that $\alpha^{X}(\omega_{i},\pi(z_{i},g_{i},\theta))=-(X_{i}-\mu^{X}(\pi(z_{i},g_{i},\theta)))\mu^{\epsilon}(\pi(z_{i},g_{i},\theta))$,
where we have $\mu^{X}(\pi(z_{i},g_{i},\theta))=\mathbb{E}[X_{i}|\pi(z_{i},g_{i},\theta)]$
and $\mu^{\epsilon}(\pi(z_{i},g_{i},\theta))=\mathbb{E}[\epsilon_{i}|\pi(z_{i},g_{i},\theta)].$
By iterated expectations $\mathbb{E}[\alpha^{X}(\omega_{i},\pi(z_{i},g_{i},\theta))]=0$,
so $\mathbb{E}[\partial\alpha^{X}(\omega_{i},\pi_{i})/\partial\theta']=\partial\mathbb{E}[\alpha^{X}(\omega_{i},\pi(z_{i},g_{i},\theta))]/\partial\theta'=0$.

Differentiating $\alpha^{X}(\omega_{i},\pi(z_{i},g_{i},\theta))$
with respect to $\theta$ at $\theta_{0}$ yields 
\[
\frac{\partial\alpha^{X}(\omega_{i},\pi_{i})}{\partial\theta'}=\left(\frac{\partial\mu^{X}(\pi_{i})}{\partial\pi_{i}}\mu^{\epsilon}(\pi_{i})-(X_{i}-\mu_{0}^{X}(\pi_{i}))\frac{\partial\mu^{\epsilon}(\pi_{i})}{\partial\pi_{i}}\right)\frac{\partial\pi(z_{i},g_{i},\theta_{0})}{\partial\theta'}.
\]
Because $\pi_{i}=\pi(z_{i},g_{i},\theta_{0})$ is bounded and $\mu^{X}(\pi_{i})$
and $\mu^{\epsilon}(\pi_{i})$ are continuously differentiable in
$\pi_{i}$ (Assumptions \ref{ass:compact}(i), \ref{ass:theta}(i),
and \ref{ass:smooth}(ii)), $\mu^{X}(\pi_{i})$, $\mu^{\epsilon}(\pi_{i})$,
$\frac{\partial\mu^{X}(\pi_{i})}{\partial\pi_{i}}$, and $\frac{\partial\mu^{\epsilon}(\pi_{i})}{\partial\pi_{i}}$
are bounded. Observe that $(z_{i},\pi_{i})$ is i.i.d.. By the law
of large numbers, we have
\[
\frac{1}{n}\sum_{i=1}^{n}\left(\frac{\partial\mu^{X}(\pi_{i})}{\partial\pi_{i}}\mu^{\epsilon}(\pi_{i})\frac{\partial\pi(z_{i},g_{i},\theta_{0})}{\partial\theta'}-\mathbb{E}\left[\frac{\partial\mu^{Z}(\pi_{i})}{\partial\pi_{i}}\mu^{\epsilon}(\pi_{i})\frac{\partial\pi(z_{i},g_{i},\theta_{0})}{\partial\theta'}\right]\right)=o_{p}(1).
\]
Moreover, following Lemma \ref{lem:xt_consist} we can show that
\[
\frac{1}{n}\sum_{i=1}^{n}\left((X_{i}-\mu_{0}^{X}(\pi_{i}))\frac{\partial\mu^{\epsilon}(\pi_{i})}{\partial\pi_{i}}\frac{\partial\pi(z_{i},g_{i},\theta_{0})}{\partial\theta'}-\mathbb{E}\left[(X_{i}-\mu_{0}^{X}(\pi_{i}))\frac{\partial\mu^{\epsilon}(\pi_{i})}{\partial\pi_{i}}\frac{\partial\pi(z_{i},g_{i},\theta_{0})}{\partial\theta'}\right]\right)=o_{p}(1).
\]
Combining the above two equations proves the lemma.
\end{proof}
\begin{lem}[Hoeffding projection]
\label{lem:weightedU}Let $\boldsymbol{q}=(q_{ij})$ denote $\boldsymbol{w}$
or $\boldsymbol{sw}^{t}$, $t=1,2$, and $t_{i}$ denote $x_{i}$\textup{
or $\lambda(\pi_{i})$}. Suppose that there exists $\tilde{\psi}_{i}$
that satisfies Assumption \ref{ass:w}(v). Consider the statistic
$W_{n}=n^{-1/2}\sum_{i=1}^{n}\sum_{j=1}^{n}q_{ij}\nu_{i}t_{j}$. Define
$W_{n}^{*}\equiv n^{-1/2}\sum_{i=1}^{n}h_{n}^{*}(\tilde{\psi}_{i},\nu_{i})$,
where $h_{n}^{*}(\tilde{\psi}_{i},\nu_{i})\equiv\sum_{j=1}^{n}\mathbb{E}[q_{ij}\nu_{i}t_{j}|\tilde{\psi}_{i},\nu_{i}]$.
Then $\|W_{n}-W_{n}^{*}\|=o_{p}(1)$.
\end{lem}
\begin{proof}
Our proof extends the results in \citet[Secion 3.7.5]{Lee1990} for
weighted $U$-statistics. In contrast to Lee's setting which assumes
deterministic weights in the $U$-statistic, our proof allows the
weights $q_{ij}$ to be random.

Let $\omega_{i}\equiv(\tilde{\psi}'_{i},\nu_{i})'$. Note that $t_{i}$
is a function of $\psi_{i}$ and thus $\tilde{\psi}_{i}$. Define
$h(\omega_{i},\omega_{j})\equiv\nu_{i}t_{j}$. We can write $W_{n}=n^{-1/2}\sum_{i}\sum_{j}q_{ij}h(\omega_{i},\omega_{j})$,
$W_{n}^{*}=n^{-1/2}\sum_{i}h_{n}^{*}(\omega_{i})$, and $h_{n}^{*}(\omega_{i})=\sum_{j}\mathbb{E}[q_{ij}h(\omega_{i},\omega_{j})|\omega_{i}]$.\footnote{\label{fn:qji}The standard Hoeffding projection is given by $\sum_{j}\mathbb{E}[q_{ij}h(\omega_{i},\omega_{j})|\omega_{i}]+\sum_{j\neq i}\mathbb{E}[q_{ji}h(\omega_{j},\omega_{i})|\omega_{i}]$,
but in our case, for $j\neq i$, $\mathbb{E}[q_{ji}h(\omega_{j},\omega_{i})|\omega_{i}]=\mathbb{E}[\mathbb{E}[q_{ji}\nu_{j}|\tilde{\boldsymbol{\psi}},\nu_{i}]|\tilde{\psi}_{i},\nu_{i}]t_{i}=\mathbb{E}[\mathbb{E}[q_{ji}|\tilde{\boldsymbol{\psi}}]\mathbb{E}[\nu_{j}|\tilde{\boldsymbol{\psi}}]|\tilde{\psi}_{i},\nu_{i}]t_{i}=0$.} Assumption \ref{ass:w}(v) implies that $\boldsymbol{q}$ and $\boldsymbol{\nu}$
are independent conditional on $\tilde{\boldsymbol{\psi}}$ and $\mathbb{E}[\nu_{i}|\tilde{\boldsymbol{\psi}}]=\mathbb{E}[\nu_{i}|\tilde{\boldsymbol{\psi}},\boldsymbol{\psi}]=\mathbb{E}[\nu_{i}|\boldsymbol{\psi}]=0$.
It follows that $\mathbb{E}[h(\omega_{i},\omega_{j})|\tilde{\boldsymbol{\psi}}]=\mathbb{E}[\nu_{i}|\tilde{\boldsymbol{\psi}}]t_{j}=0$.
Therefore, we obtain $\text{\ensuremath{\mathbb{E}}}[q_{ij}h(\omega_{i},\omega_{j})]=\text{\ensuremath{\mathbb{E}}}[\text{\ensuremath{\mathbb{E}}}[q_{ij}|\tilde{\boldsymbol{\psi}}]\text{\ensuremath{\mathbb{E}}}[h(\omega_{i},\omega_{j})|\tilde{\boldsymbol{\psi}}]]=0$
and $\text{\ensuremath{\mathbb{E}}}[h_{n}^{*}(\omega_{i})]=\sum_{j}\mathbb{E}[q_{ij}h(\omega_{i},\omega_{j})]=0$.
By Markov's inequality, it suffices if $\mathbb{E}\|W_{n}-W_{n}^{*}\|^{2}=o(1)$.

By definition, $\mathbb{E}[W'_{n}W_{n}^{*}]=n^{-1/2}\sum_{i}\mathbb{E}[W'_{n}h_{n}^{*}(\omega_{i})]$
and for each $i$, 
\begin{align*}
\mathbb{E}[W'_{n}h_{n}^{*}(\omega_{i})] & =n^{-1/2}\sum_{\tilde{i}}\sum_{j}\mathbb{E}[q_{\tilde{i}j}h(\omega_{\tilde{i}},\omega_{j})'h_{n}^{*}(\omega_{i})]\\
 & =n^{-1/2}\sum_{j}\mathbb{E}[q_{ij}h(\omega_{i},\omega_{j})'h_{n}^{*}(\omega_{i})]\\
 & =n^{-1/2}\mathbb{E}[h_{n}^{*}(\omega_{i})'h_{n}^{*}(\omega_{i})].
\end{align*}
The last equality follows by iterated expectations. The second to
last equality holds because for any $\tilde{i}\neq i$, we have $\mathbb{E}[q_{\tilde{i}j}h(\omega_{\tilde{i}},\omega_{j})'h_{n}^{*}(\omega_{i})]=\mathbb{E}[\mathbb{E}[q_{\tilde{i}j}|\tilde{\boldsymbol{\psi}}]\mathbb{E}[h(\omega_{\tilde{i}},\omega_{j})'h_{n}^{*}(\omega_{i})|\tilde{\boldsymbol{\psi}}]]=0$
as $\mathbb{E}[h(\omega_{\tilde{i}},\omega_{j})'h_{n}^{*}(\omega_{i})|\tilde{\boldsymbol{\psi}}]=\mathbb{E}[\nu_{\tilde{i}}h_{n}^{*}(\omega_{i})|\tilde{\boldsymbol{\psi}}]t'_{j}=\mathbb{E}[\nu_{\tilde{i}}|\tilde{\boldsymbol{\psi}}]\mathbb{E}[h_{n}^{*}(\omega_{i})|\tilde{\boldsymbol{\psi}}]t'_{j}=0$.
It then follows that $\mathbb{E}[W'_{n}W_{n}^{*}]=n^{-1}\sum_{i}\mathbb{E}[h_{n}^{*}(\omega_{i})'h_{n}^{*}(\omega_{i})]=\mathbb{E}\|W_{n}^{*}\|^{2}$
and thus $\mathbb{E}\|W_{n}-W_{n}^{*}\|^{2}=\mathbb{E}\|W_{n}\|^{2}-\mathbb{E}\|W_{n}^{*}\|^{2}$.
It remains to show that $\mathbb{E}\|W_{n}\|^{2}-\mathbb{E}\|W_{n}^{*}\|^{2}=o(1)$.

To show the last result, note that for any $\{i,j\}$ and $\{k,l\}$
with $i\neq k$, we can derive $\mathbb{E}[q_{ij}q_{kl}h(\omega_{i},\omega_{j})'h(\omega_{k},\omega_{l})]=\mathbb{E}[\mathbb{E}[q_{ij}q_{kl}|\tilde{\boldsymbol{\psi}}]\mathbb{E}[h(\omega_{i},\omega_{j})'h(\omega_{k},\omega_{l})|\tilde{\boldsymbol{\psi}}]]=0$.
The last equality follows because $\mathbb{E}[h(\omega_{i},\omega_{j})'h(\omega_{k},\omega_{l})|\tilde{\boldsymbol{\psi}}]=\mathbb{E}[\nu_{i}\nu_{k}|\tilde{\boldsymbol{\psi}}]t'_{j}t_{l}=\mathbb{E}[\nu_{i}|\tilde{\boldsymbol{\psi}}]\mathbb{E}[\nu_{k}|\tilde{\boldsymbol{\psi}}]t'_{j}t_{l}=0$.
Hence,
\begin{eqnarray*}
\mathbb{E}\|W_{n}\|^{2} & = & n^{-1}\sum_{i}\sum_{j}\sum_{k}\sum_{l}\mathbb{E}[q_{ij}q_{kl}h(\omega_{i},\omega_{j})'h(\omega_{k},\omega_{l})]\\
 & = & n^{-1}\sum_{i}\sum_{j}\sum_{k\neq j}\mathbb{E}[q_{ij}q_{ik}\nu_{i}^{2}t'_{j}t_{k}]+n^{-1}\sum_{i}\sum_{j}\mathbb{E}[q_{ij}^{2}\nu_{i}^{2}t'_{j}t_{j}].
\end{eqnarray*}
For comparison, because $\mathbb{E}\|W_{n}^{*}\|^{2}=n^{-1}\sum_{i}\mathbb{E}\|h_{n}^{*}(\omega_{i})\|^{2}$
we can write
\begin{eqnarray*}
\mathbb{E}\|W_{n}^{*}\|^{2} & = & n^{-1}\sum_{i}\sum_{j}\sum_{k}\mathbb{E}[\mathbb{E}[q_{ij}h(\omega_{i},\omega_{j})'|\omega_{i}]\mathbb{E}[q_{ik}h(\omega_{i},\omega_{k})|\omega_{i}]]\\
 & = & n^{-1}\sum_{i}\sum_{j}\sum_{k\neq j}\mathbb{E}[\mathbb{E}[q_{ij}\nu_{i}t'_{j}|\omega_{i}]\mathbb{E}[q_{ik}\nu_{i}t_{k}|\omega_{i}]]\\
 &  & +n^{-1}\sum_{i}\sum_{j}\mathbb{E}[\mathbb{E}[q_{ij}\nu_{i}t'_{j}|\omega_{i}]\mathbb{E}[q_{ij}\nu_{i}t_{j}|\omega_{i}]].
\end{eqnarray*}

Consider the triple sums over $i$, $j$, and $k\neq j$ in $\mathbb{E}\|W_{n}\|^{2}$
and $\mathbb{E}\|W_{n}^{*}\|^{2}$. They consist of the same number
of terms. For any $j\neq k$, we can derive
\begin{eqnarray*}
\mathbb{E}[q_{ij}q_{ik}\nu_{i}^{2}t'_{j}t_{k}] & = & \mathbb{E}[\mathbb{E}[q_{ij}q_{ik}|\tilde{\boldsymbol{\psi}},\nu_{i}]\nu_{i}^{2}t'_{j}t_{k}]\\
 & = & \mathbb{E}[\mathbb{E}[q_{ij}q_{ik}|\tilde{\boldsymbol{\psi}}]\nu_{i}^{2}t'_{j}t_{k}]\\
 & = & \mathbb{E}[\mathbb{E}[q_{ij}|\tilde{\psi}_{i},\tilde{\psi}_{j}]\mathbb{E}[q_{ik}|\tilde{\psi_{i}},\tilde{\psi}_{k}]\nu_{i}^{2}t'_{j}t_{k}]+o(n^{-2})\\
 & = & \mathbb{E}[\mathbb{E}[\mathbb{E}[q_{ij}|\tilde{\psi}_{i},\tilde{\psi}_{j}]\mathbb{E}[q_{ik}|\tilde{\psi_{i}},\tilde{\psi}_{k}]\nu_{i}^{2}t'_{j}t_{k}|\omega_{i}]]+o(n^{-2})\\
 & = & \mathbb{E}[\mathbb{E}[\mathbb{E}[q_{ij}|\tilde{\psi}_{i},\tilde{\psi}_{j}]\nu_{i}t'_{j}|\omega_{i}]\mathbb{E}[\mathbb{E}[q_{ik}|\tilde{\psi_{i}},\tilde{\psi}_{k}]\nu_{i}t_{k}|\omega_{i}]]+o(n^{-2})\\
 & = & \mathbb{E}[\mathbb{E}[q_{ij}\nu_{i}t'_{j}|\omega_{i}]\mathbb{E}[q_{ik}\nu_{i}t_{k}|\omega_{i}]]+o(n^{-2}).
\end{eqnarray*}
The second equality holds because $\boldsymbol{q}$ and $\boldsymbol{\nu}$
are independent given $\tilde{\boldsymbol{\psi}}$. The third equality
follows from $\mathbb{E}[(\mathbb{E}[q_{ij}q_{ik}|\tilde{\boldsymbol{\psi}}]-\mathbb{E}[q_{ij}|\tilde{\psi}_{i},\tilde{\psi}_{j}]\mathbb{E}[q_{ik}|\tilde{\psi}_{i},\tilde{\psi}_{k}])\nu_{i}^{2}t'_{j}t_{k}]\leq\max_{i,j,k\in\mathcal{N}:k\neq j}\mathbb{E}[(\mathbb{E}[q_{ij}q_{ik}|\tilde{\boldsymbol{\psi}}]-\mathbb{E}[q_{ij}|\tilde{\psi}_{i},\tilde{\psi}_{j}]\mathbb{E}[q_{ik}|\tilde{\psi}_{i},\tilde{\psi}_{k}])^{2}]^{1/2}\mathbb{E}[\nu_{i}^{4}]^{1/2}\max_{i}\|t_{i}\|^{2}\leq o(n^{-2})$
by Assumption \ref{ass:w}(v), the boundedness of $t_{i}$, and $\mathbb{E}[\nu_{i}^{4}]<\infty$
(Footnote \ref{fn:v_bd}). This also indicates that the $o(n^{-2})$
term does not depend on $i$, $j$ and $k$. The fifth equality follows
because for $j\neq k$, the terms $\mathbb{E}[q_{ij}|\tilde{\psi}_{i},\tilde{\psi}_{j}]\nu_{i}t_{j}$
and $\mathbb{E}[q_{ik}|\tilde{\psi}_{i},\tilde{\psi}_{k}]\nu_{i}t{}_{k}$
are independent conditional on $\omega_{i}$. The sixth equality follows
from $\mathbb{E}[q_{ij}|\tilde{\psi}_{i},\tilde{\psi}_{j}]=\mathbb{E}[q_{ij}|\tilde{\psi}_{i},\tilde{\psi}_{j},\nu_{i}]$
and iterated expectations.\footnote{Because $\boldsymbol{q}$ and $\boldsymbol{\nu}$ are independent
conditional on $\tilde{\boldsymbol{\psi}}$ and $\tilde{\psi}_{i}$
is i.i.d., we can show that $q_{ij}$ and $\nu_{i}$ are independent
conditional on $\tilde{\psi}_{i}$ and $\tilde{\psi}_{j}$.} Since the triple sums in $\mathbb{E}\|W_{n}\|^{2}$ and $\mathbb{E}\|W_{n}^{*}\|^{2}$
consist of $O(n^{3})$ terms, they differ by $n^{-1}\cdot O(n^{3})\cdot o(n^{-2})=o(1)$.

The double sums over $i$ and $j$ in $\mathbb{E}\|W_{n}\|^{2}$ and
$\mathbb{E}\|W_{n}^{*}\|^{2}$ consist of $O(n^{2})$ terms. For any
$i$ and $j$, both $\mathbb{E}[q_{ij}^{2}\nu_{i}^{2}t'_{j}t_{j}]$
and $\mathbb{E}[\mathbb{E}[q_{ij}\nu_{i}t'_{j}|\omega_{i}]\mathbb{E}[q_{ij}\nu_{i}t_{j}|\omega_{i}]]$
can be uniformly bounded by $O(n^{-2})$ because of $\mathbb{E}[\|\boldsymbol{q}\|_{\infty}^{4}]=O(n^{-4})$
(Lemma \ref{lem:network_bd} and Assumption \ref{ass:w}(ii)) and
$\mathbb{E}[\nu_{i}^{4}]<\infty$. Therefore, the second sums in $\mathbb{E}\|W_{n}\|^{2}$
and $\mathbb{E}\|W_{n}^{*}\|^{2}$ are both $n^{-1}\cdot O(n^{2})\cdot O(n^{-2})=o(1)$.
We conclude that $\mathbb{E}\|W_{n}\|^{2}-\mathbb{E}\|W_{n}^{*}\|^{2}=o(1)$.
\end{proof}

\begin{lem}
\label{lem:qvv}Let $\boldsymbol{q}=\boldsymbol{sw}$. Then 
\[
W_{n}=n^{-1/2}\sum_{i}\sum_{j}q_{ij}\nu_{i}\nu_{j}=o_{p}(1).
\]
\end{lem}
\begin{proof}
By Markov's inequality, it is sufficient to show that $\mathbb{E}[W_{n}^{2}]=o(1)$.
Write $W_{n}=n^{-1/2}(\sum_{i}q_{ii}\nu_{i}^{2}+\sum_{i}\sum_{j\neq i}q_{ij}\nu_{i}\nu_{j})$.
Recall that $\boldsymbol{q}$ is independent of $\boldsymbol{\nu}$
conditional on $\boldsymbol{\psi}=(\boldsymbol{x},\boldsymbol{z},\boldsymbol{g})$
(Assumption \ref{ass:adj_exog}), and for any $i\neq j$, $\mathbb{E}[\nu_{i}\nu_{j}|\boldsymbol{\psi}]=\mathbb{E}[\nu_{i}|\boldsymbol{\psi}]\mathbb{E}[\nu_{j}|\boldsymbol{\psi}]=0$
by i.i.d. $\psi_{i}$ and $\nu_{i}$. Therefore, we can derive
\begin{eqnarray*}
\mathbb{E}[W_{n}^{2}] & = & \frac{1}{n}\mathbb{E}(\sum_{i}q_{ii}\nu_{i}^{2}+\sum_{i}\sum_{j\neq i}q_{ij}\nu_{i}\nu_{j})^{2}\\
 & = & \frac{1}{n}\sum_{i}\mathbb{E}[q_{ii}^{2}\nu_{i}^{4}]+\frac{1}{n}\sum_{i}\sum_{j\neq i}\mathbb{E}[q_{ii}q_{jj}\nu_{i}^{2}\nu_{j}^{2}]+\frac{1}{n}\sum_{i}\sum_{j\neq i}\mathbb{E}[q_{ij}^{2}\nu_{i}^{2}\nu_{j}^{2}]\\
 &  & +\frac{1}{n}\sum_{i}\sum_{j\neq i}\sum_{k\neq i,j}\mathbb{E}[q_{ij}(q_{ik}+q_{ki})\nu_{i}^{2}\nu_{j}\nu_{k}+q_{ij}(q_{jk}+q_{kj})\nu_{i}\nu_{j}^{2}\nu_{k}]\\
 &  & +\frac{1}{n}\sum_{i}\sum_{j\neq i}\sum_{k\neq i,j}\sum_{l\neq i,j,k}\mathbb{E}[q_{ij}q_{kl}\nu_{i}\nu_{j}\nu_{k}\nu_{l}]+\frac{2}{n}\mathbb{E}(\sum_{i}q_{ii}\nu_{i}^{2})(\sum_{i}\sum_{j\neq i}q_{ij}\nu_{i}\nu_{j})\\
 & = & \frac{1}{n}\sum_{i}\mathbb{E}[q_{ii}^{2}\nu_{i}^{4}]+\frac{1}{n}\sum_{i}\sum_{j\neq i}\mathbb{E}[q_{ii}q_{jj}\nu_{i}^{2}\nu_{j}^{2}]+\frac{1}{n}\sum_{i}\sum_{j\neq i}\mathbb{E}[q_{ij}^{2}\nu_{i}^{2}\nu_{j}^{2}].
\end{eqnarray*}
The last equality follows because by iterated expectations and the
arguments above we have $\mathbb{E}[q_{ij}q_{ik}\nu_{i}^{2}\nu_{j}\nu_{k}]=\mathbb{E}[\mathbb{E}[q_{ij}q_{ik}|\boldsymbol{\psi}]\mathbb{E}[\nu_{i}^{2}\nu_{j}\nu_{k}|\boldsymbol{\psi}]]=0$
and similarly for the other terms in the third and fourth lines. These
terms have some $\nu_{i}$ that appears linearly, so they reduce to
zero.

As for the last three sums, we have $\mathbb{E}[q_{ii}^{2}\nu_{i}^{4}]\leq\mathbb{E}[q_{ii}^{4}]^{1/2}\mathbb{E}[\nu_{i}^{8}]^{1/2}=O(n^{-2})$,
$\mathbb{E}[q_{ii}q_{jj}\nu_{i}^{2}\nu_{j}^{2}]\leq\mathbb{E}[q_{ii}^{2}q_{jj}^{2}]^{1/2}\mathbb{E}[\nu_{i}^{4}]=O(n^{-2})$
and $\mathbb{E}[q_{ij}^{2}\nu_{i}^{2}\nu_{j}^{2}]\leq\mathbb{E}[q_{ij}^{4}]^{1/2}\mathbb{E}[\nu_{i}^{4}]=O(n^{-2})$
by Lemma \ref{lem:network_bd}, Assumptions \ref{ass:w}(ii) and \ref{ass:smooth}(i),
and i.i.d. $\nu_{i}$. Since each sum has at most $O(n^{2})$ terms,
we can bound $\mathbb{E}[W_{n}^{2}]$ by $n^{-1}\cdot O(n^{2})\cdot O(n^{-2})=o(1)$.
The proof is complete.
\end{proof}

\section{\protect\label{online:w}Examples of the Adjacency Matrix}

In this section, we verify Assumption \ref{ass:w} for several adjacency
matrices that are commonly used in the literature.
\begin{example}[Group averages that include oneself]
\label{ex:w.group.in}Suppose that $\boldsymbol{w}$ represents group
averages that include oneself and the group capacities are binding.
We can write $w_{ij}=\sum_{g=1}^{G}\frac{1}{n_{g}}1\{g_{i}=g\}1\{g_{j}=g\}$.
By construction, $\interleave\boldsymbol{w}\interleave_{\infty}=\max_{i\in\mathcal{N}}\sum_{j=1}^{n}|w_{ij}|=1$
and $\|\boldsymbol{w}\|_{\infty}=\max_{i,j\in\mathcal{N}}|w_{ij}|\leq\max_{g\in\mathcal{G}}\frac{1}{n_{g}}=\frac{1}{n}\max_{g\in\mathcal{G}}\frac{1}{r_{g}}$,
where $r_{g}=\frac{n_{g}}{n}>0$ for all $g\in\mathcal{G}$. Hence,
$\mathbb{E}[\|\boldsymbol{w}\|_{\infty}^{8}]\leq\frac{1}{n^{8}}\max_{g\in\mathcal{G}}\frac{1}{r_{g}^{8}}=O(n^{-8})$
and Assumptions \ref{ass:w}(i)-(ii) are satisfied. Note that $\boldsymbol{w}^{2}=\boldsymbol{w}$,
$\boldsymbol{w}'=\boldsymbol{w}$, and $\boldsymbol{s}\boldsymbol{w}=\frac{1}{1-\gamma_{1}}\boldsymbol{w}$.\footnote{For any $i,j\in\mathcal{N}$, $(\boldsymbol{w}^{2})_{ij}=\sum_{k=1}^{n}w_{ik}w_{kj}=\sum_{k=1}^{n}(\sum_{g=1}^{G}\frac{1}{n_{g}}1\{g_{i}=g\}1\{g_{k}=g\})(\sum_{g=1}^{G}\frac{1}{n_{g}}1\{g_{k}=g\}1\{g_{j}=g\})=\sum_{g=1}^{G}\sum_{k=1}^{n}\frac{1}{n_{g}^{2}}1\{g_{k}=g\}1\{g_{i}=g\}1\{g_{j}=g\}=\sum_{g=1}^{G}\frac{1}{n_{g}}1\{g_{i}=g\}1\{g_{j}=g\}=w_{ij}$,
where we have used $n_{g}=\sum_{k=1}^{n}1\{g_{k}=g\}$.} It thus suffices to verify Assumptions \ref{ass:w}(iii)-(v) for
$\boldsymbol{q}=\boldsymbol{w}$. Define $\boldsymbol{\tilde{\psi}}=\text{\ensuremath{\boldsymbol{\psi}}}=(\boldsymbol{x},\boldsymbol{z},\boldsymbol{g})$.
Because $n_{g}$ is a constant, $w_{ij}$ is a function of $g_{i}$
and $g_{j}$ -- once we know the groups that $i$ and $j$ join,
we know $w_{ij}$. In this case, $\boldsymbol{w}$ is a function of
$\boldsymbol{\psi}$ and $w_{ij}$ depends on $\boldsymbol{\psi}$
only through $\psi_{i}$ and $\psi_{j}$. Therefore, $\mathbb{E}[w_{ij}w_{kl}|\boldsymbol{\psi}]=\mathbb{E}[w_{ij}|\psi_{i},\psi_{j}]\mathbb{E}[w_{kl}|\psi_{k},\psi_{l}]=w_{ij}w_{kl}$
and $\mathbb{E}[w_{ij}|\boldsymbol{\psi}]=\mathbb{E}[w_{ij}|\psi_{i},\psi_{j}]=w_{ij}$.
Assumptions \ref{ass:w}(iii)-(v) are trivially satisfied.
\end{example}

\begin{example}[Group averages that exclude oneself]
\label{ex:w.group.ex}Suppose that $\boldsymbol{w}$ represents group
averages that exclude oneself and the group capacities are binding.
We can write $w_{ij}=\sum_{g=1}^{G}\frac{1}{n_{g}-1}1\{g_{i}=g\}1\{g_{j}=g\}$
for $i\neq j$ and $w_{ii}=0$. Similarly as in Example \ref{ex:w.group.in},
we can show that Assumptions \ref{ass:w}(i)-(ii) are satisfied. To
verify Assumptions \ref{ass:w}(iii)-(v), note that $\boldsymbol{w}'=\boldsymbol{w}$
and for $k\geq1$, the $(i,j)$ element of $\boldsymbol{w}^{k}$ takes
the form $(\boldsymbol{w}^{k})_{ij}=\sum_{g=1}^{G}c_{ij,g}(k)1\{g_{i}=g\}1\{g_{j}=g\}$,
where $c_{ij,g}(k)$ is a constant that depends on $k$ and $n_{g}$
only.\footnote{\label{fn:exclude_limit}For $k\geq2$, $c_{ij,g}(k)=\sum_{\kappa=1}^{k}(-1)^{\kappa-1}(n_{g}-1)^{-\kappa}$
for $i\neq j$ and $c_{ii,g}(k)=\sum_{\kappa=1}^{k-1}(-1)^{\kappa-1}(n_{g}-1)^{-\kappa}$.
Note that $c_{ij,g}(k)\rightarrow\frac{1}{n_{g}}$ as $k\rightarrow\infty$
(assuming $n_{g}>2$).} Moreover, $\boldsymbol{s}\boldsymbol{w}^{t}=\sum_{k=0}^{\infty}\gamma_{1}^{k}\boldsymbol{w}^{k+t}$
and $(\boldsymbol{s}\boldsymbol{w}^{r})'\boldsymbol{s}\boldsymbol{w}^{t}=\sum_{k=0}^{\infty}\sum_{l=0}^{\infty}\gamma_{1}^{k+l}\boldsymbol{w}^{k+l+r+t}$
for $r,t\geq1$. Hence, $(\boldsymbol{s}\boldsymbol{w}^{t})_{ij}=\sum_{g=1}^{G}(\sum_{k=0}^{\infty}\gamma_{1}^{k}c_{ij,g}(k+t))1\{g_{i}=g\}1\{g_{j}=g\}$
and $((\boldsymbol{s}\boldsymbol{w}^{r})'\boldsymbol{s}\boldsymbol{w}^{t})_{ij}=\sum_{g=1}^{G}(\sum_{k=0}^{\infty}\sum_{l=0}^{\infty}\gamma_{1}^{k+l}c_{ij,g}(k+l+r+t))1\{g_{i}=g\}1\{g_{j}=g\}$.
These results indicate that all the forms of $\boldsymbol{q}$ involved
in Assumptions \ref{ass:w}(iii)-(v) have the same dependence structure
as that of $\boldsymbol{w}$. Following the argument in Example \ref{ex:w.group.in},
we can show that Assumptions \ref{ass:w}(iii)-(v) are satisfied for
$\boldsymbol{\tilde{\psi}}=\boldsymbol{\psi}=(\boldsymbol{x},\boldsymbol{z},\boldsymbol{g})$.
\end{example}

\begin{example}[Dyadic networks]
\label{ex:w.fe}Suppose that individuals in a group form additional
connections (e.g., schoolmates make friends). Let $d_{ij,g}$ denote
an indicator for whether individuals $i$ and $j$ are connected in
group $g$ and $d_{i,g}\equiv\sum_{j=1}^{n}d_{ij,g}1\{g_{j}=g\}$
the number of connections that $i$ has in group $g$. Suppose that
no individual is isolated so that $d_{i,g}\geq1$ for all $i\in\mathcal{N}$
and all $g\in\mathcal{G}$. Typically, we specify $w_{ij}=\sum_{g=1}^{G}\frac{d_{ij,g}}{d_{i,g}}1\{g_{i}=g\}1\{g_{j}=g\}$
-- if both $i$ and $j$ join group $g$, then $j$'s weight on $i$
depends on whether $j$ is connected to $i$, normalized by the number
of connections that $i$ has in the group.

Following the literature on dyadic network formation with fixed effects
(\citealp{Graham2017}; \citealp{johnsson2021estimation}), we specify
$d_{ij,g}\equiv1\{f_{g}(x_{i},x_{j},a_{i},a_{j})\geq\text{\ensuremath{\zeta_{ij}}}\}$,
$\text{ \ensuremath{\forall i\neq j}}$, and $d_{ii,g}=0$, where
$a_{i}\in\mathbb{R}$ and $\zeta_{ij}\in\mathbb{R}$ represent individual-
and pair-specific unobserved heterogeneity. The links can be directed
or undirected. Without loss of generality we normalize $\zeta_{ij}\sim U[0,1]$
and assume $0\leq f_{g}\leq1$. We also assume that $a_{i}$ has finite
fourth moment. Denote $\boldsymbol{a}=(a_{1},\dots,a_{n})'$ and $\boldsymbol{\zeta}=(\zeta_{ij})$.
Let $\text{\ensuremath{\boldsymbol{\psi}}}=(\boldsymbol{x},\boldsymbol{z},\boldsymbol{g})$.
We assume that (a) $a_{i}$ and $\zeta_{ij}$ are i.i.d. and (b) $\boldsymbol{\zeta}$
is independent of $\boldsymbol{a}$ conditional on $\boldsymbol{\psi}$,
and (c) $(\boldsymbol{a},\boldsymbol{\zeta})$ is independent of $\boldsymbol{\epsilon}$
conditional on $\text{\ensuremath{\boldsymbol{\psi}}}$. The last
part is consistent with Assumption \ref{ass:adj_exog} -- conditional
on $\boldsymbol{\psi}$, $\boldsymbol{w}$ is a function of $(\boldsymbol{a},\boldsymbol{\zeta})$
and is thus independent of $\boldsymbol{\epsilon}$.\footnote{Our setting differs from those in \citet{johnsson2021estimation}
and \citet{Auerbach2022}, who consider endogenous link formation.
Their settings allow the individual effects $\boldsymbol{a}$ to be
correlated with $\boldsymbol{\epsilon}$.} Define $\tilde{\psi}_{i}\equiv(\psi'_{i},a_{i})'$ and $\text{\ensuremath{\tilde{\boldsymbol{\psi}}}}\equiv(\ensuremath{\boldsymbol{\psi}},\boldsymbol{a})$.
Conditional on $\boldsymbol{\psi}$, $\text{\ensuremath{\tilde{\boldsymbol{\psi}}}}$
is a function of $\boldsymbol{a}$ and is therefore independent of
$\boldsymbol{\epsilon}$. Similarly, conditional on $\text{\ensuremath{\tilde{\boldsymbol{\psi}}}}$,
$\boldsymbol{w}$ is a function of $\boldsymbol{\zeta}$ and is thus
independent of $\boldsymbol{\epsilon}$. Consequently, $\text{\ensuremath{\tilde{\boldsymbol{\psi}}}}$
satisfies the independence conditions specified in Assumption \ref{ass:w}(v).

Note that $\frac{1}{n-1}\mathbb{E}[d_{i,g}|\tilde{\psi}_{i}]=\frac{1}{n-1}\sum_{j\neq i}\mathbb{E}[d_{ij,g}1\{g_{j}=g\}|\tilde{\psi}_{i}]=\mathbb{E}[f_{g}(x_{i},x_{j},a_{i},a_{j})1\{g_{j}=g\}|\tilde{\psi}_{i}]$.
We assume that $\min_{g\in\mathcal{G}}\min_{i\in\mathcal{N}}\mathbb{E}[f_{g}(x_{i},x_{j},a_{i},a_{j})1\{g_{j}=g\}|\tilde{\psi}_{i}]\geq c>0$,
which implies that the network within each group is dense. Assume
$c>|\gamma_{1}|$. Each link $d_{ij,g}$ is a function of dyadic variables
$(\tilde{\psi}_{i},\tilde{\psi}_{j},\zeta_{ij})$. However, the degree
$d_{i,g}$ depends on variables from other links, though this dependence
vanishes asymptotically. Lemma \ref{lem:w.fe} below verifies that
$\boldsymbol{w}$ in this example satisfies Assumption \ref{ass:w}.
\end{example}

\begin{example}[Group averages, continued]
Examples \ref{ex:w.group.in} and \ref{ex:w.group.ex} assume that
the group capacities are binding. If a group has an infinite capacity
(as in one-sided group formation) or does not reach its capacity,
then the number of members in that group is endogenously determined.
This setting can be regarded as a special case of Example \ref{ex:w.fe},
where we set $d_{ij,g}=1$ for all $i,j\in\mathcal{N}$ (including-oneself
averages) or $d_{ij,g}=1$ for all $i\neq j$ and $d_{ii,g}=0$ (excluding-oneself
averages). Similarly as in Lemma \ref{lem:w.fe}, we can show that
the $\boldsymbol{w}$ in this case satisfies Assumption \ref{ass:w}.
\end{example}
\begin{lem}
\label{lem:w.fe}The adjacency matrix $\boldsymbol{w}$ specified
in Example \ref{ex:w.fe} satisfies Assumption \ref{ass:w}.
\end{lem}
\begin{proof}
By construction, $\interleave\boldsymbol{w}\interleave_{\infty}=1$
and $\|\boldsymbol{w}\|_{\infty}\leq\max_{g\in\mathcal{G}}\max_{i\in\mathcal{N}}\frac{1}{d_{i,g}}$.
Note that $d_{i,g}\geq\mathbb{E}[d_{i,g}|\tilde{\psi}_{i}]-|d_{i,g}-\mathbb{E}[d_{i,g}|\tilde{\psi}_{i}]|$.
Because $\min_{g\in\mathcal{G}}\min_{i\in\mathcal{N}}\frac{1}{n-1}\mathbb{E}[d_{i,g}|\tilde{\psi}_{i}]\geq c>0$
by assumption and $\max_{g\in\mathcal{G}}\max_{i\in\mathcal{N}}\frac{1}{n-1}|d_{i,g}-\mathbb{E}[d_{i,g}|\tilde{\psi}_{i}]|=o_{p}(1)$
by the uniform law of large numbers, we can derive $\mathbb{E}[\max_{g\in\mathcal{G}}\max_{i\in\mathcal{N}}(\frac{n-1}{d_{i,g}})^{8}]\leq\mathbb{E}[(c-\max_{g\in\mathcal{G}}\max_{i\in\mathcal{N}}\frac{1}{n-1}|d_{i,g}-\mathbb{E}[d_{i,g}|\tilde{\psi}_{i}]|)^{-8}]\rightarrow c^{-8}<\infty$
by Portmanteau theorem. It follows that $\mathbb{E}[\|\boldsymbol{w}\|_{\infty}^{8}]=O(n^{-8})$.
Hence, Assumptions \ref{ass:w}(i)--(ii) are satisfied.

Step 1: to verify Assumptions \ref{ass:w}(iii)--(v), we first consider
the case $\boldsymbol{q}=\boldsymbol{w}$ and then extend the argument
to other forms of $\boldsymbol{q}$. For $i\neq j$, define $w_{ij,g}\equiv\frac{d_{ij,g}}{d_{i,g}}$,
$\bar{w}_{ij,g}\equiv\frac{d_{ij,g}}{\mathbb{E}[d_{i,g}|\tilde{\psi}_{i}]}$,
and $e_{ij,g}^{w}\equiv w_{ij,g}-\bar{w}_{ij,g}=d_{ij,g}(\frac{1}{d_{i,g}}-\frac{1}{\mathbb{E}[d_{i,g}|\tilde{\psi}_{i}]})$.
By Taylor expansion,
\begin{equation}
e_{ij,g}^{w}=-\frac{d_{ij,g}}{\mathbb{E}[d_{i,g}|\tilde{\psi}_{i}]^{2}}(d_{i,g}-\mathbb{E}[d_{i,g}|\tilde{\psi}_{i}])+\frac{d_{ij,g}}{\mathbb{E}[d_{i,g}|\tilde{\psi}_{i}]^{3}}(d_{i,g}-\mathbb{E}[d_{i,g}|\tilde{\psi}_{i}])^{2}-\cdots\label{eq:eij_fe}
\end{equation}
It suffices to consider the leading term in $e_{ij,g}^{w}$. Recall
that $d_{i,g}-\mathbb{E}[d_{i,g}|\tilde{\psi}_{i}]=\sum_{j\neq i}r_{ij,g}$,
where $r_{ij,g}\equiv d_{ij,g}1\{g_{j}=g\}-\mathbb{E}[d_{ij,g}1\{g_{j}=g\}|\tilde{\psi}_{i}]$.
Note that $|r_{ij,g}|\leq1$ and $\mathbb{E}[r_{ij,g}|\tilde{\psi}_{i}]=0$.
For any $j\neq k$, conditional on $\tilde{\psi}_{i}$, $r_{ij,g}$
is a function of $(\tilde{\psi}_{j},\zeta_{ij})$ and $r_{ik,g}$
is a function of $(\tilde{\psi}_{k},\zeta_{ik})$, so $r_{ij,g}$
and $r_{ik,g}$ are independent. Therefore,
\begin{eqnarray}
\mathbb{E}[(d_{i,g}-\mathbb{E}[d_{i,g}|\tilde{\psi}_{i}])^{4}] & = & \sum_{j,k,l,m\neq i}\mathbb{E}[r_{ij,g}r_{ik,g}r_{il,g}r_{im,g}]\nonumber \\
 & = & \sum_{j\neq i}\mathbb{E}[r_{ij,g}^{4}]+\sum_{j,k\neq i,j\neq k}\mathbb{E}[r_{ij,g}^{2}r_{ik,g}^{2}]\leq O(n^{2}).\label{eq:Ed^4_fe}
\end{eqnarray}
Combining the two displays along with $\min_{g\in\mathcal{G}}\min_{i\in\mathcal{N}}\frac{1}{n-1}\mathbb{E}[d_{i,g}|\tilde{\psi}_{i}]\geq c>0$,
we obtain $\max_{g\in\mathcal{G}}\max_{i,j\in\mathcal{N}}\mathbb{E}[|e_{ij,g}^{w}|^{4}]\leq O(n^{-6})$.
Furthermore, summing over the groups we define $\bar{w}_{ij}\equiv\sum_{g=1}^{G}\bar{w}_{ij,g}1\{g_{i}=g\}1\{g_{j}=g\}$
and $e_{ij}^{w}\equiv w_{ij}-\bar{w}_{ij}$. We can derive $\max_{i,j\in\mathcal{N}}|\bar{w}_{ij}|\leq\frac{1}{\min_{g\in\mathcal{G}}\min_{i\in\mathcal{N}}\mathbb{E}[d_{i,g}|\tilde{\psi}_{i}]}\leq\frac{1}{c(n-1)}$
and $\max_{i,j\in\mathcal{N}}\mathbb{E}[|e_{ij}^{w}|^{4}]\leq G\max_{g\in\mathcal{G}}\max_{i,j\in\mathcal{N}}\mathbb{E}[|e_{ij,g}^{w}|^{4}]\leq O(n^{-6})$.

Assumption \ref{ass:w}(iii) for $\boldsymbol{q}=\boldsymbol{w}$.
Fix disjoint $\{i,j\}$ and $\{k,l\}$. Conditional on $\boldsymbol{\psi}$,
$\bar{w}_{ij}$ is a function of $(a_{i},a_{j},\zeta_{ij})$ and $\bar{w}_{kl}$
is a function of $(a_{k},a_{l},\zeta_{kl})$, so they are independent.
Moreover, $\mathbb{E}[\bar{w}_{ij}|\boldsymbol{\text{\ensuremath{\psi}}}]=\mathbb{E}[\bar{w}_{ij}|\psi_{i},\psi_{j}]$
because $\bar{w}_{ij}$ depends on $\boldsymbol{\text{\ensuremath{\psi}}}$
only through $\psi_{i}$ and $\psi_{j}$. Therefore, 
\begin{eqnarray}
\mathbb{E}[w_{ij}w_{kl}|\boldsymbol{\text{\ensuremath{\psi}}}] & = & \mathbb{E}[(\bar{w}_{ij}+e_{ij}^{w})(\bar{w}_{kl}+e_{kl}^{w})|\boldsymbol{\text{\ensuremath{\psi}}}]\nonumber \\
 & = & \mathbb{E}[\bar{w}_{ij}|\boldsymbol{\text{\ensuremath{\psi}}}]\mathbb{E}[\bar{w}_{kl}|\boldsymbol{\text{\ensuremath{\psi}}}]+\mathbb{E}[e_{ij}^{w}\bar{w}_{kl}+\bar{w}_{ij}e_{kl}^{w}+e_{ij}^{w}e_{kl}^{w}|\boldsymbol{\text{\ensuremath{\psi}}}]\nonumber \\
 & = & \mathbb{E}[\bar{w}_{ij}|\psi_{i},\psi_{j}]\mathbb{E}[\bar{w}_{kl}|\psi_{k},\psi_{l}]+\mathbb{E}[e_{ij}^{w}\bar{w}_{kl}+\bar{w}_{ij}e_{kl}^{w}+e_{ij}^{w}e_{kl}^{w}|\boldsymbol{\text{\ensuremath{\psi}}}]\nonumber \\
 & = & \mathbb{E}[w_{ij}|\psi_{i},\psi_{j}]\mathbb{E}[w_{kl}|\psi_{k},\psi_{l}]+(\mathbb{E}[e_{ij}^{w}\bar{w}_{kl}|\boldsymbol{\psi}]-\mathbb{E}[e_{ij}^{w}|\psi_{i},\psi_{j}]\mathbb{E}[\bar{w}_{kl}|\psi_{k},\psi_{l}])\nonumber \\
 &  & +(\mathbb{E}[\bar{w}_{ij}e_{kl}^{w}|\boldsymbol{\psi}]-\mathbb{E}[\bar{w}_{ij}|\psi_{i},\psi_{j}]\mathbb{E}[e_{kl}^{w}|\psi_{k},\psi_{l}])\nonumber \\
 &  & +(\mathbb{E}[e_{ij}^{w}e_{kl}^{w}|\boldsymbol{\psi}]-\mathbb{E}[e_{ij}^{w}|\psi_{i},\psi_{j}]\mathbb{E}[e_{kl}^{w}|\psi_{k},\psi_{l}]).\label{eq:w_ij*w_kl}
\end{eqnarray}
By Jensen's inequality and Cauchy-Schwarz inequality, we can bound
$\mathbb{E}[(\mathbb{E}[e_{ij}^{w}\bar{w}_{kl}|\boldsymbol{\psi}]-\mathbb{E}[e_{ij}^{w}|\psi_{i},\psi_{j}]\mathbb{E}[\bar{w}_{kl}|\psi_{k},\psi_{l}])^{2}]\leq\frac{C}{(n-1)^{2}}\max_{i,j\in\mathcal{N}}\mathbb{E}[|e_{ij}^{w}|^{2}]\leq O(n^{-5})$
uniformly and also bound $\mathbb{E}[(\mathbb{E}[e_{ij}^{w}e_{kl}^{w}|\boldsymbol{\psi}]-\mathbb{E}[e_{ij}^{w}|\psi_{i},\psi_{j}]\mathbb{E}[e_{kl}^{w}|\psi_{k},\psi_{l}])^{2}]\leq C\max_{i,j\in\mathcal{N}}\mathbb{E}[|e_{ij}^{w}|^{4}]\leq O(n^{-6})$
uniformly. It follows that $\max_{i,j,k,l\in\mathcal{N}:\{i,j\}\cap\{k,l\}=\emptyset}\mathbb{E}[(\mathbb{E}[w_{ij}w_{kl}|\boldsymbol{\psi}]-\mathbb{E}[w_{ij}|\psi_{i},\psi_{j}]\mathbb{E}[w_{kl}|\psi_{k},\psi_{l}])^{2}]\leq O(n^{-5})=o(n^{-4}/K)$
because $K/n\rightarrow0$. 

In addition, for any $i,j\in\mathcal{N}$, because $\mathbb{E}[\bar{w}_{ij}|\boldsymbol{\psi}]=\mathbb{E}[\bar{w}_{ij}|\psi_{i},\psi_{j}]$,
we obtain $\mathbb{E}[w_{ij}|\boldsymbol{\psi}]-\mathbb{E}[w_{ij}|\psi_{i},\psi_{j}]=\mathbb{E}[e_{ij}^{w}|\boldsymbol{\psi}]-\mathbb{E}[e_{ij}^{w}|\psi_{i},\psi_{j}]$.
By Jensen's inequality and Cauchy-Schwarz inequality again, we can
bound $\mathbb{E}[(\mathbb{E}[e_{ij}^{w}|\boldsymbol{\psi}]-\mathbb{E}[e_{ij}^{w}|\psi_{i},\psi_{j}])^{4}]\leq C\max_{i,j\in\mathcal{N}}\mathbb{E}[|e_{ij}^{w}|^{4}]\leq O(n^{-6})$
uniformly. Therefore, $\max_{i,j\in\mathcal{N}}\mathbb{E}[(\mathbb{E}[w_{ij}|\boldsymbol{\psi}]-\mathbb{E}[w_{ij}|\psi_{i},\psi_{j}])^{4}]\leq O(n^{-6})=o(n^{-4}/K^{2})$
because $K/n\rightarrow0$. Assumption \ref{ass:w}(iii) is satisfied
for $\boldsymbol{q}=\boldsymbol{w}$.

Assumption \ref{ass:w}(iv) for $\boldsymbol{q}=\boldsymbol{w}$ is
implied by Assumption \ref{ass:w}(iii).

Assumption \ref{ass:w}(v) for $\boldsymbol{q}=\boldsymbol{w}$. Fix
$\{i,j\}$ and $\{i,k\}$ with $j\neq k$. Conditional on $\tilde{\boldsymbol{\psi}}$,
$\bar{w}_{ij}$ is a function of $\zeta_{ij}$ and $\bar{w}_{ik}$
is a function of $\zeta_{ik}$, so they are independent. Moreover,
$\mathbb{E}[\bar{w}_{ij}|\tilde{\boldsymbol{\text{\ensuremath{\psi}}}}]=\mathbb{E}[\bar{w}_{ij}|\tilde{\psi}_{i},\tilde{\psi}_{j}]$
because $\bar{w}_{ij}$ depends on $\tilde{\boldsymbol{\text{\ensuremath{\psi}}}}$
only through $\tilde{\psi}_{i}$ and $\tilde{\psi}_{j}$. Hence, equation
(\ref{eq:w_ij*w_kl}) holds for $\{i,j\}$ and $\{i,k\}$ with $j\neq k$,
and $\tilde{\boldsymbol{\psi}}$ in place of $\boldsymbol{\psi}$,
that is,
\begin{eqnarray}
\mathbb{E}[w_{ij}w_{ik}|\tilde{\boldsymbol{\psi}}] & = & \mathbb{E}[(\bar{w}_{ij}+e_{ij}^{w})(\bar{w}_{ik}+e_{ik}^{w})|\tilde{\boldsymbol{\psi}}]\nonumber \\
 & = & \mathbb{E}[\bar{w}_{ij}|\tilde{\boldsymbol{\psi}}]\mathbb{E}[\bar{w}_{ik}|\tilde{\boldsymbol{\psi}}]+\mathbb{E}[e_{ij}^{w}\bar{w}_{ik}+\bar{w}_{ij}e_{ik}^{w}+e_{ij}^{w}e_{ik}^{w}|\tilde{\boldsymbol{\psi}}]\nonumber \\
 & = & \mathbb{E}[\bar{w}_{ij}|\tilde{\psi}_{i},\tilde{\psi}_{j}]\mathbb{E}[\bar{w}_{ik}|\tilde{\psi}_{i},\tilde{\psi}_{k}]+\mathbb{E}[e_{ij}^{w}\bar{w}_{ik}+\bar{w}_{ij}e_{ik}^{w}+e_{ij}^{w}e_{ik}^{w}|\tilde{\boldsymbol{\psi}}]\nonumber \\
 & = & \mathbb{E}[w_{ij}|\tilde{\psi}_{i},\tilde{\psi}_{j}]\mathbb{E}[w_{ik}|\tilde{\psi}_{i},\tilde{\psi}_{k}]+(\mathbb{E}[e_{ij}^{w}\bar{w}_{ik}|\tilde{\boldsymbol{\psi}}]-\mathbb{E}[e_{ij}^{w}|\tilde{\psi}_{i},\tilde{\psi}_{j}]\mathbb{E}[\bar{w}_{ik}|\tilde{\psi}_{i},\tilde{\psi}_{k}])\nonumber \\
 &  & +(\mathbb{E}[\bar{w}_{ij}e_{ik}^{w}|\tilde{\boldsymbol{\psi}}]-\mathbb{E}[\bar{w}_{ij}|\tilde{\psi}_{i},\tilde{\psi}_{j}]\mathbb{E}[e_{ik}^{w}|\tilde{\psi}_{i},\tilde{\psi}_{k}])\nonumber \\
 &  & +(\mathbb{E}[e_{ij}^{w}e_{ik}^{w}|\tilde{\boldsymbol{\psi}}]-\mathbb{E}[e_{ij}^{w}|\tilde{\psi}_{i},\tilde{\psi}_{j}]\mathbb{E}[e_{ik}^{w}|\tilde{\psi}_{i},\tilde{\psi}_{k}]).\label{eq:w_ij*w_ik}
\end{eqnarray}
Similarly as before, we can bound $\mathbb{E}[(\mathbb{E}[e_{ij}^{w}\bar{w}_{ik}|\tilde{\boldsymbol{\psi}}]-\mathbb{E}[e_{ij}^{w}|\tilde{\psi}_{i},\tilde{\psi}_{j}]\mathbb{E}[\bar{w}_{ik}|\tilde{\psi}_{i},\tilde{\psi}_{k}])^{2}]\leq O(n^{-5})$
uniformly and $\mathbb{E}[(\mathbb{E}[e_{ij}^{w}e_{ik}^{w}|\tilde{\boldsymbol{\psi}}]-\mathbb{E}[e_{ij}^{w}|\tilde{\psi}_{i},\tilde{\psi}_{j}]\mathbb{E}[e_{ik}^{w}|\tilde{\psi}_{i},\tilde{\psi}_{k}])^{2}]\leq O(n^{-6})$
uniformly. Therefore, $\max_{i,j,k\in\mathcal{N}:j\neq k}\mathbb{E}[(\mathbb{E}[w_{ij}w_{ik}|\tilde{\boldsymbol{\psi}}]-\mathbb{E}[w_{ij}|\tilde{\psi}_{i},\tilde{\psi}_{j}]\mathbb{E}[w_{ik}|\tilde{\psi}_{i},\tilde{\psi}_{k}])^{2}]\leq O(n^{-5})=o(n^{-4})$.
Assumption \ref{ass:w}(v) is satisfied for $\boldsymbol{q}=\boldsymbol{w}$.

Step 2: next we consider other forms of $\boldsymbol{q}$. Let $\bar{\boldsymbol{w}}\equiv(\bar{w}_{ij})$,
where $\bar{w}_{ii}\equiv0$, and define $\bar{\boldsymbol{s}}\equiv(I_{n}-\gamma_{1}\bar{\boldsymbol{w}})^{-1}$.
Analogously, define $\bar{\boldsymbol{q}}$ as we defined $\boldsymbol{q}$,
but with $\bar{\boldsymbol{w}}$ and $\bar{\boldsymbol{s}}$ in place
of $\boldsymbol{w}$ and $\boldsymbol{s}$. For example, if $\boldsymbol{q}=\boldsymbol{s}\boldsymbol{w}^{t}$,
then $\bar{\boldsymbol{q}}=\bar{\boldsymbol{s}}\bar{\boldsymbol{w}}^{t}$.
Because $\|\bar{\boldsymbol{w}}\|_{\infty}=\max_{i,j\in N}|\bar{w}_{ij}|\leq\frac{1}{c(n-1)}$,
we can bound $\interleave\bar{\boldsymbol{w}}\interleave_{\infty}\leq\frac{1}{c}$
and $\interleave\bar{\boldsymbol{w}}\interleave_{1}\leq\frac{1}{c}$.
Moreover, we can bound $\|\bar{\boldsymbol{w}}^{\tau}\|_{\infty}\leq\interleave\bar{\boldsymbol{w}}\interleave_{\infty}^{\tau-1}\|\bar{\boldsymbol{w}}\|_{\infty}=\frac{1}{c^{\tau}(n-1)}$
for $\tau\geq1$. Therefore, for $\bar{\boldsymbol{q}}=\boldsymbol{\bar{w}}'\bar{\boldsymbol{w}}$,
we have $\|\bar{\boldsymbol{q}}\|_{\infty}=\|\boldsymbol{\bar{w}}'\bar{\boldsymbol{w}}\|_{\infty}\leq\interleave\boldsymbol{\bar{w}}\interleave_{1}\|\boldsymbol{\bar{w}}\|_{\infty}\leq\frac{1}{c^{2}(n-1)}$.
Note that $\boldsymbol{s}\boldsymbol{w}^{t}=\sum_{\tau=t}^{\infty}\gamma_{1}^{\tau-t}\boldsymbol{w}^{\tau}$,
$\boldsymbol{w}'\boldsymbol{s}\boldsymbol{w}^{t}=\sum_{\tau=t}^{\infty}\gamma_{1}^{\tau-t}\boldsymbol{w}'\boldsymbol{w}^{\tau}$,
and $(\boldsymbol{s}\boldsymbol{w}^{r})'\boldsymbol{s}\boldsymbol{w}^{t}=(\sum_{\tau_{1}=r}^{\infty}\gamma_{1}^{\tau_{1}-r}\boldsymbol{w}^{\tau_{1}})'(\sum_{\tau_{2}=t}^{\infty}\gamma_{1}^{\tau_{2}-t}\boldsymbol{w}^{\tau_{2}})=\sum_{\tau_{1}=r}^{\infty}\sum_{\tau_{2}=t}^{\infty}\gamma_{1}^{\tau_{1}+\tau_{2}-r-t}(\boldsymbol{w}')^{\tau_{1}}\boldsymbol{w}^{\tau_{2}}$,
$r,t=1,2$. For $\bar{\boldsymbol{q}}=\bar{\boldsymbol{s}}\bar{\boldsymbol{w}}^{t}$,
we can bound $\|\bar{\boldsymbol{q}}\|_{\infty}=\|\bar{\boldsymbol{s}}\bar{\boldsymbol{w}}^{t}\|_{\infty}\leq\sum_{\tau=t}^{\infty}\gamma_{1}^{\tau-t}\|\bar{\boldsymbol{w}}^{\tau}\|_{\infty}\leq\sum_{\tau=t}^{\infty}\gamma_{1}^{\tau-t}\frac{1}{c^{\tau}(n-1)}\leq\frac{1}{c^{t-1}(c-\gamma_{1})(n-1)}$.
For $\bar{\boldsymbol{q}}=\boldsymbol{\bar{w}}'\bar{\boldsymbol{s}}\bar{\boldsymbol{w}}^{t}$,
we can derive $\|\bar{\boldsymbol{q}}\|_{\infty}=\|\boldsymbol{\bar{w}}'\bar{\boldsymbol{s}}\bar{\boldsymbol{w}}^{t}\|_{\infty}\leq\interleave\bar{\boldsymbol{w}}\interleave_{1}\|\bar{\boldsymbol{s}}\bar{\boldsymbol{w}}^{t}\|_{\infty}\leq\frac{1}{c^{t}(c-\gamma_{1})(n-1)}$.
For $\bar{\boldsymbol{q}}=(\bar{\boldsymbol{s}}\bar{\boldsymbol{w}}^{r})'\bar{\boldsymbol{s}}\bar{\boldsymbol{w}}^{t}$,
note that $\interleave\bar{\boldsymbol{w}}^{\tau}\interleave_{1}\leq(n-1)\|\bar{\boldsymbol{w}}^{\tau}\|_{\infty}\leq\frac{1}{c^{\tau}}$
and thus $\interleave\bar{\boldsymbol{s}}\bar{\boldsymbol{w}}^{r}\interleave_{1}\leq\sum_{\tau=r}^{\infty}\gamma_{1}^{\tau-r}\interleave\bar{\boldsymbol{w}}^{\tau}\interleave_{1}\leq\frac{1}{c^{r-1}(c-\gamma_{1})}$.
Hence, we obtain $\|\bar{\boldsymbol{q}}\|_{\infty}=\|(\bar{\boldsymbol{s}}\bar{\boldsymbol{w}}^{r})'\bar{\boldsymbol{s}}\bar{\boldsymbol{w}}^{t}\|_{\infty}\leq\interleave\bar{\boldsymbol{s}}\bar{\boldsymbol{w}}^{r}\interleave_{1}\|\bar{\boldsymbol{s}}\bar{\boldsymbol{w}}^{t}\|_{\infty}\leq\frac{1}{c^{r+t-2}(c-\gamma_{1})^{2}(n-1)}$.
In sum, we have the bound $\|\bar{\boldsymbol{q}}\|_{\infty}\leq\frac{C}{n-1}$
for all forms of $\boldsymbol{q}$.

Denote the $(i,j)$ element of $\bar{\boldsymbol{q}}$ as $\bar{q}_{ij}$.
Define the difference between $q_{ij}$ and $\bar{q}_{ij}$ by $e_{ij}^{q}\equiv q_{ij}-\bar{q}_{ij}$.
We will follow the proof for $\boldsymbol{q}=\boldsymbol{w}$ and
derive a bound on $\max_{i,j\in\mathcal{N}}\mathbb{E}[|e_{ij}^{q}|^{4}]$
for all forms of $\boldsymbol{q}$. Let us take $\boldsymbol{q}=(\boldsymbol{s}\boldsymbol{w}^{r})'\boldsymbol{s}\boldsymbol{w}^{t}$
as an example. Using H\"{o}lder's inequality $(\sum_{\tau}|a_{\tau}b_{\tau}|)^{4}\leq(\sum_{\tau}|a_{\tau}|^{4/3})^{3}\cdot(\sum_{\tau}|b_{\tau}|^{4})$,
we can bound
\begin{eqnarray}
 &  & \mathbb{E}[|e_{ij}^{q}|^{4}]\nonumber \\
 & = & \mathbb{E}[|((\boldsymbol{s}\boldsymbol{w}^{r})'\boldsymbol{s}\boldsymbol{w}^{t})_{ij}-((\bar{\boldsymbol{s}}\boldsymbol{\bar{w}}^{r})'\bar{\boldsymbol{s}}\boldsymbol{\bar{w}}^{t})_{ij}|^{4}]\nonumber \\
 & = & \mathbb{E}\left[\left(\sum_{\tau_{1}=r}^{\infty}\sum_{\tau_{2}=t}^{\infty}\gamma_{1}^{\frac{3}{4}(\tau_{1}+\tau_{2}-r-t)}\cdot\gamma_{1}^{\frac{1}{4}(\tau_{1}+\tau_{2}-r-t)}(((\boldsymbol{w}')^{\tau_{1}}\boldsymbol{w}^{\tau_{2}})_{ij}-((\bar{\boldsymbol{w}}')^{\tau_{1}}\bar{\boldsymbol{w}}^{\tau_{2}})_{ij})\right)^{4}\right]\nonumber \\
 & \leq & \mathbb{E}\left(\sum_{\tau_{1}=r}^{\infty}\sum_{\tau_{2}=t}^{\infty}\gamma_{1}^{\tau_{1}+\tau_{2}-r-t}\right)^{3}\cdot\left(\sum_{\tau_{1}=r}^{\infty}\sum_{\tau_{2}=t}^{\infty}\gamma_{1}^{\tau_{1}+\tau_{2}-r-t}|((\boldsymbol{w}')^{\tau_{1}}\boldsymbol{w}^{\tau_{2}})_{ij}-((\bar{\boldsymbol{w}}')^{\tau_{1}}\bar{\boldsymbol{w}}^{\tau_{2}})_{ij}|^{4}\right)\nonumber \\
 & \leq & \frac{1}{(1-\gamma_{1})^{8}}\max_{\tau_{1}\geq r,\tau_{2}\geq t}\max_{i,j\in\mathcal{N}}\mathbb{E}[|((\boldsymbol{w}')^{\tau_{1}}\boldsymbol{w}^{\tau_{2}})_{ij}-((\bar{\boldsymbol{w}}')^{\tau_{1}}\bar{\boldsymbol{w}}^{\tau_{2}})_{ij}|^{4}].\label{eq:e^q_1}
\end{eqnarray}
The cases of $\boldsymbol{q}=\boldsymbol{s}\boldsymbol{w}^{t}$ and
$\boldsymbol{q}=\boldsymbol{w}'\boldsymbol{s}\boldsymbol{w}^{t}$
yield similar bounds. Therefore, if we can derive a bound on $\mathbb{E}[|((\boldsymbol{w}')^{\tau_{1}}\boldsymbol{w}^{\tau_{2}})_{ij}-((\boldsymbol{\bar{w}}')^{\tau_{1}}\bar{\boldsymbol{w}}^{\tau_{2}})_{ij}|^{4}]$
that is uniform in $\tau_{1},\tau_{2}\geq0$ ($\tau_{1}+\tau_{2}\geq1$)
and $i,j\in\mathcal{N}$, we can get a bound on $\max_{i,j\in\mathcal{N}}\mathbb{E}[|e_{ij}^{q}|^{4}]$
for all forms of $\boldsymbol{q}$.

To simplify the exposition we introduce some notation. For any $\tau\geq1$,
denote $\dot{w}_{\iota_{0},\dots,\iota_{\tau}}\equiv\prod_{s=1}^{\tau}w_{\iota_{s-1}\iota_{s}}$
and $\dot{\bar{w}}_{\iota_{0},\dots,\iota_{\tau}}\equiv\prod_{s=1}^{\tau}\bar{w}_{\iota_{s-1}\iota_{s}}$.
We can write $(\boldsymbol{w}^{\tau})_{ij}=\sum_{(\iota_{0},\dots,\iota_{\tau}):(\iota_{0},\iota_{\tau})=(i,j)}\dot{w}_{\iota_{0},\dots,\iota_{\tau}}$,
where the sum is over all tuples $(\iota_{0},\dots,\iota_{\tau})$
such that $\iota_{0}=i$ and $\iota_{\tau}=j$, and similarly $((\boldsymbol{w}')^{\tau})_{ij}=\sum_{(\iota_{0},\dots,\iota_{\tau}):(\iota_{0},\iota_{\tau})=(i,j)}\dot{w}_{\iota_{\tau},\dots,\iota_{0}}$.
Therefore, for $\tau_{1},\tau_{2}\geq1$, we can write $((\boldsymbol{w}')^{\tau_{1}}\boldsymbol{w}^{\tau_{2}})_{ij}=\sum_{k=1}^{n}((\boldsymbol{w}')^{\tau_{1}})_{ik}(\boldsymbol{w}^{\tau_{2}})_{kj}=\sum_{(\iota_{0},\dots,\iota_{\tau_{1}+\tau_{2}}):(\iota_{0},\iota_{\tau_{1}+\tau_{2}})=(i,j)}\dot{w}_{\iota_{\tau_{1}},\dots,\iota_{0}}\dot{w}_{\iota_{\tau_{1}+1},\dots,\iota_{\tau_{1}+\tau_{2}}}$.
By replacing each $\dot{w}_{\iota_{\tau_{1}},\dots,\iota_{0}}\dot{w}_{\iota_{\tau_{1}+1},\dots,\iota_{\tau_{1}+\tau_{2}}}$
with $\dot{\bar{w}}_{\iota_{\tau_{1}},\dots,\iota_{0}}\dot{\bar{w}}_{\iota_{\tau_{1}+1},\dots,\iota_{\tau_{1}+\tau_{2}}}$,
we get a similar expression for $((\boldsymbol{\bar{w}}')^{\tau_{1}}\bar{\boldsymbol{w}}^{\tau_{2}})_{ij}$.
Define $\dot{w}_{\iota_{0},\dots,\iota_{\tau}}=\dot{\bar{w}}_{\iota_{0},\dots,\iota_{\tau}}\equiv1$
for $\tau=0$. The expression for $((\boldsymbol{w}')^{\tau_{1}}\boldsymbol{w}^{\tau_{2}})_{ij}$
can be extended to $\tau_{1}=0$ or $\tau_{2}=0$ and we can cover
the case $\boldsymbol{q}=\boldsymbol{s}\boldsymbol{w}^{t}$. Using
H\"{o}lder's inequality again, we can derive that for $\tau_{1},\tau_{2}\geq0$
(with $\tau_{1}+\tau_{2}\geq1$)
\begin{eqnarray}
 &  & \mathbb{E}[|((\boldsymbol{w}')^{\tau_{1}}\boldsymbol{w}^{\tau_{2}})_{ij}-((\boldsymbol{\bar{w}}')^{\tau_{1}}\bar{\boldsymbol{w}}^{\tau_{2}})_{ij}|^{4}]\nonumber \\
 & \leq & \mathbb{E}\left(\sum_{\substack{(\iota_{0},\dots,\iota_{\tau_{1}+\tau_{2}}):\\
(\iota_{0},\iota_{\tau_{1}+\tau_{2}})=(i,j)
}
}1\right)^{3}\left(\sum_{\substack{(\iota_{0},\dots,\iota_{\tau_{1}+\tau_{2}}):\\
(\iota_{0},\iota_{\tau_{1}+\tau_{2}})=(i,j)
}
}|\dot{w}_{\iota_{\tau_{1}},\dots,\iota_{0}}\dot{w}_{\iota_{\tau_{1}+1},\dots,\iota_{\tau_{1}+\tau_{2}}}-\dot{\bar{w}}_{\iota_{\tau_{1}},\dots,\iota_{0}}\dot{\bar{w}}_{\iota_{\tau_{1}+1},\dots,\iota_{\tau_{1}+\tau_{2}}}|^{4}\right)\nonumber \\
 & \leq & O(n^{4(\tau_{1}+\tau_{2}-1)})\max_{\substack{(\iota_{0},\dots,\iota_{\tau_{1}+\tau_{2}}):\\
(\iota_{0},\iota_{\tau_{1}+\tau_{2}})=(i,j)
}
}\mathbb{E}[|\dot{w}_{\iota_{\tau_{1}},\dots,\iota_{0}}\dot{w}_{\iota_{\tau_{1}+1},\dots,\iota_{\tau_{1}+\tau_{2}}}-\dot{\bar{w}}_{\iota_{\tau_{1}},\dots,\iota_{0}}\dot{\bar{w}}_{\iota_{\tau_{1}+1},\dots,\iota_{\tau_{1}+\tau_{2}}}|^{4}].\label{eq:e^q_2}
\end{eqnarray}
The last inequality holds because the sum over $(\iota_{0},\dots,\iota_{\tau_{1}+\tau_{2}})$
with $\iota_{0}=i$ and $\iota_{\tau_{1}+\tau_{2}}=j$ has $O(n^{\tau_{1}+\tau_{2}-1})$
terms.

Furthermore, recall that $e_{ij}^{w}=w_{ij}-\bar{w}_{ij}$ and $\max_{i,j\in\mathcal{N}}|\bar{w}_{ij}|\leq\frac{1}{c(n-1)}$.
Modifying the argument in equation (\ref{eq:Ed^4_fe}) and the subsequent
paragraph with $4$ replaced by an even number $\tau>0$, we derive
$\max_{i,j\in\mathcal{N}}\mathbb{E}[|e_{ij}^{w}|^{\tau}]\leq O(n^{-2\tau+\tau/2})$.
Therefore,
\begin{eqnarray}
 &  & \max_{(\iota_{0},\dots,\iota_{\tau_{1}+\tau_{2}}):(\iota_{0},\iota_{\tau_{1}+\tau_{2}})=(i,j)}\mathbb{E}[|\dot{w}_{\iota_{\tau_{1}},\dots,\iota_{0}}\dot{w}_{\iota_{\tau_{1}+1},\dots,\iota_{\tau_{1}+\tau_{2}}}-\dot{\bar{w}}_{\iota_{\tau_{1}},\dots,\iota_{0}}\dot{\bar{w}}_{\iota_{\tau_{1}+1},\dots,\iota_{\tau_{1}+\tau_{2}}}|^{4}]\nonumber \\
 & \leq & O(1)\cdot\max_{i,j\in\mathcal{N}}\mathbb{E}[|e_{ij}^{w}|^{4(\tau_{1}+\tau_{2})}]+O(n^{-4})\cdot\max_{i,j\in\mathcal{N}}\mathbb{E}[|e_{ij}^{w}|^{4(\tau_{1}+\tau_{2}-1)}]\nonumber \\
 &  & +\cdots+O(n^{-4(\tau_{1}+\tau_{2}-1)})\cdot\max_{i,j\in\mathcal{N}}\mathbb{E}[|e_{ij}^{w}|^{4}])\nonumber \\
 & \leq & O(n^{-6(\tau_{1}+\tau_{2})})+O(n^{-6(\tau_{1}+\tau_{2})+2})+\cdots+O(n^{-4(\tau_{1}+\tau_{2})-2})\nonumber \\
 & \leq & O(n^{-4(\tau_{1}+\tau_{2})-2})\label{eq:e^q_3}
\end{eqnarray}
for any $\tau_{1},\tau_{2}\geq0$. It follows from equations (\ref{eq:e^q_1})-(\ref{eq:e^q_3})
that $\max_{i,j\in\mathcal{N}}\mathbb{E}[|e_{ij}^{q}|^{4}]\leq O(n^{4(\tau_{1}+\tau_{2}-1)})\cdot O(n^{-4(\tau_{1}+\tau_{2})-2})=O(n^{-6})$
for $\boldsymbol{q}=(\boldsymbol{s}\boldsymbol{w}^{r})'\boldsymbol{s}\boldsymbol{w}^{t}$.
Modifying equation (\ref{eq:e^q_1}) for $\boldsymbol{q}=\boldsymbol{s}\boldsymbol{w}^{t}$
and $\boldsymbol{w}'\boldsymbol{s}\boldsymbol{w}^{t}$ and noticing
that the case of $\boldsymbol{q}=\boldsymbol{w}'\boldsymbol{w}$ is
trivially covered, we can show that $\max_{i,j\in\mathcal{N}}\mathbb{E}[|e_{ij}^{q}|^{4}]\leq O(n^{-6})$
holds for all forms of $\boldsymbol{q}$.

Assumption \ref{ass:w}(iii) for $\boldsymbol{q}=\boldsymbol{s}\boldsymbol{w}^{t}$,
$t=1,2$. Fix disjoint $\{i,j\}$ and $\{k,l\}$. Following the proof
for $\boldsymbol{q}=\boldsymbol{w}$, we can decompose
\begin{eqnarray}
\mathbb{E}[q_{ij}q_{kl}|\boldsymbol{\text{\ensuremath{\psi}}}] & = & \mathbb{E}[(\bar{q}_{ij}+e_{ij}^{q})(\bar{q}_{kl}+e_{kl}^{q})|\boldsymbol{\text{\ensuremath{\psi}}}]\nonumber \\
 & = & \mathbb{E}[\bar{q}_{ij}|\boldsymbol{\text{\ensuremath{\psi}}}]\mathbb{E}[\bar{q}_{kl}|\boldsymbol{\text{\ensuremath{\psi}}}]+\mathbb{E}[e_{ij}^{q}\bar{q}_{kl}+\bar{q}_{ij}e_{kl}^{q}+e_{ij}^{q}e_{kl}^{q}|\boldsymbol{\text{\ensuremath{\psi}}}]+\Delta_{1,ijkl}\nonumber \\
 & = & \mathbb{E}[\bar{q}_{ij}|\psi_{i},\psi_{j}]\mathbb{E}[\bar{q}_{kl}|\psi_{k},\psi_{l}]+\mathbb{E}[e_{ij}^{q}\bar{q}_{kl}+\bar{q}_{ij}e_{kl}^{q}+e_{ij}^{q}e_{kl}^{q}|\boldsymbol{\text{\ensuremath{\psi}}}]+\Delta_{1,ijkl}+\Delta_{2,ijkl}\nonumber \\
 & = & \mathbb{E}[q_{ij}|\psi_{i},\psi_{j}]\mathbb{E}[q_{kl}|\psi_{k},\psi_{l}]+\Delta_{1,ijkl}+\Delta_{2,ijkl}+\Delta_{3,ijkl},\label{eq:q_ijkl}
\end{eqnarray}
where the three discrepancy terms are defined as
\begin{eqnarray}
\Delta_{1,ijkl} & \equiv & \mathbb{E}[\bar{q}_{ij}\bar{q}_{kl}|\boldsymbol{\text{\ensuremath{\psi}}}]-\mathbb{E}[\bar{q}_{ij}|\boldsymbol{\text{\ensuremath{\psi}}}]\mathbb{E}[\bar{q}_{kl}|\boldsymbol{\text{\ensuremath{\psi}}}],\nonumber \\
\Delta_{2,ijkl} & \equiv & (\mathbb{E}[\bar{q}_{ij}|\boldsymbol{\text{\ensuremath{\psi}}}]-\mathbb{E}[\bar{q}_{ij}|\psi_{i},\psi_{j}])\mathbb{E}[\bar{q}_{kl}|\psi_{k},\psi_{l}]\nonumber \\
 &  & +(\mathbb{E}[\bar{q}_{kl}|\boldsymbol{\text{\ensuremath{\psi}}}]-\mathbb{E}[\bar{q}_{kl}|\psi_{k},\psi_{l}])\mathbb{E}[\bar{q}_{ij}|\psi_{i},\psi_{j}]\nonumber \\
 &  & +(\mathbb{E}[\bar{q}_{ij}|\boldsymbol{\text{\ensuremath{\psi}}}]-\mathbb{E}[\bar{q}_{ij}|\psi_{i},\psi_{j}])(\mathbb{E}[\bar{q}_{kl}|\boldsymbol{\text{\ensuremath{\psi}}}]-\mathbb{E}[\bar{q}_{kl}|\psi_{k},\psi_{l}]),\nonumber \\
\Delta_{3,ijkl} & \equiv & \mathbb{E}[e_{ij}^{q}\bar{q}_{kl}|\boldsymbol{\psi}]-\mathbb{E}[e_{ij}^{q}|\psi_{i},\psi_{j}]\mathbb{E}[\bar{q}_{kl}|\psi_{k},\psi_{l}]\nonumber \\
 &  & +\mathbb{E}[\bar{q}_{ij}e_{kl}^{q}|\boldsymbol{\psi}]-\mathbb{E}[\bar{q}_{ij}|\psi_{i},\psi_{j}]\mathbb{E}[e_{kl}^{q}|\psi_{k},\psi_{l}]\nonumber \\
 &  & +\mathbb{E}[e_{ij}^{q}e_{kl}^{q}|\boldsymbol{\psi}]-\mathbb{E}[e_{ij}^{q}|\psi_{i},\psi_{j}]\mathbb{E}[e_{kl}^{q}|\psi_{k},\psi_{l}].\label{eq:q_ijkl_delta}
\end{eqnarray}
We will derive uniform bounds on the second moments of these three
terms. 

\uline{Term \mbox{$\Delta_{1,ijkl}$}}. Note that $\bar{q}_{ij}=\sum_{\tau=t}^{\infty}\gamma_{1}^{\tau-t}(\bar{\boldsymbol{w}}^{\tau})_{ij}$
and $(\bar{\boldsymbol{w}}^{\tau})_{ij}=\sum_{(\iota_{0},\dots,\iota_{\tau}):(\iota_{0},\iota_{\tau})=(i,j)}\dot{\bar{w}}_{\iota_{0},\dots,\iota_{\tau}}$
for $\tau\geq1$. Conditional on $\boldsymbol{\psi}$, $\dot{\bar{w}}_{\iota_{0},\dots,\iota_{\tau}}$
is a function of $(a_{\iota_{0}},\dots,a_{\iota_{\tau}},\zeta_{\iota_{s-1}\iota_{s}},s\in\{1,\dots,\tau\})$.
Hence, $\dot{\bar{w}}_{\iota_{0},\dots,\iota_{\tau_{1}}}$ and $\dot{\bar{w}}_{\tilde{\iota}_{0},\dots,\tilde{\iota}_{\tau_{2}}}$
are independent conditional on $\boldsymbol{\psi}$ if sets $\{\iota_{0},\dots,\iota_{\tau_{1}}\}$
and $\{\tilde{\iota}_{0},\dots,\tilde{\iota}_{\tau_{2}}\}$ are disjoint.
Therefore, we derive
\begin{eqnarray*}
 &  & \mathbb{E}[\bar{q}_{ij}\bar{q}_{kl}|\boldsymbol{\text{\ensuremath{\psi}}}]-\mathbb{E}[\bar{q}_{ij}|\boldsymbol{\text{\ensuremath{\psi}}}]\mathbb{E}[\bar{q}_{kl}|\boldsymbol{\text{\ensuremath{\psi}}}]\\
 & = & \sum_{\tau_{1}=t}^{\infty}\sum_{\tau_{2}=t}^{\infty}\gamma_{1}^{\tau_{1}+\tau_{2}-2t}(\mathbb{E}[(\bar{\boldsymbol{w}}^{\tau_{1}})_{ij}(\bar{\boldsymbol{w}}^{\tau_{2}})_{kl}|\boldsymbol{\text{\ensuremath{\psi}}}]-\mathbb{E}[(\bar{\boldsymbol{w}}^{\tau_{1}})_{ij}|\boldsymbol{\text{\ensuremath{\psi}}}]\mathbb{E}[(\bar{\boldsymbol{w}}^{\tau_{2}})_{kl}|\boldsymbol{\text{\ensuremath{\psi}}}])\\
 & = & \sum_{\tau_{1}=t}^{\infty}\sum_{\tau_{2}=t}^{\infty}\gamma_{1}^{\tau_{1}+\tau_{2}-2t}\sum_{\begin{subarray}{c}
(\iota_{0},\dots,\iota_{\tau_{1}},\tilde{\iota}_{0},\dots,\tilde{\iota}_{\tau_{2}}):\\
(\iota_{0},\iota_{\tau_{1}},\tilde{\iota}_{0},\tilde{\iota}_{\tau_{2}})=(i,j,k,l)
\end{subarray}}(\mathbb{E}[\dot{\bar{w}}_{\iota_{0},\dots,\iota_{\tau_{1}}}\dot{\bar{w}}_{\tilde{\iota}_{0},\dots,\tilde{\iota}_{\tau_{2}}}|\boldsymbol{\text{\ensuremath{\psi}}}]-\mathbb{E}[\dot{\bar{w}}_{\iota_{0},\dots,\iota_{\tau_{1}}}|\boldsymbol{\text{\ensuremath{\psi}}}]\mathbb{E}[\dot{\bar{w}}_{\tilde{\iota}_{0},\dots,\tilde{\iota}_{\tau_{2}}}|\boldsymbol{\text{\ensuremath{\psi}}}])\\
 & = & \sum_{\tau_{1}=t}^{\infty}\sum_{\tau_{2}=t}^{\infty}\gamma_{1}^{\tau_{1}+\tau_{2}-2t}\sum_{\begin{subarray}{c}
\begin{subarray}{c}
(\iota_{0},\dots,\iota_{\tau_{1}},\tilde{\iota}_{0},\dots,\tilde{\iota}_{\tau_{2}}):\end{subarray}\\
(\iota_{0},\iota_{\tau_{1}},\tilde{\iota}_{0},\tilde{\iota}_{\tau_{2}})=(i,j,k,l),\\
\{\iota_{0},\dots,\iota_{\tau_{1}}\}\cap\{\tilde{\iota}_{0},\dots,\tilde{\iota}_{\tau_{2}}\}\neq\emptyset
\end{subarray}}(\mathbb{E}[\dot{\bar{w}}_{\iota_{0},\dots,\iota_{\tau_{1}}}\dot{\bar{w}}_{\tilde{\iota}_{0},\dots,\tilde{\iota}_{\tau_{2}}}|\boldsymbol{\text{\ensuremath{\psi}}}]-\mathbb{E}[\dot{\bar{w}}_{\iota_{0},\dots,\iota_{\tau_{1}}}|\boldsymbol{\text{\ensuremath{\psi}}}]\mathbb{E}[\dot{\bar{w}}_{\tilde{\iota}_{0},\dots,\tilde{\iota}_{\tau_{2}}}|\boldsymbol{\text{\ensuremath{\psi}}}]).
\end{eqnarray*}
The third sum in the last line consists of $O(n^{(\tau_{1}+\tau_{2})-3})$
terms, and each term can be bounded by $\frac{1}{c^{\tau_{1}+\tau_{2}}(n-1)^{\tau_{1}+\tau_{2}}}$
uniformly in $(i,j,k,l)$ and $\tau_{1},\tau_{2}\geq1$. Note that
$\sum_{\tau_{1}=t}^{\infty}\sum_{\tau_{2}=t}^{\infty}\gamma_{1}^{\tau_{1}+\tau_{2}-2t}\frac{1}{c^{\tau_{1}+\tau_{2}}}=\frac{1}{c^{2t}(c-\gamma_{1})^{2}}$.
Therefore, using H\"{o}lder's inequality $(\sum_{\tau}|a_{\tau}b_{\tau}|)^{2}\leq(\sum_{\tau}|a_{\tau}|^{2})\cdot(\sum_{\tau}|b_{\tau}|^{2})$,
we obtain
\begin{eqnarray*}
 &  & \mathbb{E}[|\mathbb{E}[\bar{q}_{ij}\bar{q}_{kl}|\boldsymbol{\text{\ensuremath{\psi}}}]-\mathbb{E}[\bar{q}_{ij}|\boldsymbol{\text{\ensuremath{\psi}}}]\mathbb{E}[\bar{q}_{kl}|\boldsymbol{\text{\ensuremath{\psi}}}]|^{2}]\\
 & \leq & \left(\sum_{\tau_{1}=t}^{\infty}\sum_{\tau_{2}=t}^{\infty}\gamma_{1}^{\tau_{1}+\tau_{2}-2t}\right)\cdot\left(\sum_{\tau_{1}=t}^{\infty}\sum_{\tau_{2}=t}^{\infty}\gamma_{1}^{\tau_{1}+\tau_{2}-2t}\left(\frac{O(n^{(\tau_{1}+\tau_{2})-3})}{c^{(\tau_{1}+\tau_{2})}(n-1)^{(\tau_{1}+\tau_{2})}}\right)^{2}\right)\\
 & \leq & \frac{1}{c^{4t}(c-\gamma_{1})^{4}}O(n^{(\tau_{1}+\tau_{2})-3})^{2}\cdot O(n^{-(\tau_{1}+\tau_{2})})^{2}=O(n^{-6}).
\end{eqnarray*}
We conclude that $\max_{i,j,k,l\in\mathcal{N}:\{i,j\}\cap\{k,l\}=\emptyset}\mathbb{E}[|\Delta_{1,ijkl}|^{2}]\leq O(n^{-6})$.

\uline{Term \mbox{$\Delta_{2,ijkl}$}}. Observe that $\mathbb{E}[\dot{\bar{w}}_{\iota_{0},\dots,\iota_{\tau}}|\boldsymbol{\text{\ensuremath{\psi}}}]=\mathbb{E}[\dot{\bar{w}}_{\iota_{0},\dots,\iota_{\tau}}|\psi_{\iota_{0}},\dots,\psi_{\iota_{\tau}}]$
because $\dot{\bar{w}}_{\iota_{0},\dots,\iota_{\tau}}$ depends on
$\boldsymbol{\text{\ensuremath{\psi}}}$ only through $\psi_{\iota_{0}},\dots,\psi_{\iota_{\tau}}$.
We can write
\begin{eqnarray*}
 &  & \mathbb{E}[\bar{q}_{ij}|\boldsymbol{\text{\ensuremath{\psi}}}]-\mathbb{E}[\bar{q}_{ij}|\psi_{i},\psi_{j}]\\
 & = & \sum_{\tau=t}^{\infty}\gamma_{1}^{\tau-t}(\mathbb{E}[(\bar{\boldsymbol{w}}^{\tau})_{ij}|\boldsymbol{\text{\ensuremath{\psi}}}]-\mathbb{E}[(\bar{\boldsymbol{w}}^{\tau})_{ij}|\psi_{i},\psi_{j}])\\
 & = & \sum_{\tau=t}^{\infty}\gamma_{1}^{\tau-t}\sum_{(\iota_{0},\dots,\iota_{\tau}):(\iota_{0},\iota_{\tau})=(i,j)}(\mathbb{E}[\dot{\bar{w}}_{\iota_{0},\dots,\iota_{\tau}}|\boldsymbol{\text{\ensuremath{\psi}}}]-\mathbb{E}[\dot{\bar{w}}_{\iota_{0},\dots,\iota_{\tau}}|\psi_{i},\psi_{j}])\\
 & = & \sum_{\tau=t}^{\infty}\gamma_{1}^{\tau-t}\sum_{(\iota_{1},\dots,\iota_{\tau-1})}(\mathbb{E}[\dot{\bar{w}}_{\iota_{0},\dots,\iota_{\tau}}|\psi_{\iota_{1}},\dots,\psi_{\iota_{\tau-1}};\psi_{i},\psi_{j}]-\mathbb{E}[\dot{\bar{w}}_{\iota_{0},\dots,\iota_{\tau}}|\psi_{i},\psi_{j}]).
\end{eqnarray*}
Conditional on $(\psi_{i},\psi_{j})$, each sum over $(\iota_{1},\dots,\iota_{\tau-1})$
forms a (scaled) $V$-statistic of order $\tau-1$ with the (asymmetric)
kernel function
\[
h(\psi_{\iota_{1}},\dots,\psi_{\iota_{\tau-1}};\psi_{i},\psi_{j})\equiv\mathbb{E}[\dot{\bar{w}}_{\iota_{0},\dots,\iota_{\tau}}|\psi_{\iota_{1}},\dots,\psi_{\iota_{\tau-1}};\psi_{i},\psi_{j}]-\mathbb{E}[\dot{\bar{w}}_{\iota_{0},\dots,\iota_{\tau}}|\psi_{i},\psi_{j}]
\]
and mean $0$. Note that $\dot{\bar{w}}_{\iota_{0},\dots,\iota_{\tau}}$
can be bounded by $\frac{1}{c^{\tau}(n-1)^{\tau}}$. Multiplying the
sum by $c^{\tau}(n-1)$ yields a correctly scaled $V$-statistic.
By the results on $V$-statistics \citep{Lee1990}, we can bound the
fourth moment of the sum term by
\[
\mathbb{E}\left[\left(\sum_{(\iota_{1},\dots,\iota_{\tau-1})}h(\psi_{\iota_{1}},\dots,\psi_{\iota_{\tau-1}};\psi_{i},\psi_{j})\right)^{4}\right]\leq\frac{1}{c^{4\tau}(n-1)^{4}}\cdot O(n^{-2})=\frac{1}{c^{4\tau}}O(n^{-6})
\]
uniformly in $(i,j)$. Applying H\"{o}lder's inequality similarly
as in equation (\ref{eq:e^q_1}) yields
\begin{eqnarray*}
 &  & \mathbb{E}[|\mathbb{E}[\bar{q}_{ij}|\boldsymbol{\text{\ensuremath{\psi}}}]-\mathbb{E}[\bar{q}_{ij}|\psi_{i},\psi_{j}]|^{4}]\\
 & = & \mathbb{E}\left[\left(\sum_{\tau=t}^{\infty}\gamma_{1}^{\frac{3}{4}(\tau-t)}c^{-\frac{3}{4}\tau}\cdot\gamma_{1}^{\frac{1}{4}(\tau-t)}c^{\frac{3}{4}\tau}\sum_{(\iota_{1},\dots,\iota_{\tau-1})}h(\psi_{\iota_{1}},\dots,\psi_{\iota_{\tau-1}};\psi_{i},\psi_{j})\right)^{4}\right]\\
 & \leq & \mathbb{E}\left(\sum_{\tau=t}^{\infty}\gamma_{1}^{\tau-t}c^{-\tau}\right)^{3}\cdot\left(\sum_{\tau=t}^{\infty}\gamma_{1}^{\tau-t}c^{3\tau}\left(\sum_{(\iota_{1},\dots,\iota_{\tau-1})}h(\psi_{\iota_{1}},\dots,\psi_{\iota_{\tau-1}};\psi_{i},\psi_{j})\right)^{4}\right)\\
 & \leq & \left(\sum_{\tau=t}^{\infty}\gamma_{1}^{\tau-t}c^{-\tau}\right)^{4}\max_{\tau\geq t}\max_{i,j\in\mathcal{N}}\mathbb{E}\left[c^{4\tau}\left(\sum_{(\iota_{1},\dots,\iota_{\tau-1})}h(\psi_{\iota_{1}},\dots,\psi_{\iota_{\tau-1}};\psi_{i},\psi_{j})\right)^{4}\right]\\
 & \leq & \frac{1}{c^{4(t-1)}(c-\gamma_{1})^{4}}O(n^{-6}).
\end{eqnarray*}
We conclude that $\max_{i,j\in\mathcal{N}}\mathbb{E}[|\mathbb{E}[\bar{q}_{ij}|\boldsymbol{\text{\ensuremath{\psi}}}]-\mathbb{E}[\bar{q}_{ij}|\psi_{i},\psi_{j}]|^{4}]\leq O(n^{-6})$.
Using this result along with $\|\bar{\boldsymbol{q}}\|_{\infty}\leq\frac{C}{n-1}$,
we can bound $\mathbb{E}[(\mathbb{E}[\bar{q}_{ij}|\boldsymbol{\text{\ensuremath{\psi}}}]-\mathbb{E}[\bar{q}_{ij}|\psi_{i},\psi_{j}])^{2}\mathbb{E}[\bar{q}_{kl}|\psi_{k},\psi_{l}]^{2}]\leq O(n^{-5})$
uniformly and $\mathbb{E}[(\mathbb{E}[\bar{q}_{ij}|\boldsymbol{\text{\ensuremath{\psi}}}]-\mathbb{E}[\bar{q}_{ij}|\psi_{i},\psi_{j}])^{2}(\mathbb{E}[\bar{q}_{kl}|\boldsymbol{\text{\ensuremath{\psi}}}]-\mathbb{E}[\bar{q}_{kl}|\psi_{k},\psi_{l}])^{2}]\leq O(n^{-6})$
uniformly. Hence, $\max_{i,j,k,l\in\mathcal{N}:\{i,j\}\cap\{k,l\}=\emptyset}\mathbb{E}[|\Delta_{2,ijkl}|^{2}]\leq O(n^{-5})+O(n^{-6})=O(n^{-5})$. 

\uline{Term \mbox{$\Delta_{3,ijkl}$}}. Using the results $\max_{i,j\in\mathcal{N}}\mathbb{E}[|e_{ij}^{q}|^{4}]\leq O(n^{-6})$
and $\|\bar{\boldsymbol{q}}\|_{\infty}\leq\frac{C}{n-1}$, we can
bound $\mathbb{E}[(\mathbb{E}[e_{ij}^{q}\bar{q}_{kl}|\boldsymbol{\psi}]-\mathbb{E}[e_{ij}^{q}|\psi_{i},\psi_{j}]\mathbb{E}[\bar{q}_{kl}|\psi_{k},\psi_{l}])^{2}]\leq O(n^{-5})$
uniformly and $\mathbb{E}[(\mathbb{E}[e_{ij}^{q}e_{kl}^{q}|\boldsymbol{\psi}]-\mathbb{E}[e_{ij}^{q}|\psi_{i},\psi_{j}]\mathbb{E}[e_{kl}^{q}|\psi_{k},\psi_{l}])^{2}]\leq O(n^{-6})$
uniformly. Hence, $\max_{i,j,k,l\in\mathcal{N}:\{i,j\}\cap\{k,l\}=\emptyset}\mathbb{E}[|\Delta_{3,ijkl}|^{2}]\leq O(n^{-5})+O(n^{-6})=O(n^{-5})$. 

Combining equation (\ref{eq:q_ijkl}) with the rate results on the
three discrepancy terms in (\ref{eq:q_ijkl_delta}), we derive $\max_{i,j,k,l\in\mathcal{N}:\{i,j\}\cap\{k,l\}=\emptyset}\mathbb{E}[|\mathbb{E}[q_{ij}q_{kl}|\boldsymbol{\text{\ensuremath{\psi}}}]-\mathbb{E}[q_{ij}|\psi_{i},\psi_{j}]\mathbb{E}[q_{kl}|\psi_{k},\psi_{l}]|^{2}]\leq O(n^{-5})=o(n^{-4}/K)$. 

In addition, observe that $\mathbb{E}[q_{ij}|\boldsymbol{\psi}]-\mathbb{E}[q_{ij}|\psi_{i},\psi_{j}]=\mathbb{E}[\bar{q}_{ij}|\boldsymbol{\text{\ensuremath{\psi}}}]-\mathbb{E}[\bar{q}_{ij}|\psi_{i},\psi_{j}]+\mathbb{E}[e_{ij}^{q}|\boldsymbol{\psi}]-\mathbb{E}[e_{ij}^{q}|\psi_{i},\psi_{j}]$.
We showed that $\max_{i,j\in\mathcal{N}}\mathbb{E}[|\mathbb{E}[\bar{q}_{ij}|\boldsymbol{\text{\ensuremath{\psi}}}]-\mathbb{E}[\bar{q}_{ij}|\psi_{i},\psi_{j}]|^{4}]\leq O(n^{-6})$.
Moreover, we can bound $\max_{i,j\in\mathcal{N}}\mathbb{E}[(\mathbb{E}[e_{ij}^{q}|\boldsymbol{\psi}]-\mathbb{E}[e_{ij}^{q}|\psi_{i},\psi_{j}])^{4}]\leq C\max_{i,j\in\mathcal{N}}\mathbb{E}[|e_{ij}^{q}|^{4}]\leq O(n^{-6})$
by Jensen's inequality and Cauchy-Schwarz inequality. Therefore, we
obtain $\max_{i,j\in\mathcal{N}}\mathbb{E}[|\mathbb{E}[q_{ij}|\boldsymbol{\text{\ensuremath{\psi}}}]-\mathbb{E}[q_{ij}|\psi_{i},\psi_{j}]|^{4}]\leq O(n^{-6})=o(n^{-4}/K^{2})$.
Assumption \ref{ass:w}(iii) is satisfied for $\boldsymbol{q}=\boldsymbol{s}\boldsymbol{w}^{t}$,
$t=1,2$.

Assumption \ref{ass:w}(iv) for $\boldsymbol{q}=\boldsymbol{s}\boldsymbol{w}^{t}$,
$t=1,2$, is implied by Assumption \ref{ass:w}(iii). For the other
forms of $\boldsymbol{q}$, $\boldsymbol{w}'\boldsymbol{w}$, $\boldsymbol{w}'\bar{\boldsymbol{s}}\bar{\boldsymbol{w}}^{t}$,
and $(\bar{\boldsymbol{s}}\bar{\boldsymbol{w}}^{r})'\bar{\boldsymbol{s}}\bar{\boldsymbol{w}}^{t}$,
$r,t=1,2$, given that these forms share a similar structure with
$\boldsymbol{s}\boldsymbol{w}^{t}$ and that $\boldsymbol{w}'$ has
the same dependence structure as $\boldsymbol{w}$, it is straightforward
to extend the proof for Assumption \ref{ass:w}(iii) for $\boldsymbol{s}\boldsymbol{w}^{t}$
to the other forms of $\boldsymbol{q}$.

Assumption \ref{ass:w}(v) for $\boldsymbol{q}=\boldsymbol{s}\boldsymbol{w}^{t}$,
$t=1,2$. Fix $\{i,j\}$ and $\{i,k\}$ with $j\neq k$. Similarly
as in equation \ref{eq:q_ijkl}, we can decompose
\begin{eqnarray}
\mathbb{E}[q_{ij}q_{ik}|\tilde{\boldsymbol{\text{\ensuremath{\psi}}}}] & = & \mathbb{E}[(\bar{q}_{ij}+e_{ij}^{q})(\bar{q}_{ik}+e_{ik}^{q})|\tilde{\boldsymbol{\text{\ensuremath{\psi}}}}]\nonumber \\
 & = & \mathbb{E}[\bar{q}_{ij}|\tilde{\boldsymbol{\text{\ensuremath{\psi}}}}]\mathbb{E}[\bar{q}_{ik}|\tilde{\boldsymbol{\text{\ensuremath{\psi}}}}]+\mathbb{E}[e_{ij}^{q}\bar{q}_{ik}+\bar{q}_{ij}e_{ik}^{q}+e_{ij}^{q}e_{ik}^{q}|\tilde{\boldsymbol{\text{\ensuremath{\psi}}}}]+\tilde{\Delta}_{1,ijk}\nonumber \\
 & = & \mathbb{E}[\bar{q}_{ij}|\tilde{\psi}_{i},\tilde{\psi}_{j}]\mathbb{E}[\bar{q}_{ik}|\tilde{\psi}_{i},\tilde{\psi}_{k}]+\mathbb{E}[e_{ij}^{q}\bar{q}_{ik}+\bar{q}_{ij}e_{ik}^{q}+e_{ij}^{q}e_{ik}^{q}|\tilde{\boldsymbol{\text{\ensuremath{\psi}}}}]+\tilde{\Delta}_{1,ijk}+\tilde{\Delta}_{2,ijk}\nonumber \\
 & = & \mathbb{E}[q_{ij}|\tilde{\psi}_{i},\tilde{\psi}_{j}]\mathbb{E}[q_{ik}|\tilde{\psi}_{i},\tilde{\psi}_{k}]+\tilde{\Delta}_{1,ijk}+\tilde{\Delta}_{2,ijk}+\tilde{\Delta}_{3,ijk},\label{eq:q_ijk}
\end{eqnarray}
where the three discrepancy terms are given by
\begin{eqnarray}
\tilde{\Delta}_{1,ijk} & \equiv & \mathbb{E}[\bar{q}_{ij}\bar{q}_{ik}|\tilde{\boldsymbol{\text{\ensuremath{\psi}}}}]-\mathbb{E}[\bar{q}_{ij}|\tilde{\boldsymbol{\text{\ensuremath{\psi}}}}]\mathbb{E}[\bar{q}_{ik}|\tilde{\boldsymbol{\text{\ensuremath{\psi}}}}],\nonumber \\
\tilde{\Delta}_{2,ijk} & \equiv & (\mathbb{E}[\bar{q}_{ij}|\tilde{\boldsymbol{\text{\ensuremath{\psi}}}}]-\mathbb{E}[\bar{q}_{ij}|\tilde{\psi}_{i},\tilde{\psi}_{j}])\mathbb{E}[\bar{q}_{ik}|\tilde{\psi}_{i},\tilde{\psi}_{k}]\nonumber \\
 &  & +(\mathbb{E}[\bar{q}_{ik}|\tilde{\boldsymbol{\text{\ensuremath{\psi}}}}]-\mathbb{E}[\bar{q}_{ik}|\tilde{\psi}_{i},\tilde{\psi}_{k}])\mathbb{E}[\bar{q}_{ij}|\tilde{\psi}_{i},\tilde{\psi}_{j}]\nonumber \\
 &  & +(\mathbb{E}[\bar{q}_{ij}|\tilde{\boldsymbol{\text{\ensuremath{\psi}}}}]-\mathbb{E}[\bar{q}_{ij}|\tilde{\psi}_{i},\tilde{\psi}_{j}])(\mathbb{E}[\bar{q}_{kl}|\tilde{\boldsymbol{\text{\ensuremath{\psi}}}}]-\mathbb{E}[\bar{q}_{ik}|\tilde{\psi}_{i},\tilde{\psi}_{k}]),\nonumber \\
\tilde{\Delta}_{3,ijk} & \equiv & \mathbb{E}[e_{ij}^{q}\bar{q}_{ik}|\tilde{\boldsymbol{\text{\ensuremath{\psi}}}}]-\mathbb{E}[e_{ij}^{q}|\tilde{\psi}_{i},\tilde{\psi}_{j}]\mathbb{E}[\bar{q}_{ik}|\tilde{\psi}_{i},\tilde{\psi}_{k}]\nonumber \\
 &  & +\mathbb{E}[\bar{q}_{ij}e_{ik}^{q}|\tilde{\boldsymbol{\text{\ensuremath{\psi}}}}]-\mathbb{E}[\bar{q}_{ij}|\tilde{\psi}_{i},\tilde{\psi}_{j}]\mathbb{E}[e_{ik}^{q}|\tilde{\psi}_{i},\tilde{\psi}_{k}]\nonumber \\
 &  & +\mathbb{E}[e_{ij}^{q}e_{ik}^{q}|\tilde{\boldsymbol{\text{\ensuremath{\psi}}}}]-\mathbb{E}[e_{ij}^{q}|\tilde{\psi}_{i},\tilde{\psi}_{j}]\mathbb{E}[e_{ik}^{q}|\tilde{\psi}_{i},\tilde{\psi}_{k}].\label{eq:q_ijk_delta}
\end{eqnarray}

\uline{Term \mbox{$\tilde{\Delta}_{1,ijk}$}}. Conditional on $\tilde{\boldsymbol{\text{\ensuremath{\psi}}}}$,
$\dot{\bar{w}}_{\iota_{0},\dots,\iota_{\tau}}$ is a function of $(\zeta_{\iota_{s-1}\iota_{s}},s\in\{1,\dots,\tau\})$.
Hence, $\dot{\bar{w}}_{\iota_{0},\dots,\iota_{\tau_{1}}}$ and $\dot{\bar{w}}_{\tilde{\iota}_{0},\dots,\tilde{\iota}_{\tau_{2}}}$
are independent conditional on $\tilde{\boldsymbol{\text{\ensuremath{\psi}}}}$
if sets $\{\iota_{0},\dots,\iota_{\tau_{1}}\}$ and $\{\tilde{\iota}_{0},\dots,\tilde{\iota}_{\tau_{2}}\}$
are disjoint. Therefore, following the argument for $\Delta_{1,ijkl}$
with $\{k,l\}$ replaced by $\{i,k\}$ and $\boldsymbol{\text{\ensuremath{\psi}}}$
by $\tilde{\boldsymbol{\text{\ensuremath{\psi}}}}$, we can derive
$\max_{i,j,k\in\mathcal{N}:j\neq k}\mathbb{E}[|\tilde{\Delta}_{1,ijk}|^{2}]\leq O(n^{-6})$.

\uline{Term \mbox{$\tilde{\Delta}_{2,ijk}$}}. Note that the argument
for term $\Delta_{2,ijkl}$ remains valid if we replace $\boldsymbol{\psi}$
by $\tilde{\boldsymbol{\text{\ensuremath{\psi}}}}$ and $\psi_{\iota}$
by $\tilde{\psi}_{\iota}$, $\iota\in\mathcal{N}$. It follows that
$\max_{i,j\in\mathcal{N}}\mathbb{E}[|\mathbb{E}[\bar{q}_{ij}|\boldsymbol{\tilde{\text{\ensuremath{\psi}}}}]-\mathbb{E}[\bar{q}_{ij}|\tilde{\psi}_{i},\tilde{\psi}_{j}]|^{4}]\leq O(n^{-6})$
and thus $\max_{i,j,k\in\mathcal{N}:j\neq k}\mathbb{E}[|\tilde{\Delta}_{2,ijk}|^{2}]\leq O(n^{-5})$.

\uline{Term \mbox{$\tilde{\Delta}_{3,ijk}$}}. The argument for term
$\Delta_{3,ijkl}$ holds for $\tilde{\Delta}_{3,ijk}$ as well. Hence,
we obtain $\max_{i,j,k\in\mathcal{N}:j\neq k}\mathbb{E}[|\tilde{\Delta}_{3,ijk}|^{2}]\leq O(n^{-5})$. 

Combining equation (\ref{eq:q_ijk}) with the rate results on the
three discrepancy terms in (\ref{eq:q_ijk_delta}), we have $\max_{i,j,k\in\mathcal{N}:j\neq k}\mathbb{E}[|\mathbb{E}[q_{ij}q_{ik}|\tilde{\boldsymbol{\psi}}]-\mathbb{E}[q_{ij}|\tilde{\psi}_{i},\tilde{\psi}_{j}]\mathbb{E}[q_{ik}|\tilde{\psi}_{i},\tilde{\psi}_{k}]|^{2}]\leq O(n^{-5})=o(n^{-4})$.
Assumption \ref{ass:w}(v) is satisfied for $\boldsymbol{q}=\boldsymbol{s}\boldsymbol{w}^{t}$,
$t=1,2$.
\end{proof}

\end{document}